\documentclass[
    pdftex,
    a4paper,
    11pt,
    twoside,
]{scrbook}
\usepackage[latin1]{inputenc}
\usepackage{amsmath,amsthm}
\usepackage{amsfonts}
\usepackage[pdftex]{graphics}
\usepackage{multirow}
\usepackage{dsfont}
\usepackage{tikz}
\usepackage[pdftex, pdftitle={Construction of OPE coefficients via consistency conditions}, pdfauthor={Jan Holland}, colorlinks=true]{hyperref}
\usepackage[paper=a4paper,inner=30mm,outer=20mm,top=25mm,bottom=25mm]{geometry}
\usepackage{empheq}
\usepackage{stmaryrd}
\usepackage[nocut]{thmbox}
\usepackage{simplewick}
\usepackage{multind}
\usepackage{fancyhdr}
\usepackage{longtable}
\usepackage{feynmf}
\usepackage[bottom]{footmisc}
\usepackage[toc,page]{appendix}

\usepackage{tocloft}

\makeatletter
\newcommand*\updatechaptername{%
   \addtocontents{toc}{\protect\renewcommand*\protect\cftchappresnum{\@chapapp\ }}
}
\makeatother

\makeatletter
\RequirePackage{ifthen}
\RequirePackage{multicol}
\newif\if@restonecol
\def\printindex#1#2{%
        \ifthenelse{\equal{symbols}{#1}}
        {
        \@restonecoltrue\if@twocolumn\@restonecolfalse\fi
  \columnseprule \z@ \columnsep 35pt
  \chapter*{#2}\addcontentsline{toc}{chapter}{#2}
  \begin{multicols}{2}
  \@input{#1.ind}
  \end{multicols}
  }
  {\@restonecoltrue\if@twocolumn\@restonecolfalse\fi
  \columnseprule \z@ \columnsep 35pt
  \chapter*{#2}
  \addcontentsline{toc}{chapter}{#2}
  \@input{#1.ind}}}
\makeatother

\fancypagestyle{plain}{
\fancyhf{}

}

\makeatletter
\def\cleardoublepage{\clearpage\if@twoside \ifodd\c@page\else
	\hbox{}
	\vspace*{\fill}
	\thispagestyle{empty}
	\newpage
	\if@twocolumn\hbox{}\newpage\fi\fi\fi}
\makeatother

\usepackage{fancybox}

\makeatletter
\def\tagform@#1{\maketag@@@{\cornersize{0.8}\ovalbox{
\ignorespaces\sffamily{#1}\unskip\@@italiccorr}}}
\makeatother

\usepackage[avantgarde]{quotchap}

\usepackage[calcwidth]{titlesec}
\titleformat{\section}[hang]{\sffamily\bfseries}
 {\Large\thesection}{12pt}{\Large}[{\titlerule[0.5pt]}]

\numberwithin{equation}{section}

\newtheorem{axioms}{Axiom}
\newtheorem{definitions}{Definition}[chapter]
\newtheorem{thms}{Theorem}
\newtheorem{propos}{Proposition}
\newtheorem{lemm}{Lemma}
\newtheorem{result}{Result}[section]

\definecolor{mygray}{gray}{.9}

\newenvironment{axiom}[1][]{
  \begin{axioms}[#1]\label{ax:#1}\index{thms}{Axiom!\theaxioms\, (#1)} %
}{\end{axioms}}

\newenvironment{definition}[1][]{
  \begin{definitions}[#1]\label{def:#1}\index{thms}{Definition!\thedefinitions\, (#1)}%
}{\end{definitions}}

\newenvironment{thm}[1][]{
  \begin{thms}[#1]\label{thm:#1}\index{thms}{Theorem!\thethms\, (#1)}%
}{\end{thms}}

\newenvironment{results}[1][]{
  \begin{result}[#1]\index{thms}{Result!\theresult\, (#1)}%
}{\end{result}}

\newenvironment{proposition}[1][]{
  \begin{propos}\index{thms}{Proposition!\thepropos}%
}{\end{propos}}

\newenvironment{lemma}[1][]{
  \begin{lemm}\index{thms}{Lemma!\thelemm}%
}{\end{lemm}}

\newcommand{\di}{\mathrm{d}}
\newcommand{\tdd}[2]{\frac{\di #1}{\di #2}}
\newcommand{\pdd}[2]{\frac{\partial #1}{\partial #2}}
\newcommand{\del}{\partial}
\newcommand{\C}[2]{C^{#1}_{#2}}
\newcommand{\bra}[1]{\langle #1|}
\newcommand{\ket}[1]{|#1\rangle}
\newcommand{\braket}[2]{\langle #1|#2\rangle}

\newcommand{\Co}[3]{(C_{#1})^{#2}_{#3}}
\newcommand{\Lo}[3]{\mathcal{Y}_{#1}(#2,#3)}

\newcommand{\pFq}[5]{\,_{#1}F_{#2}\left[\atop{#3}{#4};#5\right]}

\renewcommand{\atop}[2]{\genfrac{}{}{0pt}{}{#1}{#2}}
\renewcommand{\eqref}[1]{ \ref{#1}}
\newcommand{\D}[3]{\Delta_{#1}[\mathfrak{#2},r]_{#3}}
\newcommand{\Do}[4]{\Delta_{#1}^{#4}[\mathfrak{#2}]_{#3}}
\newcommand{\Ro}[5]{\Lambda_{#1}[\varphi^{#5},\mathfrak{#2},r]_{#3}}
\newcommand{\Rop}[5]{\Lambda_{#1}^{#4}[\varphi^{#5},\mathfrak{#2}]_{#3}}
\newcommand{\betrag}{|}
\newcommand{\End}{\operatorname{End}}
\renewcommand{\dotfill}{\leaders\hbox to 5pt{\hss.\hss}\hfill}

\usepackage{titlepage_diplom}

\begin{document}

\author{Jan Holland}
\title{Construction of operator product expansion coefficients via consistency conditions}
\date{Februar 2009}
\maketitle

\pagenumbering{roman}

\chapter*{Abstract}
\thispagestyle{headings}
\addcontentsline{toc}{chapter}{Abstract}
In this diploma thesis an iterative scheme for the construction of operator product expansion (OPE) coefficients is applied to determine low order coefficients in perturbation theory for a specific toy model. We use the approach to quantum field theory proposed by S. Hollands \cite{Hollands2008}, which is centered around the OPE and a number of axioms on the corresponding OPE coefficients. This framework is reviewed in the first part of the thesis.

In the second part we apply an algorithm, also proposed in \cite{Hollands2008}, for the perturbative construction of OPE coefficients to a toy model: Euclidean $\varphi^6$-theory in $3$-dimensions. Using a recently found formulation in terms of vertex operators and a diagrammatic notation in terms of trees (see also \cite{Hollands2008} and \cite{Hollandsa}), coefficients up to second order are constructed, some general features of coefficients at arbitrary order are presented and an exemplary comparison to the corresponding customary method of computation is given.


\chapter*{Acknowledgments}
\addcontentsline{toc}{chapter}{Acknowledgments}
\thispagestyle{headings}
I would like to thank Prof. S. Hollands\footnote{University of Cardiff\label{foot:UCardiff}} for the opportunity to work on the first steps in a very young approach to quantum field theory. I am thankful for his support, advice and guidance through this instructive time. I further appreciate his hospitality during my stay at Cardiff. I also want to thank Prof. Rehren for writing the additional report and for help with the official paperwork. Furthermore, I am grateful to H. Olbermann\footref{foot:UCardiff} for helpful discussions and suggestions.
My dear colleagues D.Böning, S. Grübel, P. Löptien, J. Müller and M. Theves I want to thank for sharing my way through half a decade of physics studies and also for lots of fun in our scarce free time. For a pleasant working atmosphere in office A2.113b I want to thank T. Dabrowski and A. Mann. I am further grateful to all members of the Göttingen QFT group for interesting discussions and seminars.

Last but not least I would like to thank my parents for moral and financial support, for steady encouragement and for giving me the freedom and opportunity to follow my interests.

Financial support by the Alexander von Humboldt-Stiftung is also gratefully acknowledged.

\tableofcontents

\cleardoublepage

\pagenumbering{arabic}
\updatechaptername
\chapter{Introduction}

\section{Motivation}

Various formulations of quantum field theory (QFT) have been proposed and established in the last century. The most popular ones can be split into two conceptually different categories: The path-integral and the operator approach.

The former uses the moments of some measure on the space of classical field configurations in order to construct correlation functions. As this measure is (formally) given in terms of the classical action, this formulation has the advantage of being closely related to classical field theory. In the operator approach, on the other hand, quantum fields are viewed as linear operators represented on some Hilbert space of states. Consequently, no corresponding classical theory, i.e. no Lagrangian formalism, is needed in this case. In this formalism special emphasis is put on the algebraic relations between the quantum fields. In fact, these relations may be viewed as determining the whole theory, as originally proposed by Haag and Kastler \cite{Haag200801} in the framework of algebraic quantum field theory (AQFT). Algebraic approaches have also been useful in conformal quantum field theories (cQFT), see e.g. \cite{Borcherds1986,Kac1997}, and in view of the lack of a preferred Hilbert space representation turned out to be essential in the construction of quantum field theories on curved spacetimes \cite{HollandsWald2002,HollandsWald2001,BrunettiFredenhagen2000,BrunettiFredenhagenVerch2003,Wald199411}.

In \cite{Hollands2008} a new approach to quantum field theory was proposed, where the algebraic relations between the fields (at short distances) are encoded in the Wilson operator product expansion (OPE)\cite{Wilson:1969zs}, which is elevated to the status of a defining element of the theory instead of an identity derived from it (see section \ref{subsec:historical background}). Furthermore the OPE coefficients have to obey certain constraints, in particular a \emph{factorization} relation that was observed in the construction of the OPE on curved spacetimes \cite{Hollands2007}. An axiomatic formulation of this framework is given in \cite{Hollands2008}. The key observation is that consistency conditions arising from the mentioned factorization property can be used, in combination with field equations, as a constructive tool. In addition, graphical rules for the computation of OPE coefficients within this approach have been obtained \cite{Hollandsa}. The resulting algorithm for the construction of the OPE is different from standard ones relying on divergence properties of Feynman integrals \cite{Collins1984}, but one expects the results to be equivalent (see section \ref{subsec:comparison to alternative}). A remarkable feature of the new approach is the fact that it is inherently finite (see chapter \ref{subsec:perturbation via field equation}), i.e. no renormalization procedure is needed.

This novel viewpoint has several advantageous features: 

As mentioned above, it is viable also on curved spacetimes, where it might even be necessary to elevate the existence of an OPE to axiomatic status \cite{HollandsWald}. In \cite{Hollands2008} the framework was also generalized to gauge theories, which play a central role in the description of particle physics. Another interesting feature concerns the regularity of the OPE coefficients as opposed to quantum states. In the case of simple examples it is easy to show that OPE coefficients may depend analytically on certain parameters of the theory (like e.g. the mass of a particle) where the quantum states show non-analytic behavior. As a result of such considerations, it was recently conjectured in \cite{Hollands} that even the perturbation series for the OPE coefficients may converge, i.e. it might be possible to perturbatively construct interacting quantum field theories within this approach.

This thesis will be concerned with a low order perturbative construction of this kind for the case of a simple toy model theory. It thus constitutes the first specific application of the very young framework outlined above and gives first impressions of the calculational effort involved in the iterative construction of perturbation theory. The above mentioned cancellation of divergences and also the formation of certain patterns in the mathematical structure of the coefficients can be observed explicitly. This also gives insights into the expected structure of higher order coefficients.

\section{Historical background}\label{subsec:historical background}

The singular behavior of products of quantum fields at coinciding spacetime points has been a major obstacle in the construction and understanding of quantum field theory since its discovery. Thus, the analysis of these short distance divergences was, and still is, of great interest.

The idea of a short distance expansion for products of quantum field operators was first considered by Wilson in 1964 when, according to himself inspired by axiomatic QFT, he translated some work he ``\emph{had done on Feynman diagrams with some very large momenta (...) into position space}'' \cite{Wilson1983}. This work, however, never got published and it took five more years until Wilson revisited his theory of short distance expansions adding new ideas by Kastrup and Mack concerning scale invariance in QFT.
Ironically, although in his 1969 paper Wilson introduces the OPE as an alternative framework to Lagrangian models, in the following years it became an important tool in the understanding of these theories.

In 1970 Zimmermann proved that an OPE holds in perturbative quantum field theory, thus giving its first validation also in usual Lagrangian theories \cite{Zimmerman1970}. He used this new tool in order to define \emph{normal products} of interacting quantum fields as a generalization of the normal ordered products in the free theory. Now a sensible notion of composite fields, i.e. for example higher powers of fields, could be defined in terms of normal products \cite{Zimmermann1973}.

In the following years the OPE was established in the most important branches of quantum field theory: It became a standard tool in the analysis of quantum chromodynamics (QCD), played a crucial role in the development of conformal quantum field theories, has been proven within various axiomatic settings \cite{Bostelmann2005, Fredenhagen:1994ib} and has been shown to hold order by order in perturbation theory on curved spacetime \cite{Hollands2007}. The OPE has also been used in order to prove a curved spacetime version of the spin-statistics theorem and the PCT theorem \cite{Hollands2004} and played a crucial role in the formulation of an axiomatic quantum field theory on curved spacetime, where the OPE was elevated to a fundamental status and replaces the requirement for the existence of a unique Poincare invariant state \cite{HollandsWald}. This shift of emphasis onto the OPE as defining property of the theory is in the spirit of the new approach considered in this thesis. Here the OPE is no longer viewed as simply a calculational tool that has been derived from the theory, but as a central feature around which the theory is built.

\section{Organization}

This thesis is organized as follows:

In the next chapter we introduce the new framework of quantum field theory in terms of consistency conditions as proposed in \cite{Hollands2008}. After a short motivation of the ideas leading to this approach, an axiomatic setting for quantum field theory is presented, followed by an analysis of perturbation theory in this framework. Finally, a recently discovered convenient formulation of the theory in terms of vertex operators - the \emph{fundamental left representation} - is studied.

Chapter \ref{sec:the model} presents the results of this thesis, namely the low order perturbative construction of OPE coefficients for a specific Lagrangian model theory. Our considerations naturally start at the non-interacting theory (i.e. zeroth perturbation order), where the model and some notation are introduced. We then explain a general method of perturbations via non-linear field equations, obtaining an iterative scheme for the construction of OPE coefficients. The application of this algorithm to a specific model theory gives the main results of this thesis, which are presented in section \ref{subsec:low orders} and \ref{subsec:higher order}. In the end we briefly compare our framework to ordinary methods for the calculation of OPE coefficients in terms of an example.

The thesis is closed by chapter \ref{sec:conclusions}, where our results are reviewed and interpreted. Conclusions as well as an outlook on possible future developments are given.

We refer to the appendix for the introduction to the main mathematical objects appearing in the calculations.

\chapter{QFT in terms of consistency conditions}\label{sec:OFT in terms of consistency}

In this chapter a new formulation of quantum field theory, first proposed in \cite{Hollands2008}, is described. The central object of this framework is the OPE subject to certain constraints.

\section{Motivation}\label{subsec: framework motivation}

The operator product expansion states that the product of two operators may be written as

\begin{equation}
 \big\langle\phi_a(x_1)\phi_b(x_2)\big\rangle_{\omega} \approx \sum_c \C{c}{ab}(x_1,x_2;y)\big\langle\phi_c(y)\big\rangle_{\omega}
\label{eq:framework motivation ope lang}
\end{equation}
in terms of local quantum fields $\phi_c$ and $\mathbb{C}$-number distributions $\C{c}{ab}$. Here $a,b,c$ label the composite fields of the underlying theory , $\langle\, \cdot\, \rangle_{\omega}$ is the expectation value in the state $\omega$ and ``$\approx$'' means that this identity holds as an asymptotic relation in the limit $x_1, x_2\to y$ in a suitably strong sense\footnote{See the discussion below equation \eqref{eq:framework axioms states} for more detail}. In the following we will shorten the notation by formally rewriting equation \eqref{eq:framework motivation ope lang} as

\begin{equation}
 \phi_a(x_1)\phi_b(x_2)=\sum_c \C{c}{ab}(x_1,x_2) \phi_c(x_2) \: ,
\label{eq:framework motivation ope kurz}
\end{equation}
where we have implicitly made the choice  $y=x_2$. We furthermore assume the fields to live on a real Euclidean spacetime, which can always be achieved by analytic continuation provided the spectrum condition holds in the quantum field theory. In order to motivate the condition that lies at the heart of the new framework we should study the OPE of a product of three operators.

\begin{equation}
 \phi_a(x_1)\phi_b(x_2)\phi_c(x_3)=\sum_d \C{d}{abc}(x_1,x_2,x_3) \phi_d(x_3)
\label{eq:framework motivation ope drei}
\end{equation}

Now consider a situation as depicted in figure \ref{fig:consistency regions}.
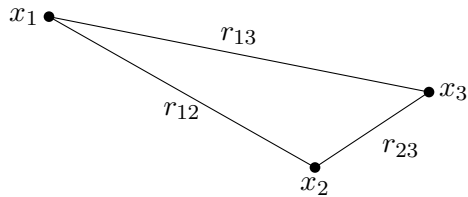
\begin{figure}[h!]
\begin{center}
\begin{tikzpicture}
 \fill (0,0) circle (2pt) node[left] {$x_1$};
\fill (5,-1) circle (2pt) node[right] {$x_3$};
\fill (3.5,-2) circle (2pt) node[below] {$x_2$};
\draw (0,0) -- (5,-1) node[above, midway] {$r_{13}$};
\draw (0,0) -- (3.5,-2) node[below, midway] {$r_{12}$};
\draw (3.5,-2) -- (5,-1) node[below right, midway] {$r_{23}$};
\end{tikzpicture}
\end{center}
\caption{A configuration of the spacetime points $x_1,x_2,x_3\in \mathbb{R}^d$}
\label{fig:consistency regions}
\end{figure}
 Defining the Euclidean distance between two points $x_i$ and $x_j$ as

\begin{equation}
 r_{ij}:=|x_i-x_j|=\sqrt{(x_i-x_j)^2}
\label{eq:framework motivation distance}
\end{equation}\index{symbols}{rij@$r_{ij}$}

figure \ref{fig:consistency regions} tells us that $r_{23}<r_{13}$. In this case we expect it to be possible to perform the OPE successively, i.e. we first expand the product of $\phi_b(x_2)\phi_c(x_3)$ in equation \ref{eq:framework motivation ope drei} around $x_3$ , multiply the result by $\phi_a(x_1)$ and perform yet another OPE around $x_3$. Writing this idea as an equation, we obtain 

\begin{equation}
 \C{d}{abc}(x_1,x_2,x_3)=\sum_e \C{e}{bc}(x_2,x_3)\C{d}{ae}(x_1,x_3) \: .
\label{eq:framework motivation split1}
\end{equation}

Similarly, as figure \ref{fig:consistency regions} also implies $r_{12}<r_{23}$, we expect that we can start by expanding the product $\phi_a(x_1)\phi_b(x_2)$ around $x_2$, multiply the result by $\phi_c(x_3)$ and finally expand around $x_3$ again. In other words:

\begin{equation}
 \C{d}{abc}(x_1,x_2,x_3)=\sum_e \C{e}{ab}(x_1,x_2)\C{d}{ec}(x_2,x_3)
\label{eq:framework motivation split2}
\end{equation}

So for constellations as in figure \ref{fig:consistency regions}, i.e. on the open domain $r_{12}<r_{23}<r_{13}$, we obtain the consistency relation that both expansion \ref{eq:framework motivation split1} and \ref{eq:framework motivation split2} must be valid and should coincide. Thus we require

\begin{equation}
 \sum_e \C{e}{bc}(x_2,x_3)\C{d}{ae}(x_1,x_3)=\sum_e \C{e}{ab}(x_1,x_2)\C{d}{ec}(x_2,x_3)
\label{eq:framework motivation consistency}
\end{equation}

when $r_{12}<r_{23}<r_{13}$. We will adopt the labeling of this constraint as "consistency-" or "associativity" condition from \cite{Hollands2008}. The basic idea of the new framework presented in this chapter is that these conditions on the 2-point OPE coefficients are stringent enough to incorporate the full information about the structure of the quantum field theory or, stated conversely, that finding a solution to these conditions effectively means that one has constructed a quantum field theory.

But if the full information on the quantum field theory is to be encoded in the constraints on the 2-point OPE coefficients, then no further constraints should appear from \emph{higher order associativity conditions}, i.e. from conditions on products of more than three fields. If, for example, we consider the OPE of four fields $\phi_a(x_1)\phi_b(x_2)\phi_c(x_3)\phi_d(x_4)$ and successively expand this product in a similar manner as above, we will obtain new relations for the 2-point OPE coefficients similar to eq. \ref{eq:framework motivation consistency}. The question is now whether these constraints are genuinely new or can be deduced from eq. \ref{eq:framework motivation consistency}. As it turns out, this problem is analogous to the analysis of the associativity condition in ordinary algebra and as in this case we will show that no further conditions arise (see chapter \ref{subsec:coherence theorem}). These considerations will also yield a unique expression of the \emph{higher order} coefficients such as $\C{e}{abcd}(x_1,x_2,x_3,x_4)$ in terms of the 2-point OPE coefficients. This result is called the \emph{coherence theorem}, because it states that the entire set of consistency conditions is coherently encoded in the associativity condition \ref{eq:framework motivation consistency}.

Now that we have identified the 2-point OPE coefficients as fundamental entities of our approach, it would be of interest to formulate perturbation theory in terms of these coefficients. With this aim in mind, let us assume the following setting: We are given a 1-parameter family of 2-point coefficients with parameter $\lambda$. These coefficients shall satisfy the associativity condition \ref{eq:framework motivation consistency} and we want to perturb around the quantum field theory described by the coefficients with $\lambda=0$. In order to avoid messy equations we introduce an index free notation getting rid of the indices $a,b,c\ldots$ above. We view the 2-point OPE coefficients collectively as a linear map $\mathcal{C}(x_1,x_2):V\otimes V\to V$, where $V$ is the \emph{space of fields}, whose basis components are given by $\C{c}{ab}(x_1,x_2)$. Then the Taylor expansion in the parameter $\lambda$ around $\lambda=0$ is

\begin{equation}
 \mathcal{C}(x_1,x_2;\lambda)=\sum_{i=0}^{\infty}\mathcal{C}_i(x_1,x_2) \lambda^i\, .
\label{eq:framework motivation perturbation}
\end{equation}

If we expand the associativity condition, eq. \ref{eq:framework motivation consistency}, in this way and assume that the condition holds at zeroth order, then we obtain the following constraint on the first order perturbation of the 2-point OPE coefficients

\begin{equation}
\begin{split}
 &\mathcal{C}_0(x_2,x_3)\Big(\mathcal{C}_1(x_1,x_2)\otimes id\Big)-\mathcal{C}_0(x_1,x_3)\Big(id \otimes \mathcal{C}_1(x_2,x_3)\Big)+\\
 &\mathcal{C}_1(x_2,x_3)\Big(\mathcal{C}_0(x_1,x_2)\otimes id\Big)-\mathcal{C}_1(x_1,x_3)\Big(id \otimes \mathcal{C}_0(x_2,x_3)\Big)=0
\end{split}
\label{eq:framework motivation perturbation consistency}
\end{equation}
which is linear in the first order coefficients and holds on the domain $r_{12}<r_{23}<r_{13}$. It was shown in \cite{Hollands2008} that this condition is of a cohomological nature and that one can identify the set of all possible first order perturbations satisfying this condition (modulo trivial field redefinitions) with the elements of a certain cohomology ring, which bears close resemblance to Hochschild cohomology \cite{MacLane199408,happel1989hcf}. This notion can be generalized to higher order perturbations, i.e. at each order the associativity condition is a potential obstacle for the continuation of the perturbation series. This obstruction is then again an element of our cohomology ring.

The above definition of perturbation theory is very general in the sense that the physical meaning of the parameter $\lambda$ is completely open. $\lambda$ might for instance measure the strength of some coupling in a Lagrangian theory, like e.g. the self interaction in the theory described by the classical Lagrangian $L=(\del\phi)^2+\lambda \phi^4$. The perturbation would in this case be around the free theory where the OPE coefficients are known. It is also possible to perturb more general, not necessarily Lagrangian, theories, e.g. more general conformal field theories. One could also take $SU(N)$ Yang-Mills theory as yet another example. Here one could chose $\lambda=1/N$, where $N$ is the number of \emph{colors} of the theory, and perturb around the large-$N$-limit of the theory. This general theory of perturbations is described in more detail in chapter \ref{subsec:general perturbations}.

\section{Axiomatic framework}\label{subsec:axioms}

The aim of the present section is to give a precise formulation of the ideas informally presented above. In this approach quantum field theory is defined by an axiomatic setup that was first proposed in \cite{Hollands2008} (in \cite{HollandsWald} a basically similar framework was introduced for quantum field theory on curved spacetime). 

First we define the playground of our quantum field theory to be an infinite dimensional vector space $V$, whose elements can be thought of as the components of the various composite scalar, spinor and tensor fields. For example, in a theory containing only one scalar field $\varphi$, the elements of $V$ would be in one-to-one correspondence with the monomials in $\varphi$ and its derivatives. One would naturally assume $V$ to be graded in various ways, as it should be possible to classify the different quantum fields in the theory by characteristic properties, such as spin, dimension, Bose/Fermi character, etc. As we are considering Euclidean quantum field theories, we expect $V$ to carry a representation of the $D$-dimensional rotation group $SO(D)$, or of its covering group Spin$(D)$ respectively if spinor fields are present. This representation can be decomposed into unitary, finite-dimensional irreducible representations $V_S$ characterized by the eigenvalues $S=(\lambda_1,\ldots,\lambda_r)$ of the $r$ Casimir operators associated with $SO(D)$. Thus, we introduce a grading by these irreducible representations (irrep's):\index{symbols}{V@$V$}

\begin{equation}
 V=\bigoplus_{\Delta\in\mathbb{R}_+}\bigoplus_{S\in \text{irrep}} \mathbb{C}^{N(\Delta, S)}\otimes V_S
\label{eq:framework axioms V}
\end{equation}

An additional grading, which will later be related to the \emph{dimension} of the quantum fields, is provided by the numbers $\Delta\in\mathbb{R}_+$. The numbers $N(\Delta, S)\in\mathbb{N}$, here supposed to be finite, express the multiplicity of the fields with a given dimension $\Delta$ and spin $S$. Here we should remark that the infinite sums in this decomposition are understood without any closure taken, meaning that the elements of $V$ are in one-to-one correspondence with sequences of the form $(\ket{v_1},\ket{v_2},\ldots,\ket{v_n},0,0,\ldots)$, where only finitely many zeros appear and $\ket{v_i}$ is a vector in the $i$-th summand of the decomposition. As further structure on $V$ we demand the existence of an anti-linear, involutive operation $\star : V\to V$ which should be thought of as hermitian conjugation of the quantum fields. Additionally we would like to have a linear grading map $\gamma : V\to V$ satisfying $\gamma^2=id$ which is to be thought of as a grading with respect to \emph{bosonic} (eigenvalue $+1$) and \emph{fermionic} (eigenvalue $-1$) vectors. Finally, we demand the existence of $D$ derivations on $V$, i.e. linear maps $\del_{\mu}:V\to V$ with $\mu\in\{1,\ldots,D\}$ satisfying the Leibnitz rule and $\del_{\mu}\circ\gamma=\gamma\circ\del_{\mu}$. These derivations increase the dimension $\Delta$ of the vectors in $V$ by 1.

The linear space defined in this way is nothing more than a list of objects that we think of as labeling the composite fields of the theory, but so far the dynamics and the quantum nature of the theory, i.e. the information that is of most interest, have not been addressed. This information is now encoded in the OPE coefficients associated with the quantum fields. This is a hierarchy

\begin{equation}
 \mathcal{C}=\Big(\mathcal{C}(x_1,x_2),\mathcal{C}(x_1,x_2,x_3),\mathcal{C}(x_1,x_2,x_3,x_4),\ldots\Big)\: ,
\label{eq:framework axioms hierarchy}
\end{equation}
where the $\mathcal{C}(x_1,\ldots,x_n)$\index{symbols}{C@$\mathcal{C}(x_1,\ldots,x_n)$} are analytic functions on the \emph{configuration space}

\begin{equation}
 M_n:=\{(x_1,\ldots,x_n)\in(\mathbb{R}^D)^n\: |\: x_i\neq x_j,\: \forall\: 1\leq i<j\leq n \}
\label{eq:framework axioms configuration space}
\end{equation}
taking values in the linear maps

\begin{equation}
 \mathcal{C}(x_1,\ldots,x_n): \underbrace{V\otimes \cdots \otimes V}_{n-\text{factors}}\to V \: .
\label{eq:framework axioms map C}
\end{equation}
For the one-point coefficient we set $\mathcal{C}(x_1)= id: V\to V$. By taking the components of these maps in a basis of $V$ the OPE coefficients from the previous section can be retrieved. So, if $\{\ket{v_a}\}$ denotes a basis of $V$ in adapted to the grading with the corresponding basis $\{\bra{v^a}\}$ of the dual space

\begin{equation}
 V^* =\bigoplus_{\Delta\in\mathbb{R}_+}\bigoplus_{S\in \text{irrep}} \mathbb{C}^{N(\Delta, S)}\otimes V_{\overline{S}}\: ,
\label{eq:framework axioms Vdual}
\end{equation}
with $V_{\overline{S}}$ denoting the conjugate representation, $\braket{v^b}{v_a}=\delta^b_a$, then we have\index{symbols}{Cab@$\C{b}{a_1\ldots a_n}(x_1,\ldots,x_n)$}

\begin{equation}
 \C{b}{a_1\ldots a_n}(x_1,\ldots,x_n)=\braket{v^b}{\mathcal{C}(x_1,\ldots,x_n)|v_{a_1}\otimes\cdots\otimes v_{a_n}}\: ,
 \label{eq:framework axioms OPE basis}
\end{equation}
using the customary bra-ket notation $\ket{v_{a_1}\otimes\cdots\otimes v_{a_n}}:=\ket{v_{a_1}}\otimes\cdots\otimes\ket{v_{a_n}}$. In the following we express the basic properties of quantum field theory as axioms on the structure of the OPE coefficients:

\begin{axiom}[Hermitian conjugation]
{\ \\} Denoting by $\iota: V\to V$ the anti-linear map given by the star operation $\star$, we have

\begin{equation}
\overline{\mathcal{C}(x_1,\ldots,x_n)}=\iota \mathcal{C}(x_1,\ldots,x_n)\iota^n\: ,
\label{eq:framework axioms hermitian conjugation}
\end{equation}
where $\iota^n:=\iota\otimes\cdots\otimes\iota$ denotes the n-fold tensor product of the map $\iota$.

\end{axiom}

\begin{axiom}[Euclidean invariance]
{\ \\} Let $R$ be the representation of Spin$(D)$ on $V$, let $a\in\mathbb{R}^D$ and $g\in$Spin$(D)$. Then the equation

\begin{equation}
\mathcal{C}(gx_1+a,\ldots, gx_n+a)=R^*(g)\mathcal{C}(x_1,\ldots,x_n)R(g)^n
\label{eq:framework axioms euclidean invariance}
\end{equation}
holds, where $R(g)^n:=R(g)\otimes\cdots\otimes R(g)$ again is the $n$-fold tensor product of $R(g)$.
\end{axiom}

\begin{axiom}[Bosonic nature]
{\ \\} The OPE coefficients themselves should be "bosonic" in the sense that

\begin{equation}
\mathcal{C}(x_1,\ldots,x_n)=\gamma \mathcal{C}(x_1,\ldots,x_n)\gamma^n\: ,
\label{eq:framework axioms bosonic nature}
\end{equation}
where the shorthand notation for the $n$-fold tensor product was used again.
\end{axiom}

\begin{axiom}[Identity element]
{\ \\} There exists a unique element $\mathds{1}\in V$ of dimension $\Delta=0$ satisfying the properties $\mathds{1}^*=\mathds{1}$, $\gamma(\mathds{1})=\mathds{1}$ and

\begin{equation}
\mathcal{C}(x_1,\ldots,x_n)\ket{v_1\otimes\cdots\mathds{1}\otimes\cdots v_{n-1}}=\mathcal{C}(x_1,\ldots,\hat{x}_i,\ldots,x_n)\ket{v_1\otimes\cdots\otimes v_{n-1}}
\label{eq:framework axioms identity element}
\end{equation}
where $\mathds{1}$ is in the $i$-th tensor position, with $i<n$. If $\mathds{1}$ is in the $n$-th position, the relation

\begin{equation}
\mathcal{C}(x_1,\ldots,x_n)\ket{v_1\otimes\cdots v_{n-1}\otimes\mathds{1}}=t(x_{n-1},x_n)\mathcal{C}(x_1,\ldots,x_{n-1})\ket{v_1\otimes\cdots\otimes v_{n-1}}
\label{eq:framework axioms identity element2}
\end{equation}
has to hold, where $t$ is a "Taylor expansion map" characterized below.

\end{axiom}

The reason for the slightly more complicated form of eq. \ref{eq:framework axioms identity element2} is that $x_n$ is the point we expand around and thus the corresponding $n$-th tensor entry stands on a different footing than the other entries. In order to heuristically motivate the form of equation \ref{eq:framework axioms identity element2} and to specify the map $t$ more precisely, we consider the following situation:

Let $\phi_a$ be a quantum (or classical) field. Then we can formally perform a Taylor expansion

\begin{equation}
\phi_a(x_1)=\sum_{i=0}^{\infty} \frac{1}{i!} y^{\mu_1}\cdots y^{\mu_i}\del_{\mu_1}\cdots\del_{\mu_i}\phi_a(x_2)
\label{eq:framework axioms identity element taylor}
\end{equation}
with $y=x_1-x_2$. As each field $\del_{\mu_1}\cdots\del_{\mu_i}\phi_a$ is just another composite field of the theory denoted for example by $\phi_b$, we might rewrite equation \ref{eq:framework axioms identity element taylor} in the form $\phi_a(x_1)=\sum t^b_a(x_1,x_2)\phi_b(x_2)$. The $t^b_a$ are defined using the above Taylor expansion, at least up to potential trivial changes which take into account the fact that a derivative of the field $\phi_a$ might correspond to a linear combination of other fields in the particular labeling we have chosen for the fields. Application of these ideas formally yields

\begin{equation}
\begin{split}
\sum_b\C{b}{a_1\ldots a_{n-1}\mathds{1}}(x_1,\ldots,x_n)\phi_b(x_n)&=\phi_{a_1}(x_1)\cdots\phi_{a_{n-1}}(x_{n-1})\, \mathds{1}\\
	&=\sum_b\C{b}{a_1\ldots a_{n-1}}(x_1,\ldots,x_{n-1})\,\phi_b(x_{n-1})\\
	&=\sum_{b,c} \C{c}{a_1\ldots a_{n-1}}(x_1,\ldots,x_{n-1})\, t^b_c(x_{n-1},x_n)\phi_b(x_n)\: ,
\end{split}
\label{eq:framework axioms identity element heuristic}
\end{equation}
which suggests that

\begin{equation}
\C{b}{a_1\ldots a_{n-1}\mathds{1}}(x_1,\ldots,x_n)=\sum_c t^b_c(x_{n-1},x_n)\, \C{c}{a_1\ldots a_{n-1}}(x_1,\ldots,x_{n-1})\: .
\label{eq:framework axioms identity element heuristic2}
\end{equation}
We now argue that this sum if finite, so one does not have to bother with convergence issues. Note that in eq. \ref{eq:framework axioms identity element taylor} the dimension of the operators on the right hand side should not be smaller than that of the operator on the left hand side, because the former contain additional derivatives. We thus conclude that $t^b_a(x_1,x_2)$ is only nonzero as long as the dimension of the operator $\phi_a$ is greater or equal to the dimension of $\phi_b$. Since there are only finitely many operators up to a given dimension, our proposition is confirmed.

Abstracting the features derived from these heuristic considerations, we are now ready to give a definition of the map $t$ appearing in eq. \ref{eq:framework axioms identity element2}. We postulate the existence of a linear map $t(x_1,x_2):V\to V$ for all $x_1,x_2\in\mathbb{R}^D$ satisfying the following conditions. First, the map should have the same transformation properties as the OPE with respect to rotations (see axiom \ref{ax:Euclidean invariance}). Defining $V^{\Delta}$ as the subspace of $V$ in the decomposition \ref{eq:framework axioms V} spanned by vectors of dimension $\Delta$, we furthermore require

\begin{equation}
 t(x_1,x_2)V^{\Delta}\subset \bigoplus_{\hat{\Delta}\geq \Delta}V^{\hat{\Delta}}\: .
\label{eq:framework axioms identity element t dimension} 
\end{equation} 
Next, we have the cocycle relation

\begin{equation}
t(x_1,x_2)t(x_2,x_3)=t(x_1,x_3)
\label{eq:framework axioms identity element t cocycle}\: .
\end{equation}
Finally, the restriction of any vector of $t(x_1,x_2)V^{\Delta}$ to any subspace of $V^{\hat{\Delta}}$ is supposed to depend polynomially on $x_1-x_2$. This finishes the characterization of the "Taylor expansion map" $t$ and thus completes the formulation of axiom \ref{ax:Identity element}. Note that this axiom implies in particular the relation

\begin{equation}
t(x_1,x_2)\ket{v}=\mathcal{C}\ket{v\otimes\mathds{1}}\: ,
\label{eq:framework axioms identity element remark}
\end{equation}
i.e. $t(x_1,x_2)$ uniquely determines the 2-point OPE coefficients with an identity operator and vice versa. As a special case, we have $t(x_1,x_2)\mathds{1}=\mathds{1}$ using eq. \ref{eq:framework axioms identity element} and $\mathcal{C}(x_1)=id$, which implies that the identity operator does not depend on a \emph{reference point}.

Before we come to the next axiom, some notation is introduced. Let $I_1,\ldots,I_r$ denote a partition of the set $\{1,\ldots,n\}$ into disjoint ordered subsets, where all elements in $I_i$ are greater than all elements in $I_{i-1}$ for all $i$. An example for such a partition in the case $n=5$ would be $I_1\{1\}$, $I_2\{2,3,4\}$, $I_3=\{5,6\}$. For each ordered subset $I\subset\{1,\ldots,n\}$ we define $X_I$ to be the ordered tuple $(x_i)_{i\in I}\in (\mathbb{R}^D)^{|I|}$ and we set $m_k:=\max(I_k)$ and $\mathcal{C}(X_I):=id$ if $I$ consists of only one element. Furthermore, let $d(X_I)$ be the set of relative Euclidean distances between the points in a collection $X_I=(x_i)_{i\in I}$, defined as the set of positive real numbers

\begin{equation}
d(X_I):=\{r_{ij}\: |\: i.j\in I, i\neq j \}
\label{eq:framework axioms factorization distances}
\end{equation}
We are now ready for the\index{symbols}{D@$\mathcal{D}[\{I_1,\ldots,I_r\}]$}

\begin{axiom}[Factorization]
{\ \\} The identity

\begin{equation}
\mathcal{C}(X_{\{1,\ldots,n\}})=\mathcal{C}(X_{\{m_1,\ldots,m_r\}})\Big(\mathcal{C}(X_{I_1})\otimes\cdots\otimes\mathcal{C}(X_{I_r})\Big)
\label{eq:framework axioms factorization}
\end{equation}
holds on the open domain

\begin{equation}
\mathcal{D}[\{I_1,\ldots,I_r\}]:=\Big\{(x_1,\ldots,x_n)\in M_n | \min d(X_{\{m_1,\ldots,m_r\}})>\max(d(X_{I_1}),\ldots,d(X_{I_r}))\Big\}
\label{eq:framework axioms factorization domain}
\end{equation}

\end{axiom}

At this point some remarks are in order. First one should note that the factorization identity \ref{eq:framework axioms factorization} expressed in a basis of $V\otimes\cdots\otimes V$ involves an $r$-fold infinite sum on the right side of the equation. In fact, it is the statement of the factorization property that these infinite sums converge on the indicated domain. Outside the domain the sums are not restricted at all and one would expect them to diverge. The axiom can be generalized to arbitrary partitions of $\{1,\ldots,n\}$ making use of the (anti-)symmetry axiom \ref{ax:(Anti-)symmetry} below. If fermionic fields are included, then $\pm$ signs will appear. It is also important to remark that the factorization relation can be iterated on suitable domains, i.e. if for example the subset $I_j$ is itself partitioned into subsets, then the coefficient $\mathcal{C}(X_{I_j})$ will itself factorize on a suitable subdomain. Such subsequent partitions may naturally be identified with trees. A version of the factorization property in terms of trees was given in \cite{Hollands2007} and will also be given below \ref{eq:framework coherence factorization basis}.

\begin{axiom}[Scaling]
{\ \\} Let $v_{a_1},\ldots,v_{a_n}\in V$ be vectors with dimension $\Delta_1,\ldots,\Delta_n$ (remember the decomposition of $V$ in eq. \ref{eq:framework axioms V}) respectively and let $v^b\in V^*$ be an element in the dual space of $V$ with dimension $\Delta_{n+1}$. Then the scaling degree of the $\mathbb{C}$-valued distribution \ref{eq:framework axioms OPE basis} should be estimated by

\begin{equation}
sd\, \C{b}{a_1\ldots a_n}\leq \Delta_1+\ldots +\Delta_n-\Delta_{n+1}\: .
\label{eq:framework axioms scaling}
\end{equation}\index{symbols}{sd@$sd\, \C{b}{a_1\ldots a_n}$}

\end{axiom}
By \emph{scaling degree} we mean
\begin{equation}
 sd\, \C{b}{a_1\ldots a_n}=\inf_{p\in\mathbb{R}}\Big(\lim_{\epsilon\to 0}\epsilon^{p}\C{b}{a_1\ldots a_n}(\epsilon x_1,\ldots,\epsilon x_n)=0 \text{ for all } (x_1,\ldots,x_n)\in M_n\Big)\, .
\label{eq:framework axioms scaling degree}
\end{equation}
Further, if $v^b$ is an element of the one-dimensional subspace of dimension-0 fields spanned by the identity operator $\mathds{1}\in V$, and if $n=2$ and $v_{a_1}=v_{a_2}^*\neq 0$, then the inequality in the eq.\eqref{eq:framework axioms scaling} is required to hold.

\begin{axiom}[(Anti-)symmetry]
 {\ \\} Let $\tau_{i-1,i}$ be the permutation exchanging the $(i-1)$-th and the $i$-th tensor factor in an element of $V\otimes \cdots\otimes V$. Then we have (for $1<i<n$)

\begin{equation}
\mathcal{C}(x_1,\ldots,x_{i-1},x_i,\ldots,x_n)\tau_{i-1,i}= \mathcal{C}(x_1,\ldots,x_i,x_{i-1},\ldots,x_n)(-1)^{F_{i-1}F_i}
\label{eq:framework axioms symmetry}
\end{equation}
with

\begin{equation}
 F_i:=\frac{1}{2}\, id^{i-1}\otimes (id-\gamma)\otimes id^{n-i}\: .
\label{eq:framework axioms symmetry F}
\end{equation}
In the case $i=n$ we require

\begin{equation}
\mathcal{C}(x_1,\ldots,x_{n-1},x_n)\tau_{n-1,n}= t(x_{n-1},x_n) \mathcal{C}(x_1,\ldots,x_n,x_{n-1})(-1)^{F_{n-1}F_n}
\label{eq:framework axioms symmetry2}
\end{equation}

\end{axiom}
The last factor in eqs. \ref{eq:framework axioms symmetry} and \ref{eq:framework axioms symmetry2} makes the OPE coefficients of bosonic fields symmetric and the OPE coefficients of fermionic fields anti-symmetric. Also notice that again we had to treat the case involving the $n$-th point and the $n$-th tensor factor separately. The reasons for this are obviously similar to those leading to eq. \ref{eq:framework axioms identity element2}.

\begin{axiom}[Derivations]
The derivations $\del_{\mu}$ on $V$ are compatible with partial derivatives of the OPE coefficients with respect to the spacetime arguments $x_i\in\mathbb{R}^D$ in the sense that

\begin{equation}
 \mathcal{C}(x_1,\ldots,x_n)\, \ket{v_1\otimes\cdots \del_{\mu}v_i\otimes\cdots v_n}=\del_{(x_i)_{\mu}}\mathcal{C}(x_1,\ldots,x_n)\, \ket{v_1\otimes\cdots v_i\otimes\cdots v_n}
\label{eq:framework axioms derivations1}
\end{equation}
for $i< n$. If there is a derivative in the last tensor position, then the identity

\begin{equation}
 \mathcal{C}(x_1,\ldots,x_n)\, \ket{v_1\otimes\cdots \otimes \del_{\mu}v_n}=\left[\del_{(x_n)_{\mu}}\mathcal{C}(x_1,\ldots,x_n)+\del_{\mu}\circ\mathcal{C}(x_1,\ldots,x_n)\right]\, \ket{v_1\otimes\cdots \otimes v_n}
\label{eq:framework axioms derivations2}
\end{equation}
holds.
\end{axiom}
This axiom, which has not been required in \cite{Hollands2008}, will later allow us to transform field equations, i.e. partial differential equations on $V$, into similar relations involving the OPE coefficients (see section \ref{subsec:perturbation via field equation}). \\

This finishes our presentation of the axioms for the new framework. As already mentioned in section \ref{subsec: framework motivation}, the factorization property, axiom \ref{ax:Factorization}, lies at the heart of the theory. The stringent constraints it imposes on the possible consistent hierarchies $(\mathcal{C}(x_1,x_2),\mathcal{C}(x_1,x_2,x_3),\ldots)$ carry the main information in the approach. The Euclidean invariance condition, axiom \ref{ax:Euclidean invariance}, implies translation invariance of the OPE coefficients and links the decomposition \ref{eq:framework axioms V} of the field space $V$ into sectors of different spin to the transformation properties of the OPE coefficients under rotations. Likewise, the scaling property links the grading of $V$ with respect to the dimension to the scaling properties of the OPE coefficients. Furthermore, the (anti-)symmetry requirement, axiom \ref{ax:(Anti-)symmetry}, is a replacement for local (anti-)commutativity (Einstein causality) in the Euclidean setting. Note also that we have not required a spin-statistics or $PCT$-theorem to hold, as these can in fact be derived from the axiomatic framework just introduced (see \cite{Hollands2004}).

So in summary, a quantum field theory is defined as a pair $(V,\mathcal{C})$ consisting of a vector space $V$ with the above properties and a hierarchy of OPE coefficients $\mathcal{C}:=(\mathcal{C}(x_1,x_2),\mathcal{C}(x_1,x_2,x_3),\ldots)$ satisfying axioms \ref{ax:Hermitian conjugation} to \ref{ax:Derivations}. Now the issue of equivalence of different theories of this kind arises. One would naturally identify quantum field theories that only differ by a redefinition of their fields, where, informally, a field redefinition means that the definition of the quantum fields of the theory is changed by a transformation of the form $\phi_a(x)$ to $\hat{\phi}_a(x)=z^b_a\phi_b(x)$, where $z_a^b$ is some matrix on field space. The following definition carries over these ideas to our framework.

\begin{definition}[Equivalence]
 {\ \\} Let $(V,\mathcal{C})$ and $(\hat{V},\hat{\mathcal{C}})$ be two quantum field theories. If there exists an invertible map $z:V\to \hat{V}$ satisfying

\begin{equation}
 zR(g)=\hat{R}(g)z\:, \hspace{1cm} z\gamma=\hat{\gamma}z\:, \hspace{1cm} z\,\star=\hat{\star}\,z
\label{eq:framework axioms redifinition z}
\end{equation}
as well as

\begin{equation}
\mathcal{C}(x_1,\ldots,x_n)=z^{-1}\hat{\mathcal{C}}(x_1,\ldots,x_n)z^n
\label{eq:framework axioms redifinition z2}
\end{equation}
for all $n$, where $z^n=z\otimes\cdots\otimes z$, then the two quantum field theories are said to be equivalent and $z$ is called a field redefinition.

\end{definition}

Before concluding this chapter one further condition is imposed, namely that the quantum field theory $(V,\mathcal{C})$ exhibits a vacuum state. The appropriate notion of quantum state in our Euclidean setting is a collection of correlation functions, which we will write as $\langle \phi_{a_1}(x_1)\cdots\phi_{a_n}(x_n) \rangle_{\Omega}$ with arbitrary $n$ and $a_1,\ldots,a_n$. These functions should be analytic on $M_n$ and satisfy the Osterwalder-Schrader (OS) axioms for the vacuum state $\Omega$ (see \cite{Osterwalder:1974tc,Osterwalder:1973dx}) and also the OPE in the sense that

\begin{equation}
\langle \phi_{a_1}(x_1)\cdots\phi_{a_n}(x_n) \rangle_{\Omega} \sim \sum_b \C{b}{a_1\ldots a_n}(x_1,\ldots,x_n)\langle \phi_b(x_n)\rangle_{\Omega}\: .
\label{eq:framework axioms states}
\end{equation}
The symbol `` $\sim$ '' here indicates that the difference between the left and right side of this expression is a distribution on $M_n$ with smaller scaling degree than any given number $\delta$ provided the above sum goes over all of the finitely many fields $\phi_b$ whose dimension is smaller than some number $\Delta=\Delta(\delta)$. Then by the OS-reconstruction theorem the theory can be continued back to Minkowski spacetime and the fields can be represented as linear operators on a Hilbert space $\mathcal{H}$ of states. It may also be of interest in some settings (e.g. theories with unbounded potentials) to drop the notion of unique vacuum state by leaving out those OS-axioms which involve statements about invariance under Euclidean transformations.

Obviously, if we require a quantum state to satisfy the OS-axioms, new constraints on the OPE coefficients are expected to appear. These constraints will not be discussed here, as our focus is on the algebraic conditions satisfied by the OPE coefficients. The additional restrictions imposed by the OS-axioms are genuinely new, so in some contexts one might even want to drop them (e.g. the condition of OS-positivity does not hold in some systems in statistical mechanics and in gauge theories before taking the quotient by the BRST-differential).

\section{The coherence theorem}\label{subsec:coherence theorem}

It has already been mentioned a few times that we think of the factorization property, axiom \ref{ax:Factorization}, as the key condition on the OPE coefficients. These restrictions on the OPE coefficients $\mathcal{C}(x_1,\ldots,x_n)$ (with $n\geq 2$) are expected to be very stringent and encode most of the non-trivial information of the theory. It is the purpose of this section to analyze the interdependence of these conditions for different $n$, i.e. for OPE coefficients coming from the expansion of a product of $n$ fields. The result will be that all the \emph{higher} constraints (i.e. larger $n$) are already encoded in the first non-trivial constraint arising for $n=3$. As this implies that all the factorization constraints can be coherently described by a single condition, this result was named the \emph{coherence theorem} in \cite{Hollands2008}.

It is instructive to consider an analog from ordinary algebra before treating our case in more detail. Let us consider a finite dimensional associative algebra $\mathbf{A}$. Thus 

\begin{equation}
 (AB)C=A(BC)\hspace{2cm}\forall\: A,B,C\in\mathbf{A}
\label{eq:framework coherence standard associativity}
\end{equation}
holds, or rewritten in terms of the linear product map $m:\mathbf{A}\otimes\mathbf{A}\to\mathbf{A}$ with $m(A,B)=AB$:

\begin{equation}
 m(id\otimes m)=m(m\otimes id)
\label{eq:framework coherence standard associativity m}
\end{equation}
Both sides of this equation are maps $\mathbf{A}\otimes\mathbf{A}\otimes\mathbf{A}\to\mathbf{A}$. By successively applying eq. \ref{eq:framework coherence standard associativity} it is easy to prove that also

\begin{equation}
 (AB)(CD)=(A(BC))D\hspace{2cm}\forall\: A,B,C,D\in\mathbf{A}
\label{eq:framework coherence standard higher associativity}
\end{equation}
holds without further restrictions on the map $m$. In fact, it is an elementary result in algebra that no such \emph{higher associativity conditions} arise. It is not difficult to see that this result is analogous to our coherence theorem, where eq. \ref{eq:framework coherence standard associativity} plays the role of the three point factorization constraint and eq. \ref{eq:framework coherence standard higher associativity} together with similar higher order conditions is identified with the $n>3$ factorization conditions.

Now let us come back to our original setting and consider axiom \ref{ax:Factorization} for $n=3$. We encounter three different partitions of the set $\{1,2,3\}$ corresponding to non-trivial factorization conditions \ref{eq:framework axioms factorization}, namely $\mathbf{T}_3:=\{\{1,2\}\{3\}\}$, $\mathbf{T}_2:=\{\{1,3\}\{2\}\}$ and $\mathbf{T}_1:=\{\{2,3\}\{1\}\}$. The respective domains on which the factorization identities are supposed to be valid are, according to eq. \ref{eq:framework axioms factorization domain}, given by

\begin{eqnarray}
 \mathcal{D}[\mathbf{T}_1]\,&=\,\{(x_1,x_2,x_3)\,|\, r_{23}<r_{13}\}\: , \\
 \mathcal{D}[\mathbf{T}_2]\,&=\,\{(x_1,x_2,x_3)\,|\, r_{13}<r_{23}\}\: , \\
 \mathcal{D}[\mathbf{T}_3]\,&=\,\{(x_1,x_2,x_3)\,|\, r_{12}<r_{23}\}\: .
\label{eq:framework coherence 3 point domains}
\end{eqnarray}
The first two domains in the above equations are clearly disjoint, but both have a non-empty intersection with the third domain. Thus, according to axiom \ref{ax:Factorization}, on these intersections both factorizations should give $\mathcal{C}(x_1,x_2,x_3)$ and hence should also be equal to each other. We can write this as

\begin{equation}
 \mathcal{C}(x_2,x_3)\Big(\mathcal{C}(x_1,x_2)\otimes id\Big)=\mathcal{C}(x_1,x_3)\Big(id\otimes\mathcal{C}(x_2,x_3)\Big)
\label{eq:framework coherence 3 point constraint}
\end{equation}
where the spacetime arguments lie in the intersection $\mathcal{D}[\mathbf{T}_1]\cap\mathcal{D}[\mathbf{T}_3]$, i.e. $r_{12}<r_{23}<r_{13}$. A similar relation holds on the domain $\mathcal{D}[\mathbf{T}_2]\cap\mathcal{D}[\mathbf{T}_3]$, but it turns out that this relation can be derived from \ref{eq:framework coherence 3 point constraint} with the help of the symmetry condition, axiom \ref{ax:(Anti-)symmetry},

\begin{equation}
 \mathcal{C}(x_1,x_2)=t(x_1,x_2)\mathcal{C}(x_2,x_1)\tau_{1,2}
\label{eq:framework coherence symmetry axiom}
\end{equation}

and the relation

\begin{equation}
 \mathcal{C}(x_1,x_3)=\mathcal{C}(x_2,x_3)\Big(t(x_1,x_2)\otimes id\Big)
\label{eq:framework coherence relation}
\end{equation}
for $r_{12}<r_{23}$. So in the case $n=3$ (i.e. three spacetime points) there exists only one independent consistency condition, eq. \ref{eq:framework coherence 3 point constraint}, which has already been given in component form in \ref{eq:framework motivation consistency}.

Remember, the aim of this chapter is to analyze higher factorization conditions, i.e. axiom \ref{ax:Factorization} for $n>3$. We have seen in the analogous problem in ordinary algebra that all higher associativity conditions can be derived from eq. \ref{eq:framework coherence standard associativity m}, which is the analogue to our eq. \ref{eq:framework coherence 3 point constraint}. In the following it will be shown that, as in the example of associative algebra, we will not encounter any higher factorization conditions and that the coefficients $\mathcal{C}(x_1,\ldots,x_n)$, the analogue of a product of $n$ elements of our algebra $\mathbf{A}$, are completely determined by the coefficients $\mathcal{C}(x_1,x_2)$, which can be identified with a product of two elements of $\mathbf{A}$. 

Thus, our first task is to write all the factorization equations only in terms of the $\mathcal{C}(x_1,x_2)$. In this context the language of rooted trees is natural and useful (see also \cite{Hollands2007}). By a rooted tree on $n$ elements $\{1,\ldots,n\}$ we will mean a set $\{S_1,\ldots,S_k\}$ of nested subsets $S_i\subset\{1,\ldots,n\}$, such that each $S_i$ is either contained in another set of the family of subsets, or disjoint from it. The set $\{1,\ldots,n\}$, called the root, is by definition not in the tree. One can then visually think of the sets $S_i$ as the nodes of the tree, which are connected by branches to all those nodes that are subsets of $S_i$, but not proper subsets of any element of the tree other than $S_i$. The leaves of the tree are the nodes that do not contain any other sets of the tree as subsets, so they are of the form $S_i=\{i\}$. Let us introduce some further notation: If $\mathbf{T}$ is a tree on $n$ elements of a set, then we denote by $|\mathbf{T}|$ the elements of this set. Furthermore, let $\mathbf{T}$ be of the form $\mathbf{T}=\{\mathbf{T}_1,\ldots,\mathbf{T}_r\}$, where each $\mathbf{T}_i$ is itself a tree on a proper subset of $\{1,\ldots,n\}$, so that $|\mathbf{T}_1|\cup\cdots\cup|\mathbf{T}_r|=\{1,\ldots,n\}$ is a partition into disjoint subsets. For such trees we recursively define an open, non-empty domain of $M_n$, by\index{symbols}{DT@$\mathcal{D}[\mathbf{T}]$}

\begin{equation}
\begin{split}
 \mathcal{D}[\mathbf{T}]\:=\:&\Big\{(x_1,\ldots,x_n)\in M_n\: |\: X_{|\mathbf{T}_1|}\in \mathcal{D}[\mathbf{T}_1],\ldots, X_{|\mathbf{T}_r|}\in \mathcal{D}[\mathbf{T}_r];\\
&\min d(X_{\{m_1,\ldots,m_r\}})> \max(d(X_{|\mathbf{T}_1|}),\ldots,d(X_{|\mathbf{T}_r|}))\Big\}\: .
\end{split}
\label{eq:framework coherence tree domain}
\end{equation}

Here $m_i$ is the maximum element element upon which the tree $\mathbf{T}_i$ is built. Otherwise the notation is as introduced above axiom \ref{ax:Factorization}. The domain presented in eq. \ref{eq:framework axioms factorization domain} is recovered from this definition, if the $\mathbf{T}_i$ are the trees with only a single node apart from the leaves. Then the $I_i$ of eq. \ref{eq:framework axioms factorization domain} is given by the elements of the $i$-th subtree $\mathbf{T}_i$. Thus, in this case, the factorization property \ref{eq:framework axioms factorization} holds on the given domain. But if we consider more complex trees, the corresponding domains according to eq. \ref{eq:framework coherence tree domain} will just be proper open subsets of the domain we just considered, so the factorization condition is still satisfied on these domains. Hence, the factorization identity, eq, \ref{eq:framework axioms factorization} holds on $\mathcal{D}[\mathbf{T}]$ in any case. In addition, we may now iterate this factorization, as the factors $\mathcal{C}(X_{|\mathbf{T}_i|})$ themselves factorize on $\mathcal{D}[\mathbf{T}]$, because that domain is defined in such a way that $X_{|\mathbf{T}_i|}\in\mathcal{D}[\mathbf{T}_i]$. So, successively making use of such factorization identities, we obtain a nested factorization property on each of the domains $\mathcal{D}[\mathbf{T}]$.

In order to explicitly write down these identities, some further notation is useful. Let $S\in\mathbf{T}$. Then we write $l(1),\ldots,l(j)\subset_{\mathbf{T}}S$, if $l(1),\ldots,l(j)$ are the branches descending from $S$. By $m_i$ we denote the largest element in $l(i)$ and assume an ordering of the branches, such that $m_1<\ldots<m_j$. Also, as above in eq. \ref{eq:framework axioms OPE basis}, we will work in a basis of the linear maps $\mathcal{C}(x_1,\ldots,x_n):V^{\otimes n}\to V$ with components $\C{b}{a_1\ldots a_n}(x_1,\ldots,x_n)$. Then for each tree $\mathbf{T}$ an iteration of axiom \ref{ax:Factorization} as described above will lead to the following factorization identity on the domain $\mathcal{D}[\mathbf{T}]$:

\begin{equation}
 \C{b}{a_1\ldots a_n}(x_1,\ldots,x_n)=\sum_{a_S:\,S\in\mathbf{T}}\left(\prod_{S:\, l(1),\ldots,l(j)\subset_{\mathbf{T}}S}\C{a_S}{a_{l(1)}\ldots a_{l(j)}}(x_{m_1},\ldots,x_{m_j})\right)\: ,
\label{eq:framework coherence factorization basis}
\end{equation}
where the sums are over all $a_S$ with $S$ a subset in the tree excluding the root $\{1,\ldots,n\}$ and the leaves ${1},\ldots,\{n\}$. For the latter we set $a_{\{1\}}:=a_1,\ldots,a_{\{n\}}=a_n$ and $a_{\{1,\ldots,n\}}:=b$. The hierarchical order by which the nested infinite sums are carried out is determined by the tree, with the sums corresponding to the nodes closest to the leaves coming first. Now, if we assume $\mathcal{T}$ to be a binary tree, i.e. one with exactly two branches descending from every node, then by the above formula we have expressed the $n$-point OPE coefficient $\mathcal{C}(x_1,\ldots,x_n)$ in terms of products of 2-point coefficients on the open domain $\mathcal{D}[\mathbf{T}]\subset M_n$. Remembering that by definition $\mathcal{C}(x_1,\ldots,x_n)$ is an analytic function on the open, connected domain $M_n$ and using the fact that an analytic function on a connected domain is uniquely determined by its restriction to an open set, we propose:

\begin{proposition}\label{prop:factorization}
The $n$-point OPE coefficients $\mathcal{C}(x_1,\ldots,x_n)$ are uniquely determined by the 2-point coefficients $\mathcal{C}(x_1,x_2)$. In particular, if two quantum field theories have equivalent 2-point OPE coefficients (see definition \ref{def:Equivalence}), then they are equivalent.
\end{proposition}

Thus, our first task is achieved. However, the question of higher factorization conditions is still not solved. It remains to show that the factorization condition \ref{eq:framework coherence factorization basis} for binary trees does not impose any further restrictions on $\mathcal{C}(x_1,x_2)$ apart from eq. \ref{eq:framework coherence 3 point constraint}. Therefore, let us study the following expression

\begin{equation}
 (f_{\mathbf{T}})^b_{a_1\ldots a_n}(x_1,\ldots,x_n):=\sum_{a_S:\,S\in\mathbf{T}}\left(\prod_{S:\, l(1),l(2)\subset_{\mathbf{T}}S}\C{a_S}{a_{l(1)}a_{l(2)}}(x_{m_1},x_{m_2})\right)\: ,
\label{eq:framework coherence f}
\end{equation}
on the domain $\mathcal{D}[\mathbf{T}]$ for any binary tree $\mathbf{T}$. In other words, $f_{\mathbf{T}}(x_1,\ldots,x_n)$ is just the expression for $\mathcal{C}(x_1,\ldots,x_n)$ in the factorization identity \ref{eq:framework coherence factorization basis} for the binary tree $\mathbf{T}$. So, as we have just argued, $f_{\mathbf{T}}$ can be analytically continued to an analytic function on $M_n$, which we will denote by $f_{\mathbf{T}}$ as well and which does not depend on the choice of binary tree $\mathbf{T}$. As we want to analyze the constraints imposed on the 2-point coefficients $\mathcal{C}(x_1,x_2)$ by the properties we just stated, we now drop these assumptions and only assume that the sums in eq. \ref{eq:framework coherence f} converge and define an analytic function $f_{\mathbf{T}}$ on $\mathcal{D}[\mathbf{T}]$, which can be analytically continued to $M_n$ for all $n$ and all binary trees $\mathbf{T}$ on $n$ elements. For the sake of the argument, we in particular do \emph{not} assume the $f_{\mathbf{T}}$ to coincide for different binary trees, except in the case $n=3$, where the assumption that the $f_{\mathbf{T}}$ coincide for the three possible binary trees on the respective domains is equivalent to the factorization condition for three points, eq. \ref{eq:framework coherence 3 point constraint}, plus the symmetry and normalization conditions, eqs. \ref{eq:framework coherence symmetry axiom} and \ref{eq:framework coherence relation}. These three conditions will be assumed to hold.

Now we want to show that these assumptions suffice to deduce that all $f_{\mathbf{T}}$ coincide for all binary trees $\mathbf{T}$, thus implying the absence of further consistency conditions on $\mathcal{C}(x_1,x_2)$ beyond those for $n=3$. We graphically present the corresponding proof, which is not very difficult and quite similar to the proof of the analogous statement in our example of ordinary algebra. We start with the case $n=3$. Here, the assumption that all $f_{\mathbf{T}}$ agree for the three trees is graphically expressed by fig. \ref{fig:associativity 3}.

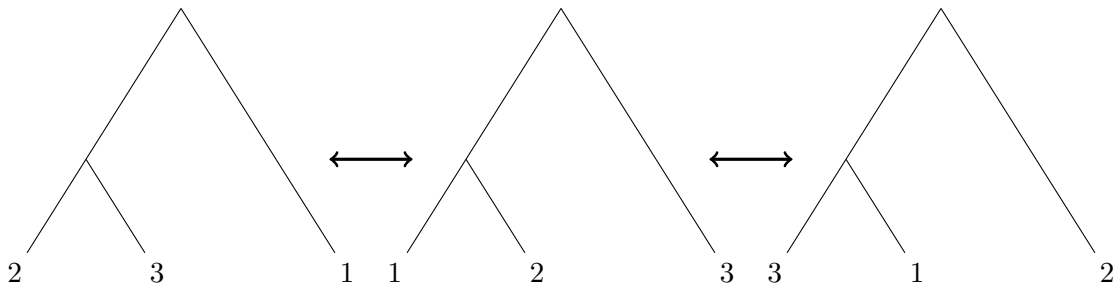
\begin{figure}[h!]
\begin{center}
\begin{tikzpicture}
\tikzstyle{level 1}=[sibling distance = 2.5cm, level distance = 2cm]
\tikzstyle{level 2}=[sibling distance = 1.875cm, level distance = 1.5cm]

 \node at (0,0) [shape=coordinate] {}
	child {child {node {2}} child {node {3}}} 
	child {node [shape=coordinate] (a) {} child[fill=none]{edge from parent [draw=none]} child{node{1}}};
\node at (5,0) [shape=coordinate] {}
	child {node [shape=coordinate] (b) {} child {node {1}} child {node {2}}} 
	child {node [shape=coordinate] (c){} child [fill=none]{edge from parent [draw=none]} child{node{3}}};
\node at (10,0) [shape=coordinate] {}
	child {node [shape=coordinate] (d) {} child {node {3}} child {node {1}}} 
	child {child [fill=none]{edge from parent [draw=none]} child{node{2}}};
	\draw [<->, shorten >= 20pt, shorten <= 20pt,very thick] (a.north east) -- (b.north west);
	\draw [<->, shorten >= 20pt, shorten <= 20pt,very thick] (c.north east) -- (d.north west);
\end{tikzpicture}
\end{center}
\caption{The associativity condition in graphical notation. Double arrows indicate that the OPE's represented by the respective trees coincide on the (non-empty) intersection of the associated domains $\mathcal{D}[\mathbf{T}_i]$.	Note that these arrows are not a transitive relation: The domains associated with the left- and rightmost tree are disjoint. }
\label{fig:associativity 3}
\end{figure}
Each tree in this graph symbolizes the corresponding expression $f_{\mathbf{T}}$ and the arrows denote the following relations: (i) the corresponding domains (see eq. \ref{eq:framework coherence 3 point domains}) are not disjunct and (ii) the expressions coincide on the intersection. Analyticity of the $f_{\mathbf{T}}$ then implies that the $f_{\mathbf{T}}$'s agree on the whole of $M_n$. Let us now move on to the $n>3$ case and let $\mathbf{T}$ be an arbitrary tree on $n$ elements. The idea of the proof is the following: Find a sequence $\mathbf{T}_0,\mathbf{T}_1,\ldots,\mathbf{T}_r$ of trees such that $\mathbf{T}_0=\mathbf{T}$ and $\mathbf{T}_r=\mathbf{S}$, where $\mathbf{S}$ is the \emph{reference tree}

\begin{equation}
 \mathbf{S}:=\Big\{\{n\},\{n-1,n\},\{n-2,n-1,n\},\ldots,\{1,2,\ldots,n\}\Big\}
\label{eq:framework coherence reference tree}
\end{equation}
as drawn in fig. \ref{fig:reference tree S}.

\begin{figure}[h!]
\begin{center}
\begin{tikzpicture}[level distance = 1cm]
 \coordinate
	child {child{ child{ child[dashed]{child[solid]{node{1}} child[solid]{node{2}}} child[fill=none]{edge from parent [draw=none]}{} } child{ child[fill=none]{edge from parent [draw=none]} child{ child[fill=none]{edge from parent [draw=none]} child{node{n-2} }} } } child{ child[fill=none]{edge from parent [draw=none]} child{ child[fill=none]{edge from parent [draw=none]} child{ child[fill=none]{edge from parent [draw=none]} child{node{n-1}} } } }}
	child {child[fill=none]{edge from parent [draw=none]} child{child[fill=none]{edge from parent [draw=none]} child{child[fill=none]{edge from parent [draw=none]} child{child[fill=none]{edge from parent [draw=none]} child{node{n}}}}}};
\end{tikzpicture}
\end{center}
\caption{The reference tree $S$ defined in eq.\ref{eq:framework coherence reference tree}}
\label{fig:reference tree S}
\end{figure}
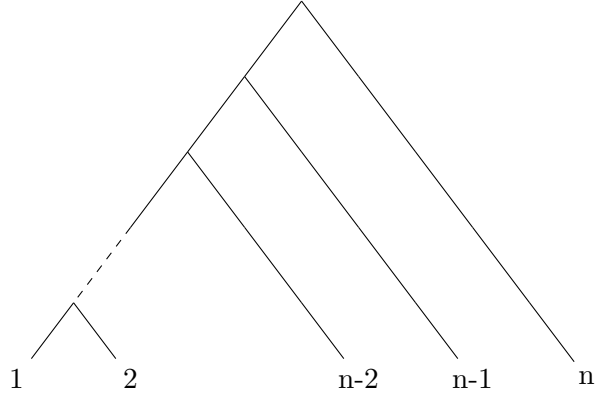
Furthermore, we want a relation as above, i.e. non-disjointness and equivalence on the intersection of the domains, to hold between elements $\mathbf{T}_i$ and $\mathbf{T}_{i-1}$ in our sequence of trees. Making use of the analyticity properties of the $f_{\mathbf{T}}$ as in the $n=3$ case, this would yield $f_{\mathbf{T}}=f_{\mathbf{S}}$ on $M_n$, and hence all $f_{\mathbf{T}}$ would be equal.

The construction of the desired sequence of trees is presented in the following by inductive methods. Our starting point is the binary tree $\mathbf{T}=\mathbf{T}_0$ as depicted on the left of fig. \ref{fig:elementary manipulation 1}, where shaded regions represent subtrees whose particular form is irrelevant at this stage.

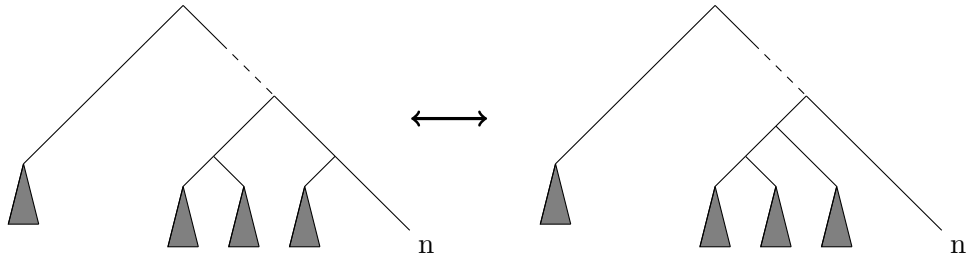
\begin{figure}[h!]
\begin{center}
\begin{tikzpicture}
\tikzstyle{level 1}=[sibling distance = 1cm, level distance = .5cm]
\tikzstyle{level 2}=[sibling distance = 1cm, level distance = .5cm]
\tikzstyle{level 3}=[sibling distance = 1cm, level distance = .5cm]
\tikzstyle{level 4}=[sibling distance = .8cm, level distance = .4cm]
\tikzstyle{level 5}=[sibling distance = .8cm, level distance = .4cm]
\tikzstyle{level 6}=[sibling distance = .4cm, level distance = .8cm]

\node at (0,0) [shape=coordinate] {}
	child{child[level distance=1.6cm, sibling distance=3.2cm]{node(triang a1)[shape=coordinate]{} child[level distance=.8cm, sibling distance=.4cm]{node(triang a2)[shape=coordinate]{} } child[level distance=.8cm, sibling distance=.4cm]{node(triang a3)[shape=coordinate]{} } } child[fill=none]{edge from parent [draw=none]}} child[level distance=.5cm]{child[fill=none]{edge from parent [draw=none]} child[dashed, level distance= .7cm, sibling distance=1.4cm ]{child[solid, level distance=.8cm, sibling distance=1.6cm]{  child{ node(triang b1)[shape=coordinate]{} child[solid,level distance=.8cm, sibling distance=.4cm]{ node(triang b2)[shape=coordinate]{} } child[solid,level distance=.8cm, sibling distance=.4cm]{ node(triang b3)[shape=coordinate]{} } } child[solid]{node(triang c1)[shape=coordinate]{} child[solid,level distance=.8cm, sibling distance=.4cm]{ node(triang c2)[shape=coordinate]{} } child[solid,level distance=.8cm, sibling distance=.4cm] { node(triang c3)[shape=coordinate]{} } }} child[solid, level distance=.8cm, sibling distance=1.6cm]{child[solid]{ node(triang d1)[shape=coordinate]{} child[solid,level distance=.8cm, sibling distance=.4cm]{ node(triang d2)[shape=coordinate]{} } child[solid,level distance=.8cm, sibling distance=.4cm]{ node(triang d3)[shape=coordinate]{} } } child[solid,level distance=1.2cm, sibling distance=2.4cm]{node{n}}   } }  }
	;
\filldraw[fill opacity=.5] (triang a1) -- (triang a2) -- (triang a3);
\filldraw[fill opacity=.5] (triang b1) -- (triang b2) -- (triang b3);
\filldraw[fill opacity=.5] (triang c1) -- (triang c2) -- (triang c3);
\filldraw[fill opacity=.5] (triang d1) -- (triang d2) -- (triang d3);

\node at (7,0) [shape=coordinate] {}
	child{child[level distance=1.6cm, sibling distance=3.2cm]{node(2triang a1)[shape=coordinate]{} child[level distance=.8cm, sibling distance=.4cm]{node(2triang a2)[shape=coordinate]{} } child[level distance=.8cm, sibling distance=.4cm]{node(2triang a3)[shape=coordinate]{} } } child[fill=none]{edge from parent [draw=none]}} child[level distance=.5cm]{child[fill=none]{edge from parent [draw=none]} child[dashed, level distance= .7cm, sibling distance=1.4cm ]{ child[solid, level distance=.4cm, sibling distance=.8cm]{child[level distance=.4cm, sibling distance=.8cm]{  
	child{ node(2triang b1)[shape=coordinate]{} child[level distance=.8cm, sibling distance=.4cm]{ node(2triang b2)[shape=coordinate]{} } child[level distance=.8cm, sibling distance=.4cm]{ node(2triang b3)[shape=coordinate]{} } } child{node(2triang c1)[shape=coordinate]{} child[level distance=.8cm, sibling distance=.4cm]{ node(2triang c2)[shape=coordinate]{} } child[level distance=.8cm, sibling distance=.4cm] { node(2triang c3)[shape=coordinate]{} } }} child[level distance=.8cm, sibling distance=1.6cm]{  
	node(2triang d1)[shape=coordinate]{} child[level distance=.8cm, sibling distance=.4cm]{ node(2triang d2)[shape=coordinate]{} } child[level distance=.8cm, sibling distance=.4cm]{ node(2triang d3)[shape=coordinate]{} }
	  }  } child[solid,level distance=2cm, sibling distance=4cm]{node{n}}   } } 
	;
\filldraw[fill opacity=.5] (2triang a1) -- (2triang a2) -- (2triang a3);
\filldraw[fill opacity=.5] (2triang b1) -- (2triang b2) -- (2triang b3);
\filldraw[fill opacity=.5] (2triang c1) -- (2triang c2) -- (2triang c3);
\filldraw[fill opacity=.5] (2triang d1) -- (2triang d2) -- (2triang d3);
\draw[<->, very thick] (3,-1.5) -- (4,-1.5);

\end{tikzpicture}
\end{center}
\caption{The first elementary manipulation of the tree $\mathbf{T}=\mathbf{T}_0$ on the left into the tree $\mathbf{T}_1$ on the right. Shaded triangles represent subtrees whose specific form is irrelevant.}
\label{fig:elementary manipulation 1}
\end{figure}
The next tree in the sequence, $\mathbf{T}_1$, is drawn on the right of fig. \ref{fig:elementary manipulation 1}. As stated above, we now need a relation between those trees. Simply using the definitions one easily finds $\mathcal{D}[\mathbf{T}_0]\cap\mathcal{D}[\mathbf{T}_1]\neq\emptyset$, i.e. the corresponding domains have a non-empty intersection. Furthermore, these trees only differ by an elementary manipulation of the kind already considered in fig. \ref{fig:associativity 3}. Therefore we can use the result from the three point consistency condition and conclude that the corresponding expressions $f_{\mathbf{T}_0}$ and $f_{\mathbf{T}_1}$ coincide on the open domain $\mathcal{D}[\mathbf{T}_0]\cap\mathcal{D}[\mathbf{T}_1]$ and hence everywhere on $M_n$ because of analyticity. We repeat this procedure until we arrive at the tree $\mathbf{T}_{r_1}$ on the left of fig. \ref{fig:elementary manipulation 2}, which has the property that the $n$-th leaf is directly connected to the root.

\begin{figure}[h!]
\begin{center}
\begin{tikzpicture}
\tikzstyle{level 1}=[sibling distance = 1cm, level distance = .5cm]
\tikzstyle{level 2}=[sibling distance = 1cm, level distance = .5cm]
\tikzstyle{level 3}=[sibling distance = 1cm, level distance = .5cm]
\tikzstyle{level 4}=[sibling distance = .8cm, level distance = .4cm]
\tikzstyle{level 5}=[sibling distance = .8cm, level distance = .4cm]
\tikzstyle{level 6}=[sibling distance = .8cm, level distance = .4cm]

\node at (0,0) [shape=coordinate] {} child{
	child{child[dashed,level distance=.5cm, sibling distance=1cm]{ child[level distance=.9cm,sibling distance=1.8cm, solid]{ node(triang a1)[shape=coordinate]{} child[level distance=.8cm, sibling distance=.4cm]{node(triang a2)[shape=coordinate]{} } child[level distance=.8cm, sibling distance=.4cm]{node(triang a3)[shape=coordinate]{} } } child[fill=none]{edge from parent [draw=none]} } child[fill=none]{edge from parent [draw=none]}} child[level distance=.5cm]{child[fill=none]{edge from parent [draw=none]} child[dashed, level distance= .7cm, sibling distance=1.4cm ]{child[solid, level distance=.8cm, sibling distance=1.6cm]{  child{ node(triang b1)[shape=coordinate]{} child[solid,level distance=.8cm, sibling distance=.4cm]{ node(triang b2)[shape=coordinate]{} } child[solid,level distance=.8cm, sibling distance=.4cm]{ node(triang b3)[shape=coordinate]{} } } child[solid]{node(triang c1)[shape=coordinate]{} child[solid,level distance=.8cm, sibling distance=.4cm]{ node(triang c2)[shape=coordinate]{} } child[solid,level distance=.8cm, sibling distance=.4cm] { node(triang c3)[shape=coordinate]{} } }} child[solid, level distance=.8cm, sibling distance=1.6cm]{child[solid]{ node(triang d1)[shape=coordinate]{} child[solid,level distance=.8cm, sibling distance=.4cm]{ node(triang d2)[shape=coordinate]{} } child[solid,level distance=.8cm, sibling distance=.4cm]{ node(triang d3)[shape=coordinate]{} } } child[solid,level distance=1.2cm, sibling distance=2.4cm]{node{n-1}}   } }  }}
	child[level distance=3.7cm, sibling distance=7.4cm]{node{n}}
	;
\filldraw[fill opacity=.5] (triang a1) -- (triang a2) -- (triang a3);
\filldraw[fill opacity=.5] (triang b1) -- (triang b2) -- (triang b3);
\filldraw[fill opacity=.5] (triang c1) -- (triang c2) -- (triang c3);
\filldraw[fill opacity=.5] (triang d1) -- (triang d2) -- (triang d3);

\node at (8,0) [shape=coordinate] {} child{
	child{child[dashed,level distance=.5cm, sibling distance=1cm]{ child[level distance=.9cm,sibling distance=1.8cm, solid]{ node(2triang a1)[shape=coordinate]{} child[level distance=.8cm, sibling distance=.4cm]{node(2triang a2)[shape=coordinate]{} } child[level distance=.8cm, sibling distance=.4cm]{node(2triang a3)[shape=coordinate]{} } } child[fill=none]{edge from parent [draw=none]} } child[fill=none]{edge from parent [draw=none]}}child[level distance=.5cm]{child[fill=none]{edge from parent [draw=none]} child[dashed, level distance= .7cm, sibling distance=1.4cm ]{ child[solid, level distance=.4cm, sibling distance=.8cm]{child[level distance=.4cm, sibling distance=.8cm]{  
	child{ node(2triang b1)[shape=coordinate]{} child[level distance=.8cm, sibling distance=.4cm]{ node(2triang b2)[shape=coordinate]{} } child[level distance=.8cm, sibling distance=.4cm]{ node(2triang b3)[shape=coordinate]{} } } child{node(2triang c1)[shape=coordinate]{} child[level distance=.8cm, sibling distance=.4cm]{ node(2triang c2)[shape=coordinate]{} } child[level distance=.8cm, sibling distance=.4cm] { node(2triang c3)[shape=coordinate]{} } }} child[level distance=.8cm, sibling distance=1.6cm]{  
	node(2triang d1)[shape=coordinate]{} child[level distance=.8cm, sibling distance=.4cm]{ node(2triang d2)[shape=coordinate]{} } child[level distance=.8cm, sibling distance=.4cm]{ node(2triang d3)[shape=coordinate]{} }
	  }  } child[solid,level distance=2cm, sibling distance=4cm]{node{n-1}}   } } }
	  child[level distance=3.7cm, sibling distance=7.4cm]{node{n}}
	;
\filldraw[fill opacity=.5] (2triang a1) -- (2triang a2) -- (2triang a3);
\filldraw[fill opacity=.5] (2triang b1) -- (2triang b2) -- (2triang b3);
\filldraw[fill opacity=.5] (2triang c1) -- (2triang c2) -- (2triang c3);
\filldraw[fill opacity=.5] (2triang d1) -- (2triang d2) -- (2triang d3);
\draw[<->, very thick] (3.5,-1.5) -- (4.5,-1.5);

\end{tikzpicture}
\end{center}
\caption{Another elementary manipulation.}
\label{fig:elementary manipulation 2}
\end{figure}
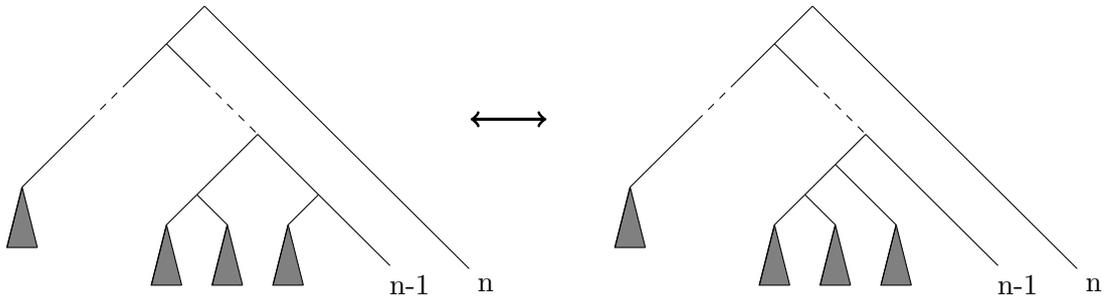
The transformation of trees we used so far is no longer applicable on this tree, so we now perform a manipulation as depicted in fig. \ref{fig:elementary manipulation 2}. Again it is easy to convince oneself that the desired relation holds between these trees. This process can be repeated until we reach the tree $\mathbf{T}_{r_2}$ drawn in fig. \ref{fig:Tr2}. By now it is clear how to proceed the iteration until we reach the desired tree $\mathbf{S}$ in fig. \ref{fig:reference tree S}, thus finishing the proof.

\begin{figure}[h!]
\begin{center}
\begin{tikzpicture}
 \coordinate
	child[level distance=.5cm, sibling distance=1cm]{child[level distance=1cm,sibling distance=2cm]{node(triang1)[shape=coordinate]{}  child[level distance=2cm,sibling distance=4cm]{node(triang2)[shape=coordinate]{}}  child[level distance=2cm,sibling distance=4cm]{node(triang3)[shape=coordinate]{}} }  child[level distance=3cm, sibling distance=6cm]{node{n-1}}   }
	child[level distance=3.5cm, sibling distance=7cm]{node{n}}
	;
	\filldraw[fill opacity=.5] (triang1) -- (triang2) -- (triang3);
\end{tikzpicture}
\end{center}
\caption{The tree $\mathbf{T}_{r_2}$}
\label{fig:Tr2}
\end{figure}
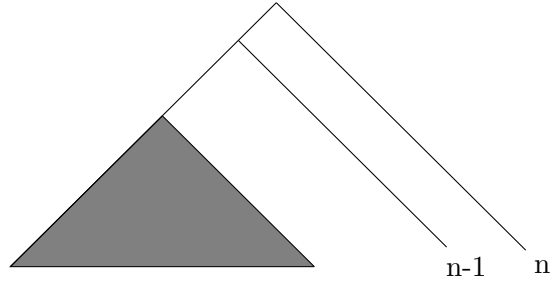
As a result, we obtain the following theorem:

\begin{thm}[Coherence Theorem]\label{thm:coherence theorem}
 {\ \\} Let $f_{\mathbf{T}}$ be defined as in eq. \ref{eq:framework coherence f} on the domain $\mathcal{D}[\mathbf{T}]$ for any binary tree $\mathbf{T}$ as a convergent power series expansion and assume that $f_{\mathbf{T}}$ has an analytic extension to all of $M_n$. Furthermore, assume that the associativity condition \ref{eq:framework coherence 3 point constraint} and the symmetry and normalization conditions, eqs. \ref{eq:framework coherence symmetry axiom} and \ref{eq:framework coherence relation}, hold, i.e. suppose that all $f_{\mathbf{T}}$ coincide for trees with three leaves. Then $f_{\mathbf{T}}=f_{\mathbf{S}}$ for any pair of binary trees $\mathbf{T},\mathbf{S}$.
\end{thm}

\section{General theory of perturbations}\label{subsec:general perturbations}

The concept of perturbations of a quantum field theory is essential in the extraction of explicit measurable predictions from the theory. Thus, we would like to implement this notion in our framework as well. This section will be concerned with the description of perturbation theory in the new framework. According to our definition, a perturbed quantum field theory should correspond to a perturbation series in some parameter $\lambda$ for the OPE coefficients. As these coefficients are required to satisfy the constraints given in section \ref{subsec:axioms}, the perturbations of the coefficients will also have to satisfy corresponding constraints. It will be shown in this section that these constraints are of a cohomological nature.

We have seen in the previous section that, up to technicalities related to the convergence of various series, the constraints on the OPE coefficients imposed by the factorization condition, axiom \ref{ax:Factorization}, can be formulated as a single ``associativity condition'' on the 2-point coefficients only, see eq.\ref{eq:framework coherence 3 point constraint}. The perturbed 2-point OPE coefficients will have to satisfy a perturbed version of this constraint, which turns out to be essentially the only constraint. We now want to study this perturbed version of the associativity condition.

The analogy between our framework (in particular the factorization condition) and ordinary algebra has already been emphasized in the previous section. We now carry this analogy a bit further, as the following discussion is closely analogous to the well-known characterization of perturbations, or in this context \emph{deformations}, of an ordinary finite dimensional algebra. Let us therefore recall the basic theory of deformations of finite dimensional algebras (see \cite{happel1989hcf, Gerstenhaber:1963zz}). Let $\mathbf{A}$ be a finite dimensional algebra over $\mathbb{C}$, whose product is denoted as usual by $\mathbf{A}\otimes\mathbf{A}\to\mathbf{A}$, $A\otimes B\mapsto AB$ for all $A,B\in \mathbf{A}$. Then a 1-parameter family of products $A\otimes B\mapsto A\bullet_{\lambda}B$, where $\lambda\in\mathbb{R}$\index{symbols}{lambda@$\lambda$} is a smooth deformation parameter, is called a deformation. We define the product $A\bullet_0 B$ to be the original product $AB$, but for non-zero $\lambda$ we obtain a new product on $\mathbf{A}$, or alternatively on the ring of formal power series $\mathbb{C}(\lambda)\otimes \mathbf{A}$ if we merely consider perturbations in the sense of formal power series. As argued above, this new product has to satisfy the strong constraints imposed by the associativity condition. Denoting the $i$-th order perturbation of the product by

\begin{equation}
 m_i(A,B)=\frac{1}{i!}\tdd{^i}{\lambda^i} A\bullet_{\lambda}B\Big|_{\lambda=0}
\label{eq:framework perturbation mi}
\end{equation}
the associativity law yields to first order

\begin{equation}
 m_0(id\otimes m_1)-m_0(m_1\otimes id)+m_1(id\otimes m_0)-m_1(m_0\otimes id)=0
\label{eq:framework perturbation associativity 1st}
\end{equation}
as a map $\mathbf{A}\otimes\mathbf{A}\otimes\mathbf{A}\to\mathbf{A}$, in an obvious tensor product notation. Similarly, one obtains conditions for higher derivatives $m_i$ of the new product, which for $i\geq 2$ are of the form

\begin{equation}
 m_0(id\otimes m_i)-m_0(m_i\otimes id)+m_i(id\otimes m_0)-m_i(m_0\otimes id)=-\sum_{j=1}^{i-1}m_{i-j}(id\otimes m_j)-m_{i-j}(m_j\otimes id) \: .
\label{eq:framework perturbation associativity ith}
\end{equation}
In this discussion we want to exclude the trivial case, i.e. a simple $\lambda$ dependent redefinition of the generators of $\mathbf{A}$. Such a redefinition may be expressed in terms of a 1-parameter family of invertible maps $\alpha_{\lambda}:\mathbf{A}\to\mathbf{A}$, such that the corresponding trivially deformed product can be written as

\begin{equation}
 A\bullet_{\lambda}B=\alpha_{\lambda}^{-1}\left[\alpha_{\lambda}(A)\alpha_{\lambda}(B)\right]\: .
\label{eq:framework perturbation trivial deformation}
\end{equation}
So $\alpha_{\lambda}$ can be viewed as an isomorphism between $(\mathbf{A},\bullet_0)$ and $(\mathbf{A},\bullet_{\lambda})$, which suggests that the latter should not be regarded as a new algebra. To first order, the trivially deformed product is given by

\begin{equation}
 m_1=m_0(id\otimes \alpha_1)-m_0(\alpha_1\otimes id)-\alpha_1m_0
\label{eq:framework perturbation trivial deformation 1st}
\end{equation}
where $\alpha_i=\frac{1}{i!}\tdd{^i}{\lambda^i}\alpha_{\lambda}\Big|_{\lambda=0}$. Similar formulas hold for higher orders.

We now want to give a more elegant formulation and interpretation of our conditions for the $i$-th order deformations of the associative product, eq. \ref{eq:framework perturbation associativity ith}, using the language of cohomology theory \cite{Gerstenhaber:1963zz}. For this purpose, we introduce the linear space $\Omega^n(\mathbf{A})$ of all linear maps $\psi_n:\mathbf{A}\otimes\ldots\otimes\mathbf{A}\to\mathbf{A}$ and the linear operator $d:\Omega^n\to\Omega^{n+1}$ defined by

\begin{equation}
\begin{split}
 (d\psi_n)(A_1,\ldots,A_{n+1})=&A_1\psi_n(A_2,\ldots,A_{n+1})-(-1)^n\psi_n(A_1,\ldots,A_n)A_{n+1}\\
 +&\sum_{j=1}^n(-1)^j\psi_n(A_1,\ldots,A_jA_{j+1},\ldots,A_{n+1})\: .
\label{eq:framework perturbation differential d}
\end{split}
\end{equation}
Using this definition and the associativity law for the original product on the algebra $\mathbf{A}$ one can show that $d^2=0$, i.e. $d$ is a differential with a corresponding cohomology complex, the so called \emph{Hochschild complex} (see e.g. \cite{connes1985ncd}). More precisely, let $Z^n(\mathbf{A})$ be the space of all closed maps $\psi_n$, i.e. those satisfying $d\psi_n=0$, and $B^n(\mathbf{A})$ the space of all exact $\psi_n$, that means those for which $\psi_n=d\psi_{n-1}$ for some $\psi_{n-1}$. Then the $n$-th Hochschild cohomology $HH^n(\mathbf{A})$ is defined to be the quotient $Z^n(\mathbf{A})/B^n(\mathbf{A})$. In this language, we may now identify the first order associativity condition, eq. \ref{eq:framework perturbation associativity 1st}, with the statement $dm_1=0$, or equivalently $m_1\in Z^2(\mathbf{A})$. In addition, if the new product arises from just a trivial redefinition, as in eq. \ref{eq:framework perturbation trivial deformation}, then it follows that $m_1=d\alpha_1$, which means $m_1\in B^2(\mathbf{A})$. So indeed one finds that the non-trivial first order perturbations $m_1$ of the algebra product correspond to the non-trivial classes $[m_1]\in HH^2(\mathbf{A})$. Hence, $HH^2(\mathbf{A})\neq 0$ is necessary for the existence of non-trivial deformations. We now want to continue our analysis at higher orders. For this purpose, let us assume a non-trivial first order deformation to exist and let us study the second order deformations. Thus, we consider eq.\ref{eq:framework perturbation associativity ith} for $i=2$ and start with the right side of this equation, which can be viewed as a map $\omega_2\in\Omega^3(\mathbf{A})$. Computation shows that $d\omega_2=0$, so $\omega_2\in Z^3(\mathbf{A})$. The left side of equation \ref{eq:framework perturbation associativity ith} for $i=2$ turns out to be just $dm_2\in B^3(\mathbf{A})$. Thus, if the second order associativity condition is to hold, we must have $\omega_2=dm_2\in B^3(\mathbf{A})$, or in other words, the class $[\omega_2]\in HH^3(\mathbf{A})$ must vanish. It follows that there is an obstruction to lift the perturbation to second order in the case of non-trivial deformations, i.e. for $HH^3(\mathbf{A})\neq 0$. We can analogously continue to third order, obtaining the corresponding potential obstruction that $[\omega_3]\in HH^3(\mathbf{A})$ vanishes, and so on. In summary, the space of non-trivial perturbations corresponds to elements of $HH^2(\mathbf{A})$, while the obstructions lie in $HH^3(\mathbf{A})$.

Let us now conclude this example and try to carry over the concepts we just used to the case of perturbations of a quantum field theory in our framework. A short reminder: A quantum field theory as defined in section \ref{subsec:axioms} is given by the pair $(V,\mathcal{C})$ of a vector space $V$, as defined in eq. \ref{eq:framework axioms V}, and a hierarchy of OPE coefficients $\mathcal{C}$ with certain properties. In section \ref{subsec:coherence theorem} we argued that all higher $n$-point coefficients are uniquely determined by the 2-point coefficients $\mathcal{C}(x_1,x_2)$. Furthermore, we were able to show that, up to technical assumptions concerning the convergence of the series in eq. \ref{eq:framework coherence f}, the key constraints on the $n$-point OPE coefficients are encoded in the associativity condition, eq. \ref{eq:framework coherence 3 point constraint} for the 2-point coefficient, which we repeat for convenience:

\begin{equation}
 \mathcal{C}(x_2,x_3)\Big(\mathcal{C}(x_1,x_2)\otimes id\Big)-\mathcal{C}(x_1,x_3)\Big(id\otimes\mathcal{C}(x_2,x_3)\Big)=0\hspace{1cm} \text{for }r_{12}<r_{23}<r_{13}
\label{eq:framework perturbation 3 point constraint}
\end{equation}
We want to study the following problem: When is it possible to find a 1-parameter deformation $\mathcal{C}(x_1,x_2;\lambda)$ of the OPE coefficients which again satisfies the associativity condition, at least in the sense of formal power series in the deformation parameter $\lambda$? In fact, the symmetry condition \ref{eq:framework coherence symmetry axiom}, the normalization condition \ref{eq:framework coherence relation} and the axioms from section \ref{subsec:axioms}, except for axiom \ref{ax:Factorization}, should hold for the perturbation as well. However, as these conditions are linear in $\mathcal{C}(x_1,x_2)$, they are much more trivial in nature than eq. \ref{eq:framework perturbation 3 point constraint}. Therefore, for the rest of this section, we will not include these conditions in our discussion, but instead continue with the main point, i.e. the implications of the associativity condition \ref{eq:framework perturbation 3 point constraint} for the perturbed OPE coefficients.

In analogy to the example from ordinary algebra, we will again find a characterization of perturbations in a cohomological framework. We now want to define a linear operator $b$, which defines the cohomology in question and therefore plays the role of the $d$ in our example. However, because the definition of this operator will involve infinite sums (just as eq.\ref{eq:framework perturbation 3 point constraint}) and as such sums are typically only convergent on certain domains, we have to specify a set of domains that will be stable under the action of $b$ and is suitable for our application. The choice of a set of this kind is by far not unique, and different choices will yield different rings. For simplicity and definiteness, we chose the non-empty, open domains of $(\mathbb{R}^D)^n$ defined by

\begin{equation}
 \mathcal{F}=\{(x_1,\ldots,x_n)\in M_n;\: r_{1\, (i-1)}<r_{(i-1)\, i}<r_{(i-2)\, i}<\ldots <r_{1\, i},\: 1<i\leq n\}\subset M_n\: .
\label{eq:framework perturbation domains F}
\end{equation}
It is possible to express these domains in terms of the $\mathcal{D}[\mathbf{T}]$ defined in eq.\eqref{eq:framework coherence tree domain}, but this will not be necessary. Note also that the associativity condition \ref{eq:framework perturbation 3 point constraint} holds on the domain $\mathcal{F}_3\{r_{12}<r_{23}<r_{13}\}$.

We also need some objects for $b$ to act upon. Therefore we define $\Omega^n(V)$ to be the set of all holomorphic functions $f_n$ on the domain $\mathcal{F}_n$ that are valued in the linear maps

\begin{equation}
 f_n(x_1,\ldots,x_n):V\otimes\cdots\otimes V\to V, \hspace{2cm} (x_1,\ldots,x_n)\in\mathcal{F}_n\: .
\label{eq:framework perturbation maps f}
\end{equation}
Now we are ready to introduce the boundary operator $b:\Omega^n(V)\to\Omega^{n+1}(V)$ by the formula

\begin{equation}
\begin{split}
 (bf_n)&(x_1,\ldots,x_{n+1}):=\mathcal{C}(x_1,x_{n+1})(id\otimes f_n(x_2,\ldots,x_{n+1}))\\
  &+\sum_{i=1}^n (-1)^i f_n(x_1,\ldots,\widehat{x}_i,\ldots,x_{n+1})(id^{i-1}\otimes\mathcal{C}(x_i,x_{i+1})\otimes id^{n-i})\\
  &+(-1)^{n+1}\mathcal{C}(x_n,x_{n+1})(f_n(x_1,\ldots,x_n)\otimes id)\: ,
\end{split}
\label{eq:framework perturbation operator b}
\end{equation}
where $\mathcal{C}(x_1,x_2)$ is the undeformed OPE coefficient and a hat $\,\widehat{\,}\,$ means omission. Note that this definition involves a composition of $\mathcal{C}$ with $f_n$, which, when expressed in a basis of $V$, implicitly involves an infinite summation over these basis elements. It is therefore necessary to assume here (and also in the following for similar formulas) that these sums converge on the set of points $(x_1,\ldots,x_{n+1})\in \mathcal{F}_{n+1}$. Thus, whenever we write $bf_n$, it is understood that $f_n\in\Omega^n(V)$ is in the domain of $b$. We now need the following lemma:

\begin{lemma}\label{lemma:differential b}
The map $b$ is a differential, i.e. $b^2f_n=0$ for $f_n$ in the domain of $b$, such that $bf_n$ is also in the domain of $b$.
\end{lemma}
The corresponding proof is essentially straightforward computation and was given in \cite{Hollands2008}, so it will not be repeated here. With the help of this lemma, we can define a cohomology ring associated to the differential $b$ as

\begin{equation}
 H^n(V;\mathcal{C}):=\frac{Z^n(V;\mathcal{C})}{B^n(V;\mathcal{C})}=\frac{\{\ker b:\Omega^n(V)\to\Omega^{n+1}(V)\}}{\{\operatorname{ran} b:\Omega^{n-1}(V)\to\Omega^n(V)\}}\quad .
\label{eq:framework perturbation cohomology ring}
\end{equation}
Now that we have introduced the necessary concepts from cohomology theory into our framework, we will, as in the case of our example from ordinary algebra, be able to find an elegant and compact formulation of the problem to find a 1-parameter family of perturbations $\mathcal{C}(x_1,x_2;\lambda)$ such that our associativity condition, eq.\eqref{eq:framework perturbation 3 point constraint}, continues to hold to all orders in $\lambda$. Introducing the grading of the 2-point OPE coefficients with respect to the perturbation order by\index{symbols}{Ci@$\mathcal{C}_i(x_1,x_2)$}

\begin{equation}
 \mathcal{C}_i(x_1,x_2)=\frac{1}{i!}\tdd{^i}{\lambda^i}\mathcal{C}(x_1,x_2;\lambda)\Big|_{\lambda=0}\quad ,
\label{eq:framework perturbation grading}
\end{equation}
we note that the first order associativity condition

\begin{equation}
 \begin{split}
  &\mathcal{C}_0(x_2,x_3)\Big(\mathcal{C}_1(x_1,x_2)\otimes id\Big)-\mathcal{C}_0(x_1,x_3)\Big(id \otimes \mathcal{C}_1(x_2,x_3)\Big)\\
 +\,&\mathcal{C}_1(x_2,x_3)\Big(\mathcal{C}_0(x_1,x_2)\otimes id\Big)-\mathcal{C}_1(x_1,x_3)\Big(id \otimes \mathcal{C}_0(x_2,x_3)\Big)=0\, ,
 \end{split}
\label{eq:framework perturbation associativity C 1st}
\end{equation}
which is valid for $(x_1,x_2,x_3)\in\mathcal{F}_3$, can equivalently be stated as

\begin{equation}
 b\, \mathcal{C}_1=0\quad .
\label{eq:framework perturbation associativity C 1st b}
\end{equation}
Here and in the following $b$ is defined in terms of the unperturbed OPE coefficient $\mathcal{C}_0$. We conclude that $\mathcal{C}_1$ has to be an element of $Z^2(V;\mathcal{C}_0)$. Let us consider coefficients $\mathcal{C}(x_1,x_2)$ and $\mathcal{C}(x_1,x_2;\lambda)$ connected by a $\lambda$-dependent field redefinition $z(\lambda):V\to V$ in the sense of defn.\ref{def:Equivalence}. To first order, this implies

\begin{equation}
 \mathcal{C}_1(x_1,x_2)=-z_1\mathcal{C}_0(x_1,x_2)+\mathcal{C}_0(x_1,x_2)(z_1\otimes id+id\otimes z_1)\, ,
\label{eq:framework perturbation redefinition}
\end{equation}
which is equivalent to $bz_1=\mathcal{C}_1$, where again $z_i=\frac{1}{i!}\tdd{^i}{\lambda^i}z(\lambda)|_{\lambda=0}$. Thus, again in analogy to our example from the beginning of this section, the first order deformations of $\mathcal{C}_0$ modulo the trivial ones are given by the classes in $H^2(V;\mathcal{C}_0)$. In order to generalize this result to arbitrary order in $\lambda$, we assume all perturbations up to order $i-1$ to exist and state the associativity condition for the $i$-th order perturbation as the following condition for $(x_1,x_2,x_3)\in\mathcal{F}_3$:

\begin{equation}
 \begin{split}
  &\mathcal{C}_0(x_2,x_3)\Big(\mathcal{C}_i(x_1,x_2)\otimes id\Big)-\mathcal{C}_0(x_1,x_3)\Big(id \otimes \mathcal{C}_i(x_2,x_3)\Big)\\
 +\,&\mathcal{C}_i(x_2,x_3)\Big(\mathcal{C}_0(x_1,x_2)\otimes id\Big)-\mathcal{C}_i(x_1,x_3)\Big(id \otimes \mathcal{C}_0(x_2,x_3)\Big)=\omega_i(x_1,x_2,x_3)\, ,
 \end{split}
\label{eq:framework perturbation associativity C ith}
\end{equation}
where $\omega_i\in\Omega^3(V)$ is defined by

\begin{equation}
 \omega_i(x_1,x_2,x_3):=-\sum_{j=1}^{i-1}\mathcal{C}_{i-j}(x_1,x_3)\Big(id\otimes\mathcal{C}_j(x_2,x_3)\Big)-\mathcal{C}_{i-j}(x_2,x_3)\Big(\mathcal{C}_j(x_1,x_2)\otimes id\Big)\: .
\label{eq:framework perturbation omega}
\end{equation}
At this stage we again encounter infinite sums when a basis of $V$ is introduced into the above equation. We assume these sums to converge on $\mathcal{F}_3$ as well. Then eq.\eqref{eq:framework perturbation associativity C ith} can also be put into the elegant form

\begin{equation}
 b\mathcal{C}_i=\omega_i\: .
\label{eq:framework perturbation associativity C ith b}
\end{equation}
A solution to this equation defines the $i$-th order perturbation. It is obvious that a necessary condition on such a solution is $b\omega_i=0$, or in other words $\omega\in Z^3(V;\mathcal{C}_0)$. In \cite{Hollands2008} it has been shown by the following lemma that this is indeed the case.

\begin{lemma}\label{lemma:w in Z}
 If $\omega_i$ is in the domain of $b$, and if $b\mathcal{C}_j=\omega_j$ for all $j<i$, then $b\omega_i=0$.
\end{lemma}
Again, we do not repeat the proof of this lemma here, but refer the reader to \cite{Hollands2008}. Now if a solution to eq.\eqref{eq:framework perturbation associativity C ith b} exists, i.e. if $\omega_i\in B^3(V;\mathcal{C}_0)$, then any other solution will differ from this one by a solution to the corresponding \emph{homogeneous} equation. A trivial solution to the homogeneous equation of the form $bz_i$ again corresponds to an $i$-th order field redefinition and is not counted as a genuine perturbation.

We conclude with a summary of our findings in this section: We have found that the perturbation series can be continued at $i$-th order, if $[\omega_i]$ is the trivial class in $H^3(V;\mathcal{C}_0)$, which is a potential obstruction. In the case where this imposes no obstruction, the space of non-trivial $i$-th order perturbations is given by $H^2(V;\mathcal{C}_0)$. In particular, if we knew $H^2(V;\mathcal{C}_0)\neq 0$ and $H^3(V;\mathcal{C}_0)=0$, then perturbations could be defined to arbitrary order in $\lambda$.

\section{The fundamental left representation}\label{subsec:FLR}

In the previous sections we have often used examples from ordinary algebra in order to motivate concepts in our framework. In the following we will introduce another such parallel, namely a construction in our framework which has some features in common with a \emph{representation} of an algebra.

In order to motivate certain aspects of our approach, we sometimes wrote formal relations like

\begin{equation}
 "\phi_a(x_1)\phi_b(x_2)=\C{c}{ab}(x_1,x_2)\phi_c(x_2)"\quad ,
\label{eq:framework FLR formal}
\end{equation}
where summation over repeated indices is understood. It is, however, important to note that these relations were only heuristic in the sense that none of our required properties of the OPE coefficients relied on the existence or properties of the hypothetical operators $\phi_a$, which merely served as "dummy variables". Our viewpoint here is again similar to the standard viewpoint taken in algebra, where an abstract algebra $\mathbf{A}$ is entirely defined in terms of its product, i.e. a linear map $m:\mathbf{A}\otimes\mathbf{A}\to\mathbf{A}$. But, as in our case, the algebra elements need not be represented a priori by linear operators on a vector space. Instead one is free to chose any representation, i.e. a linear map $\pi:\mathbf{A}\to\operatorname{End} (H)$, which preserves the product structure, $\pi[m(A,B)]=\pi(A)\pi(B)$ for all $A,B\in\mathbf{A}$. By this line of thought it seems natural to look for a construction similar to a representation in our context. As we will argue in the following, it is indeed possible to find a "canonical" construction of this kind, which will be called the \emph{fundamental left representation}. This leads us to the following definition:\index{symbols}{Y@$\mathcal{Y}(\ket{v_a}, x)$}

\begin{definition}[The fundamental left representation]
We define ``vertex operators'' $\mathcal{Y}(\ket{v}, x):V\to V$, also referred to as ``left representatives'' in subsequent chapters, by the formula 

\begin{equation}
\mathcal{Y}(\ket{v}, x)\ket{w}=\mathcal{C}(x,0)\, \Big(\ket{v}\otimes \ket{w}\Big)\quad. 
\end{equation}
for any two vectors $\ket{v},\ket{w}\in V$ and $x\neq 0$. Choosing a basis $\{\ket{v_a}\}$ of our vector space $V$, the matrix components of this map are then given by

\begin{equation}
 [\mathcal{Y}(\ket{v_a}, x)]^c_b:=\bra{v_c}\mathcal{Y}(\ket{v_a}, x)\ket{v_b}=\C{c}{ab}(x,0)\quad .
\label{eq:framework FLR matrix element}
\end{equation}
This notion will be referred to as fundamental left (or vertex algebra) representation.

\end{definition}
Note that axiom \ref{ax:Derivations} on the OPE coefficients implies

\begin{equation}
 \mathcal{Y}(\ket{\del_{\mu}v},x)=\del_{\mu}\mathcal{Y}(\ket{v},x)\quad,
\label{eq:framework FLR derivatives}
\end{equation}
where on the right side $\del_{\mu}=\del_{x_\mu}$ denotes the usual partial derivative with respect to $x_{\mu}$. Further, by the associativity condition on the OPE coefficients, eq. \ref{eq:framework perturbation 3 point constraint}, one can deduce

\begin{equation}
 \mathcal{Y}(\ket{v_a}, x)\mathcal{Y}(\ket{v_b}, y)=\sum_c\C{c}{ab}(x,y)\mathcal{Y}(\ket{v_c}, y)\quad , \quad \text{for }0<|x-y|<|y|<|x|\, ,
\label{eq:framework FLR consistency}
\end{equation}
which implies that the linear operators $\mathcal{Y}(\ket{v_a},x)$ satisfy the operator product expansion. Thus, we may formally view them as forming a "representation" of the heuristic field operators, i.e. formally "$\pi[\phi_a(x)]=\mathcal{Y}(\ket{v_a}, x)$" is a "representation" of the algebra defined by the OPE coefficients. Note also that eq.\eqref{eq:framework FLR consistency} may equivalently be written as

\begin{equation}
 \mathcal{Y}(\ket{v_a}, x)\mathcal{Y}(\ket{v_b}, y)=\mathcal{Y}(\mathcal{Y}(\ket{v_a}, x-y)\, \ket{v_b}, y)
\label{eq:framework FLR borcherds identity}
\end{equation}
on the same domain as above. This is a standard identity in the theory of vertex operator algebras \cite{Borcherds1986,Kac1997, FrenkelLepowskyMeurman198903, Nikolov2004}. The relation between the approach to quantum field theory as outlined in this chapter and the mathematical theory of vertex operator algebras will be pursued further in \cite{Hollandsa}.

\chapter{The model}\label{sec:the model}

Whereas chapter \ref{sec:OFT in terms of consistency} was concerned with the definition and general features of our new approach, we will in this chapter give an explicit construction of a model theory within this framework. Our model is a massless, scalar Lagrangian theory on 3-dimensional Euclidean space and the construction will be up to low orders in perturbation theory.

We will proceed as follows: Our starting point will be the construction of a scalar, massless, non-interacting quantum field theory in arbitrary dimension $D\geq 3$, which is presented in section \ref{subsec:free field}. In this context, we will make use of the fundamental left representation (see section \ref{subsec:FLR}) in order to introduce a formulation of the theory in terms of the familiar concept of \emph{creation} and \emph{annihilation} operators on a Fock space. This convenient formulation was developed in joint work by Hollands and Olbermann \cite{Hollands2008}. Then, in section \ref{subsec:perturbation via field equation}, an algorithm for the iterative construction of perturbations of an interacting Lagrangian quantum field theory is developed (also first presented in \cite{Hollands2008}). Following this scheme, low order calculations for a specific 3-dimensional toy model theory are carried out in section \ref{subsec:low orders}. Finally, in section \ref{subsec:higher order} we give some results on OPE coefficients at arbitrary perturbation order, which follow from patterns emerging in the mentioned iterative scheme. These last two sections constitiute the main results of this thesis.

\section{The free massless field}\label{subsec:free field}

The aim of this section is to construct a quantum field theory in the sense outlined in chapter \ref{sec:OFT in terms of consistency} for the ``simplest possible example'', namely for a free, massless scalar field on $D$-dimensional Euclidean space, which is classically described by a Lagrangian

\begin{equation}
 \mathcal{L}_{free}(\varphi,\del_{\mu}\varphi)=-\frac{1}{\sigma_D}\int d^Dy\, \del_{\mu}\varphi(y)\del^{\mu}\varphi(y)\quad,
\label{eq:model free field lagrangian}
\end{equation}\index{symbols}{Lagrangianfree@$\mathcal{L}_{free}$}%
which leads to the field equation

\begin{equation}
 -\frac{1}{\sigma_D}\square\varphi=\tilde{\square}\varphi =0\: ,
\label{eq:model free field field equation}
\end{equation}\index{symbols}{LaplaceTilde@$\tilde{\square}$}
with $\square =\delta^{\mu\nu}\del_{\mu}\del_{\nu}$\index{symbols}{Laplace@$\square$} and where $\sigma_D=\frac{2\pi^{D/2}}{\Gamma(D/2)}$\index{symbols}{sigmaD@$\sigma_D$} is the surface area of the $D$-dimensional unit sphere. Here the prefactor $-\frac{1}{\sigma_D}$ is chosen for later convenience, since it leads to the particularly simple form

\begin{equation}
 G_F(x)=r^{2-D}
\label{eq:model free field propagator}
\end{equation}\index{symbols}{GF@$G_F(x)$}
for the Green's function of the operator $\tilde{\square}$, where $r=|x|$. In our framework, construction of a corresponding quantum field theory means that we are to find the OPE coefficients $\mathcal{C}(x_1,x_2)$ satisfying the axioms of section \ref{subsec:axioms} for this model. Additionally, also according to section \ref{subsec:axioms}, we need to define a vector space $V$ as characterized in that section. We chose this vector space, assuming $D\geq 3$ for convenience, according to the following definition:

\begin{definition}[Vector space $V$]\index{symbols}{V@$V$}
 {\ \\} Let $V$ be the unital, commutative $\mathbb{C}$-module generated as a module (i.e. under addition, multiplication and scalar multiplication) by formal expressions of the form $\del_{\{\mu_1}\ldots\del_{\mu_N\}}\varphi$ and unit $\mathds{1}\index{symbols}{1@$\mathds{1}$}$, where $\mu_i\in\{1,\ldots,D\}$ and curly brackets denote the totally symmetric, trace-free part, i.e. by definition

\begin{equation}
 \delta^{\mu_i\mu_j}\del_{\{\mu_1}\ldots\del_{\mu_N\}}\varphi=0\: .
\label{eq:model free field symmetric traceless}
\end{equation}

\end{definition}
The trace free condition has been imposed because any trace would give rise to an expression containing $\square\varphi$, which should vanish in order to satisfy the field equation on the level of $V$. The next step is to find a basis of $V$ which is most convenient for our purpose. As an intermediate step, let us first consider a basis of $\mathbb{R}^D$ in terms of totally symmetric, trace free, rank-$l$ tensors. For given $l\geq 0$, the dimension of this space is $N(l,D)$, where\index{symbols}{NlD@$N(l,D)$}

\begin{equation}
 N(l,D)=\left\{\begin{array}{clcl} &1  &\text{for }l=0 \\ 
 &\frac{(2l+D-2)(l+D-3)!}{(D-2)!l!}\: &\text{for }l>0 \end{array}\right.\: ,
\label{eq:model free field dimension basis}
\end{equation}
so for example $N(l,3)=2l+1$ and $N(l,4)=(l+1)^2$. We denote the corresponding basis elements by $(t_{lm})^{\mu_1\ldots\mu_l}$, $m\in\{1,\ldots,N(l,D)\}$, and for convenience require these to be orthonormal with respect to the natural hermitian inner product on $(\mathbb{R}^D)^{\otimes l}$ induced by the Euclidean metric on $\mathbb{R}^D$. Let us first define\index{symbols}{philm@$\varphi_{lm}$}

\begin{equation}
 \varphi_{lm}=F(l)\, t_{lm}\, \del^l\varphi:=F(l)(t_{lm})^{\mu_1\ldots\mu_l}\del_{\mu_1}\ldots\del_{\mu_l}\varphi\: ,
\label{eq:model free field basis philm}
\end{equation}
where summation over repeated spacetime indices $\mu_i$ is understood  and where\index{symbols}{Fl@$F(l)$}

\begin{equation}
 F(l):=\sqrt{\frac{\Gamma(D/2-1)}{2^ll!\Gamma(l+D/2-1)}}\: ,
\label{eq:model free field basis F}
\end{equation}
is a normalization coefficient chosen in a way to later obtain a simple form for the OPE coefficients. Then a basis of $V$ as a $\mathbb{C}$-vector space is given by $\mathds{1}$, together with the elements\index{symbols}{0va@$\ket{v_a}$}

\begin{equation}
 \ket{v_a}:=\prod_{l,m}\frac{1}{(a_{lm}!)^{1/2}}(\varphi_{lm})^{a_{lm}}\: ,
\label{eq:model free field basis elements}
\end{equation}
where $a=\{a_{lm}|\, l\geq 0\:;\: 0<m\leq N(l,D)\}$ is a multi-index of non-negative integers $a_{lm}$, only finitely many of which are non-zero. The canonical dimension of such an element shall be defined as

\begin{equation}
 |a|:=\sum_{l,m}a_{lm}\left(\frac{D-2}{2}+l\right)\: .
\label{eq:model free field basis elements dimension}
\end{equation}\index{symbols}{0a@$\betrag a\betrag$}

One may formally view $V$ as a ``Fock-space``, where $a_{lm}$ is the ''occupation number`` of the ''mode`` labeled by the quantum numbers $l,m$.
That means, we may decompose $V$ into subspaces of different "particle number" (sum of all ''occupation numbers``) , i.e.

\begin{equation}
 V=\bigoplus_{n=0}^{\infty} V_n=\bigoplus_{n=0}^{\infty}\otimes^n V_1 \, 
\label{eq:model diagrams Fock space}
\end{equation}
where $V_n=\otimes^n V_1$ is the ''$n$-particle" subspace. As usual, it is also possible to define creation and annihilation operators, $\mathbf{b}_{lm}^{\dagger}$ and $\mathbf{b}_{lm}$, on the Fock-space, as linear maps $\mathbf{b}_{lm}^{\dagger}:V_n\to V_{n+1}$ and $\mathbf{b}_{lm}:V_n\to V_{n-1}$, whose action on the basis elements in our case is given by\index{symbols}{blm@$\mathbf{b}_{lm}^{\dagger}, \mathbf{b}_{lm}$}

\begin{eqnarray}
 \mathbf{b}_{lm}^{\dagger}\ket{v_a}& := &(a_{lm}+1)^{1/2}\:\ket{v_{a+e_{lm}}} \label{eq:model free field ladder operators1} \\
 \mathbf{b}_{lm}\ket{v_a}& := &(a_{lm})^{1/2}\:\ket{v_{a-e_{lm}}}\: ,
\label{eq:model free field ladder operators2}
\end{eqnarray}
where $e_{lm}$\index{symbols}{elm@$e_{lm}$} is the multi-index with unit entry at position $l,m$ and zeros elsewhere. Thus, $V_n$ is generated from the unit element $\mathds{1}$, i.e. the ``vacuum'', by $\operatorname{span} (\mathbf{b}_{l_1 m_1}\ldots \mathbf{b}_{l_n m_n})$. These operators satisfy canonical commutation relations

\begin{eqnarray}
 \left[\mathbf{b}_{lm},\mathbf{b}_{l'm'}^{\dagger}\right]&=&\delta_{ll'}\delta_{mm'}\, id \\ \Big[\mathbf{b}_{lm},\mathbf{b}_{l'm'}\Big]&=&0=\left[\mathbf{b}_{lm}^{\dagger},\mathbf{b}_{l'm'}^{\dagger}\right]
\label{eq:model free field ladder operators commutation relations}
\end{eqnarray}
where $id$ is the identity operator on $V$.

We now want to present the explicit form of the OPE coefficients of the model in terms of the above operators. In order to further simplify the form of the coefficients, we introduce spherical harmonics in $D$-dimensions (see \cite{Muller1966} and appendix \ref{app:spherical symmetries}) and establish an isomorphism between the totally symmetric, trace-free tensors $t_{lm}$ and the mentioned spherical harmonics $Y_{lm}$ by\index{symbols}{tlm@$t_{lm}$}

\begin{equation}
 (t_{lm})^{\mu_1\ldots\mu_l}=c_l\int_{S^{D-1}} d\Omega \: x^{\{\mu_1}\cdots x^{\mu_l\}} Y_{lm}(x)\: ,
\label{eq:model free field isomorphism tensors-harmonics}
\end{equation}
where we integrate over the $D-1$-dimensional unit sphere $S^{D-1}$.
%
The constant $c_l$ can be determined by our requirement of orthonormality of the tensors $t_{lm}$, with the result (see eq.\eqref{eq:app spherical symm arbitrary definition cl}).\index{symbols}{cl@$c_l$}

\begin{equation}
 c_l=\left(\frac{2^l\, \Gamma(l+D/2)}{l!\,\Gamma(D/2)\sigma_D}\right)^{1/2}\: .
\label{eq:model free field normalization cl} 
\end{equation}
With this notation in place, we want to proceed to the actual construction of the OPE coefficients $\mathcal{C}(x_1,x_2)$. For this purpose, it is sufficient to consider the left-representatives $\mathcal{Y}(\ket{v_a},x):V\to V$ for all $\ket{v_a}\in V$, since the matrix elements $[\mathcal{Y}(\ket{v_a},x)]^c_b=\C{c}{ab}(x,0)$ are exactly the OPE coefficients, see section \ref{subsec:FLR}. We will start our investigation with the simplest non-trivial left-representative, $\mathcal{Y}(\varphi,x)$, corresponding to the basic field, which is defined as

\begin{equation}
\begin{split}
 \mathcal{Y}(\varphi,x)=\sqrt{\sigma_D}\; r^{-(D-2)/2}\sum_{l=0}^{\infty}&\sum_{m=1}^{N(l,D)}\left(\frac{D-2}{2l+D-2}\right)^{1/2}\times \\
 &\left[r^{l+(D-2)/2}\, Y^{lm}(\hat{x})\, \mathbf{b}_{lm}^{\dagger}+r^{-l-(D-2)/2}\, Y_{lm}(\hat{x})\, \mathbf{b}_{lm}\right]\: ,
\end{split}
\label{eq:model free field Lvarphi} 
\end{equation}
where $r=|x|$, $\hat{x}=x/|r|$ and $Y^{lm}(\hat{x})=\overline{Y_{lm}(\hat{x})}$.
Notice that this equation has the familiar form of a free field operator with an ''emissive`` and an ''absorptive`` part. However, this is less surprising if one remembers that $\mathcal{Y}(\varphi,x)$ is in a sense the ''representative`` of the (formal) field operator $\varphi(x)$ on $V$. We will now ''derive`` eq.\eqref{eq:model free field Lvarphi} from the standard quantum field theory formalism, i.e. the vectors $\ket{v_a}\in V$ are now really viewed as quantum fields in the usual sense. In order to determine $\mathcal{Y}(\varphi,x)$ let us consider the product of $\varphi$ with an arbitrary element $\ket{v_a}\in V$. Using \emph{Wick's theorem}, this can be written as:

\begin{equation}
\begin{split}
 \varphi(x) \,\ket{v_a(0)}=&\varphi(x)\left(\prod_{l,m}\frac{1}{(a_{lm}!)^{1/2}}(\varphi_{lm})^{a_{lm}}(0)\right)\\
 =&\prod_{l,m}\frac{1}{(a_{lm}!)^{1/2}}\Big( :\varphi(x)(\varphi_{lm})^{a_{lm}}(0):+\text{ all possible contractions } \Big)\\
 =&\prod_{l,m}\frac{1}{(a_{lm}!)^{1/2}}\, :\varphi(x)(\varphi_{lm})^{a_{lm}}(0):\\
 + \sum_{l',m'} a_{l'm'}&F(l')\, t_{l'm'}\, (-1)^{l'}\del^{l'} r^{2-D}(\varphi_{l'm'})^{a_{l'm'}-1}(0)\prod_{\atop{l,m}{(l,m)\neq (l',m')}}\frac{1}{(a_{lm}!)^{1/2}}(\varphi_{lm})^{a_{lm}}(0)
\end{split}
\label{eq:model free field wicks theorem} 
\end{equation}
Here we have simply inserted the explicit form of $\ket{v_a}$ as defined in eq.\eqref{eq:model free field basis elements}, applied Wick's theorem and in the last step used the fact that the propagator in our theory is just $G_F(x)=r^{2-D}$ as stated in eq.\eqref{eq:model free field propagator}. Double dots $:\,\cdot\,:$ here denote the \emph{normal ordered product} of the standard free quantum field theory formalism. Due to the analyticity properties of this normal ordered product, we can perform a Taylor expansion of the corresponding term in eq.\eqref{eq:model free field wicks theorem} in $x$ around $0$, which yields

\begin{equation}
\begin{split}
 \varphi(x) \,\ket{v_a(0)}=&\prod_{l,m}\frac{1}{(a_{lm}!)^{1/2}}\left[\sum_{l'}\frac{1}{l'!}x^{\mu_1}\cdots x^{\mu_{l'}} \,(\del_{\mu_1}\cdots\del_{\mu_{l'}}\varphi)(0)\right](\varphi_{lm})^{a_{lm}}(0)\\
+ \sum_{l',m'} a_{l'm'}&F(l')\, t_{l'm'}\, (-1)^{l'}\del^{l'} r^{2-D}(\varphi_{l'm'})^{a_{l'm'}-1}(0)\prod_{\atop{l,m}{(l,m)\neq (l',m')}}\frac{1}{(a_{lm}!)^{1/2}}(\varphi_{lm})^{a_{lm}}(0)
\end{split}
\label{eq:model free field Taylor expansion} 
\end{equation}
This rather lengthy expression can be simplified by the observation that the term $x^{\mu_1}\cdots x^{\mu_{l'}}=x^{\otimes l}$ may be replaced by its trace-free, anti-symmetric part $x^{\{\mu_1}\cdots x^{\mu_{l'}\}}$ due to the field equation \eqref{eq:model free field field equation}. We may thus use the identity

\begin{equation}
 x^{\{\mu_1}\cdots x^{\mu_{l}\}}=c_l^{-1}\, (t_{lm})^{\mu_1\ldots\mu_l}\, r^l Y^{lm}(\hat{x})\quad ,
\label{eq:model free field simplification 1} 
\end{equation}
which holds as a result of eq.\eqref{eq:model free field isomorphism tensors-harmonics}. Furthermore, we will need the relation

\begin{equation}
 (-1)^{l}\del^l r^{2-D}=c_l^{-1}2^l\,\frac{\Gamma(l+(D-2)/2)}{\Gamma(D-2)/2}\cdot t^{lm}\, Y_{lm}(\hat{x})r^{2-D-l}\: ,
\label{eq:model free field simplification 2} 
\end{equation}
with $t^{lm}=\overline{t_{lm}}$, which is derived in appendix \ref{app:spherical symmetries}. By substitution of eqs.\eqref{eq:model free field simplification 1} and \eqref{eq:model free field simplification 2} into eq.\eqref{eq:model free field Taylor expansion} we obtain

\begin{equation}
 \begin{split}
  \varphi(x) \,\ket{v_a(0)}=\prod_{l,m}\frac{1}{(a_{lm}!)^{1/2}}\left[\sum_{l'}\frac{1}{l'!}c_{l'}^{-1}r^{l'} Y^{l'm'}(\hat{x})t_{l'm'} \,\del^{l'}\varphi(0)\right](\varphi_{lm})^{a_{lm}}(0&)\\
+ \sum_{l',m'} a_{l'm'}F(l')t_{l'm'}\left[c_{l'}^{-1}2^{l'}\,\frac{\Gamma(l'+(D-2)/2)}{\Gamma(D-2)/2}t^{l'm'} Y_{l'm'}(\hat{x})r^{2-D-l'}\right](\varphi_{l'm'})^{a_{l'm'}-1}(0&)\\
 \times\prod_{\atop{l,m}{(l,m)\neq (l',m')}}\frac{1}{(a_{lm}!)^{1/2}}(\varphi_{lm})^{a_{lm}}(0&)
 \end{split}
\label{eq:model free field substitution}
\end{equation}
A closer inspection of this equation yields the following simplifications: Remembering the definition of $\varphi_{lm}(x)$ in eq.\eqref{eq:model free field basis philm}, we can use

\begin{equation}
 t_{l'm'} \,\del^{l'}\varphi(0)=F(l')^{-1}\varphi_{l'm'}(0)
\label{eq:model free field simplification 3} 
\end{equation}
in the first line of equation \eqref{eq:model free field substitution}. Secondly, we can just drop the contraction over the tensors $t_{l'm'}$ in the second line, because in our above construction we chose these tensors to be orthonormal. Finally, if we insert the explicit form of the constants $c_l$ and $F(l)$ into the equation and perform simple algebraic manipulations, we obtain the convenient expression

\begin{equation}
 \begin{split}
  \varphi(x) \,\ket{v_a(0)}=&\sqrt{\sigma_D} r^{-(D-2)/2}\sum_{lm}\left(\frac{D-2}{2l+D-2}\right)^{1/2}\\
 \times &\left[(a_{lm}+1)^{1/2}r^{l+(D-2)/2}Y^{lm}(\hat{x})\ket{v_{a+e_{lm}}}+(a_{lm})^{1/2}r^{-l-(D-2)/2}Y_{lm}(\hat{x})\ket{v_{a-e_{lm}}}\right]\, .
 \end{split}
\label{eq:model free field wick final form}
\end{equation}
From the above equation we can simply read off the desired OPE coefficients.

\begin{equation}
 \begin{split}
  \C{b}{\varphi a}(x,0)=&\sqrt{\sigma_D} r^{-(D-2)/2}\sum_{lm}\left(\frac{D-2}{2l+D-2}\right)^{1/2}\\
 \times &\left[(a_{lm}+1)^{1/2}r^{l+(D-2)/2}Y^{lm}(\hat{x})\delta_{a+e_{lm},b}+(a_{lm})^{1/2}r^{-l-(D-2)/2}Y_{lm}(\hat{x})\delta_{a-e_{lm},b}\right]\, 
 \end{split}
\label{eq:model free field OPE coeffs}
\end{equation}
Now, by definition \ref{def:The fundamental left representation} we have $[\mathcal{Y}(\varphi,x)]^b_a=\C{b}{\varphi a}(x,0)$, so using the ladder operators introduced above we finally obtain

\begin{equation}
\begin{split}
 \mathcal{Y}(\varphi,x)=\sqrt{\sigma_D}\; r^{-(D-2)/2}\sum_{l=0}^{\infty}&\sum_{m=1}^{N(l,D)}\left(\frac{D-2}{2l+D-2}\right)^{1/2}\times \\
 &\left[r^{l+(D-2)/2}\, Y^{lm}(\hat{x})\, \mathbf{b}_{lm}^{\dagger}+r^{-l-(D-2)/2}\, Y_{lm}(\hat{x})\, \mathbf{b}_{lm}\right]\: ,
\end{split}
\label{eq:model free field Lvarphi 2}
\end{equation}
which is just eq.\eqref{eq:model free field Lvarphi}. As stated above, this is only the simplest non-trivial left-representative. The corresponding formula for a general $\ket{v_a}\in V$, as defined in eq.\eqref{eq:model free field basis elements}, is

\begin{equation}
 \mathcal{Y}(\ket{v_a},x)=\: :\prod_{l,m}\frac{1}{(a_{lm}!)^{1/2}}[F(l)t_{lm}\, \del^l\mathcal{Y}(\varphi,x)]^{a_{lm}} :
\label{eq:model free field La}
\end{equation}
Again double dots denote normal ordering, which in our Fock space formulation simply means that all creation operators are to the left of the annihilation operators. This formula can be derived in two ways: One can either proceed in analogy to the simple $\mathcal{Y}(\varphi,x)$ case presented above, or use the factorization condition, axiom \ref{ax:Factorization}, in order to iteratively construct $\mathcal{Y}(\ket{v_a},x)$ out of $\mathcal{Y}(\varphi,x)$. If we use the former approach, we again import information from the standard formulation of quantum field theory. We thus have to check whether the OPE coefficients found this way are compatible with our framework, i.e. we have to check whether the axioms of section \ref{subsec:axioms} hold. This can indeed be done, where most effort again goes into the proof of the consistency condition, eq.\eqref{eq:framework perturbation 3 point constraint}. Here we will neither give this proof, nor the derivation of eq.\eqref{eq:model free field La} by this method. Instead we take the above mentioned alternative road to this equation. In this derivation we use the form of $\mathcal{Y}(\varphi,x)$ in eq.\eqref{eq:model free field Lvarphi 2}, and proceed in our framework, as defined in chapter \ref{sec:OFT in terms of consistency}. So at this stage we impose all our axioms, in particular the factorization axiom and therefore also the associativity condition \eqref{eq:framework perturbation 3 point constraint}, on the further OPE coefficients. In other words, our axioms will not have to be checked afterwards, but are required to hold initially. This iterative construction of $\mathcal{Y}(\ket{v_a},x)$ out of $\mathcal{Y}(\varphi,x)$ will be described in the next section, where it will serve as an instructive example of the procedure presented there. One should mention here that both approaches to the derivation of eq.\eqref{eq:model free field La} work equally well and give the same results.

Let us conclude this section with a brief summary of our results. We have constructed the full quantum field theory, i.e. is the pair $(V,\mathcal{C})$, for the non-interacting model described by the field equation \eqref{eq:model free field field equation} in $D$-dimensions ($D> 2$). The corresponding vector space $V$ is defined in eq.\eqref{eq:model free field basis elements} and the OPE coefficients $\mathcal{C}(x_1,x_2)$ can be obtained from the fundamental left representation, given in eqs.\eqref{eq:model free field Lvarphi 2} and \eqref{eq:model free field La}, as described in definition \ref{def:The fundamental left representation}. Furthermore, by the coherence theorem, or more specifically by proposition \ref{prop:factorization}, the n-point OPE coefficients $\mathcal{C}(x_1,\ldots,x_n)$ are uniquely determined by the 2-point coefficients $\mathcal{C}(x_1,x_2)$.

\section{Perturbations via non-linear field equations}\label{subsec:perturbation via field equation}

Having constructed free quantum field theory (for a massless scalar field), we now want to focus on the more interesting case of a theory with interaction. In section \ref{subsec:general perturbations} we have discussed perturbations in our framework as an analog of deformations of an algebra. This setting was very general, in the sense that it also holds for non-Lagrangian models. In the following, we are going to deal with the special case of theories which have a classical counterpart with a Lagrangian

\begin{equation}
 \mathcal{L}=\mathcal{L}_{free}+\mathcal{L}_{int}
\label{eq:model perturbations lagrangian}
\end{equation}\index{symbols}{Lagrangian@$\mathcal{L}$}\index{symbols}{Lagrangianfree@$\mathcal{L}_{free}$}\index{symbols}{Lagrangianint@$\mathcal{L}_{int}$}%
where $\mathcal{L}_{free}$ is given by eq.\eqref{eq:model free field lagrangian} and with the interaction part

\begin{equation}
 \mathcal{L}_{int}=\frac{-\lambda}{(k+1)\sigma_D}\int d^Dy\, \varphi^{k+1}(y)
\label{eq:model perturbations interaction lagrangian}
\end{equation}\index{symbols}{lambda@$\lambda$}
with\footnote{In fact, we only want to consider renormalizable theories, so $k$ should be chosen appropriately, see also section \ref{subsec:ambiguities}} $2\leq k\in \mathbb{N}$ and $\lambda\in\mathbb{C}$. This choice leads to an equation of motion of the form

\begin{equation}
\tilde{\square}\varphi=-\frac{\lambda}{\sigma_D}\varphi^k\quad\Rightarrow\quad \square\varphi =\lambda \varphi^k
 \label{eq:model perturbations equation of motion}\, .
\end{equation}
In other words, we consider massless, scalar $\varphi^{k+1}$-theory with interaction parameter $\lambda$ on $D$-dimensional Euclidean space. In this setting, the following theorem holds:

\begin{thm}[Perturbations via field equations]
An interacting quantum field theory $(\mathcal{C},V)$ obeying the field equation \eqref{eq:model perturbations equation of motion} can be constructed perturbatively up to arbitrary orders in the coupling constant $\lambda$ from the underlying free theory by an algorithm which relies on the successive application of the associativity condition \eqref{eq:framework perturbation 3 point constraint} and the differential equation \eqref{eq:model perturbations equation of motion}.
\end{thm}
This algorithm will be outlined in the following.

Recall from the previous section that we defined the vector space $V$ to be spanned by trace-free expressions of the form $\del_{\{\mu_1}\ldots\del_{\mu_N\}}\varphi$. This was motivated by the fact that all expressions containing a trace would vanish in the free theory due to the field equation \ref{eq:model free field field equation}. In the present context, i.e. with a non-linear field equation of the form given above, this argument clearly does not hold anymore. Therefore we consider from now on the vector space $\hat{V}$ which is spanned by the unit element and all expressions of the form $\del_{\mu_1}\ldots\del_{\mu_N}\varphi$. Then the vertex operators of the interacting theory are maps from $\hat{V}$ to $\End(\hat{V})$. In the remainder of this thesis we will drop the caret over $\hat{V}$ again to lighten the notation, but it is always understood that from now on also expressions containing traces are allowed.

Let us chose a basis of the vector space $V$ with elements $\ket{v_a}$ as in the previous section, see \eqref{eq:model free field basis elements}. Then we may transfer the field equation to the level of vertex operators by the identity \eqref{eq:framework FLR derivatives}, which yields

\begin{equation}
 \square \mathcal{Y}(\varphi,x)=\lambda \mathcal{Y}(\varphi^k,x)\quad .
\label{eq:model perturbations equation of motion LR}
\end{equation}
On the level of OPE coefficients\footnote{It follows from the Euclidean invariance axiom that $\C{c}{ab}(x,y)=\C{c}{ab}(x-y)$, so $\C{c}{ab}(x,0)=\C{c}{ab}(x)$}, this implies

\begin{equation}
\square  \C{b}{\varphi a}(x) = \square \bra{v_b}\mathcal{Y}(\varphi,x)\ket{v_a}=\lambda\,\bra{v_b} \mathcal{Y}(\varphi^k,x)\, \ket{v_a} =\lambda \C{b}{\varphi^k a}(x)
\label{eq:model perturbations OPE field equation}
\end{equation}
As described in section \ref{subsec:general perturbations}, perturbations in our framework imply a grading of the OPE coefficients $\mathcal{C}$ as in eq.\eqref{eq:framework perturbation grading}. In terms of the basis elements of the OPE coefficients, this grading takes the form\index{symbols}{Ci@$\Co{i}{c}{ab}(x)$}

\begin{equation}
 \Co{i}{c}{ab}(x):=\frac{1}{i!}\tdd{^i}{\lambda^i}\C{c}{ab}(x)\Big|_{\lambda=0}\: ,
\label{eq:model perturbations grading}
\end{equation}
or equivalently for the left representatives\index{symbols}{Yi@$\Lo{i}{\ket{v_a}}{x}$}

\begin{equation}
 \Lo{i}{\ket{v_a}}{x}:= \frac{1}{i!}\tdd{^i}{\lambda^i}\mathcal{Y}(\ket{v_a},x)\Big|_{\lambda=0}\quad.
\label{eq:model perturbations grading LR}
\end{equation}
Substitution of this grading into eq.\eqref{eq:model perturbations OPE field equation} and comparison of the outermost left and right expressions in this equation yields an infinite number of relations

\begin{equation}
\square \Co{i}{b}{\varphi a}(x)=\Co{i-1}{b}{\varphi^k a}(x)\; ,
\label{eq:model perturbations relations}
\end{equation}
at any order $i>0$ in $\lambda$. The perturbation orders of the coefficients in this equation differ by one, as a result of the appearance of the parameter $\lambda$ on the right of eq.\eqref{eq:model perturbations OPE field equation}. It is this fact, which makes eq.\eqref{eq:model perturbations relations} so powerful. At a first glance, it almost seems as if these relations were already enough in order to establish an iterative pattern, which allows for the construction of the perturbed OPE coefficients up to arbitrary order, starting from the free theory. However, it is easy to see that we quickly run into problems. Let us briefly go through the procedure until these obstacles appear: The zeroth-order coefficients are known from section \ref{subsec:free field}. Then we solve the differential equation \eqref{eq:model perturbations relations} obtaining the coefficients $\Co{1}{b}{\varphi a}$.
Obviously, we would now like to apply this equation again and proceed to second order. For this purpose, however, we would need the coefficients $\Co{1}{b}{\varphi^k a}$, which a priori we know nothing about. So already at first order our iteration seems to break down. At this stage we introduce the second ingredient into the construction, namely we assume that the associativity condition, eq.\eqref{eq:framework perturbation 3 point constraint}, holds. Suppose for the moment that all perturbations are known to order $i-1$. Then, according to eq.\eqref{eq:framework perturbation associativity C ith}, the associativity condition on the $i$-th order perturbation can be written as (see also eq.\eqref{eq:framework motivation consistency})

\begin{equation}
 \sum_{j=0}^{i}\sum_e \Co{j}{e}{ab}(x_1,x_2)\Co{i-j}{d}{ec}(x_2,x_3)=\sum_{j=0}^{i}\sum_e \Co{j}{e}{bc}(x_2,x_3)\Co{i-j}{d}{ae}(x_1,x_3)\; ,
\label{eq:model perturbations consistency}
\end{equation}
on the domain $\mathcal{F}_3=\{r_{12}<r_{23}<r_{13}\}$, see eq.\eqref{eq:framework perturbation domains F}. Let us consider the special case $\ket{v_a}=\varphi=\ket{v_b}$.

\begin{equation}
 \sum_{j=0}^{i}\sum_e \Co{j}{e}{\varphi \varphi}(x_1,x_2)\Co{i-j}{d}{ec}(x_2,x_3)=\sum_{j=0}^{i}\sum_e \Co{j}{e}{\varphi c}(x_2,x_3)\Co{i-j}{d}{\varphi e}(x_1,x_3)\; .
\label{eq:model perturbations consistency special}
\end{equation}
Next, we are interested in the limit $x_1\to x_2$. Clearly, the coefficient $\Co{j}{e}{\varphi \varphi}(x_1,x_2)$ on the left side of the above equation will be most dramatically affected by this procedure. By the scaling dimension condition, axiom \ref{ax:Scaling}, we have for this coefficient

\begin{equation}
 sd\, \Co{j}{e}{\varphi \varphi}\leq D-2-|e|\, , 
\label{eq:model perturbations scaling}
\end{equation}
where we also used the fact that $|\varphi|=(D-2)/2$ from eq.\eqref{eq:model free field basis elements dimension}. As all coefficients $\Co{j}{e}{\varphi \varphi}(x_1,x_2)$ with negative scaling degree will vanish in the limit, only few terms, namely those with $|e|\leq |\varphi^2|$, will contribute to the sum on the left side of eq.\eqref{eq:model perturbations consistency special}. 

Let us consider the case $j=0$ in the sum on the left side of eq.\eqref{eq:model perturbations consistency special}. With the results of the previous section it can easily be derived that only one term in the sum over $e$ will give a contribution in this case. This can be seen as follows: The condition $|e|\leq |\varphi^2|$ restricts $e$ to be either $\mathds{1}$, $\varphi_{lm}$ with $l\leq (D-2)/2$ or $\varphi^2$. However, orthonormality of our basis and the form of the free theory left representative $\Lo{0}{\varphi}{x}$ imply that the OPE coefficients $\Co{0}{\varphi_{lm}}{\varphi\,\varphi}=\bra{\varphi_{lm}}\Lo{0}{\varphi}{x}\ket{\varphi}$ vanish for any value of $l$, since a single ladder operator does not suffice to transform the vector $\ket{\varphi}$ into $\ket{\varphi_{lm}}$. Now let us discuss the choice $e=\mathds{1}$, i.e. we consider the product $\Co{0}{\mathds{1}}{\varphi \varphi}(x_1,x_2)\Co{i}{d}{\mathds{1} c}(x_2,x_3)$. Here the second coefficient vanishes for $i\neq 0$, because the requirement for an identity element, axiom \ref{ax:Identity element}, necessarily requires

\begin{equation}
 \C{b}{\mathds{1} a}(x_1,x_2;\lambda)=\delta^b_a
\label{eq:model perturbations identity}
\end{equation}
which obviously implies $\Co{i}{b}{\mathds{1} a}(x_1,x_2)=0$ for $i\neq 0$. Thus, for $j=0$ and $i>0$ the only contribution to the sum over $e$ on the left side of eq.\eqref{eq:model perturbations consistency special} is the product $\Co{i}{d}{\varphi^2 c}(x_2,x_3)\Co{0}{\varphi^2}{\varphi\, \varphi}(x_1,x_2)$. Applying the results of the previous section one finds $\Co{0}{\varphi^2}{\varphi\, \varphi}(x_1,x_2)=1$.
Thus, if we shovel all terms except $\Co{i}{d}{\varphi^2 c}(x_2,x_3)$ to the right side of the equality sign in eq.\eqref{eq:model perturbations consistency special}, we have the equation\footnote{Here and in the following the limit is understood as $\lim\limits_{x_1\nearrow x_2}$, i.e. $|x_1|$ approaches $|x_2|$ from below} (for $i>0$)

\begin{equation}
\begin{split}
 \Co{i}{d}{\varphi^2 c}(x_2,x_3)=\lim_{x_1\to x_2}\Bigg(&\sum_e\sum_{j=0}^{i} \Co{j}{e}{\varphi c}(x_2,x_3)\Co{i-j}{d}{\varphi e}(x_1,x_3)\\ -&\sum_{e}^{|e|\leq |\varphi^2|}\sum_{j=1}^{i} \Co{j}{e}{\varphi \varphi}(x_1,x_2)\Co{i-j}{d}{ec}(x_2,x_3)\Bigg)
\end{split}
\label{eq:model perturbations limit}
\end{equation}
Now suppose we already know $\Co{i}{b}{\varphi a}(x_1,x_2)$ for all $\ket{v_a},\ket{v_b}\in V$ in addition to all the lower order coefficients. Then all expressions appearing on the right side of the above equation are known. That means, if we have all coefficients up to $\Co{i}{b}{\varphi a}(x_1,x_2)$, then we can uniquely determine $\Co{i}{b}{\varphi^2 a}(x_1,x_2)$ by this equation, which is just the kind of identity we were looking for. Before we go on with our iterative procedure, let us first view the above equation from a different perspective. Remember, e.g. from eq.\eqref{eq:framework motivation consistency}, that the first sum just gives the 3-point coefficient $\Co{i}{d}{\varphi \varphi c}(x_1,x_2,x_3)$ on the domain $\mathcal{F}_3$. Thus, the relation may equivalently be written as

\begin{equation}
\begin{split}
 \Co{i}{d}{\varphi^2 c}(x_2,x_3)=\lim_{x_1\to x_2}\left(\Co{i}{d}{\varphi \varphi c}(x_1,x_2,x_3)-\sum_e^{|e|\leq |\varphi^2|}\sum_{j=1}^{i} \Co{j}{e}{\varphi\,  \varphi}(x_1,x_2)\Co{i-j}{d}{ec}(x_2,x_3)\right)
\end{split}
\label{eq:model perturbations limit 3point}
\end{equation}
for $r_{23}<r_{13}$. This suggests the following interpretation: Naively, one might try to obtain the desired coefficient $\Co{i}{d}{\varphi^2 c}(x_2,x_3)$ by just letting two points of the above three point coefficient approach each other. Similarly, one might try to obtain $\varphi^2(x_2)$ as the coincidence limit of the product $\varphi(x_1) \varphi(x_2)$. It is a well known feature of quantum field theory that this naive limit does not exist, as one is dealing with distributional objects (operator valued distributions in the standard formulation, distributional OPE coefficients in our framework). So in order to make sense of expressions like $\varphi^2(x_2)$ or $\Co{i}{d}{\varphi^2 c}(x_2,x_3)$, one has to subtract counterterms before performing the limit in order to obtain well defined objects. That is precisely the meaning of the sum on the right side of equation \eqref{eq:model perturbations limit 3point}. So in a sense, this equation may be viewed as our analogue of \emph{renormalization}, where the counterterms are represented by the finite sum that is subtracted. Note, however, that this identity is an intrinsic feature of our theory resulting from the associativity condition. Therefore we do not have to apply any kind of external renormalization as the framework is inherently finite at any order.

Let us conclude this interpretational interlude and come back to our iteration procedure. Remember that we have just found a way to express the coefficients of the form $\Co{i}{d}{\varphi^2 c}(x_1,x_2)$ in terms of the $\Co{i}{d}{\varphi c}(x_1,x_2)$ and lower order coefficients. This procedure can be repeated. That means we start at eq.\eqref{eq:model perturbations consistency} again, but this time we chose $a=\varphi^2$ and $b=\varphi$. Following the steps that lead to eq.\eqref{eq:model perturbations limit}, we again take the limit $x_1\to x_2$. We will encounter one summand of the form $\Co{0}{\varphi^3}{\varphi^2 \varphi}(x_1,x_2)\Co{i}{d}{\varphi^3 c}(x_2,x_3)$. The first factor here is again just $1$, so we isolate $\Co{i}{d}{\varphi^3 c}(x_2,x_3)$ on the left side of the equation, obtaining
\begin{equation}
\begin{split}
 &\Co{i}{d}{\varphi^3 c}(x_2,x_3)=\lim_{x_1\to x_2} \Bigg[\sum_{e}\sum_{j=0}^{i} \Co{j}{e}{\varphi c}(x_2,x_3)\Co{i-j}{d}{\varphi^2 e}(x_1,x_3)\\&-\sum_{e}^{|e|\leq |\varphi^3|}\sum_{j=1}^{i} \Co{j}{e}{\varphi^2 \varphi}(x_1,x_2)\Co{i-j}{d}{ec}(x_2,x_3)-\sum_{e\neq \varphi^3}^{|e|\leq |\varphi^3|}\Co{0}{e}{\varphi^2 \varphi}(x_1,x_2)\Co{i}{d}{e c}(x_2,x_3)\Bigg]
\end{split}
\label{eq:model perturbations limit 2}
\end{equation}
Again, all coefficients on the right side are known. 
Clearly, this scheme can be continued iteratively. Assume we know all coefficients up to those of the type $\Co{i}{d}{\varphi^{k-1} c}(x_2,x_3)$, then we can uniquely construct the coefficients $\Co{i}{d}{\varphi^{k} c}(x_2,x_3)$ by the formula

\begin{equation}
 \begin{split}
 &\Co{i}{d}{\varphi^k c}(x_2,x_3)=\lim_{x_1\to x_2} \Bigg[\sum_{e}\sum_{j=0}^{i} \Co{j}{e}{\varphi c}(x_2,x_3)\Co{i-j}{d}{\varphi^{k-1} e}(x_1,x_3)\\&-\sum_{e}^{|e|\leq |\varphi^k|}\sum_{j=1}^{i} \Co{j}{e}{\varphi^{k-1} \varphi}(x_1,x_2)\Co{i-j}{d}{ec}(x_2,x_3)
 -\sum_{e\neq \varphi^k}^{|e|\leq |\varphi^k|}\Co{0}{e}{\varphi^{k-1} \varphi}(x_1,x_2)\Co{i}{d}{e c}(x_2,x_3)\Bigg]
\end{split}
\label{eq:model perturbations limit k-1}
\end{equation}
This procedure solves the problem we encountered in the discussion below eq.\eqref{eq:model perturbations relations} in our approach towards a construction of the perturbed OPE coefficients just using the field equation. There we got stuck already at first perturbation order, as we could not relate the coefficients $\Co{1}{b}{\varphi a}$ and $\Co{1}{b}{\varphi^k a}$. This relation can now be achieved by $k-1$ iterations of the above equation.

Let us briefly sum up the algorithm we just found: The starting point of the iteration is the free quantum field theory $(V,\mathcal{C}_0)$ as described in section \ref{subsec:free field}. Next we use the non-linear field equation, more precisely we solve the differential equation \eqref{eq:model perturbations relations}, in order to construct the first order coefficients of the form $\Co{1}{b}{\varphi a}$. Then we repeatedly apply eq.\eqref{eq:model perturbations limit k-1}, which follows from the associativity condition, obtaining the coefficients $\Co{1}{b}{\varphi^k a}$, which allow for the construction of the second order coefficients $\Co{2}{b}{\varphi a}$ via the differential equation \eqref{eq:model perturbations relations} again. Exploiting the associativity condition and the field equation once more yields the third order coefficients, and so on. The procedure is summarized schematically in the following diagram:

\begin{equation}
\begin{tikzpicture}[rounded corners,level distance=4cm]
\node[rectangle, fill=mygray] {$\Co{0}{c}{ab}$}[grow'=right]
child[-latex,very thick] {node[rectangle, fill=mygray] {$\Co{1}{b}{\varphi a}$} child[-latex, very thick] {node[rectangle, fill=mygray] {$\Co{1}{c}{ab}$} child[-latex,very thick] {node[rectangle, fill=mygray] {$\Co{2}{b}{\varphi a}$} child[very thick,level distance=1cm,-] {child[dotted, very thick,level distance=1cm,-]{}} edge from parent node[above, sloped, text=black] {'' $\square^{-1}$ ``} } edge from parent node[above, sloped, text=black] {\begin{small}associativity\end{small}
} node[below, sloped, text=black] {\begin{small}condition\end{small}
} } edge from parent node[above, sloped, text=black] {'' $\square^{-1}$ ``}
};
\end{tikzpicture}
\end{equation}

From the standard formulation of quantum field theory we have learned the lesson, that the construction of higher order perturbations should be expected to run into serious calculational difficulties already at rather low orders. These difficulties appear in the calculation of Feynman integrals containing loops, which has become an independent branch of theoretical physics over the last decades. Therefore, we should also expect serious calculational effort to go into the explicit construction of perturbed OPE coefficients in our framework. Where do we encounter these difficulties? Reviewing our algorithm, we find that it essentially consists of two calculational steps: Solving the differential equation \eqref{eq:model perturbations relations} and performing the sums and the limit in eq.\eqref{eq:model perturbations limit k-1}. The former does not cause any trouble, as all 2-point coefficients can be expressed as a linear combination of terms of the form $c\cdot Y_{lm}(\hat{x})r^a(\log r)^b $ for some constant $c\in\mathbb{C}$ and parameters $a,b,l,m\in\mathbb{N}$. As we will see below, it is not very difficult to invert the Laplace operator on such an expression (see appendix \ref{app:characteristic diff eq}). It turns out that most calculational work has to be put into eq.\eqref{eq:model perturbations limit k-1}, especially into the first sum over $e$. This is an infinite sum over all basis elements of our vector space $V$, i.e. a sum, or rather a multiple sum, over the multi-index $e$. Before performing the limit in eq.\eqref{eq:model perturbations limit k-1}, one has to put these sums into a form that makes it possible to control the cancellation of infinities with the counterterms. So we have in a sense traded the problematic higher loop Feynman integrals for multiple infinite sums. It is not clear a priori which way is more convenient for explicit calculations and it is one aim of this thesis to give some first impressions of the effort that goes into the explicit construction of perturbations of a quantum field theory in the iterative scheme outlined above. 

\subsection{Ambiguities}\label{subsec:ambiguities}
In the usual approaches to renormalization, certain ambiguities are typically present. It is therefore natural to ask ''how unique`` the results obtained from our method are. As mentioned above, our algorithm consists of two computational steps: Solving the differential equation \eqref{eq:model perturbations relations} and applying the consistency condition (eq.\eqref{eq:model perturbations limit k-1}). The latter uniquely determines the coefficients $\Co{i}{b}{\varphi^k a}$ from the coefficients $\Co{i}{b}{\varphi^{k-1} a}$, so possible ambiguities can only arise in the differential equation. In appendix \ref{app:characteristic diff eq} we present a special solution to this equation, which is used throughout this thesis, and the freedom in the choice of this solution is also briefly discussed. To sum up, we may add to any solution $\Lo{i}{\varphi}{x}$ of the field equation an $\End(V)$-valued harmonic polynomial in $x$, i.e.

\begin{equation}
 \mathcal{Y}_i'(\varphi,x)=\Lo{i}{\varphi}{x}+K(x)
\end{equation}
with

\begin{equation}
 K(x)=\sum_{J=0}^{\infty} K_{J}r^JY_{JM}(\hat{x};D)+\sum_{J=0}^{\infty} K_{2-J-D}r^{2-J-D}Y_{JM}(\hat{x};D)
\end{equation}
where $K_{J}\in\End(V)$. Our choice of left representatives $\Lo{i}{\varphi}{x}$ is further restricted by the axioms of our framework (see section \ref{subsec:axioms}). Most importantly, the scaling dimension
constraint, axiom \ref{ax:Scaling}, restricts $\Lo{i}{\varphi}{x}$ to be of the general form

\begin{equation}
 \Lo{i}{\varphi}{x}=\sum A_{i,d,q,J,M}(\ket{v})r^d(\log r)^q Y_{JM}(\hat{x},D)
\end{equation}
with

\begin{equation}
 A_{i,d,q,J,M}\in\left\{A\in\End(V)\,:\, |A(v)|-|v|-|\varphi|\geq d \quad\forall v\in V \right\}\quad,
\end{equation}
where $|\cdot|$ denotes the canonical dimension of the vectors in $V$ as defined in eq.\eqref{eq:model free field basis elements dimension}. The fact that in the above condition only an inequality is required to hold means that we may add contributions of lower scaling dimension to any OPE coefficient. Thus, the grading of the OPE coefficients by dimension, which held in our construction of the free theory in the previous chapter, is replaced by a filtration at higher perturbation orders. 

The convention used for the solution of the differential equation within this thesis is justified by the following proposition:

\begin{proposition}\label{prop:our solution}%
Given left representatives $\Lo{i-1}{\varphi^k}{x}$ satisfying the axioms of section \ref{subsec:axioms}, the solution $\Lo{i}{\varphi}{x}=\square^{-1}\Lo{i-1}{\varphi^k}{x}$ satisfies these axioms as well (except possibly the factorization axiom). Further, if the OPE coefficients at order $i-1$ are graded by dimension (i.e. if the equality holds in the scaling dimension axiom), then this grading is preserved by the solution obtained with the help of $\square^{-1}$.
\end{proposition}
Remark: The factorization axiom can not be checked with the knowledge of $\Lo{i}{\varphi}{x}$ alone, see eq.\eqref{eq:framework FLR borcherds identity}. Recall from the previous discussion that we \emph{assume} this property to hold, which allows us to determine the other $i$-th order left representatives by the algorithm outlined above.
\begin{proof}
Axioms \ref{ax:Hermitian conjugation}, \ref{ax:Bosonic nature}, \ref{ax:Identity element} and \ref{ax:(Anti-)symmetry} are satisfied trivially, simply because they were also satisfied by the zeroth order coefficients and since the operator $\square^{-1}$ does not change any relevant properties for these axioms. Similarly, the zeroth order coefficients satisfy the Euclidean invariance property. Since $\square^{-1}$ acts rotationally invariant (i.e. it does not change the spherical harmonics $Y_{JM}$) our choice is also consistent with this axiom. Further, the constraint implied by axiom \ref{ax:Derivations} is precisely the field equation as a differential equation on $\mathbb{R}^{D}$. Since $\square^{-1}$ is a solution to that equation, it also satisfies this axiom. It remains to discuss the scaling axiom and the claim that the grading by dimension is conserved. Assume the axiom and the grading to hold at order $i-1$. Then $sd\, \Co{i-1}{b}{\varphi^ka}= |a|+|\varphi^k|-|b|$ holds. Using our solution to the field equation, which increases the scaling degree by $2$, it then follows that

\begin{equation}
 sd\, \Co{i}{b}{\varphi a}=sd\, \Co{i-1}{b}{\varphi^ka}-2=|a|+|\varphi^k|-|b|-2
\label{eq:ambiguities scaling}
\end{equation}
Now we claim that in a renormalizable theory $|\varphi^k|-|\varphi|= 2$ always holds. This can be deduced from the field equation as follows: Comparing the dimensions of both sides of the equation using eq.\eqref{eq:model free field basis elements dimension} we can solve for the dimension\footnote{Recall that the derivations $\del_\mu$ on $V$ increase the dimension by $1$} of the coupling constant $\lambda$.

\begin{equation}
 |\square\varphi|-|\varphi^k|=2-|\varphi^{k-1}|=|\lambda|
\end{equation}
By definition, a theory is renormalizable if $|\lambda|= 0$, super-renormalizable if $|\lambda|< 0$ and non-renormalizable if $|\lambda|>0$. Since we are interested only in the first case, the relation

\begin{equation}
 |\varphi^k|-|\varphi|=|\varphi^{k-1}|= 2
\end{equation}
does indeed hold. Using this relation in eq.\eqref{eq:ambiguities scaling} we find the desired equation

\begin{equation}
 sd\, \Co{i}{b}{\varphi a}= |a|+|\varphi|-|b|
\end{equation}
confirming axiom \ref{ax:Scaling} and preserving the grading. Since both the axiom and the grading hold in the free theory, they hold to all perturbation orders using our solution.
\end{proof}
Since the OPE is a short distance expansion, the following theorem should be of interest:

\begin{thm}[Ambiguities at short distances]%
Let $\ket{v_a}\in V$, $\bra{v_b}\in V^{*}$ be arbitrary with $\Co{i-1}{b}{\varphi^k a}(x)\neq 0$. Then the most rapidly divergent part of the OPE coefficients $\Co{i}{b}{\varphi a}(x)$ in the limit $x\to 0$ is uniquely determined by the solution $\square^{-1}\Co{i-1}{b}{\varphi^k a}(x)$.
\end{thm}
\begin{proof}
 In order to find the part of an OPE coefficient with the strongest divergence, we first extract the contribution of maximal scaling degree. From the scaling axiom we know that this is the contribution proportional to $r^{|c|-|a|-|b|}$ for a coefficient $\Co{i}{c}{ab}$. Let us denote the projection on the contribution of scaling degree $d$ by $Sc_d$. Thus

\begin{equation}
 \Co{i}{c}{ab}(x)=Sc_{|a|+|b|-|c|}\,\Co{i}{c}{ab}(x)+O(r^{d})
\end{equation}
with $d>|c|-|a|-|b|$. As the OPE coefficients may also contain powers of logarithms, it is not guaranteed that all terms in $Sc_{|a|+|b|-|c|}\,\Co{i}{c}{ab}(x)$ have the same divergent behavior for $x\to 0$. In order to find the dominating contribution to $Sc_{|a|+|b|-|c|}\,\Co{i}{c}{ab}(x)$ at short distances, we have to pick out the terms including the highest powers of $\log r$.

Now the claim is that this contribution containing the highest power of $\log r$ among the terms of maximal scaling degree is uniquely determined by $\square^{-1}\Co{i-1}{b}{\varphi^k a}(x)$ (provided $\Co{i-1}{b}{\varphi^k a}(x)\neq 0$). Recall from proposition \ref{prop:our solution} that $\square^{-1}\Co{i-1}{b}{\varphi^k a}(x)$ exclusively contains terms of maximal scaling degree, i.e.

\begin{equation}
 Sc_{|a|+|\varphi|-|b|}\, \Co{i}{b}{\varphi a}(x)=\square^{-1}\Co{i-1}{b}{\varphi^k a}(x)+\text{ ''ambiguities``}.
\end{equation}
To finish the proof, we show in the following that $\square^{-1}\Co{i-1}{b}{\varphi^k a}(x)$ always contains higher powers of $\log r$ than any possible ambiguity. Recall from the discussion at the beginning of this section that the ambiguities have to be in the kernel of $\square$ and in addition have to be compatible with the axioms of section \ref{subsec:axioms}. The elements in the kernel of the Laplacian are the harmonic polynomials, hence the ambiguities have to be of the form

\begin{equation}
K_J(x)=\begin{cases}
 r^JY_{JM}(\hat{x})\quad&\text{for }J\geq 0\\
r^JY_{-(J+D-2)M}(\hat{x})\quad&\text{for }J< 0
\end{cases}
\end{equation}
Since we only consider ambiguities of maximal scaling degree, we are only interested in $J=|b|-|\varphi|-|a|$. We distinguish two cases:

\begin{enumerate}
 \item $K_{(|b|-|\varphi|-|a|)}(x)$ is incompatible with the Euclidean invariance axiom. Thus, no terms of maximal scaling degree may be added to $\square^{-1}\Co{i-1}{b}{\varphi^k a}(x)$ and we have 

\begin{equation}
 Sc_{|a|+|\varphi|-|b|}\, \Co{i}{b}{\varphi a}(x)=\square^{-1}\Co{i-1}{b}{\varphi^k a}(x)
\end{equation}
in accordance with the theorem
\item $K_{(|b|-|\varphi|-|a|)}(x)$ is compatible with the Euclidean invariance axiom. Then $\square^{-1}\Co{i-1}{b}{\varphi^k a}(x)$ contains a contribution proportional to the same spherical harmonic as $K_{(|a|+|\varphi|-|b|)}(x)$, since $\square^{-1}\Co{i-1}{b}{\varphi^k a}(x)$ itself also satisfies the Euclidean invariance axiom. It follows from the definition of $\square^{-1}$ in eq.\eqref{eq:app diff eq inverse} that any contribution to $\square^{-1}\Co{i-1}{b}{\varphi^k a}(x)$ proportional to a harmonic polynomial includes at least one power of $\log r$. Since the ambiguities contain no logarithms, they diverge less rapidly.

\end{enumerate}

\end{proof}
Remark: The theorem does not imply that $\square^{-1}$ determines the short distance behavior of the complete OPE, since we excluded coefficients $\Co{i}{b}{\varphi a}$ with $\Co{i-1}{b}{\varphi^k a}=0$.

It has been mentioned a few times that our construction does not rely on any renormalization prescription. Nevertheless the ambiguities in our approach share some intriguing similarities with the well known remormalization ambiguities of standard quantum field theory.

\begin{itemize}
 \item Renormalized composite quantum fields are filtered objects by scaling dimension (as opposed to graded). Similarly, our OPE coefficients are filtered by scaling dimension in the interacting case.

\item One may change the definition of $\square^{-1}$ by introducing a complex parameter $\mu$ in every logarithmic contribution, i.e. $\log r\to \log(\mu r)$. This solution also satisfies proposition \ref{prop:our solution}. The free parameter $\mu$ is reminiscent of the choice of \emph{renormalization scale} in the standard formulation of quantum field theory (see also section \ref{subsec:comparison to alternative}).
\end{itemize}
A deeper understanding of the ambiguities and their relation to renormalization theory may be a topic for future research (see chapter \ref{sec:conclusions}).

\subsection{Construction of \texorpdfstring{$\Lo{0}{\ket{v_a}}{x}$}{L(0)a}}\label{subsec:Construction of L(0)a}

As already promised in the previous section, we conclude with the construction of the general left-representative $\mathcal{Y}(\ket{v_a},x)$ in the free theory, i.e. we give a proof of eq.\eqref{eq:model free field La} (in this section we should actually write $\mathcal{Y}_0(\ket{v_a},x)$ indicating zeroth perturbation order, but for convenience we will not do this in the following discussion). Our starting point is the left representative $\mathcal{Y}(\varphi,x)$,see eq.\eqref{eq:model free field Lvarphi 2} , which has been ''derived`` from standard quantum field theory. Remembering the identity $[\mathcal{Y}(\ket{v_a},x)]^c_b=\C{c}{ab}$, it is clear that the iterative scheme of this section, in particular eq.\eqref{eq:model perturbations limit k-1}, is exactly the needed tool for our purpose, as it allows for the construction of $\mathcal{Y}(\varphi^k,x)$ starting from $\mathcal{Y}(\varphi,x)$. It is then possible, by taking the appropriate derivatives and multiplying the right constants, to determine $\mathcal{Y}(\ket{v_a},x)$ from $\mathcal{Y}(\varphi^k,x)$. Thus, the following calculation is also a good practice for later applications of the consistency condition, eq.\eqref{eq:model perturbations limit k-1}. Let us begin with the first application of this equation:

\begin{equation}
\begin{split}
 \C{b}{\varphi^2 a}(x_2,x_3)=\lim_{x_1\to x_2}\left(\sum_c\C{c}{\varphi a}(x_2,x_3)\C{b}{\varphi c}(x_1,x_3)- \C{\mathds{1}}{\varphi \varphi}(x_1,x_2)\,\delta^b_a\right)
\end{split}
\label{eq:model perturbations Lphi^2}
\end{equation}
This is just eq.\eqref{eq:model perturbations limit} for $i=0$ with an additional counterterm due to eq.\eqref{eq:model perturbations identity}. We can rewrite this equation as

\begin{equation}
\begin{split}
 \mathcal{Y}(\varphi^2,x)=\lim_{y\to x}\left(\mathcal{Y}(\varphi,x)\mathcal{Y}(\varphi,y)- [\mathcal{Y}(\varphi,x-y)]^{\mathds{1}}_{\varphi}\, \mathds{1}\right)\, .
\end{split}
\label{eq:model perturbations Lphi^2 matrix}
\end{equation}
By the definition of $\mathcal{Y}(\varphi,x)$ in eq.\ref{eq:model free field Lvarphi 2}, we can write this product as

\begin{equation}
 \begin{split}
 \mathcal{Y}(\varphi,x)\mathcal{Y}(\varphi,y)=\sigma_D\; r_y^{-(D-2)/2}&r_x^{-(D-2)/2}\sum_{l,l'}\sum_{m,m'}\left(\frac{D-2}{2l+D-2}\right)^{1/2}\left(\frac{D-2}{2l+D-2}\right)^{1/2}\times \\
  &\left[r_x^{l'+(D-2)/2}\, Y^{l'm'}(\hat{x})\, \mathbf{b}_{l'm'}^{\dagger}+r_x^{-l'-(D-2)/2}\, Y_{l'm'}(\hat{x})\, \mathbf{b}_{l'm'}\right] \times \\ &\left[r_y^{l+(D-2)/2}\, Y^{lm}(\hat{y})\, \mathbf{b}_{lm}^{\dagger}+r_y^{-l-(D-2)/2}\, Y_{lm}(\hat{y})\, \mathbf{b}_{lm}\right]\: ,
\end{split}
\label{eq:model perturbations Lphi Lphi}
\end{equation}
where $r_x=|x|$ and $r_y=|y|$ respectively. Let us focus on the following partial sum

\begin{equation}
 \sigma_D\; r_y^{-(D-2)/2}r_x^{-(D-2)/2}\sum_{l=0}^{\infty}\sum_{m=1}^{N(l,D)}\frac{D-2}{2l+D-2}\left[r_y^{l+(D-2)/2}r_x^{-l-(D-2)/2}\, Y^{lm}(\hat{y})Y_{lm}(\hat{x})\, \mathbf{b}_{lm} \mathbf{b}_{lm}^{\dagger}\right]
\label{eq:model perturbations Lphi Lphi subsum}
\end{equation}
This sum is of particular interest, because it is the only partial sum of eq.\eqref{eq:model perturbations Lphi Lphi} with infinitely many non-vanishing contributions after taking matrix elements $\bra{v_b}\, \cdot\, \ket{v_a}$. This is due to the fact that the successive application of creation and annihilation operator with the same indices on some vector $\ket{v_a}\in V$ gives, according to the definitions \eqref{eq:model free field ladder operators1} and \eqref{eq:model free field ladder operators2}, just a prefactor of $(a_{lm}+1)$, where $a_{lm}$ is the "occupation number" of the respective mode. It is easy to convince oneself that all other contributions to $\mathcal{Y}(\varphi,x)\mathcal{Y}(\varphi,y)$ have finite matrix elements. First, consider the partial sum of eq.\eqref{eq:model perturbations Lphi Lphi} with two annihilation operators. Again, according to eq.\eqref{eq:model free field ladder operators2}, the action of an annihilation operator $\mathbf{b}_{lm}$ on a vector $\ket{v_a}$ gives a prefactor of $(a_{lm})^{1/2}$ and reduces the index $a_{lm}$ by 1. Recall that in the definition of our basis $V$ we demanded the multi-indices $a$ of the basis elements $\ket{v_a}$ to contain only a finite number of non-zero entries. Hence, the aforementioned prefactor resulting from the application of an annihilation operator makes sure that a sum of the type $\sum_{l=0}^{\infty}\mathbf{b}_{lm}\ket{v_a}$ is always finite. So if we sum over the product of two annihilation operators, both of these sums will be finite. Concerning the case of two creations operators, we can deduce from orthogonality of the basis elements that $\bra{v_b}\mathbf{b}_{lm}^{\dagger}\ket{v_a}=0$ for $a+e_{lm}\neq b$, so the sum $\sum_{l=0}^{\infty}\bra{v_b}\mathbf{b}_{lm}^{\dagger}\ket{v_a}$ has at most one non-zero summand and is thus finite. Finally, let us come to the partial sums with a pair of one annihilation and one creation operator. If the former stands to the right of the latter, we can simply use the same argumentation as in the case of two annihilation operators and find that also here no infinite sums appear. This only leaves the case where the creation operator acts before the annihilator. However, if the indices on these two operators do not coincide, it follows from the commutation relations \eqref{eq:model free field ladder operators commutation relations} that we may simply exchange their order, which leads us to the case that we just argued to be finite. Hence, the only possibility for infinities to appear in an arbitrary matrix element of eq.\eqref{eq:model perturbations Lphi Lphi} is the matrix element $\bra{v_a}\,\cdot\, \ket{v_a}$ of the partial sum in eq.\eqref{eq:model perturbations Lphi Lphi subsum}.

We can further simplify the analysis of this infinite sum by exploiting the commutation relation of the ladder operators also in this case, which suggests that in addition to the term where the order of the two operators is switched we pick up a term with the identity operator replacing the pair of ladder operators, i.e.

\begin{equation}
 \mathbf{b}_{lm}\mathbf{b}_{lm}^{\dagger}=\mathbf{b}_{lm}^{\dagger}\mathbf{b}_{lm}+\, id
\end{equation}
 Then, the first term is finite again, since now the annihilation operator acts first, and the remaining infinite sum has the form of eq.\eqref{eq:model perturbations Lphi Lphi subsum} with the ladder operators replaced by the identity. We now want to find a closed form expression for this sum. Using the addition theorem for the $D$-dimensional spherical harmonics (see eq.\eqref{eq:app spherical symm arbitrary addition theorem}),

\begin{equation}
 \sum_m Y_{lm}(\hat{x})Y^{lm}(\hat{y})=\frac{N(l,D)}{\sigma_D}P_l(D,\hat{x}\cdot \hat{y})\,
\label{eq:model perturbations addition theorem}
\end{equation}
where $P_l(D,\hat{x})$ is the $D$-dimensional Legendre polynomial and $N(l,D)$ and $\sigma_D$ are defined as in section \ref{subsec:free field}, we can perform the sum over $m$ and arrive at the formula

\begin{equation}
 r_x^{-(D-2)}\sum_{l=0}^{\infty}\left(\atop{l+D-3}{l}\right)\, \left(\frac{r_y}{r_x}\right)^l\, P_l(D,\hat{x}\cdot \hat{y})
\label{eq:model perturbations Lphi Lphi subsum 2}
\end{equation}

 In view of eq.\eqref{eq:model perturbations Lphi^2 matrix}, we are interested in the limit $y\to x$ in this expression, so for convenience we may chose $x$ and $y$ collinear, i.e. $\hat{x}=\hat{y}$. As the Legendre polynomials are normalized, $P_l(D,1)=1$, the equation further simplifies

\begin{equation}
 r_x^{-(D-2)}\sum_{l=0}^{\infty}\left(\atop{l+D-3}{l}\right)\, \left(\frac{r_y}{r_x}\right)^l= r_x^{-(D-2)}\sum_{l=0}^{\infty} C_l^{\left(\frac{D-2}{2}\right)}(1)\, \left(\frac{r_y}{r_x}\right)^l\, .
\label{eq:model perturbations Lphi Lphi subsum 3}
\end{equation}
Here $C^{\nu}_l(x)$ are the Gegenbauer polynomials (see e.g. \cite{Erdelyi1953} or \cite{Abramowitz1965}) and should not be confused with an OPE coefficient. It is useful to write the equation in this form, because the generating function of the Gegenbauer polynomials is well known:

\begin{equation}
 \sum_{l=0}^{\infty} C_l^{\nu}(z)\, h^n=(1-2hz+h^2)^{-\nu}\qquad \text{for } |h|<|z\pm (z^2-1)^{1/2}|
\label{eq:model perturbations gegenbauer generating function}
\end{equation}
Since we know that $r_x>r_y$ in our case, we can apply this identity to eq.\eqref{eq:model perturbations Lphi Lphi subsum 3} with the result

\begin{equation}
 r_x^{-(D-2)}\, \left(1-\frac{r_y}{r_x}\right)^{-(D-2)}=\left(\frac{1}{r_x-r_y}\right)^{D-2}\, ,
\label{eq:model perturbations Lphi Lphi subsum 4}
\end{equation}
so this partial sum of eq.\eqref{eq:model perturbations Lphi Lphi} is in fact divergent in the limit $y\to x$. Hence, there should either be another divergent part of this sum with opposite sign, or the "counterterm" has to cancel the divergence. Our formula for $\mathcal{Y}(\varphi,x)$, eq.\eqref{eq:model free field Lvarphi 2}, yields for the counterterm 

\begin{equation}
\C{\mathds{1}}{\varphi \varphi}(x-y)=|x-y|^{-(D-2)}\, ,
 \label{eq:model perturbations Lphi Lphi counterterm}
\end{equation}
which is indeed equal to eq.\eqref{eq:model perturbations Lphi Lphi subsum 4} in the limit $y\to x$. As there are no additional counterterms, and as the rest of eq.\eqref{eq:model perturbations Lphi Lphi} is finite, the intrinsic renormalization procedure has indeed worked out. Note that we can reduce the calculational effort drastically if we just drop the counterterm together with the partial sum considered above from eq.\eqref{eq:model perturbations Lphi^2 matrix}. This can be achieved in a very elegant manner, namely by \emph{normal ordering}, i.e. by simply rearranging the order of the ladder operators in eq.\eqref{eq:model perturbations Lphi Lphi} in such a way that the annihilation operators always act first. If these operators are commuting, nothing changes in this trivial process. Hence the only case where this procedure actually manifests itself is if we put products of the form $\mathbf{b}_{lm}\mathbf{b}_{lm}^{\dagger}$ into normal order. Usually we would have obtained an additional term with the ladder operators replaced by the identity, if we were to exchange the order of this expression \emph{by hand}, due to the commutation relations. If we require normal ordering, this extra term is neglected. Now recall form eq.\eqref{eq:model perturbations Lphi Lphi subsum} and the subsequent discussion that it was precisely this extra term that was responsible for the divergence, which canceled with the counterterm. Hence, we may simply forget about this extra term and the counterterm altogether, which is expressed in the following formula:

\begin{equation}
 \mathcal{Y}(\varphi^2,x)=\lim_{y\to x}\left(:\mathcal{Y}(\varphi,x)\mathcal{Y}(\varphi,y):\right)=\, :\Big(\mathcal{Y}(\varphi,x)\Big)^2:
\label{eq:model perturbations Lphi^2 matrix normal order}
\end{equation}
As above, double dots denote normal ordering.

This procedure can again be iterated. Let us proceed with the analog of equation \eqref{eq:model perturbations Lphi^2 matrix} for the next left representative

\begin{equation}
 \mathcal{Y}(\varphi^3,x)=\lim_{y\to x}\Big[\mathcal{Y}(\varphi,x)\mathcal{Y}(\varphi^2,y)-\sum_{e\neq \varphi^3}^{|e|\leq |\varphi^3|}[\mathcal{Y}(\varphi,x-y)]^{e}_{\varphi}\mathcal{Y}(\ket{v_e},x)\Big]\, ,
\label{eq:model perturbations Lphi^3 matrix}
\end{equation}
which follows from eq.\ref{eq:model perturbations limit 2}. The sum over $e$ can be further specified by noting that the restrictions $|e|\leq |\varphi^3|$ and $e\neq \varphi^3$ constrain $e$ to be of the form $\mathds{1}$, $\varphi_{lm}$ with $l\leq (D-2)$ or $\varphi_{lm}\varphi_{l'm'}$ with $l+l'\leq (D-2)/2$. Eq.\eqref{eq:model perturbations Lphi^2 matrix normal order} suggests that $\mathcal{Y}(\varphi^2,x)$ acts by two ladder operators, so it follows that only the option $e=\varphi_{lm}$ gives a non-vanishing OPE coefficient $C^{e}_{\varphi^2\, \varphi}$, since in the other two cases no combination of the two ladder operators can transform $\varphi$ into $e$. Therefore, the above equation may be simplified to

\begin{equation}
 \mathcal{Y}(\varphi^3,x)=\lim_{y\to x}\Big[\mathcal{Y}(\varphi,x)\mathcal{Y}(\varphi^2,y)-\sum_{l=0}^{D-2}\, C^{\varphi_{lm}}_{\varphi^2\, \varphi}(x-y)\mathcal{Y}(\varphi_{lm},x)\Big]\, ,
\label{eq:model perturbations Lphi^3 matrix 2}
\end{equation}
Now let us see what happens is we again just prescribe normal ordering to the product of the left representatives $\mathcal{Y}(\varphi^2,y)=\, :[\mathcal{Y}(\varphi,y)]^2:$ and $\mathcal{Y}(\varphi,x)$. This time we encounter products of three ladder operators, two of which are already normal ordered. Normal ordering is trivial except for terms of the form  $:\mathbf{b}_{l'm'}^{\dagger}\mathbf{b}_{lm}:\mathbf{b}_{lm}^{\dagger}$ or $:\mathbf{b}_{l'm'}\mathbf{b}_{lm}:\mathbf{b}_{lm}^{\dagger}$. Ignoring the sum over the primed indices in these expressions, we can carry out the sums over $m$ and $l$ just as we did in the calculations leading to eq.\eqref{eq:model perturbations Lphi^2 matrix normal order}. Therefore, the partial sum containing all expressions of the form just mentioned can be simplified to

\begin{equation}
 2(r_x-r_y)^{-(D-2)}\mathcal{Y}(\varphi,y)\, ,
\label{eq:model perturbations Lphi^3 partial sum}
\end{equation}
where the term in brackets results from the sum over $l$ and $m$, $\mathcal{Y}(\varphi,y)$ represents the sum over $l'$ and $m'$ and the additional factor 2 comes into the equation, because the sums over expressions including $:\mathbf{b}_{l'm'}^{\dagger}\mathbf{b}_{lm}:\mathbf{b}_{lm}^{\dagger}$ or $:\mathbf{b}_{lm}\mathbf{b}_{l'm'}^{\dagger}:\mathbf{b}_{lm}^{\dagger}$ are the same due to normal ordering. Since $\mathcal{Y}(\varphi,x)$ is analytic around $y=x$, we may perform a Taylor expansion

\begin{equation}
  2(r_y-r_x)^{-(D-2)}\mathcal{Y}(\varphi,y)= 2(r_y-r_x)^{-(D-2)}\sum_{l}\frac{(y-x)^l}{l!}\partial^l\mathcal{Y}(\varphi,x)
\label{eq:model perturbations Lphi^3 partial sum taylor}
\end{equation}
 But this is exactly equal to the counterterms $C^{\varphi_{lm}}_{\varphi^2\, \varphi}(x-y)\mathcal{Y}(\varphi_{lm},x)$ in the limit $x\to y$, so we can drop all these terms, require normal ordering and write

\begin{equation}
 \mathcal{Y}(\varphi^3,x)=:\mathcal{Y}(\varphi,x) \mathcal{Y}(\varphi,x) \mathcal{Y}(\varphi,x):\quad .
\label{eq:model perturbations Lphi^3 result}
\end{equation}
The generalization of this result to arbitrary powers uses the same argumentation, and it is then a straightforward calculation to verify the formula for the general left representative $\mathcal{Y}(\ket{v_a},x)$ given in eq.\eqref{eq:model free field La}.

\section{Diagrammatic notation}\label{subsec:diagrammatic notation}

The algorithm presented in the previous section can be neatly visualized in terms of rooted trees, which should not be confused with the trees we used in chapter \ref{subsec:coherence theorem}. These trees will help us keeping track of all the infinite sums and subsequent inversions of the Laplace operator, which are the main steps in the scheme of the previous section. They might further be helpful to distinguish emerging patterns in our construction. In this section we introduce this diagrammatic notation, defining the explicit correspondence between parts of the graphs and the expressions appearing in the scheme of section \ref{subsec:perturbation via field equation}. Its usefulness will become clear in the following sections, where it will be widely applied.

Let $V$ be the Fock space defined in section \ref{subsec:free field} and let (see also appendix \ref{app:characteristic diff eq})

\begin{equation}
\mathbb{Y}(x)=\mathbb{C}\llbracket r,r^{-1},\log r \rrbracket\otimes \{Y_n(\hat{x};D)\}\otimes\End(V)\quad.
\end{equation}\index{symbols}{y@$\mathbb{Y}(x)$}
Then the following linear map defines our diagrammatic notation:

\begin{definition}[Diagrammatic notation]
{\ \\} Let $\mathbf{T}$ be  a rooted tree. To each leaf of this tree assign a label of the form $\pm lm$, where $l\in\mathbb{N}$ and $m\in\{1,\ldots,N(l,D)\}$ and $N(l,D)$ is defined as in section \ref{subsec:free field}. Let us call these labeled trees $\mathbf{T}^l$. Further, let $x\in\mathbb{R}^{D}$, $r=|x|$ and $\hat{x}=x/r$ as in the previous section.
Then we define the linear map $\mathcal{G}_x:\mathbf{T}^l\to\mathbb{Y}(x)$ by the following rules:\\

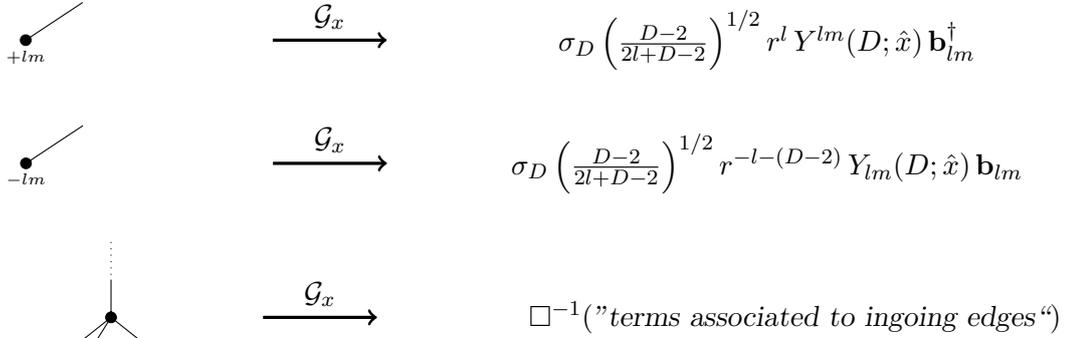
\begin{tikzpicture}
 \node at (5,0) [coordinate]{} child[level distance=.5cm]{[fill] circle (2pt) node[below]{\begin{tiny}$+lm$\end{tiny}}} child[fill=none]{edge from parent [draw=none]};
 \node at (14,-.5)  {$\sigma_D\left(\frac{D-2}{2l+D-2}\right)^{1/2}r^{l}\, Y^{lm}(D;\hat{x})\, \mathbf{b}_{lm}^{\dagger}$};
 \draw[->, very thick] (7.5,-.5) -- node[above]{$\mathcal{G}_x$} (9,-.5);
\end{tikzpicture}

\begin{tikzpicture}
 \node at (5,0) [coordinate]{} child[level distance=.5cm]{[fill] circle (2pt) node[below]{\begin{tiny}$-lm$\end{tiny}}} child[fill=none]{edge from parent [draw=none]};
 \node at (14,-.5)  {$\sigma_D\left(\frac{D-2}{2l+D-2}\right)^{1/2}r^{-l-(D-2)}\, Y_{lm}(D;\hat{x})\, \mathbf{b}_{lm}$};
 \draw[->, very thick] (7.5,-.5)-- node[above]{$\mathcal{G}_x$}(9,-.5);
\end{tikzpicture}

\begin{tikzpicture}
\tikzstyle{level 3}=[level distance=.8cm, sibling distance=.5cm]
\tikzstyle{level 4}=[level distance=.4cm, sibling distance=.25cm]

 \node[coordinate] at (5,0){} child[dotted,level distance=.5cm]{child[solid,level distance=.5cm]{[fill] circle (2pt) child{child[dotted]   child[fill=none]{edge from parent [draw=none]} child[fill=none]{edge from parent [draw=none]} child[fill=none]{edge from parent [draw=none]} child[fill=none]{edge from parent [draw=none]}}   child{ node[coordinate](a){} child[fill=none]{edge from parent [draw=none]} child[dotted]  child[fill=none]{edge from parent [draw=none]} child[fill=none]{edge from parent [draw=none]} child[fill=none]{edge from parent [draw=none]} }   child[fill=none]{edge from parent [draw=none]} child[fill=none]{edge from parent [draw=none]} child{node[coordinate](b){} child[fill=none]{edge from parent [draw=none]} child[fill=none]{edge from parent [draw=none]} child[fill=none]{edge from parent [draw=none]} child[fill=none]{edge from parent [draw=none]} child[dotted] } }};
 \draw[dotted, shorten <= 10pt, shorten >= 10pt,very thick] (a.north east) -- (b.north west);
 \node at (14,-1)  {$\square^{-1}(\text{''terms associated to ingoing edges``})$};
\draw[->, very thick] (7,-1)--node[above]{$\mathcal{G}_x$}(8.5,-1);
\end{tikzpicture}\\
Here by $\square^{-1}$ we mean inversion of the Laplace operator (see appendix \ref{app:characteristic diff eq}). The order of the ladder operators $\mathbf{b}$ is according to the corresponding labels on the trees, i.e. if the label $+lm$ is left of the label $-l'm'$, then the operator $\mathbf{b}_{lm}^{\dagger}$ is left of $\mathbf{b}_{l'm'}$.

\end{definition}
Let us explain these relations: By the first two assignments, the leaves of the tree are closely related to $\mathcal{Y}_0(\varphi,x)$, see eq.\eqref{eq:model free field Lvarphi 2}. In fact, they are either the part of $\mathcal{Y}_0(\varphi,x)$ containing the annihilation operator, or the part containing the creation operator. This is where one starts to calculate the expression corresponding to a tree. Whenever edges meet at a vertex, the outgoing (parent) edge is recursively determined by its children, i.e. the ingoing edges, by an \emph{inversion} of the Laplace operator on the expressions associated to these ingoing edges (see appendix \ref{app:characteristic diff eq} for more information on this inversion). It should be noted that the root is not a vertex (although incoming edges meet), so no differential equation must be solved there. 

As an example, let us write the left representative $\mathcal{Y}(\varphi^2,x)$ in this notation (compare eqs. \ref{eq:model perturbations Lphi^2 matrix normal order} and \ref{eq:model free field Lvarphi 2}):

\begin{equation}
\begin{tikzpicture}

\node at (0,0) {$ \mathcal{Y}(\varphi^2,x)=\sum\limits_{l,l'}\sum\limits_{m,m'}\, :\mathcal{G}_x\Big($};
\node[coordinate] at (3,.7) {} child[level distance=1cm]{[fill] circle (2pt) node[below]{\begin{tiny}$+lm$\end{tiny}}} child[level distance=1cm]{[fill] circle (2pt) node[below]{\begin{tiny}$+l'm'$\end{tiny}}};
\node[coordinate] at (5.5,.7) {} child[level distance=1cm]{[fill] circle (2pt) node[below]{\begin{tiny}$-lm$\end{tiny}}} child[level distance=1cm]{[fill] circle (2pt) node[below]{\begin{tiny}$-l'm'$\end{tiny}}};
\node[coordinate] at (8,.7) {} child[level distance=1cm]{[fill] circle (2pt) node[below]{\begin{tiny}$+lm$\end{tiny}}} child[level distance=1cm]{[fill] circle (2pt) node[below]{\begin{tiny}$-l'm'$\end{tiny}}};
\node[coordinate] at (10.5,.7) {} child[level distance=1cm]{[fill] circle (2pt) node[below]{\begin{tiny}$-lm$\end{tiny}}} child[level distance=1cm]{[fill] circle (2pt) node[below]{\begin{tiny}$+l'm'$\end{tiny}}};
\node at (4.25,0){$+$};
\node at (6.75,0){$+$};
\node at (9.25,0){$+$};
\node at (12,0){$\Big):$};
\end{tikzpicture}
\label{eq:model diagrams Lphi^2}
\end{equation}
Double dots again denote normal ordering, which in this case means that in the last tree the leaves are switched if the indices coincide, i.e. if $(l,m)=(l',m')$. In explicit calculations, such relations will come in handy, especially at higher perturbation orders. Some additional conventions will also be helpful in more complicated examples:

\begin{itemize}
\item from here on we use the following shorthand notation (analogously for more complicated trees):

\begin{equation}
\begin{tikzpicture}
\node at (1.5,0) {$ \mathcal{G}_x\Big($};
\node[coordinate] at (2.8,.7) {} child[level distance=1cm]{[fill] circle (2pt) node[below]{\begin{tiny}$+lm$\end{tiny}}} child[level distance=1cm]{[fill] circle (2pt) node[below]{\begin{tiny}$+l'm'$\end{tiny}}};
\node at (5.5,.7) {$x$} child[level distance=1cm]{[fill] circle (2pt) node[below]{\begin{tiny}$+lm$\end{tiny}}} child[level distance=1cm]{[fill] circle (2pt) node[below]{\begin{tiny}$+l'm'$\end{tiny}}};
\node at (4.25,0){$\Big)=$};
\end{tikzpicture}
\end{equation}

\item whenever we draw a tree with unlabeled leaves, it is understood that we sum over all possible labels for those leaves
\item whenever two leaves are connected (''from below``) by a line with an arrow, then the labels of these leaves are \emph{contracted}; more precisely we attach the label $-(lm)$ to the leaf that is pointed at by the arrow, attach the label $+(lm)$ to the leaf at the other end of the line and sum over $l$ and $m$

\begin{equation}
\begin{tikzpicture}
  \fill (0,0) circle (2pt) node(1) {};
 \fill (2,0) circle (2pt) node(2){};
\path [<-, draw, thick] (1) -- +(0,-.5) -| (2);
\path [-, draw] (0,0) --  (.2,.5);
\path [-, draw,dotted] (.2,.5) -- (.3,.75);
\path [-, draw] (2,0) --  (2,.5);
\path [-, draw,dotted] (2,.5) -- (2,.75);
\path [-, draw,dotted] (.75,0) -- (1.25,0);
\node at (4,0) {$=\quad \sum\limits_{l,m}$};
\fill (6,0) circle (2pt) node[below] {\begin{tiny}-$(lm)$\end{tiny}};
 \fill (8,0) circle (2pt) node[below]{\begin{tiny}$+(lm)$\end{tiny}};
\path [-, draw] (6,0) --  (6.2,.5);
\path [-, draw,dotted] (6.2,.5) -- (6.3,.75);
\path [-, draw] (8,0) --  (8,.5);
\path [-, draw,dotted] (8,.5) -- (8,.75);
\path [-, draw,dotted] (6.75,0) -- (7.25,0);
 \end{tikzpicture} 
\end{equation}

\item if two leaves are connected by a line from below and if this line is dashed, then we contract the labels of these leaves, but replace the corresponding ladder operators by the identity operator $id$ on $V$
\end{itemize}
These conventions substantially simplify the diagrammatic notation, e.g. eq.\eqref{eq:model diagrams Lphi^2} may now be expressed as

\begin{equation}
 \begin{tikzpicture}
  \node at (0,0) {$ \mathcal{Y}(\varphi^2,x)=\quad:$};
\node at (2,.5) {$x$} child[level distance=1cm]{[fill] circle (2pt) } child[level distance=1cm]{[fill] circle (2pt) };
\node at (3.5,0) {$ :$};
 \end{tikzpicture}
\end{equation}
More complicated examples can be found in the following sections.

\section{Low order computations: Next to leading order}\label{subsec:low orders}

From now on we focus on 3-dimensional scalar, massless quantum field theory with $\varphi^6$-interaction (see \eqref{eq:model perturbations equation of motion}). This allows for more explicit results and facilitates calculations. The choice to work on a 3-dimensional spacetime was made in order to keep representation theory simple, as the 3-dimensional rotation group SO$(3)$ is familiar and well understood. The exponent of the interaction term in the field equation was chosen in such a way that the coupling constant is dimensionless.

\subsection{Construction of \texorpdfstring{$\Lo{1}{\varphi}{x}$}{L(1)phi}} \label{subsubsec:L(1)phi}

Recall the results of the previous two sections: In sec. \ref{subsec:free field} we constructed the complete left representation, and hence the complete quantum field theory, for the non-interacting model obeying a field equation of the form \eqref{eq:model free field field equation}. Secondly, in sec. \ref{subsec:perturbation via field equation}, we presented an iterative scheme that allows for the perturbative construction of an interacting quantum field theory, with a non-homogeneous field equation of the form \eqref{eq:model perturbations equation of motion}, given that the corresponding free theory is known (more precisely, we actually only need the left representative $\mathcal{Y}_0(\varphi,x)$ as starting point of the iteration). Combining the information from these chapters, it is clear that we should be able to construct perturbations of the free theory in our setting. The first step beyond leading perturbation order is the calculation of the left representative $\mathcal{Y}_1(\varphi,x)$ according to our algorithm, which is the aim of this section.

Remember from section \ref{subsec:perturbation via field equation} that the field equation can be exploited in order to find relations between OPE coefficients of different perturbation orders, see eq.\eqref{eq:model perturbations relations}. The corresponding relation connecting first order coefficients to those of the free theory is of the form

\begin{equation}
 \square\Co{1}{b}{\varphi a}(x)=\Co{0}{b}{\varphi^k b}(x)\, .
\label{eq:model NLO Lphi connection general}
\end{equation}
As mentioned above, for calculational ease we will in the following consider the specific theory with $\varphi^6$-interaction, i.e. $k=5$ in the above formula, in 3 dimensions. In this context we have the equation

\begin{equation}
 \Delta\Co{1}{b}{\varphi a}(x)=\Co{0}{b}{\varphi^5 b}(x)\, 
\label{eq:model NLO Lphi connection specific}
\end{equation}
where $\Delta=\del^2_{x_1}+\del^2_{x_2}+\del^2_{x_3}$\index{symbols}{Delta@$\Delta$} is the familiar three dimensional Laplace operator. In terms of left representatives, this equation implies

\begin{equation} 
\Delta \mathcal{Y}_1(\varphi,x)=\mathcal{Y}_0(\varphi^5,x)\, .
\label{eq:model NLO Lphi connection specific L}
\end{equation}
Therefore, to reach the aim of this section, all we have to do is solve this differential equation. Let us first recall our results for the free theory from section \ref{subsec:free field}. In $D=3$ dimensions, the leading order of our theory is characterized by the equations

\begin{equation}
 \mathcal{Y}_0(\varphi,x)=\sum_{l=0}^{\infty}\sum_{m=-l}^{l}\left(\frac{4\pi}{2l+1}\right)^{1/2} \left[r^{l}\, Y^{lm}(\hat{x})\, \mathbf{b}_{lm}^{\dagger}+r^{-(l+1)}\, Y_{lm}(\hat{x})\, \mathbf{b}_{lm}\right]\: ,
\label{eq:model NLO Lphi Lvarphi0}
\end{equation}
which is the $D=3$ case of the general expression for $\mathcal{Y}_0(\varphi,x)$ in eq.\eqref{eq:model free field Lvarphi 2}, and
\begin{equation}
 \mathcal{Y}_0(\ket{v_a},x)=\: :\prod_{l,m}\frac{1}{(a_{lm}!)^{1/2}}[t_{lm}\, \del^l\mathcal{Y}_0(\varphi,x)]^{a_{lm}} :\, .
\label{eq:model NLO Lphi La0}
\end{equation}
An even more compact form of eq.\eqref{eq:model NLO Lphi Lvarphi0} can be obtained with the help of the \emph{modified} spherical harmonics defined as

\begin{equation}
 S_{lm}(\hat{x}):=\left(\frac{4\pi}{2l+1}\right)^{1/2}\, Y_{lm}(\hat{x})\, ,
\label{eq:model NLO Lphi modified spherical harmonics}
\end{equation}
see \cite{Brink1962} and eq.\eqref{eq:app spherical symm D=3 spher harm modified}, which gives
\begin{equation}
 \mathcal{Y}_0(\varphi,x)=\sum_{l=0}^{\infty}\sum_{m=-l}^{l} \left[r^{l}\, S^{lm}(\hat{x})\, \mathbf{b}_{lm}^{\dagger}+r^{-(l+1)}\, S_{lm}(\hat{x})\, \mathbf{b}_{lm}\right]\: .
\label{eq:model NLO Lphi Lvarphi0 kompakt}
\end{equation}
In view of eq.\eqref{eq:model NLO Lphi connection specific L}, the left representative $\mathcal{Y}_0(\varphi^5,x)$ is of special interest to us. According to the equations above, it takes the form

\begin{equation}
\begin{split}
 &\Lo{0}{\varphi^5}{x}=\sum_{l_1,\ldots,l_5}\sum_{m_1,\ldots, m_5}\Big( r^{l_1+\ldots+l_5}\cdot S^{l_1m_1}(\hat{x})\cdots S^{l_5m_5}(\hat{x})\cdot\mathbf{b}_{l_1m_1}^{\dagger}\cdots\mathbf{b}_{l_5m_5}^{\dagger}\\
&+5\cdot r^{l_1+\ldots+l_4-l_5-1}\cdot S^{l_1m_1}(\hat{x})\cdots S^{l_4m_4}(\hat{x}) S_{l_5m_5}(\hat{x})\cdot\mathbf{b}_{l_1m_1}^{\dagger}\cdots\mathbf{b}_{l_4m_4}^{\dagger}\mathbf{b}_{l_5m_5}\\
&+10\cdot r^{l_1+l_2+l_3-l_4-l_5-2}\cdot S^{l_1m_1}(\hat{x})S^{l_2m_2}(\hat{x}) S^{l_3m_3}(\hat{x})S_{l_4m_4}(\hat{x}) S_{l_5m_5}(\hat{x})\cdot\mathbf{b}_{l_1m_1}^{\dagger}\mathbf{b}_{l_2m_2}^{\dagger}\mathbf{b}_{l_3m_3}^{\dagger}\mathbf{b}_{l_4m_4}\mathbf{b}_{l_5m_5}\\
&+10\cdot r^{l_1+l_2-l_3-l_4-l_5-3}\cdot S^{l_1m_1}(\hat{x})S^{l_2m_2}(\hat{x}) S_{l_3m_3}(\hat{x})S_{l_4m_4}(\hat{x}) S_{l_5m_5}(\hat{x})\cdot\mathbf{b}_{l_1m_1}^{\dagger}\mathbf{b}_{l_2m_2}^{\dagger}\mathbf{b}_{l_3m_3}\mathbf{b}_{l_4m_4}\mathbf{b}_{l_5m_5}\\
&+5\cdot r^{l_1-l_2-\ldots-l_5-4}\cdot S^{l_1m_1}(\hat{x})S_{l_2m_2}(\hat{x})\cdots  S_{l_5m_5}(\hat{x})\cdot\mathbf{b}_{l_1m_1}^{\dagger}\mathbf{b}_{l_2m_2}\cdots\mathbf{b}_{l_5m_5}\\
&+r^{-l_1-\ldots-l_5-5}\cdot S_{l_1m_1}(\hat{x})\cdots S_{l_5m_5}(\hat{x})\cdot\mathbf{b}_{l_1m_1}\cdots\mathbf{b}_{l_5m_5}\Big)\, .
\end{split}
\label{eq:model NLO Lphi Lvarphi5 0}
\end{equation}
Here normal ordering has already been carried out, which explains the numerical prefactors in lines 2-5, as e.g. the sums over $:\mathbf{b}_{l_1m_1}^{\dagger}\cdots\mathbf{b}_{l_4m_4}^{\dagger}\mathbf{b}_{l_5m_5}:$ and $:\mathbf{b}_{l_1m_1}^{\dagger}\cdots\mathbf{b}_{l_4m_4}\mathbf{b}_{l_5m_5}^{\dagger}:$ are equivalent after normal ordering. In our graphical notation, this could simply be written as

\begin{equation}
\begin{tikzpicture}
 \node at (0,0) {$\Lo{0}{\varphi^5}{x}=\quad :$};
 \node at (4.3,.9) {$x$} child {[fill] circle (2pt) } child {[fill] circle (2pt) } child {[fill] circle (2pt) } child {[fill] circle (2pt) } child {[fill] circle (2pt) };
\node at (7.5,0) {$:$};
\end{tikzpicture}
\label{eq:model NLO Lphi Lvarphi5 0 baum}
\end{equation}
where we remind the reader that summation over all possible labels for the 5 unlabeled leaves is implicitly assumed.

Before solving the differential equation \eqref{eq:model NLO Lphi connection specific L}, we want to put $\Lo{0}{\varphi^5}{x}$ into a form better suited for this task. The main obstacle in trying to invert the Laplace operator on eq.\eqref{eq:model NLO Lphi Lvarphi5 0} is the product of the spherical harmonics. The differential equation is simplified decisively if these products are decomposed into their irreducible parts, i.e. if we "couple" the spherical harmonics using as intertwiners the familiar \emph{Clebsch-Gordan} coefficients (or, equivalently, \emph{Wigner $3j$-symbols}). For details on the representation theory of the 3-dimensional rotation group, see appendix \ref{app:3-dim symmetries} and references therein. The identity we are looking for is eq.\eqref{eq:app spherical symm D=3 spherical harmonics coupling completely general}, which was derived in the appendix, and we repeat it here for reference:

\begin{equation}
\begin{split}
 S^{l_1m_1}(\hat{x})\cdots S^{l_pm_p}(\hat{x})&S_{l_{p+1}m_{p+1}}(\hat{x})\cdots S_{l_qm_q}(\hat{x})=\\
 &\sum_J T[+(l_1m_1),\ldots +(l_pm_p),-(l_{p+1}m_{p+1}),\ldots. -(l_qm_q)]_{JM}\, S_{JM}(\hat{x})
\end{split}
\label{eq:app spherical symm D=3 spherical harmonics coupling completely general ref}
\end{equation}
Here $T$ is a tensor, which contains the consecutive application of the intertwiners as defined in eqs.\eqref{eq:app spherical symm D=3 spherical harmonics coupling tensor T} and \eqref{eq:app spherical symm D=3 spherical harmonics coupling tensor T general}. Before we apply this equation on the product of the spherical harmonics in eq.\eqref{eq:model NLO Lphi Lvarphi5 0}, let us first introduce some abbreviations:

\begin{definition}[Abbreviated indices]
Let $q,i,l_i\in\mathbb{N}$ and $m_i\in\{-l_i,\ldots,l_i\}$. Then we define the abbreviated notation

\begin{equation}
 \mathfrak{l}_i^q=\left\{\begin{array}{ll} +(l_im_i)\quad &\text{for }i> q \\  -(l_im_i) &\text{for }i\leq q   \end{array}\right.\quad.
\label{eq:model NLO Lphi C1phi tree index}
\end{equation} 
Further, let $\mathbf{b}_{+(lm)}:=\mathbf{b}_{lm}^{\dagger}$ and $\mathbf{b}_{-(lm)}:=\mathbf{b}_{lm}$, $\sum\limits_{\mathfrak{l}^q_i=a}^b=\sum\limits_{l_i=a}^b\sum\limits_{m_i=-l_i}^{l_i}$ and $e_{\pm(lm)}=\pm e_{lm}$ (recall that $e_{lm}$ is the multiindex with unit entry at ''position`` $(lm)$).
\end{definition}\index{symbols}{liq@$\mathfrak{l}^q_i$}
With this notation and representation theoretic machinery, we can put eq.\eqref{eq:model NLO Lphi Lvarphi5 0} into the convenient form

\begin{equation}
 \Lo{0}{\varphi^5}{x}=\sum_{q=0}^5\sum_{\mathfrak{l}^q_1,\ldots,\mathfrak{l}^q_5=0}^{\infty}\sum_J \left(\atop{5}{q}\right) \, r^{(-l_1-\ldots-l_q+l_{q+1}+\ldots+l_5-q)}\cdot T[\mathfrak{l}^q_1,\ldots,\mathfrak{l}^q_5]_{JM}S_{JM}(\hat{x})\, \mathbf{b}_{\mathfrak{l}^q_5}\cdots\mathbf{b}_{\mathfrak{l}^q_1}\quad ,
\label{eq:model NLO Lphi Lvarphi5 0 irrep}
\end{equation}
which is much better suited for our purpose, i.e. for solving the differential equation \eqref{eq:model NLO Lphi connection specific L}, since now we have expressed $\Lo{0}{\varphi^5}{x}$ as an element of the ring of functions $\mathbb{Y}(x)$ defined in eq.\eqref{eq:app diff eq ring Y}. Assuming that any left representative can be written as an element of this ring, a solution to the differential equation is simply found by application of the operator $\square^{-1}\in\End(\mathbb{Y}(x))$, which we call $\Delta^{-1}$\index{symbols}{Deltainv@$\Delta^{-1}$} in three spacetime dimensions. Recall, however, that this solution is not unique. For a discussion of these ambiguities see section \ref{subsec:ambiguities}.

Before we give the result for the left representative $\Lo{1}{\varphi}{x}$, let us first introduce some more notation for convenience:

\begin{definition}[Gradings of left representatives]{\ \\}
 The decomposition of an arbitrary left representative $\Lo{n}{\ket{v_a}}{x}$ with respect to its dimension is written as

\begin{equation}
\Lo{n}{\ket{v_a}}{x}=:\sum_{d=-\infty}^{\infty}\, \Lo{n}{\ket{v_a}}{x;d}\cdot r^d\, ,
\label{eq:model NLO Lphi grading dimension d}
\end{equation}
while the grading by irreducible representations of the rotation group SO$(3)$ is introduced as

\begin{equation}
\Lo{n}{\ket{v_a}}{x}=:\sum_{J=0}^{\infty}\sum_{M=-J}^J\, \Lo{n}{\ket{v_a}}{x}_{JM}\cdot S_{JM}(\hat{x})\, .
\label{eq:model NLO Lphi grading irrep J}
\end{equation}\index{symbols}{Yidj@$\Lo{n}{\ket{v_a}}{x;d}_J$}

\end{definition}
Then we finally obtain $\Lo{1}{\varphi}{x}$ by application of the operator $\Delta^{-1}$, i.e. eq.\eqref{eq:app diff eq inverse} in $D=3$ dimensions, to $\Lo{0}{\varphi^5}{x}$ in the form of eq.\eqref{eq:model NLO Lphi Lvarphi5 0 irrep}. Using the gradings introduced above, we find:

\begin{results}[The left representative $\Lo{1}{\varphi}{x}$]
\begin{equation}
 \Lo{1}{\varphi}{x}=
\sum_{d=-\infty}^{\infty}\sum_{J=0}^{\infty}\sum_{M=-J}^J\Lo{0}{\varphi^5}{x;d}_{JM}\cdot r^{d+2} S_{JM}(\hat{x})\cdot D(d+2,J,r)   \, ,
\label{eq:model NLO Lphi L1phi}
\end{equation}
\end{results}
where we defined (see also eq.\eqref{eq:app diff eq D} for $q=0$ and $D=3$)

\begin{equation}
 D(d,J,r):=\left\{\begin{array}{cl}
 \frac{1}{d(d+1)-J(J+1)}\quad & \text{for }\min\{|d|,|d+1|\}\neq J\\ & \\
 \frac{\log r}{2d+1} & \text{for }\min\{|d|,|d+1|\}= J
\end{array}
\right.
\label{eq:model NLO Lphi D}
\end{equation}\index{symbols}{D0@$D(d,J,r), D^{(0)}(d,J,r)$}
for the sake of brevity. The diagram corresponding to this equation is simply 

\begin{equation}
\begin{tikzpicture}
 \node at (0,0) {$\Lo{1}{\varphi}{x}=\quad :$};
 \node at (3.5,.9) {$x$} child[level distance=.5cm]{ [fill] circle (2pt) child[level distance=1cm, sibling distance=1cm] {[fill] circle (2pt) } child[level distance=1cm, sibling distance=1cm] {[fill] circle (2pt) } child[level distance=1cm, sibling distance=1cm] {[fill] circle (2pt) } child[level distance=1cm, sibling distance=1cm] {[fill] circle (2pt) } child[level distance=1cm, sibling distance=1cm] {[fill] circle (2pt) }};
\node at (6.7,0) {$:\quad =$};
\node at (9.5,.9) {$x$} child[level distance=.5cm]{ [fill] circle (2pt) child[level distance=1cm, sibling distance=1cm] {[fill] circle (2pt) } child[level distance=1cm, sibling distance=1cm] {[fill] circle (2pt) } child[level distance=1cm, sibling distance=1cm] {[fill] circle (2pt) } child[level distance=1cm, sibling distance=1cm] {[fill] circle (2pt) } child[level distance=1cm, sibling distance=1cm] {[fill] circle (2pt) }};
\end{tikzpicture}
\label{eq:model NLO Lphi Lvarphi 1 baum}
\end{equation}
The difference to the diagram for the left representative $\Lo{0}{\varphi^5}{x}$ in eq.\eqref{eq:model NLO Lphi Lvarphi5 0 baum} is the appearance of a vertex below the root, which according to the rules of def. \ref{def:Diagrammatic notation} stands for the inversion of the Laplace operator. In the second equality we dropped the double dots, as from here on we implicitly assume normal ordering whenever 5 leaves meet at a vertex. In summary, we have successfully determined the first NLO left representative. \\

We end this section with the computation of the explicit form of the first order OPE coefficient $\Co{1}{b}{\varphi a}$. We will proceed as follows: After introducing some additional notation to keep the equations at a reasonable length, we first give a general result for the OPE coefficients of the free theory. Then the differential equation \eqref{eq:model NLO Lphi connection specific} will be applied in order to find the desired first order coefficients.

\begin{definition}[Metric on $V$]
For two vectors $\ket{v_a},\ket{v_b}\in V$ we introduce the metric\footnotemark{} $g:V\otimes V\to \mathbb{N}$
\begin{equation}
 g(\ket{v_a},\ket{v_b})=g(a,b):=\sum_{l=0}^{\infty}\sum_{m=-l}^l |a_{lm}-b_{lm}|\quad ,
\label{eq:model NLO Lphi multiindex metric}
\end{equation}
where $|\cdot|$ denotes the usual absolute value of an integer. This notion measures how much $\ket{v_a}$ differs from $\ket{v_b}$, in the sense that at least $g(a,b)$ ladder operators are needed to transform one of the vectors into the other. Note that $g(a,b)$ is always finite, since we required the multi-indices labeling our basis elements to have only finitely many non-zero entries (see eq.\ref{eq:model free field basis elements}).
\end{definition}\index{symbols}{g(ab)@$g(a,b)$}\footnotetext{This notion is not directly related to and should not be confused with the canonical dimension $|a|$ defined in eq.\eqref{eq:model free field basis elements dimension}.}
As we will see below, OPE coefficients $\Co{1}{b}{\varphi^k\, a}$ with equal values of $g(a,b)$ exhibit structural similarities, so we will often classify OPE coefficients by this value. Further, the notion of \emph{multiset} will be frequently used (see appendix \ref{app:multisets}). In this context, we make the following definitions:

\begin{definition}[Multisets of indices]\index{symbols}{Iabn@$\mathcal{I},\mathcal{I}(n),\mathcal{I}^b_a(n)$}
Let $\mathfrak{l}^q_i$ be defined as in def. \ref{def:Abbreviated indices} and let $\ket{v_a},\ket{v_b}\in V$. We define the following three sets:

\begin{align}
 \mathcal{I}&=\Big\{ \mathfrak{A}=\Lbag \mathfrak{l}^q_1,\ldots,\mathfrak{l}^q_n \Rbag\, |\, n\in\mathbb{N}\, , q\in\{0,\ldots,n\} \Big\}\\
\mathcal{I}(n)&=\Big\{ \mathfrak{A}\in\mathcal{I}\, |\, \operatorname{card}(\mathfrak{A})=n \Big\}=\Big\{ \mathfrak{A}=\Lbag \mathfrak{l}^q_1,\ldots,\mathfrak{l}^q_n \Rbag\, | q\in\{0,\ldots,n\} \Big\}\\
\mathcal{I}^b_a(n)&=\Big\{ \mathfrak{A}\in\mathcal{I}(n)\, |\, \bra{v_b}:\prod_{\mathfrak{l}^q_i\in\mathfrak{A}}\mathbf{b}_{\mathfrak{l}^q_i}:\ket{v_a}\neq 0 \Big\}
\end{align}

\end{definition}
The set $\mathcal{I}$ may be interpreted as the set of all possible labels for any number of leaves of a tree (see section \ref{subsec:diagrammatic notation}). Then $\mathcal{I}(n)$ is the subset of $\mathcal{I}$ containing all labels for trees with $n$ leaves. Recall from section \ref{subsec:diagrammatic notation} that we associate ladder operators to the leaves of a tree. The set $\mathcal{I}^b_a(n)$ is the restriction of $\mathcal{I}(n)$ to the labels of those trees that transform the vector $\ket{v_a}$ into $\ket{v_b}$ (recall that our basis is orthogonal, thus $\braket{v_b}{v_a}=0$ for $a\neq b$) after normal ordering of the corresponding ladder operators.

\begin{proposition}\label{propos:labelings}
From the definitions above it follows that

\begin{equation}
 \operatorname{card}\Big(\mathcal{I}^b_a(n)\Big)=1 \quad\text{if }g(a,b)=n
\label{eq:model NLO Lphi labels propos1}
\end{equation}
and
\begin{equation}
\mathcal{I}^b_a(n)=\emptyset\quad \text{for }g(a,b)>n\text{ or }g(a,b)+n=\text{odd .}
\label{eq:model NLO Lphi labels propos2}
\end{equation}
\end{proposition}
\begin{proof}
 The first statement, eq.\eqref{eq:model NLO Lphi labels propos1} follows straightforwardly from the definition of $g(a,b)$: The condition $g(a,b)=n$ tells us that at least $n$ ladder operators are needed to transform $\ket{v_a}$ into $\ket{v_b}$. In other words, the multiindices $a$ and $b$ differ in $n$ entries. Therefore the multiset $\mathfrak{A}=\Lbag\mathfrak{l}^q_1,\ldots,\mathfrak{l}^q_n\Rbag\in\mathcal{I}^b_a(n)$ is uniquely fixed by the condition $\bra{v_b}:\prod_{\mathfrak{l}^q_i\in\mathfrak{A}}\mathbf{b}_{\mathfrak{l}^q_i}:\ket{v_a}\neq 0$, so $\mathcal{I}^b_a(n)$ contains only one element if $g(a,b)=n$.

Let us come to the second part of the proposition. For $g(a,b)>n$ the set $\mathcal{I}^b_a(n)$ is empty, because more than $n$ ladder operators are needed to tranform $\ket{v_a}$ into $\ket{v_b}$. Thus, the condition on the multisets $\mathfrak{A}$ in the definition of $\mathcal{I}^b_a(n)$ cannot be satisfied due to orthogonality of our basis. Hence

\begin{equation}
 \mathcal{I}^b_a(n)=\emptyset \text{ for }g(a,b)>n.
\end{equation}
Similarly, the set is empty for $g(a,b)+n=\text{odd}$, which can be seen as follows: If $g(a,b)$ is an even number, then we need an even number of ladder operators to transform $\ket{v_a}$ into $\ket{v_b}$ (the minimum number $g(a,b)$ plus pairs of \emph{contracted} operators, i.e. creation and annihilation operators with the same indices). Thus, $\mathcal{I}^b_a(n)$ vanishes if $n=\text{odd}$ in this case. Similarly, the set is empty for $n=\text{even}$ if $g(a,b)$ is odd, so

\begin{equation}
 \mathcal{I}^b_a(n)=\emptyset \text{ for }g(a,b)+n=\text{odd .}
\end{equation}

\end{proof}

\begin{definition}[Shorthand notation for dimension of OPE coefficients]
Given any three basis vectors $\ket{v_a},\ket{v_b},\ket{v_c}\in V$ , we define

\begin{equation}
 \mathbf{d}=|c|-|a|-|b|=-sd\,\Co{i}{c}{ab}
\end{equation}
(recall the definition of $|\cdot|$ for multiindices from eq.\eqref{eq:model free field basis elements dimension} and the definition of scaling degree from eq.\eqref{eq:framework axioms scaling degree}). Further, given a multiset $\mathfrak{A}=\Lbag\mathfrak{l}^q_1,\ldots,\mathfrak{l}^q_n\Rbag\in\mathcal{I}(n)$ and two basis vectors $\ket{v_a},\ket{v_b}\in V$ satisfying $b=a+\sum_{i=1}^n e_{\mathfrak{l}_i^q}$, let
\begin{equation}
 d_{\mathfrak{A}}:= \sum_{i=q+1}^{n}\left(l_i+\frac{1}{2}\right)-\sum_{j=1}^{q}\left(l_j+\frac{1}{2}\right)-\frac{1}{2}=-sd\, \Co{i}{b}{\varphi a}\quad.
\label{eq:model NLO Lphi2 additional notation d}
\end{equation}\index{symbols}{d@$\mathbf{d}$}\index{symbols}{da@$d_{\mathfrak{A}}$}
\end{definition}

\begin{definition}[Shorthand notation for coupling tensors]
\index{symbols}{TA@$T[\mathfrak{A}]_{J}$}
Let $\mathfrak{A}\in\mathcal{I}$. Then we define
\begin{equation}
 T[\mathfrak{A}=\Lbag \mathfrak{l}_1^q,\ldots,\mathfrak{l}_n^q\Rbag]_{J}:=T[\mathfrak{l}^q_{i_1},\ldots,\mathfrak{l}^q_{i_n}]_{JM}\qquad \text{with }l_{i_{j}}\geq l_{i_{j+1}} \, \forall \, j\in\{1,\ldots,n-1\}\quad.
\end{equation}
\end{definition}
Note that on the left hand side we suppressed the \emph{magnetic quantum number} $M$, which is uniquely encoded in the multiset $\mathfrak{A}$ as the sum $m_1+\ldots+m_n$, for the sake of brevity. The definition of $T[\mathfrak{A}]$ is motivated by the fact that in many expressions we are free to choose a coupling order for the spherical harmonics appearing in the construction. For such cases we introduce the convention that always the spherical harmonics of highest degree are coupled first. Then the coupling tensor is uniquely characterized by a multiset of indices instead of a (ordered) tuple. Note that there seems to be an ambiguity in this prescription, e.g. if we consider the multiset $\mathfrak{A}=\Lbag +(l_1m_1),+(l_2m_2),-(l_1m_1)  \Rbag$. In this case, it is not clear whether to couple $+(l_1m_1)$ before $-(l_1m_1)$ or the other way around. However, in such cases, i.e. whenever $l_i=l_j$ with $i\leq q < j$, we may apply the addition theorem of the spherical harmonics, eq.\eqref{eq:app spherical symm D=3 addition theorem}, and omit the indices $\mathfrak{l}^q_i$ and $\mathfrak{l}^q_j$ altogether.

\begin{equation}
 T[\Lbag\mathfrak{l}^q_1,\ldots,\mathfrak{l}^q_i,\ldots,\mathfrak{l}^q_j,\ldots,\mathfrak{l}^q_n\Rbag]_{J}=T[\Lbag\mathfrak{l}^q_1,\ldots,\widehat{\mathfrak{l}^q_i},\ldots,\widehat{\mathfrak{l}^q_j},\ldots,\mathfrak{l}^q_n\Rbag]_{J}\qquad \text{if }l_i=l_j\text{ and }i\leq q < j
\end{equation}
This solves the apparent ambiguity. Another property to be remarked is that \emph{zero entries} may be dropped from the coupling tensor, i.e.

\begin{equation}
 T[\Lbag\mathfrak{l}^q_1,\ldots,\mathfrak{l}^q_{i-1},\mathfrak{l}^q_i,\mathfrak{l}^q_{i+1},\ldots,\mathfrak{l}^q_n\Rbag]_{J}=T[\Lbag\mathfrak{l}^q_1,\ldots,\mathfrak{l}^q_{i-1},\mathfrak{l}^q_{i+1},\ldots,\mathfrak{l}^q_n\Rbag]_{J}\qquad \text{if }l_i=0
\end{equation}
This may be seen either from the property

\begin{equation}
 \braket{l0\, m0}{JM}=\delta_{Jl}\delta_{Mm}
\end{equation}
of the Clebsch-Gordan coefficients, or from the fact that the coupling of the spherical harmonic $S_{00}(\hat{x})=1$ is trivial. Is also follows that

\begin{equation}
 T[\Lbag\mathfrak{l}^q_1\Rbag]_J=\delta_{Jl_1}\quad.
\end{equation}

\begin{definition}[Symmetry factor]
\index{symbols}{sA@$s[\mathfrak{A}]$}Let $s:\mathcal{I}\to \mathbb{N}$ be the map defined by

\begin{equation}
 s[\mathfrak{A}\in\mathcal{I}]=\operatorname{card}\left(\Big\{(a_1,\ldots,a_n)\,|\, \Lbag a_1,\ldots,a_{n}\Rbag = \mathfrak{A}\Big\}\right)\quad,
\end{equation}
where $(a_1,a_2,\ldots,a_n)$ denotes (ordered) tuples. Thus, $s[\mathfrak{A}]$ is the number of different tuples constructed from the elements of $\mathfrak{A}$.
\end{definition}
Remarks: Equivalently, $s[\mathfrak{A}]$ may be interpreted as the number of different trees obtained from

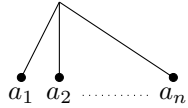
\begin{equation}
 \begin{tikzpicture}
 \node[coordinate] at (0,0){} child[level distance=1cm, sibling distance=.5cm]{[fill] circle (1.5pt) node[below]{\begin{small}$a_1$\end{small}
}} child[level distance=1cm, sibling distance=1cm]{[fill] circle (1.5pt) node[below](a){\begin{small}$a_2$\end{small}
}} child[level distance=1cm, sibling distance=1.5cm]{[fill] circle (1.5pt) node[below](b){\begin{small}$a_n$\end{small}
}};
\draw[dotted] (a) -- (b);
\end{tikzpicture}
\end{equation}
by permutations of the leaves. An explicit formula for $s[\mathfrak{A}]$ is most conveniently obtained if the multiset $\mathfrak{A}$ is expressed in terms of a set $A$ and a function $f:A\to\mathbb{N}$ (see def. \ref{def:Multiset}). Then

\begin{equation}
 s[\mathfrak{A}=(A,f)\in\mathcal{I}(n)]=\frac{n!}{\prod_{a\in A}f(a)!}\quad.
\end{equation}

\begin{definition}[Prefactors from ladder operators]
\index{symbols}{faA@$f^b_a[\mathfrak{A}]$}Let $f:V\otimes V \otimes \mathcal{I}(n)\to \mathbb{R}$ be the map defined by the following property: Given two vectors $\ket{v_a},\ket{v_b}\in V$, a number $n\in\mathbb{N}$ and a multiset $\mathfrak{A}\in\mathcal{I}^b_a(n)$, let

\begin{equation}
 f^b_a[\mathfrak{A}]=\bra{v_b}:\prod_{\mathfrak{l}^q_i\in\mathfrak{A}}\mathbf{b}_{\mathfrak{l}^q_i}:\, \ket{v_a}\quad.
\label{eq:model NLO Lphi prefactor f}
\end{equation}
Explicit expressions for $f^b_a(\mathfrak{A})$ can be obtained using eqs.\eqref{eq:model free field ladder operators1} and \eqref{eq:model free field ladder operators2}.
\end{definition}
Now we are ready to express the zeroth order OPE coefficients in the compact form

\begin{results}[OPE coefficients $\Co{0}{b}{\varphi^k\, a}$]\label{res:model NLO Lphi C0phik}%
Using the notation introduced above, the zeroth order OPE coefficients take the compact form
\begin{equation}
\Co{0}{b}{\varphi^k\, a}(x)=\begin{cases}
                           0 & \text{for }g(a,b)>k\\
 0 & \text{for }g(a,b)+k=\text{odd}\\
 \sum\limits_{\mathfrak{A}\in\mathcal{I}^b_a(k)}\sum\limits_{J}f_a^b[\mathfrak{A}]\cdot \Delta_{0}[\mathfrak{A}]_{J}\cdot S_J(\hat{x})\cdot r^{\mathbf{d}} \quad & \text{otherwise}
                          \end{cases}
\label{eq:model NLO Lphi C0phik result}
\end{equation}
with\footnotemark{}
\begin{equation}
 \Delta_{0}[\mathfrak{A}]_{J}=s[\mathfrak{A}]\cdot T[\mathfrak{A}]_J
\label{eq:model NLO Lphi2 additional notation D0}
\end{equation}\index{symbols}{delta00@$\Delta_{0}[\mathfrak{A}]_{J}$}
\end{results}\footnotetext{This function $\Delta_0:\mathcal{I}\to\mathbb{N}$, and also the functions $\Delta_1,\ldots$ to be defined below, should not be confused with the Laplace operator $\Delta$.}%
Remark: Here we also suppressed the magnetic quantum number of the spherical harmonic $S_{JM}(\hat{x})$. As in the definition of the coupling tensor $T[\mathfrak{A}]_J$, the corresponding number $M$ can always be uniquely retrieved from the elements of the multiset $\mathfrak{A}$.
\begin{proof}
Using the graphical notation of section \ref{subsec:diagrammatic notation}, we may express the left representative $\Lo{0}{\varphi^k}{x}$ by a sum over labeled trees with $k$ leaves. If we take the matrix elements $\bra{v_b}\cdot\ket{v_a}$ of this expression, then all contributions to this sum which do not transform $\ket{v_a}$ into $\ket{v_b}$ vanish due to orthogonality of the basis. Thus, the labels $\mathfrak{l}^q_1,\ldots,\mathfrak{l}^q_k$ of the $k$ leaves have to satisfy $b=a+\sum_{i=1}^ke_{\mathfrak{l}^q_i}$. Further, recall that the operators corresponding to the leaves of the tree are normal ordered, so the order of the labelings does not matter. In other words, we obtain the same result for all permutations of the leaves. So instead of summing over all possible tress satisfying the mentioned property, we might as well restrict the sum to such trees that can not be transformed into one another by permutations of the leaves and multiply by a symmetry factor of the respective tree. These arguments imply the equation

\begin{equation}
 \begin{tikzpicture}
  \node at (.5,-.05){$\Co{0}{b}{\varphi^k a}(x)=\bra{v_b}\, :$};
\node at (3,.5) {$x$}   child[level distance=1cm, sibling distance=.5cm] {[fill] circle (1.5pt) } child[level distance=1cm,sibling distance=1cm] {[fill] circle (1.5pt) node(c){} } child[level distance=1cm,sibling distance=1cm] {[fill] circle (1.5pt) node(d){} };
\node at (5.6,-.2){$:\, \ket{v_a}=\sum\limits_{\mathfrak{A}\in\mathcal{I}^b_a(k)}\bra{v_b}\, :$};
\node at (8.3,.5) {$x$}  child[level distance=1cm,sibling distance=.5cm] {[fill] circle (1.5pt) node[below]{\begin{tiny}$\mathfrak{l}_1^q$\end{tiny}} } child[level distance=1cm,sibling distance=.5cm] {[fill] circle (1.5pt) node[below](a){\begin{tiny}$\mathfrak{l}_2^q$\end{tiny}} }  child[level distance=1cm,sibling distance=1cm] {[fill] circle (1.5pt) node[below](b){\begin{tiny}$\mathfrak{l}_k^q$\end{tiny}} };
\node at (10.6,0){$:\, \ket{v_a}\cdot s[\mathfrak{A}]$};
\draw[dotted] (a) -- (b);
\draw[dotted] (c) -- (d);
 \end{tikzpicture}\, ,
\label{eq:model NLO Lphi C1phi g(a,b)=5 diagram}
\end{equation}
since the elements of the multisets in $\mathcal{I}^b_a(k)$ satisfy the desired property, since the order of the entries in a multiset does not matter and since the factor $s[\mathfrak{A}]$ was defined to be the mentioned symmetry factor of the diagram. We found in proposition \ref{propos:labelings} that the set $\mathcal{I}^b_a(k)$ is empty for $g(a,b)>k$ or $g(a,b)+k=\text{odd}$, which implies $\Co{0}{b}{\varphi^k\, a}(x)=0$ in those cases just as we have claimed in the result.

Now it remains to translate these diagrams into explicit formulae using the rules of section \ref{subsec:diagrammatic notation}. To begin with, the action of the ladder operators associated to the leaves of the tree on the vector $\ket{v_a}$ gives a numerical prefactor, which by definition \ref{def:Prefactors from ladder operators} is denoted by $f^b_a[\mathfrak{A}]$. Further, for each leaf we obtain a spherical harmonic whose indices are determined by the label of the leaf. ``Coupling'' of these spherical harmonics, i.e. a decomposition into irreducible parts, yields a factor $\sum_JT[\mathfrak{A}]_JS_J(\hat{x})$. It remains to determine the power of $r$, which by definition is just $\mathbf{d}=|b|-|a|-k/2$, so we indeed arrive at eq.\eqref{eq:model NLO Lphi C0phik result}.
\end{proof}
With this result at hand, it is an easy task to find the

\begin{results}[OPE coefficients $\Co{1}{b}{\varphi\, a}$]\label{res:model NLO Lphi C1phi}%
Using the right inverse $\Delta^{-1}$ to solve the differential equation \eqref{eq:model NLO Lphi connection specific}, one obtains
\begin{equation}
\Co{1}{b}{\varphi\, a}(x)=\begin{cases}
                           0 & \text{for }g(a,b)>5\\
 0 & \text{for }g(a,b)=\text{even}\\
 \sum\limits_{\mathfrak{A}\in\mathcal{I}^b_a(5)}\sum\limits_{J}
f^b_a[\mathfrak{A}]\cdot
S_J(\hat{x})\cdot r^{\mathbf{d}} \D{1}{A}{J} 
\quad & \text{otherwise}
                          \end{cases}
\label{eq:model NLO Lphi C1phi result}
\end{equation}
with
\begin{equation}
 \D{1}{A}{J}=s[\mathfrak{A}]\cdot T[\mathfrak{A}]_JD[d_{\mathfrak{A}},J,r]
\label{eq:model NLO Lphi2 additional notation D1}
\end{equation}\index{symbols}{delta1p@$\D{1}{A}{J}$}%
and where $D(d,r,J)$ is defined as in eq.\eqref{eq:model NLO Lphi D}.
\end{results}

\begin{proof}
 The differential equation to be solved is

\begin{equation}
 \Delta\Co{1}{b}{\varphi a}(x)=\Co{0}{b}{\varphi^5 b}(x)\, .
\label{eq:model NLO Lphi connection specific2}
\end{equation}
Inserting the concrete form of the zeroth order coefficient, we arrive at the equation

\begin{equation}
 \Delta\Co{1}{b}{\varphi a}(x)=\begin{cases}
                           0 & \text{for }g(a,b)>5\\
 0 & \text{for }g(a,b)+5=\text{odd}\\
 \sum\limits_{\mathfrak{A}\in\mathcal{I}^b_a(5)}\sum\limits_{J}f^b_a[\mathfrak{A}]\cdot \Delta_{0}[\mathfrak{A}]_{J}\cdot S_J(\hat{x})\cdot r^{\mathbf{d}} \quad & \text{otherwise}
                          \end{cases}
\end{equation}
Application of the operator $\Delta^{-1}$ to the right side of the equation yields $\Co{1}{b}{\varphi a}(x)=0$ for $g(a,b)>5$ or $g(a,b)=\text{even}$ and otherwise

\begin{equation}
 \Co{1}{b}{\varphi a}(x)= \sum\limits_{\mathfrak{A}\in\mathcal{I}^b_a(5)}\sum\limits_{J} f^b_a[\mathfrak{A}]\Delta_{0}[\mathfrak{A}]_{J}\cdot S_J(\hat{x})\cdot r^{\mathbf{d}} D(\mathbf{d},J,r)\quad .
\end{equation}
Note that in the case at hand $\mathbf{d}=d_{\mathfrak{A}}$, so the result is confirmed.
\end{proof}
In summary, we have found explicit expressions for all OPE coefficients of the type $\Co{1}{b}{\varphi a}$. Some particularly simple coefficients are tabulated in appendix \ref{app:OPE table}. As an illustration of how to obtain these expressions from result \ref{res:model NLO Lphi C1phi}, we close this section with an example computation.

\begin{example}

We want to determine the OPE coefficient $\Co{1}{\varphi^{p+2}\del^{l'}\varphi}{\varphi\, \varphi^p}$ with $p,l'\neq 0$. As a first step, let us determine the set $\mathcal{I}^{\varphi^{p+2}\del^{l'}\varphi}_{\varphi^p}(5)$. The constraint $\bra{v_b}:\prod_{\mathfrak{l}^q_i\in\mathfrak{A}}\mathbf{b}_{\mathfrak{l}^q_i}:\ket{v_a}\neq 0$ in the definition of the set $\mathcal{I}^b_a$ leaves only one allowed multiset, namely

\begin{equation}
 \mathcal{I}^{\varphi^{p+2}\del^{l'}\varphi}_{\varphi^p}(5)=\Big\{\mathfrak{A}=\Lbag +(00),+(00),+(00),-(00),+(l'm') \Rbag\Big\}\quad.
\end{equation}
The numerical prefactor $f^b_a[\mathfrak{A}]$ is in this case

\begin{equation}
 f^{{\varphi^{p+2}\del^{l'}\varphi}}_{\varphi^p}[\mathfrak{A}=\Lbag +(00),+(00),+(00),-(00),+(l'm') \Rbag]=p
\end{equation}
because the creation operators do not yield any prefactor, while the action of the annihilation operator gives the factor $p$, i.e. $\mathbf{b}^{\dagger}_{00}\varphi^p=\varphi^{(p+1)}$ but $\mathbf{b}_{00}\varphi^p=p\varphi^{(p-1)}$. Now we come to $\D{0}{A}{J}$. First, the symmetry factor is

\begin{equation}
 s[\mathfrak{A}=\Lbag +(00),+(00),+(00),-(00),+(l'm') \Rbag]=5\cdot 4=20\quad.
\end{equation}
Recall that the product $T[\mathfrak{A}]_JS_J(\hat{x})$ resulted from the coupling of five spherical harmonics whose indices are the elements of $\mathfrak{A}$. Since $S_{00}(\hat{x})=1$, we only have one non-trivial spherical harmonic, $S^{l'm'}(\hat{x})$, in the case at hand. Thus

\begin{equation}
 T[\mathfrak{A}=\Lbag +(00),+(00),+(00),-(00),+(l'm') \Rbag]_JS_J(\hat{x})=S^{l'm'}(\hat{x})\quad.
\end{equation}
The value of $\mathbf{d}=-sd\,\Co{1}{\varphi^{p+2}\del^{l'}\varphi}{\varphi\, \varphi^p}$ of the coefficient is

\begin{equation}
 \mathbf{d}=|\varphi^{p+2}\del^{l'}\varphi|-|\varphi^p|-|\varphi|=l'+1\quad.
\end{equation}
Finally, we find
\begin{equation}
 D[d_{\mathfrak{A}}=\mathbf{d}=l'+1,J=l',r]=\frac{1}{2(l'+1)}\quad,
\end{equation}
which gives the total result (see also appendix \ref{app:OPE table})
\begin{equation}
 \Co{1}{\varphi^{p+2}\del^{l}\varphi}{\varphi\, \varphi^p}=10p\sum_{m}S^{lm}(\hat{x})r^{l+1}\cdot \frac{1}{l+1}\quad .
\end{equation}

\end{example}

\newpage
\subsection{Construction of \texorpdfstring{$\Lo{1}{\varphi^2}{x}$}{L(1)phi2}}\label{subsubsec:L1phi^2}

According to our algorithm of section \ref{subsec:perturbation via field equation}, it is now possible to construct all first order left representatives $\Lo{1}{\ket{v_a}}{x}$ (and thus the complete quantum field theory at first perturbation order) just from $\Lo{1}{\varphi}{x}$ by application of the associativity condition. In this section we take the first step into this direction by constructing the left representative $\Lo{1}{\varphi^2}{x}$. As we will see in the following, this task turns out to be comparably simple, as no "renormalization" is needed at this stage in our toy model. More serious calculational effort will be needed in the next section.

The relation between $\Lo{1}{\varphi}{x}$, which is known from the previous section, and the desired left representative $\Lo{1}{\varphi^2}{x}$, is

\begin{equation}
 \Lo{1}{\varphi^2}{x}=\lim_{y\to x}\left[\Lo{0}{\varphi}{y}\Lo{1}{\varphi}{x}+\Lo{1}{\varphi}{y}\Lo{0}{\varphi}{x}-\sum_{e}^{|e|\leq\varphi^2}\Co{1}{e}{\varphi\,, \varphi}(x-y)\Lo{0}{\ket{v_e}}{x}\right]\, ,
\label{eq:model NLO Lphi2 limit}
\end{equation}
which is just eq.\eqref{eq:model perturbations limit} at first order expressed in terms of left representatives. Let us take a closer look at the counterterms, i.e. the expressions subtracted on the right side. The fact that $\Co{1}{b}{\varphi\, a}$ vanishes for $g(a,b)=$even together with the restriction $|e|\leq|\varphi^2|$ suggests that only the terms including $\Co{1}{\mathds{1}}{\varphi\, \varphi}(x-y)$ or $ \Co{1}{\varphi^2}{\varphi\, \varphi}(x-y)$ survive the limit. Due to the specific form of $\Lo{1}{\varphi}{x}$ from the previous section, however, these matrix elements, or OPE coefficients, both vanish, see eq.\eqref{eq:model NLO Lphi C1phi result} (the set $\mathcal{I}^b_a(5)$ is empty in those cases) or table \ref{app:OPE table}. The absence of counterterms from eq.\eqref{eq:model NLO Lphi2 limit} implies that the remaining expressions are finite, so we may simply perform the limit and write

\begin{equation}
 \Lo{1}{\varphi^2}{x}=\Lo{0}{\varphi}{x}\Lo{1}{\varphi}{x}+\Lo{1}{\varphi}{x}\Lo{0}{\varphi}{x} .
\label{eq:model NLO Lphi2 limit 2}
\end{equation}
In our diagrammatic notation of section \ref{subsec:diagrammatic notation} this equation takes the form

\begin{equation}
\begin{tikzpicture}
\node at (1,0) {$ \Lo{1}{\varphi^2}{x}=$};
\node at (4,1) {$x$} child[level distance = 2cm,sibling distance=3cm]{[fill] circle (2pt) } child[level distance=1cm,sibling distance=1.5cm] {[fill] circle (2pt) child[sibling distance=.5cm] {[fill] circle (2pt) } child[sibling distance=.5cm] {[fill] circle (2pt) } child[sibling distance=.5cm] {[fill] circle (2pt) } child[sibling distance=.5cm] {[fill] circle (2pt) } child[sibling distance=.5cm] {[fill] circle (2pt) }};
\node at (9,1) {$x$}  child[level distance=1cm,sibling distance=1.5cm] {[fill] circle (2pt) child[sibling distance=.5cm] {[fill] circle (2pt) } child[sibling distance=.5cm] {[fill] circle (2pt) } child[sibling distance=.5cm] {[fill] circle (2pt) } child[sibling distance=.5cm] {[fill] circle (2pt) } child[sibling distance=.5cm] {[fill] circle (2pt) }}child[level distance = 2cm,sibling distance=3cm]{[fill] circle (2pt) };
\node at (6.5,0){$+$};
\end{tikzpicture}
\label{eq:model NLO Lphi2 diagram}
\end{equation}
where again summation over all combinations of labels for the leaves is understood. 

Let us investigate this expression more closely and try to understand why no counterterms are necessary. As could be seen in section \ref{subsec:Construction of L(0)a}, infinities occur whenever infinite sums over a product of the form $\mathbf{b}_{lm}\mathbf{b}^{\dagger}_{lm}$ are performed, i.e., diagrammatically speaking, whenever the indices of two leaves are \emph{contracted}, with the leave on the right corresponding to a creation operator. In the free theory computations we then made  use of the identity $\mathbf{b}_{lm}\mathbf{b}_{lm}^{\dagger}=\mathbf{b}_{lm}^{\dagger}\mathbf{b}_{lm}+id$, which follows from the commutation relation \eqref{eq:model free field ladder operators commutation relations}, to separate these expressions into a normal ordered (finite) part and a divergent part. We then saw that this divergent part precisely cancels with the counterterms, which suggests that the renormalization procedure is equivalent to normal ordering. Obviously, it would be convenient to carry over this approach to the first order calculations and to include also the sixth operator in the above diagrams into the normal ordering process. However, since no counterterms are present in eq.\eqref{eq:model NLO Lphi2 limit 2}, it is not clear what to make of the terms which are neglected if we demand normal ordering. As it turns out, the infinite parts of these extra terms cancel out, but there is a finite remainder, which we call $(R_1)_{\varphi^2}$. Introducing the

\begin{definition}[Additional grading of left representatives]{\ \\}\index{symbols}{Yidq@$\Lo{n}{\ket{v_a}}{x;d,q}$}%
By $\Lo{n}{\ket{v_a}}{x;d,q}$ we denote the contribution to $\Lo{n}{\ket{v_a}}{x;d}$ comprised of $q\in\mathbb{N}$ annihilation operators.
\end{definition}
we obtain

\begin{results}[Left representative $\Lo{1}{\varphi^2}{x}$ in normal ordered form]\label{res:model NLO Lphi2}
Using the left representative $\Lo{1}{\varphi}{x}$ given in the previous section, one obtains
\begin{equation}
\begin{split}
\Lo{1}{\varphi^2}{x}=&:\Lo{1}{\varphi}{x}\Lo{0}{\varphi}{x}:+:\Lo{0}{\varphi}{x}\Lo{1}{\varphi}{x}:+(R_{1})_{\varphi^2}(x)\\
 =&2:\Lo{1}{\varphi}{x}\Lo{0}{\varphi}{x}:+(R_{1})_{\varphi^2}(x)=2:\Lo{0}{\varphi}{x}\Lo{1}{\varphi}{x}:+(R_{1})_{\varphi^2}(x)
\end{split}
\label{eq:model NLO Lphi2 normal ordered}
\end{equation}
with $(R_{1})_{\varphi^2}(x)\in\mathbb{Y}(x)$ given by

\begin{equation}
\begin{split}
(R_{1})_{\varphi^2}(x):= 5\sum_{q=0}^4\sum_{d=-\infty}^{\infty}\sum_{j=0}^{\infty} r^{d+1}S_{jm}(\hat{x}) \Lo{0}{\varphi^4}{x;d,q}_{jm} (R_1)_{\varphi^2}(d,j,q,r)
\end{split}
\label{eq:model NLO Lphi2 remainder 1}
\end{equation}\index{symbols}{R12x@$(R_{1})_{\varphi^2}(x)$}
where

\begin{equation}
 (R_1)_{\varphi^2}(d,j,q\in\{0,4\},r)=0\quad ,
\label{eq:model NLO Lphi2 remainder vanish}
\end{equation}

\begin{equation}
\begin{split}
 (R_1)_{\varphi^2}(d,j,q\in\{1,3\},r):= \sum_{l=0}^{|d+2|-1}&\left\{\sum_{\atop{J=|l-j|}{\text{denom.}\neq 0}}^{l+j}  \frac{\braket{j l 0 0}{J 0}^2}{(l-|d+2|)(l-|d+2|+1)-J(J+1)}\right.\\+& \left.\operatorname{sign}(d+2)\, \begin{pmatrix}j & l & |d+2|-1-l\\ 0 & 0 & 0    \end{pmatrix}\log r\right\}
\end{split}
\label{eq:model NLO Lphi2 remainder 2}
\end{equation}
and

\begin{equation}
 (R_1)_{\varphi^2}(d,j,q=2,r):= R_{(q=2)}(d,j,r)+\operatorname{sign}(d+2)\log r\, \sum_{l=0}^{|d+2|-1} \begin{pmatrix}j & l & |d+2|-1-l\\ 0 & 0 & 0    \end{pmatrix}\quad .
\label{eq:model NLO Lphi2 remainder 3}
\end{equation}
Here $R_{(q=2)}(d,j,r)$ is a finite, real valued function defined below in eq.\eqref{eq:model NLO Lphi2 infinite sums formula result general}.
\end{results}
By $\operatorname{sign}(x)$ we denote the sign function

\begin{equation}
 \operatorname{sign}(x)=\left\{\begin{array}{lc}-1\quad &\text{if }x<0 \\
 0 &\text{if }x=0\\ 1 &\text{if }x>0 
\end{array}\right.\quad .
\label{eq:model NLO Lphi2 sign function}
\end{equation}\index{symbols}{sign@$\operatorname{sign}(x)$}

\begin{proof}

As mentioned above, we want to exploit the commutation relations of the ladder operators to bring the operators in eq.\eqref{eq:model NLO Lphi2 diagram} into normal order. This is trivial, unless we have to switch the order of a pair of the form $\mathbf{b}_{lm}\mathbf{b}_{lm}^{\dagger}$. Since the operators associated to the left representative $\Lo{1}{\varphi}{x}$ are already normal ordered (see section \ref{subsubsec:L(1)phi}), only the two diagrams

\begin{equation}
\begin{tikzpicture}
\node at (4,1) {$x$} child[level distance = 2cm,sibling distance=3cm]{[fill] circle (2pt) node(1){} } child[level distance=1cm,sibling distance=1.5cm] {[fill] circle (2pt) child[sibling distance=.5cm] {[fill] circle (2pt)node(2){}  } child[sibling distance=.5cm] {[fill] circle (2pt) } child[sibling distance=.5cm] {[fill] circle (2pt) } child[sibling distance=.5cm] {[fill] circle (2pt) } child[sibling distance=.5cm] {[fill] circle (2pt) }};
\node at (9,1) {$x$}  child[level distance=1cm,sibling distance=1.5cm] {[fill] circle (2pt) child[sibling distance=.5cm] {[fill] circle (2pt) } child[sibling distance=.5cm] {[fill] circle (2pt) } child[sibling distance=.5cm] {[fill] circle (2pt) } child[sibling distance=.5cm] {[fill] circle (2pt) } child[sibling distance=.5cm] {[fill] circle (2pt)node(3){}  }}child[level distance = 2cm,sibling distance=3cm]{[fill] circle (2pt)node(4){}  };
\node at (6.5,0){$+$};
\path [<-, draw, thick] (1) -- +(0,-.5) -| (2);
\path [<-, draw, thick] (3) -- +(0,-.5) -| (4);
\end{tikzpicture}
\label{eq:model NLO Lphi2 infinite sums diagram}
\end{equation}
contain expressions of this kind (recall from chapter \ref{subsec:diagrammatic notation} that the arrows denote contraction). Here we have to apply $\mathbf{b}_{lm}\mathbf{b}_{lm}^{\dagger}=\mathbf{b}_{lm}^{\dagger}\mathbf{b}_{lm}+id$ and obtain in addition to the normal ordered products a sum where the two contracted operators are replaced by the identity. In diagrams, this may be written as (recall that dashed lines denote replacement of the corresponding ladder operators with the identity)

\begin{equation}
 \begin{tikzpicture}
\node at (0,1) {$x$} child[level distance = 1cm,sibling distance=1.5cm]{[fill] circle (1.5pt) node(1){} } child[level distance=.5cm,sibling distance=.75cm] {[fill] circle (1.5pt) child[sibling distance=.25cm] {[fill] circle (1.5pt)node(2){}  } child[sibling distance=.25cm] {[fill] circle (1.5pt) } child[sibling distance=.25cm] {[fill] circle (1.5pt) } child[sibling distance=.25cm] {[fill] circle (1.5pt) } child[sibling distance=.25cm] {[fill] circle (1.5pt) }};
\node at (2.5,1) {$x$}  child[level distance=.5cm,sibling distance=.75cm] {[fill] circle (1.5pt) child[sibling distance=.25cm] {[fill] circle (1.5pt) } child[sibling distance=.25cm] {[fill] circle (1.5pt) } child[sibling distance=.25cm] {[fill] circle (1.5pt) } child[sibling distance=.25cm] {[fill] circle (1.5pt) } child[sibling distance=.25cm] {[fill] circle (1.5pt)node(3){}  }}child[level distance = 1cm,sibling distance=1.5cm]{[fill] circle (1.5pt)node(4){}  };
\node at (1.25,.5){$+$};
\node at (5,1) {$x$} child[level distance = 1cm,sibling distance=1.5cm]{[fill] circle (1.5pt) node(5){} } child[level distance=.5cm,sibling distance=.75cm] {[fill] circle (1.5pt) child[sibling distance=.25cm] {[fill] circle (1.5pt)node(6){}  } child[sibling distance=.25cm] {[fill] circle (1.5pt) } child[sibling distance=.25cm] {[fill] circle (1.5pt) } child[sibling distance=.25cm] {[fill] circle (1.5pt) } child[sibling distance=.25cm] {[fill] circle (1.5pt) }};
\node at (3.75,.5){$=\, :$};
\node at (7.7,1) {$x$} child[level distance=.5cm,sibling distance=.75cm] {[fill] circle (1.5pt) child[sibling distance=.25cm] {[fill] circle (1.5pt) } child[sibling distance=.25cm] {[fill] circle (1.5pt) } child[sibling distance=.25cm] {[fill] circle (1.5pt) } child[sibling distance=.25cm] {[fill] circle (1.5pt) } child[sibling distance=.25cm] {[fill] circle (1.5pt)node(3){}  }}child[level distance = 1cm,sibling distance=1.5cm]{[fill] circle (1.5pt)node(4){}  };
\node at (6.25,.5){$:\, +\, :$};
\node at (10.2,1) {$x$} child[level distance = 1cm,sibling distance=1.5cm]{[fill] circle (1.5pt) node(7){} } child[level distance=.5cm,sibling distance=.75cm] {[fill] circle (1.5pt) child[sibling distance=.25cm] {[fill] circle (1.5pt)node(8){}  } child[sibling distance=.25cm] {[fill] circle (1.5pt) } child[sibling distance=.25cm] {[fill] circle (1.5pt) } child[sibling distance=.25cm] {[fill] circle (1.5pt) } child[sibling distance=.25cm] {[fill] circle (1.5pt) }};
\node at (12.9,1) {$x$}  child[level distance=.5cm,sibling distance=.75cm] {[fill] circle (1.5pt) child[sibling distance=.25cm] {[fill] circle (1.5pt) } child[sibling distance=.25cm] {[fill] circle (1.5pt) } child[sibling distance=.25cm] {[fill] circle (1.5pt) } child[sibling distance=.25cm] {[fill] circle (1.5pt) } child[sibling distance=.25cm] {[fill] circle (1.5pt)node(9){}  }}child[level distance = 1cm,sibling distance=1.5cm]{[fill] circle (1.5pt)node(10){}  };
\node at (8.75,.5){$:\, +5$};
\path [<-, draw, thick, dashed] (7) -- +(0,-.5) -| (8);
\path [<-, draw, thick, dashed] (9) -- +(0,-.5) -| (10);
\node at (11.45,.5){$+5$};
\end{tikzpicture}
\end{equation}
This formula is exactly equivalent to the first line of eq.\eqref{eq:model NLO Lphi2 normal ordered} if we define $(R_1)_{\varphi^2}$ to be equal to the two diagrams including the dashed contractions. The second line of eq.\eqref{eq:model NLO Lphi2 normal ordered} simply follows from the fact that the order of ladder operators inside normal ordering signs does not matter, so the two normal ordered diagrams in the equation above are in fact equal. The equation above may also be written in terms of the left representatives as

\begin{equation}
\begin{split}
 \Lo{0}{\varphi}{x}\Lo{1}{\varphi}{x}+\Lo{1}{\varphi}{x}\Lo{0}{\varphi}{x}=&\, :\Lo{0}{\varphi}{x}\Lo{1}{\varphi}{x}:\, +\, :\Lo{1}{\varphi}{x}\Lo{0}{\varphi}{x}:\\
&+\,  \contraction{}{\Lo{0}{\varphi}{x}}{}{\Lo{1}{\varphi}{x}}\Lo{0}{\varphi}{x}\Lo{1}{\varphi}{x}+\contraction{}{\Lo{1}{\varphi}{x}}{}{\Lo{0}{\varphi}{x}}\Lo{1}{\varphi}{x}\Lo{0}{\varphi}{x}
\end{split}
\end{equation}
where the line connecting the left representatives denotes contraction:

\begin{definition}[Contraction of left representatives]\label{def:contractions}\index{symbols}{yizz@$\contraction{}{\Lo{0}{\varphi}{x}}{} {\Lo{n}{\ket{v_a}}{y}}\Lo{0}{\varphi}{x}\Lo{n}{\ket{v_a}}{y}$}%
The commutators

\begin{equation}
 \left[\mathcal{Y}^{-}_0(\varphi,x),\Lo{n}{\ket{v_a}}{y}\right]=\contraction{}{\Lo{0}{\varphi}{x}}{}{\Lo{n}{\ket{v_a}}{y}}\Lo{0}{\varphi}{x}\Lo{n}{\ket{v_a}}{y}
\end{equation}
and

\begin{equation}
 \left[\Lo{n}{\ket{v_a}}{y},\mathcal{Y}^{+}_0(\varphi,x)\right]=\contraction{}{\Lo{n}{\ket{v_a}}{y}}{}{\Lo{0}{\varphi}{x}}\Lo{n}{\ket{v_a}}{y}\Lo{0}{\varphi}{x}
\end{equation}
are called ``contractions'' of the left representatives $\Lo{0}{\varphi}{x}$ and $\Lo{n}{\ket{v_a}}{y}$, where $\mathcal{Y}^{\pm}_0(\varphi,x)$ is the part of $\Lo{0}{\varphi}{x}$ containing only creation ($+$) or only annihilation $(-)$ operators.
\end{definition}
It remains to determine $(R_1)_{\varphi^2}$. This rather involved computation can be found in appendix \ref{app:proofs}.

\end{proof}
To sum up the results of the above discussion, we have most importantly verified that in accordance with eq.\eqref{eq:model NLO Lphi2 limit 2} no subtraction of counterterms is needed in order to cure possible divergences in the calculation of $\Lo{1}{\varphi^2}{x}$. Further we have succeeded in writing $\Lo{1}{\varphi^2}{x}$ purely in terms of normal ordered expressions. As we will see in the following, this is particularly helpful in the calculation of matrix elements of this left representative, since no infinite sums have to be performed.\\

We conclude this section with explicit results for OPE coefficients obtained from the left representative $\Lo{1}{\varphi^2}{x}$ by taking the according matrix elements. 
Again it is helpful to introduce some additional notation in order to keep equations short.

\begin{definition}[Partitions of multisets]\label{def:partitions of multisets}\index{symbols}{Pij@$\mathcal{P}_{i,j}[\mathfrak{A}]$}%
By $\mathcal{P}_{i,j}(\mathfrak{A})$ we denote the set of partitions of any multiset $\mathfrak{A}$ of cardinality $i+j$ into two submultisets of cardinalty $i$ and $j$ respectively, whose sum is $\mathfrak{A}$, i.e.

\begin{equation}
 \mathcal{P}_{i,j}\left[\mathfrak{A}=\Lbag a_1,\ldots,a_{i+j}\Rbag \right]=\left\{  \mathfrak{P}_1=\Lbag a_{n_1},\ldots,a_{n_i}\Rbag, \mathfrak{P}_2=\Lbag a_{n_{i+1}},\ldots,a_{n_{i+j}}\Rbag \, \Big|\, \mathfrak{P}_1\uplus \mathfrak{P}_2=\mathfrak{A} \right\} 
\end{equation}
\end{definition}

\begin{results}[OPE coefficients $\Co{1}{b}{\varphi^2\, a}$]\label{res:model NLO Lphi C1phi2}%
The matrix elements of the left representative $\Lo{1}{\varphi^2}{x}$ given in result \ref{res:model NLO Lphi2} are
\begin{equation}
\Co{1}{b}{\varphi^2\, a}(x)=
                           0 \qquad \text{for }g(a,b)>6 \text{ or }g(a,b)=\text{odd}\text{ or }g(a,b)=0
\end{equation}
and otherwise
\begin{align}
\Co{1}{b}{\varphi^2\, a}(x)=& 2\sum\limits_{\mathfrak{A}\in\mathcal{I}^b_a(6)}\sum\limits_{\mathcal{P}_{5,1}[\mathfrak{A}]}\sum\limits_{J,J_1,J_2}
f^b_a[\mathfrak{A}]\, \D{1}{P_1}{J_1}\,\Delta_{0}[\mathfrak{P_2}]_{J_2}\,T[J_1,J_2]_{J}\, S_J(\hat{x})\, r^{\mathbf{d}}
\notag\\
+&\sum\limits_{\mathfrak{B}\in\mathcal{I}^b_a(4)}\sum\limits_{J}
f^b_a[\mathfrak{B}]\, \Ro{1}{B}{J}{p}{2}\cdot S_J(\hat{x})\cdot r^{\mathbf{d}}
\label{eq:model NLO Lphi C1phi2 result}
\end{align}
with

\begin{equation}
 \Lambda_1[\varphi^2,\mathfrak{A}=\Lbag \mathfrak{l}_1^q,\ldots,\mathfrak{l}_{4}^q\Rbag,r ]_{J}:=\,  s[\mathfrak{A}]\cdot T[\mathfrak{A}]_{J}\, (R_1)_{\varphi^2}[d_{\mathfrak{A}}-3/2,J,q,r] \quad .
\label{eq:model NLO Lphi2 additional notation R}
\end{equation}\index{symbols}{Lambda1p@$\Lambda_1[\varphi^2,\mathfrak{A}]_{J}$}

\end{results}

\begin{proof}

The first two statements, i.e. vanishing of the OPE coefficient for $g(a,b)>6$ and $g(a,b)=$odd, follow simply from proposition \ref{propos:labelings}.
For the remaining values of $g(a,b)$ we first note that

\begin{equation}
\begin{split}
 2\bra{v_b} :\Lo{1}{\varphi}{x}\Lo{0}{\varphi}{x}: \ket{v_a}=&2\bra{v_b} : \begin{tikzpicture}[baseline=0cm]
                                                                      \node at (0,.7) {$x$} child[level distance = 1cm,sibling distance=1.5cm]{[fill] circle (1.5pt) node(1){} } child[level distance=.5cm,sibling distance=.75cm] {[fill] circle (1.5pt) child[sibling distance=.25cm] {[fill] circle (1.5pt)node(2){}  } child[sibling distance=.25cm] {[fill] circle (1.5pt) } child[sibling distance=.25cm] {[fill] circle (1.5pt) } child[sibling distance=.25cm] {[fill] circle (1.5pt) } child[sibling distance=.25cm] {[fill] circle (1.5pt) }};
                                                                     \end{tikzpicture}: 
 \ket{v_a}\\=&2\sum_{\mathfrak{A}\in\mathcal{I}^b_a(6)}\sum_{\mathcal{P}_{5,1}[\mathfrak{A}]}\bra{v_b} : \begin{tikzpicture}[baseline=0cm]
                                                                      \node at (0,.7) {$x$} child[level distance = 1cm,sibling distance=1.5cm]{[fill] circle (1.5pt) node[below]{\begin{small} $\mathfrak{P}_2$ \end{small} } } child[level distance=.5cm,sibling distance=.75cm] {[fill] circle (1.5pt) child[sibling distance=.25cm] {[fill] circle (1.5pt)node(2){}  } child[sibling distance=.25cm] {[fill] circle (1.5pt) } child[sibling distance=.25cm] {[fill] circle (1.5pt) node[below]{\begin{small} $\mathfrak{P}_1$ \end{small} } } child[sibling distance=.25cm] {[fill] circle (1.5pt) } child[sibling distance=.25cm] {[fill] circle (1.5pt) }};
                                                                     \end{tikzpicture} :
 \ket{v_a}\cdot s[\mathfrak{P}_1]\, s[\mathfrak{P}_2]
\end{split}
\end{equation}
Here the second equality holds for the following reasons: For the matrix element not to vanish, the labels of the six leaves above have to be in $\mathcal{I}^b_a(6)$ by definition, which explains the first sum on the right side.
Now we have to attach five of those labels to leaves entering the vertex, and one to the leaf directly connected to the root. Therefore, we split the multiset $\mathfrak{A}$ into two submultiset: The multiset $\mathfrak{P}_1$ of cardinality $5$ and the multiset $\mathfrak{P}_2$ of cardinality $1$. The labels in $\mathfrak{P}_1$ are then attached to the five leaves entering the vertex in any order and we have to multiply by the symmetry factor $s[\mathfrak{P}_1]$ in order to account for all the diagrams that may be obtained by permutations of those five leaves. Analogously, we attach the one label in $\mathfrak{P}_2$ to the remaining leaf and multiply by $s[\mathfrak{P}_2]$, obtaining the right side of the equation above.

Now it remains to translate the diagram into a formula. Again the numerical factor obtained from the action of the six ladder operators on $\ket{v_a}$ is denoted by $f^b_a[\mathfrak{A}]$. Further, coupling of the five spherical harmonics associated to the leaves entering the vertex gives a factor of $\sum_{J_1}T[\mathfrak{P}_1]_{J_1}S_{J_1}(\hat{x})$. The remaining leaf contributes one additional spherical harmonic, which we write as $\sum_{J_2}T[\mathfrak{P}_2]_{J_2}S_{J_2}(\hat{x})$. Coupling these two contributions then yields the factor\linebreak $\sum_{J,J_1,J_2}T[\mathfrak{P}_1]_{J_1}T[\mathfrak{P}_2]_{J_2}T[J_1,J_2]_JS_{J}(\hat{x})$. The vertex in the diagram gives an additional factor $D[d_{\mathfrak{P}_1},J_1,r]$ and the overall scaling degree is again by definition $-\mathbf{d}$. In summary, we have found

\begin{equation}
\begin{split}
 &2\bra{v_b} :\Lo{1}{\varphi}{x}\Lo{0}{\varphi}{x}: \ket{v_a}=\\&2\sum_{\mathfrak{A}\in\mathcal{I}^b_a(6)}\sum_{\mathcal{P}_{5,1}[\mathfrak{A}]} f^b_a[\mathfrak{A}] s[\mathfrak{P}_1]\, s[\mathfrak{P}_2]\sum_{J,J_1,J_2}T[\mathfrak{P}_1]_{J_1}T[\mathfrak{P}_2]_{J_2}T[J_1,J_2]_JS_{J}(\hat{x})D[d_{\mathfrak{P}_1},J_1,r] r^{\mathbf{d}}
\end{split}
\end{equation}
Recalling the definitions of $\Delta_{0}[\mathfrak{A}]_{J}$ and $\D{1}{A}{J}$, we indeed verify the first line of eq.\eqref{eq:model NLO Lphi C1phi2 result}.

Now consider the contribution from the remainder term. Here we find

\begin{equation}
\begin{split}
 \bra{v_b}(R_1)_{\varphi^2}(x)\ket{v_a}=&\sum_J(R_1)_{\varphi^2}(\mathbf{d}-1,J,q,r)S_J(\hat{x})\,r^{\mathbf{d}}\cdot \bra{v_b}\Lo{0}{\varphi^4}{x;\mathbf{d}-1,q}_{J}\ket{v_a}\\
=&\sum_{\mathfrak{B}\in\mathcal{I}^b_a(4)}\sum_J(R_1)_{\varphi^2}(d_{\mathfrak{B}}-3/2,J,q,r)S_J(\hat{x})\,r^{\mathbf{d}}\cdot f^b_a[\mathfrak{B}]\, \Delta_{0}[\mathfrak{B}]_{J}
\end{split}
\end{equation}
where in the first equality eq.\eqref{eq:model NLO Lphi2 remainder 1} was used. The second equality follows from result \ref{res:model NLO Lphi C1phi} and from the dimensional analysis

\begin{equation}
\mathbf{d}-1=-sd\, \Co{1}{b}{\varphi^2\, a}-1=|b|-|a|-2
\end{equation}
and

\begin{equation}
d_{\mathfrak{A}}=-sd\, \Co{i}{b}{\varphi a}=|b|-|a|-1/2\quad \text{for }\mathfrak{A}\in\mathcal{I}^b_a(n)
\end{equation}
This confirms eq.\eqref{eq:model NLO Lphi C1phi2 result}.

To finish the proof we have to show that the OPE coefficient vanishes for $g(a,b)=0$, i.e. for $a=b$. We start the computation in the first line of eq.\eqref{eq:model NLO Lphi C1phi2 result} and want to determine the set of multisets $\mathcal{I}^a_a(6)$. The condition $a=a+\sum_{i=1}^6e_{\mathfrak{l}^q_i}$ restricts the multisets in $\mathcal{I}^a_a(6)$ to the form $\mathfrak{A}=\Lbag+(l_1m_1),-(l_1m_1),\ldots,+(l_3m_3),-(l_3m_3)\Rbag$. Hence, the submultisets are either of the form

\begin{equation}
\mathfrak{P}_1=\Lbag+(l_im_i),-(l_im_i),+(l_jm_j),-(l_jm_j),+(l_km_k)\Rbag \text{ and } \mathfrak{P}_2=\Lbag-(l_km_k)\Rbag  
\label{eq:model NLO Lphi2 Cphi2 decomp1}
\end{equation}
or
\begin{equation}
\mathfrak{P}_1=\Lbag+(l_im_i),-(l_im_i),+(l_jm_j),-(l_jm_j),-(l_km_k)\Rbag \text{ and } \mathfrak{P}_2=\Lbag+(l_km_k)\Rbag  
\label{eq:model NLO Lphi2 Cphi2 decomp2}
\end{equation}
with $i,j,k\in\{1,2,3\}$. For the first case, we obtain

\begin{equation}
\begin{split}
 \sum\limits_{J,J_1,J_2}
\D{1}{P_1}{J_1}\,\Delta_{0}[\mathfrak{P}_2]_{J_2}\,T[J_1,J_2]_{J}\, S_J(\hat{x})
=s[\mathfrak{P}_1]\frac{\log r}{2l_k+1}
\end{split}
\end{equation}
Here no spherical harmonics and coupling tensors appear, as we may simply apply the addition theorem three times. Further we used the fact that 

\begin{equation}
D[d_{\mathfrak{P}_1}=l_k,J_1=l_k,r]=\frac{\log r}{2l_k+1}\quad.
\end{equation}
Similarly, we find for the second type of decomposition, eq.\eqref{eq:model NLO Lphi2 Cphi2 decomp2}, that

\begin{equation}
\begin{split}
 \sum\limits_{J,J_1,J_2}
\D{1}{P_1}{J_1}\,\Delta_{0}[\mathfrak{P}_2]_{J_2}\,T[J_1,J_2]_{J}\, S_J(\hat{x})
=-s[\mathfrak{P}_1]\frac{\log r}{2l_k+1}\quad,
\end{split}
\end{equation}
where again we applied the addition theorem of the spherical harmonics and where we used

\begin{equation}
D[d_{\mathfrak{P}_1}=-l_k-1,J_1=l_k,r]=-\frac{\log r}{2l_k+1}\quad.
\end{equation}
Thus, the contributions from the two decompositions of any multiset $\mathfrak{A}\in\mathcal{I}^b_a(6)$ cancel each other, implying that the first line in eq.\eqref{eq:model NLO Lphi C1phi2 result} vanishes.

It remains to consider the contribution from the remainder term. Here the multisets $\mathfrak{B}\in\mathcal{I}^a_a(4)$ are restricted to be of the form $\mathfrak{B}=\Lbag+(l_1m_1),-(l_1m_1),+(l_2m_2),-(l_2m_2)\Rbag$, which implies $d_{\mathfrak{B}}=-1/2$ and $J=0$. However, it follows from eqs.\eqref{eq:model NLO Lphi2 remainder 3} and \eqref{eq:model NLO Lphi2 infinite sums formula result general} that

\begin{equation}
 (R_1)_{\varphi^2}(d_{\mathfrak{B}}-3/2=-2, J=0, r;q=2)=0
\end{equation}
which follows from inspection of the summation limits in those expression. Therefore, we also have $\Ro{1}{B}{J}{p}{2}=0$ for all $\mathfrak{B}\in\mathcal{I}^a_a(4)$, so the second line of eq.\eqref{eq:model NLO Lphi C1phi2 result} vanishes as well. Hence, we have found that $\Co{1}{a}{\varphi\, a}(x)=0$.

\end{proof}

Again we compute a specific example OPE coefficient in order to illustrate the application of result \ref{res:model NLO Lphi C1phi2}.

\begin{example} Consider the coefficient $\Co{1}{\varphi^p}{\varphi^2\, \varphi^{p+3}\del^{l}\varphi}$ for $l\neq 0$. We start with the determination of the multisets in $\mathcal{I}^{\varphi^p}_{\varphi^{p+3}\del^{l}\varphi}(6)$. Again this set consists of only one element

\begin{equation}
 \mathcal{I}^{\varphi^p}_{\varphi^{p+3}\del^{l}\varphi}(6)=\{\mathfrak{A}=\Lbag -(00),-(00),-(00),-(00),+(00),-(lm) \Rbag\}
\end{equation}
The numerical prefactor then becomes

\begin{equation}
 f^{\varphi^p}_{\varphi^{p+3}\del^{l}\varphi}[\mathfrak{A}]=(p+3)(p+2)(p+1)p=\frac{(p+3)!}{(p-1)!}\quad.
\end{equation}
Since all zero entries in the coupling tensors $T$ may be dropped, only the product $T[-(lm)]_JS_J(\hat{x})=S_{lm}(\hat{x})$ remains in the case at hand (equivalently, the coupling of $S_{00}\cdots S_{00}S_{lm}=S_{lm}$). The multiset $\mathfrak{A}$ may be split into two submultisets of cardinality $5$ and $1$ in three different ways:

\begin{align}
 &\mathfrak{P}_{1a}=\Lbag -(00),-(00),-(00),+(00),-(lm) \Rbag\qquad \mathfrak{P}_{2a}=\Lbag-(00)\Rbag\notag\\
 &\mathfrak{P}_{1b}=\Lbag -(00),-(00),-(00),-(00),-(lm) \Rbag\qquad \mathfrak{P}_{2b}=\Lbag+(00)\Rbag\notag\\
 &\mathfrak{P}_{1c}=\Lbag -(00),-(00),-(00),-(00),+(00)\Rbag\qquad \mathfrak{P}_{2c}=\Lbag-(lm)\Rbag
\end{align}
As the second multiset only contains one element, its symmetry factor $s[\mathfrak{P}_2]$ is always 1. For the other multisets we obtain

\begin{align}
 &s[\mathfrak{P}_{1a}]=20\\
 &s[\mathfrak{P}_{1b}]=s[\mathfrak{P}_{1c}]=5\quad.
\end{align}
It remains to determine the factors

\begin{align}
 &D[d_{\mathfrak{P}_{1a}}=-l-2,J=l,r]=\frac{1}{2(l+1)}\notag\\
 &D[d_{\mathfrak{P}_{1b}}=-l-3,J=l,r]=\frac{1}{4l+6)}\notag\\
 &D[d_{\mathfrak{P}_{1c}}=-2,J=0,r]=\frac{1}{2}\\
\end{align}
and the power of $r$

\begin{equation}
\mathbf{d}= |\varphi^p|-|\varphi^{p+3}\del^{l}\varphi|-|\varphi^2|=-l-3\quad.
\end{equation}
This finishes the discussion of the contribution from the first line of eq.\eqref{eq:model NLO Lphi C1phi2 result}. In the second line we find again only one allowed multiset

\begin{equation}
 \mathcal{I}^{\varphi^p}_{\varphi^{p+3}\del^{l}\varphi}(4)=\{\mathfrak{B}=\Lbag -(00),-(00),-(00),-(lm) \Rbag\}\quad.
\end{equation}
Note, however, that all indices in this multiset correspond to annihiliation operators (i.e. they have a negative sign), so we have to determine remainder terms of the form $(R_1)_{\varphi^2}(d_{\mathfrak{B}}-3/2,J=l,q=4)$, which vanish by eq.\eqref{eq:model NLO Lphi2 remainder vanish}.

Putting all the pieces together (recall also the factor $2$ in eq.\eqref{eq:model NLO Lphi C1phi2 result}), we arrive at the result

\begin{equation}
\Co{1}{\varphi^p}{\varphi^2\, \varphi^{p+3}\del^{l}\varphi}(x)=\frac{(p+3)!}{(p-1)!}\sum_{m}S_{lm}(\hat{x})r^{-l-3}\left(5+\frac{5}{2l+3}+\frac{20}{l+1}\right)
\end{equation}
in accordance with appendix \ref{app:OPE table}.

\end{example}

\subsection{Construction of \texorpdfstring{$\Lo{1}{\varphi^3}{x}$}{L(1)phi3}}\label{subsec:L1phi3}

In the previous section we were able to perform the first iteration step at next to leading order, i.e. we constructed $\Lo{1}{\varphi^2}{x}$ out of $\Lo{1}{\varphi}{x}$ and $\Lo{0}{\varphi}{x}$, without the need to subtract any counterterms, see eq.\eqref{eq:model NLO Lphi2 limit}. The first true example of our analog of renormalization will be encountered in the present section. The iteration step yielding $\Lo{1}{\varphi^3}{x}$ is expressed by the formula (see \eqref{eq:model perturbations limit 2})

\begin{equation}
\begin{split}
 \Lo{1}{\varphi^3}{x}=\lim_{y\to x}\Big[&\Lo{0}{\varphi}{x}\Lo{1}{\varphi^2}{y}+\Lo{1}{\varphi}{x}\Lo{0}{\varphi^2}{y}-\Co{1}{\varphi^3}{\varphi^2\, \varphi}(x-y)\Lo{0}{\varphi^3}{x}\\
 -& \Co{0}{\varphi}{\varphi^2\,\varphi}(x-y)\Lo{1}{\varphi}{x}- \Co{0}{\varphi_{1m}}{\varphi^2\,\varphi}(x-y)\Lo{1}{\varphi_{1m}}{x} \Big]\, .
\end{split}
\label{eq:model NLO Lphi3 limit}
\end{equation}
Here the expressions (assuming $x,y$ collinear and $|x|>|y|$)

\begin{align}
 &\Co{1}{\varphi^3}{\varphi^2\, \varphi}(x-y) =  20\log |x-y|=20\left[\log|x|+\log\left(1-\frac{|y|}{|x|}\right)\right]\label{eq:model NLO Lphi3 CT log} \\
&\Co{0}{\varphi}{\varphi^2\,\varphi}(x-y)  =  \frac{2}{|x-y|}
\label{eq:model NLO Lphi3 CT}
\end{align}
are divergent in the limit $y\to x$. Therefore in our analysis of the first two summands on the right side of eq.\eqref{eq:model NLO Lphi3 limit} we have to find terms that precisely cancel these infinities, and the remaining expressions should then be finite. This procedure might be interpreted as a sort of renormalization of the products $\Lo{0}{\varphi}{x}\Lo{1}{\varphi^2}{y}$ and $\Lo{1}{\varphi}{x}\Lo{0}{\varphi^2}{y}$. Note that the two counterterms have different divergent behavior: Polynomial and logarithmic. We have already encountered polynomial counterterms in our constructions at zeroth order, and we will see in the following that indeed the polynomial divergences of the present section are very closely related to those of the free theory computations. The logarithmic counterterm, however, is a completely new feature at first perturbation order, and it will turn out to be considerably more complicated to find the corresponding expressions in eq.\eqref{eq:model NLO Lphi3 limit} canceling this term. 

Our plan for this section is the following: As in previous calculations, it is the aim to write $\Lo{1}{\varphi^3}{x}$ as a normal ordered product of left representatives plus some finite remainder $(R_1)_{\varphi^3}$, simliar to eq.\eqref{eq:model NLO Lphi2 normal ordered}. Basically, the procedure is analogous to the $\Lo{1}{\varphi^2}{x}$ case: We exploit the commutation relations for the ladder operators, eq.\eqref{eq:model free field ladder operators commutation relations}, to bring all expressions into normal order, picking up additional terms whenever we have to exchange the order of a creation and annihilation operator with the same index. These additional terms include infinite sums, whose divergent behavior has to be analyzed. In contrast to the previous section, we expect these sums to be divergent in the limit $y\to x$. This divergence should be cured by the subtraction of the counterterms, see eq.\eqref{eq:model NLO Lphi3 CT}. Hence, the above mentioned remainder $(R_1)_{\varphi^3}$ is precisely the difference of the additional terms we pick up in the process of bringing all expressions into normal order, and the counterterms, where we take the limit $y\to x$. The section is again concluded by a discussion of the OPE coefficients obtained from $\Lo{1}{\varphi^3}{x}$ by taking matrix elements.\\

To begin with, we present

\begin{results}[Left representative $\Lo{1}{\varphi^3}{x}$ in normal ordered form]\label{res:model NLO Lphi3}
 
\begin{equation}
\Lo{1}{\varphi^3}{x}=\,3:\Lo{1}{\varphi}{x}\Lo{0}{\varphi^2}{x}:\, + 3:\Lo{0}{\varphi}{x}(R_1)_{\varphi^2}(x):+ (R_1)_{\varphi^3}(x)
\label{eq:model NLO Lphi3 remainder form}
\end{equation}
where

\begin{equation}
\begin{split}
&\begin{tikzpicture}
  \node at (4.6,0) {$x$} child[level distance = 1cm,sibling distance=1cm]{[fill] circle (1.5pt) node(5){} };
\node at (6,0) {$y$} child[level distance = 1cm,sibling distance=1.5cm]{[fill] circle (1.5pt) node(6){} } child[level distance=.5cm,sibling distance=1cm] {[fill] circle (1.5pt) child[sibling distance=.25cm] {[fill] circle (1.5pt) node(7){} } child[sibling distance=.25cm] {[fill] circle (1.5pt)  } child[sibling distance=.25cm] {[fill] circle (1.5pt)  } child[sibling distance=.25cm] {[fill] circle (1.5pt)  } child[sibling distance=.25cm] {[fill] circle (1.5pt) node(8){} }};
\node at (7.6,0) {$x$} child[level distance = 1cm,sibling distance=1cm]{[fill] circle (1.5pt) node(9){} }; \node at (9,0) {$y$} child[level distance=.5cm,sibling distance=.75cm] {[fill] circle (1.5pt) child[sibling distance=.25cm] {[fill] circle (1.5pt) node(10){} } child[sibling distance=.25cm] {[fill] circle (1.5pt)  } child[sibling distance=.25cm] {[fill] circle (1.5pt) } child[sibling distance=.25cm] {[fill] circle (1.5pt)  } child[sibling distance=.25cm] {[fill] circle (1.5pt) node(11){} }}child[level distance = 1cm,sibling distance=1cm]{[fill] circle (1.5pt) node(12){} };
\node at (7.35,-.5){$+$};
\node at (9.8,-.5){$+$};
\path [<-, draw,dashed] (5) -- +(0,-.5) -| (8);
\path [<-, draw,dashed] (6) -- +(0,-.3) -| (7);
\path [<-, draw,dashed] (9) -- +(0,-.5) -| (10);
\path [<-, draw,dashed] (11) -- +(0,-.5) -| (12);
\node at (10.6,0) {$x$}  child[level distance=.5cm,sibling distance=1cm] {[fill] circle (1.5pt) child[sibling distance=.25cm] {[fill] circle (1.5pt) node(13){} } child[sibling distance=.25cm] {[fill] circle (1.5pt)  } child[sibling distance=.25cm] {[fill] circle (1.5pt) } child[sibling distance=.25cm] {[fill] circle (1.5pt) } child[sibling distance=.25cm] {[fill] circle (1.5pt) node(14){} }};
\node at (12.2,0) {$y$} child[level distance = 1cm,sibling distance=1cm]{[fill] circle (1.5pt) node(15){} }  child[level distance = 1cm,sibling distance=1cm]{[fill] circle (1.5pt) node(16){} };
\path [<-, draw,dashed] (13) -- +(0,-.5) -| (16);
\path [<-, draw,dashed] (14) -- +(0,-.3) -| (15);
\node at (14.5,-.5){$-\log|x-y|\Lo{0}{\varphi^3}{x}\Big]$};
\node at (2.6,-.5){$(R_1)_{\varphi^3}(x)=20\lim\limits_{y\to x}\Big[$};
\end{tikzpicture}\\
&=\lim_{y\to x}\Big[ \contraction[.5ex]{\Lo{0}{\varphi}{x}}{\Lo{0}{\varphi}{y}}{}{\Lo{1}{\varphi}{y}}\contraction[1ex]{}{\Lo{0}{\varphi}{x}}{\Lo{0}{\varphi}{y}}{\Lo{1}{\varphi}{y}} \Lo{0}{\varphi}{x}\Lo{0}{\varphi}{y}\Lo{1}{\varphi}{y} + \contraction{}{\Lo{0}{\varphi}{x}}{}{\Lo{1}{\varphi}{y}}\contraction[.5ex]{\Lo{0}{\varphi}{x}}{\Lo{1}{\varphi}{y}}{}{\Lo{0}{\varphi}{y}}  \Lo{0}{\varphi}{x}\Lo{1}{\varphi}{y}\Lo{0}{\varphi}{y}+\contraction[.5ex]{}{\Lo{1}{\varphi}{x}}{}{\Lo{0}{\varphi}{y}} \contraction{}{\Lo{1}{\varphi}{x}}{\Lo{0}{\varphi}{y}}{\Lo{0}{\varphi}{y}}  \Lo{1}{\varphi}{x}\Lo{0}{\varphi}{y}\Lo{0}{\varphi}{y}\\
 &\qquad-20\log |x-y|\Lo{0}{\varphi^3}{x} \Big]
\end{split}
\label{eq:model NLO Lphi3 remainder form b}
\end{equation}\index{symbols}{R13@$(R_1)_{\varphi^3}(x)$}

\end{results}
\begin{proof}

Starting at eq.\eqref{eq:model NLO Lphi3 limit}, we can bring $\Lo{1}{\varphi^3}{x}$ into the desired form simply by applying the results of the previous chapters, and without the need to perform any new calculations. This can be seen from the following observations: First, we substitute eq.\eqref{eq:model NLO Lphi2 normal ordered} for the left representative $\Lo{1}{\varphi^2}{x}$ into eq.\eqref{eq:model NLO Lphi3 limit}, obtaining

\begin{equation}
\begin{split}
 \Lo{1}{\varphi^3}{x}=\lim_{y\to x}\Big[&\Lo{0}{\varphi}{x}\,2:\Lo{1}{\varphi}{y}\Lo{0}{\varphi}{y}:+\Lo{0}{\varphi}{x}(R_1)_{\varphi^2}(y)+\Lo{1}{\varphi}{x}\,:\Lo{0}{\varphi}{y}\Lo{0}{\varphi}{y}:\\
 -&\frac{2}{x-y}\Lo{1}{\varphi}{x} - 20\log |x-y|\Lo{0}{\varphi^3}{x}- \Co{0}{\varphi_{1m}}{\varphi^2\,\varphi}(x-y)\Lo{1}{\varphi_{1m}}{x} \Big]\, .
\end{split}
\label{eq:model NLO Lphi3 limit substitution}
\end{equation}
Now let us try to bring the operators in the first line of this equation into normal order. For the first term we find

\begin{equation}
\begin{split}
2\lim_{y\to x}\, \Big[\Lo{0}{\varphi}{x}\,:\Lo{1}{\varphi}{y}&\Lo{0}{\varphi}{y}:\Big]=2:\Lo{1}{\varphi}{x}\Lo{0}{\varphi^2}{x}:\\
+&2\lim_{y\to x}\Big[\contraction{}{\Lo{0}{\varphi}{x}}{}{\Lo{1}{\varphi}{y}} \Lo{0}{\varphi}{x}\Lo{1}{\varphi}{y}\Lo{0}{\varphi}{y}
 + \contraction{}{\Lo{0}{\varphi}{x}}{\Lo{1}{\varphi}{y}}{\Lo{0}{\varphi}{y}} \Lo{0}{\varphi}{x}\Lo{1}{\varphi}{y}\Lo{0}{\varphi}{y}\Big]
\end{split}
\label{eq:model NLO Lphi3 limit NO 1}
\end{equation}
where we have already performed the limit in the (finite) normal ordered term. Similarly, we may write the last term in the first line of eq.\eqref{eq:model NLO Lphi3 limit substitution} as

\begin{equation}
\begin{split}
\lim_{y\to x}\, \Big[\Lo{1}{\varphi}{x}\,:\Lo{0}{\varphi^2}{y}:\Big]\, =\, &:\Lo{1}{\varphi}{x}\Lo{0}{\varphi^2}{x}:\\ +&\lim_{y\to x}\Big[2\contraction{}{\Lo{1}{\varphi}{x}}{}{\Lo{0}{\varphi}{y}} \Lo{1}{\varphi}{x}\Lo{0}{\varphi}{y}\Lo{0}{\varphi}{y}
 +\contraction[.5ex]{}{\Lo{1}{\varphi}{x}}{}{\Lo{0}{\varphi}{y}} \contraction{}{\Lo{1}{\varphi}{x}}{\Lo{0}{\varphi}{y}}{\Lo{0}{\varphi}{y}}  \Lo{1}{\varphi}{x}\Lo{0}{\varphi}{y}\Lo{0}{\varphi}{y}\Big]
\end{split}
\label{eq:model NLO Lphi3 limit NO 2}
\end{equation}
Finally, remembering the definition of the remainder term $(R_1)_{\varphi^2}$, eq.\eqref{eq:model NLO Lphi2 infinite sums formula}, the last remaining term in the first line of eq.\eqref{eq:model NLO Lphi3 limit substitution} may be put into the form

\begin{equation}
\begin{split}
&\lim_{x\to y}\, \Big[\Lo{0}{\varphi}{x}(R_1)_{\varphi^2}(y)\Big]=\, :\Lo{0}{\varphi}{x}(R_1)_{\varphi^2}(x):+\lim_{y\to x}\contraction{}{\Lo{0}{\varphi}{x}}{}{(R_1)_{\varphi^2}} \Lo{0}{\varphi}{x}(R_1)_{\varphi^2}(y)\\
=&\, :\Lo{0}{\varphi}{x}(R_1)_{\varphi^2}(x):+\lim_{y\to x}\Big[ \contraction[.5ex]{\Lo{0}{\varphi}{x}}{\Lo{0}{\varphi}{y}}{}{\Lo{1}{\varphi}{y}}\contraction[1ex]{}{\Lo{0}{\varphi}{x}}{\Lo{0}{\varphi}{y}}{\Lo{1}{\varphi}{y}} \Lo{0}{\varphi}{x}\Lo{0}{\varphi}{y}\Lo{1}{\varphi}{y} + \contraction{}{\Lo{0}{\varphi}{x}}{}{\Lo{0}{\varphi}{y}}\contraction[.5ex]{\Lo{0}{\varphi}{x}}{\Lo{0}{\varphi}{y}}{}{\Lo{1}{\varphi}{y}}  \Lo{0}{\varphi}{x}\Lo{1}{\varphi}{y}\Lo{0}{\varphi}{y}\Big]
\end{split}
\label{eq:model NLO Lphi3 limit NO 3}
\end{equation}
Now recall from our construction of the free theory left representatives, in particular from the computation of $\Lo{0}{\varphi^2}{x}$ in section \ref{subsec:Construction of L(0)a}, that

\begin{equation}
 \lim_{y\to x}\contraction{}{\Lo{0}{\varphi}{x}}{}{\Lo{0}{\varphi}{y}} \Lo{0}{\varphi}{x}\Lo{0}{\varphi}{y}=\lim\limits_{y\to x}\,\frac{1}{x-y}\, .
\label{eq:model NLO Lphi3 subtree polynomial divergence}
\end{equation}
which implies for the last summand in eq.\eqref{eq:model NLO Lphi3 limit NO 1}

\begin{equation}
2\lim_{y\to x}\contraction{}{\Lo{0}{\varphi}{x}}{\Lo{1}{\varphi}{y}}{\Lo{0}{\varphi}{y}} \Lo{0}{\varphi}{x}\Lo{1}{\varphi}{y}\Lo{0}{\varphi}{y}=2\lim_{y\to x}\,\frac{1}{x-y}\, \Lo{1}{\varphi}{y}\, .
\end{equation}
A Taylor expansion of $\Lo{1}{\varphi}{y}$ around $x$ shows that this contribution cancels with the polynomial counterterms in the second line of eq.\eqref{eq:model NLO Lphi3 limit}. Furthermore, we can write the sum of the remaining expressions with one contraction as

\begin{equation}
2\lim_{y\to x}\Big[\contraction{}{\Lo{0}{\varphi}{x}}{}{\Lo{1}{\varphi}{y}} \Lo{0}{\varphi}{x}\Lo{1}{\varphi}{y}\Lo{0}{\varphi}{y}+\contraction{}{\Lo{1}{\varphi}{x}}{}{\Lo{0}{\varphi}{y}} \Lo{1}{\varphi}{x}\Lo{0}{\varphi}{y}\Lo{0}{\varphi}{y}\Big] =2:\Lo{0}{\varphi}{x}(R_1)_{\varphi^2}(x):
\end{equation}
where the definition of the remainder $(R_1)_{\varphi^2}$ has been used. Thus, if we plug eqs.\eqref{eq:model NLO Lphi3 limit NO 1}, \eqref{eq:model NLO Lphi3 limit NO 2} and \eqref{eq:model NLO Lphi3 limit NO 3} into eq.\eqref{eq:model NLO Lphi3 limit substitution} and further use the two identities above, we end up with result \ref{res:model NLO Lphi3} as claimed.

\end{proof}
The first two terms in eq.\eqref{eq:model NLO Lphi3 remainder form} are known, so it remains to find an explicit formula for the new remainder term $(R_1)_{\varphi^3}$. Unfortunately, we have not been able to determine this operator completely, but we have found the following partial results.

\begin{results}[The remainder operator $(R_1)_{\varphi^3}$]\label{res:model NLO Rphi3}\index{symbols}{R13@$(R_1)_{\varphi^3}(x)$}%
Let

\begin{equation}
 (R_1)_{\varphi^3}(x)=\sum_{q=0}^3\sum_{d=-\infty}^{\infty}\sum_{j=0}^{\infty} r^dS_{jm}(\hat{x})\Lo{0}{\varphi^3}{x;d,q}_{jm}(R_1)_{\varphi^3}(d,j,q;r)\quad.
\label{eq:model NLO Lphi3 remainder grading}
\end{equation}
Then

\begin{equation}
\begin{split}
 &(R_1)_{\varphi^3}(d,j,q=0;r)=20 \log r\left[\sum_{l,l'=0}^{l+l'\leq d}\sum_{J_1} \braket{jl\, 00}{J_1\, 0}^2 \begin{pmatrix}
                        J_1 & l' & d-l-l'\\ 0 & 0 & 0
                       \end{pmatrix}^2-1\right]\\
+&20\sum_{\atop{J_1,J_2}{M_1,M_2}}\left(\sum_{\atop{l,l'}{l+l'\leq d}}+2\sum_{l=0}^d\sum_{l'=d+1-l}^{\frac{J_1+d-l}{2}}\right)\frac{\braket{jl\, 00}{J_1\, 0}^2\braket{J_1 l'\, 00}{J_2\, 0}^2}{(l+l'-d-1)(l+l'-d)-J_2(J_2+1)}\Big|_{\text{denom.}\neq 0}\\
 &-\frac{10}{\pi}\sum_{k=0}^j\sum_{l'=0}^{(j-2k+d)/2}\frac{\Gamma(k+\frac{1}{2})\Gamma(j-k+\frac{1}{2})\Gamma(\frac{d+j-2k+1}{2}-l')\Gamma(\frac{d+j-2k}{2}+1)}{\Gamma(k+1)\Gamma(j-k+1)\Gamma(\frac{d+j-2k}{2}-l'+1)\Gamma(\frac{d+j-2k+3}{2})}\\
 &\times L_5\left[\atop{1,d+j-2k+\frac{5}{2},d-k+\frac{3}{2},d+j-k+2,\frac{d+j-2k-l'}{2}+1,\frac{d+j-2k+l'+3}{2}}{d+j-2k+\frac{3}{2},d-k+2,d+j-k+\frac{5}{2},\frac{d+j-2k-l'+3}{2},\frac{d+j-2k+l'}{2}+2}\right]\, ,
\end{split}
\label{eq:model NLO Lphi3 remainder q=0}
\end{equation}
and

\begin{equation}
\begin{split}
 &(R_1)_{\varphi^3}(-d-3,j,q=3;r)=-20 \log r\left[\sum_{l,l'=0}^{l+l'\leq d}\sum_{J_1} \braket{jl\, 00}{J_1\, 0}^2 \begin{pmatrix}
                        J_1 & l' & d-l-l'\\ 0 & 0 & 0
                       \end{pmatrix}^2+1\right]\\
+&20\sum_{\atop{J_1,J_2}{M_1,M_2}}\left(\sum_{\atop{l,l'}{l+l'\leq d}}+2\sum_{l=0}^d\sum_{l'=d+1-l}^{\frac{J_1+d-l}{2}}\right)\frac{\braket{jl\, 00}{J_1\, 0}^2\braket{J_1 l'\, 00}{J_2\, 0}^2}{(l+l'-d-1)(l+l'-d)-J_2(J_2+1)}\Big|_{\text{denom.}\neq 0}\\ &-\frac{10}{\pi}\sum_{k=0}^j\sum_{l'=0}^{(j-2k+d)/2}\frac{\Gamma(k+\frac{1}{2})\Gamma(j-k+\frac{1}{2})\Gamma(\frac{d+j-2k+1}{2}-l')\Gamma(\frac{d+j-2k}{2}+1)}{\Gamma(k+1)\Gamma(j-k+1)\Gamma(\frac{d+j-2k}{2}-l'+1)\Gamma(\frac{d+j-2k+3}{2})}\\
 &\times L_5\left[\atop{1,d+j-2k+\frac{5}{2},d-k+\frac{3}{2},d+j-k+2,\frac{d+j-2k-l'}{2}+1,\frac{d+j-2k+l'+3}{2}}{d+j-2k+\frac{3}{2},d-k+2,d+j-k+\frac{5}{2},\frac{d+j-2k-l'+3}{2},\frac{d+j-2k+l'}{2}+2}\right]\, ,
\end{split}
\label{eq:model NLO Lphi3 remainder q=5}
\end{equation}
where $L_5$ is the non divergent part of the hypergeometric series $_6F_5$, see eq.\eqref{eq:app hypergeos L(p)}.
\end{results}
The derivation of this result can be found in appendix \ref{app:proofs}. Remark: For $d=j=0$ we obtain the simple results

\begin{align}
 &(R_1)_{\varphi^3}(d=0,j=0,q=0,r)=0\\
&(R_1)_{\varphi^3}(d=-3,j=0,q=3,r)=-40\log r
\end{align}
This may be seen as follows: The second line in eqs.\eqref{eq:model NLO Lphi3 remainder q=0} and \eqref{eq:model NLO Lphi3 remainder q=5} does not give any contribution in the case at hand, since for $l=l'=0$ the denominator in the summand would vanish. Further, for $L_5$, i.e. the finite part of the hypergemetric series, we find (for $d=j=0$)

\begin{equation}
 \lim_{\epsilon\to 0}\pFq{6}{5}{1,\frac{5}{2},\frac{3}{2},2,1,\frac{3}{2}}{\frac{3}{2},2,\frac{5}{2},\frac{3}{2},2}{1-\epsilon}=\lim_{\epsilon\to 0}\pFq{2}{1}{1,1}{2}{1-\epsilon}=\lim_{\epsilon\to 0}\frac{\log \epsilon}{\epsilon-1}
\end{equation}
Since there is no finite contribution to this hypergeometric series, we have $L_5=0$ in the case at hand. Concerning the logarithmic contribution, we derive the results given above using $\braket{0000}{00}=1$.

It remains to determine the operators $(R_1)_{\varphi^3}(x;q=1)$ and $(R_1)_{\varphi^3}(x;q=2)$. Unfortunately, in this case we have neither been able to find a closed form expression, nor have we been able to verify the cancellation of infinities. The additional difficulty here is due to the lack of relations like $d\geq j$ or $-d-3>j$, which were a result of the fact that the three ladder operators in $(R_1)_{\varphi^3}(x,q=0)$ and $(R_1)_{\varphi^3}(x,q=3)$ were either all creators, or all annihilators. Thus, we can not use the simplifications discussed in appendix \ref{app:characteristic sum}.

The end of this section is again devoted to matrix elements of the left representative $\Lo{1}{\varphi^3}{x}$, i.e. to OPE coefficients of the form $\Co{1}{b}{\varphi^3 a}$.

\begin{results}[OPE coefficients $\Co{1}{b}{\varphi^3\, a}$]\label{res:model NLO Lphi C1phi3}%
The matrix elements of the left representative $\Lo{1}{\varphi^3}{x}$ given in result \ref{res:model NLO Lphi3} are
\begin{equation}
\Co{1}{b}{\varphi^3\, a}(x)=
                           0 \qquad \text{for }g(a,b)>7 \text{ or }g(a,b)=\text{even}
\label{eq:model NLO Lphi C1phi3 result vanishing}
\end{equation}
and
\begin{align}
\Co{1}{b}{\varphi^3\, a}(x)=& 3\sum\limits_{\mathfrak{A}\in\mathcal{I}^b_a(7)}\sum\limits_{\mathcal{P}_{5,2}[\mathfrak{A}]}\sum\limits_{J,J_1,J_2}
f^b_a[\mathfrak{A}]\, \D{1}{P_1}{J_1}\,\Delta_{0}[\mathfrak{P}_2]_{J_2}\,T[J_1,J_2]_{J}\, S_J(\hat{x})\, r^{\mathbf{d}}
\notag\\ +&3\sum\limits_{\mathfrak{B}\in\mathcal{I}^b_a(5)}\sum\limits_{\mathcal{P}_{4,1}[\mathfrak{B}]}\sum\limits_{J,J_1,J_2}
f^b_a[\mathfrak{B}]\, \Ro{1}{P_1}{J_1}{p}{2}\,\Delta_{0}[\mathfrak{P}_2]_{J_2}\,T[J_1,J_2]_{J}\, S_J(\hat{x})\cdot r^{\mathbf{d}}
\notag\\
+&\sum\limits_{\mathfrak{C}\in\mathcal{I}^b_a(3)}\sum\limits_{J}
f^b_a[\mathfrak{C}]\, \Ro{1}{C}{J}{p}{3}\cdot S_J(\hat{x})\cdot r^{\mathbf{d}}
\label{eq:model NLO Lphi C1phi3 result}
\end{align}
for $g(a,b)>1$, with

\begin{equation}
 \Lambda_1[\varphi^3,\mathfrak{A}=\Lbag \mathfrak{l}_1^q,\mathfrak{l}_2^q,\mathfrak{l}_{3}^q\Rbag]_{J}:=\,  s[\mathfrak{A}]\cdot T[\mathfrak{A}]_{J}\, (R_1)_{\varphi^3}[d_{\mathfrak{A}}-1,J,q,r] \quad .
\label{eq:model NLO Lphi3 additional notation R}
\end{equation}\index{symbols}{Lambda1p3@$\Lambda_1[\varphi^3,\mathfrak{A}]_{J}$}
\end{results}
Remark: The family of coefficients with $g(a,b)=1$ is not considered here, since all these coefficients necessarily involve contributions from remainder terms $(R_1)_{\varphi^3}(x,q=1)$ or $(R_1)_{\varphi^3}(x,q=2)$, which are unknown as mentioned above.

\begin{proof}

The argumentation here is basically the same as in the previous section. Eq.\eqref{eq:model NLO Lphi C1phi3 result vanishing} is just a consequence of proposition \ref{propos:labelings}.
Eq.\eqref{eq:model NLO Lphi C1phi3 result} may be derived analogously to eq.\eqref{eq:model NLO Lphi C1phi2 result} of the previous section, so we will discuss it only briefly. In the first line the only difference is that the submultiset $\mathfrak{P}_2$ now has cardinality 2. This is due to the fact that in the case at hand we have one additional power of $\Lo{0}{\varphi}{x}$, so we sum over diagrams of the form

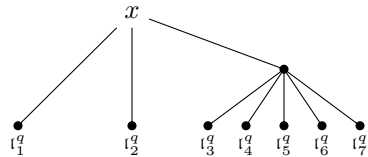
\begin{equation}
 \begin{tikzpicture}
   \node at (0,-.25) {$x$}  child[level distance = 1.5cm,sibling distance=1.5cm]{[fill] circle (1.5pt) node[below]{\begin{tiny}$\mathfrak{l}_1^q$\end{tiny}} } child[level distance = 1.5cm,sibling distance=2.5cm]{[fill] circle (1.5pt) node[below]{\begin{tiny}$\mathfrak{l}_2^q$\end{tiny}} } child[level distance=.75cm,sibling distance=2cm] {[fill] circle (1.5pt) child[sibling distance=.5cm] {[fill] circle (1.5pt) node[below]{\begin{tiny}$\mathfrak{l}_3^q$\end{tiny}} } child[sibling distance=.5cm] {[fill] circle (1.5pt) node[below]{\begin{tiny}$\mathfrak{l}_4^q$\end{tiny}} } child[sibling distance=.5cm] {[fill] circle (1.5pt) node[below]{\begin{tiny}$\mathfrak{l}_5^q$\end{tiny}} } child[sibling distance=.5cm] {[fill] circle (1.5pt) node[below]{\begin{tiny}$\mathfrak{l}_6^q$\end{tiny}} } child[sibling distance=.5cm] {[fill] circle (1.5pt) node[below]{\begin{tiny}$\mathfrak{l}_7^q$\end{tiny}} }};
 \end{tikzpicture}
\end{equation}
The additional leaf directly attached to the root is the cause of the additional element in $\mathfrak{P}_2$. The second line of the result accounts for the contribution from the product $3:(R_1)_{\varphi^2}(x)\Lo{0}{\varphi}{x}:$. In total this product consists of 5 ladder operators, so we have to sum over all multisets $\mathfrak{B}\in\mathcal{I}^b_a(5)$. Then these multisets are split into a submultiset of cardinality 4, which is associated to the contribution $\Lambda_1$ from the remainder operator, and a submultiset of cardinality 1, which is the argument of the contribution $\Delta_0$ from the zeroth order term. Coupling all spherical harmonics, we obtain the second line of eq.\eqref{eq:model NLO Lphi C1phi3 result}. The third line of that equation follows if we insert eq.\eqref{eq:model NLO Lphi3 remainder grading} for the remainder operator and use result \ref{res:model NLO Lphi C0phik} in order to determine the matrix elements of this expression.

\end{proof}

\subsection{Construction of \texorpdfstring{$\Lo{1}{\varphi^4}{x}$}{L(1)phi4}}

Next in line is the left representative $\Lo{1}{\varphi^4}{x}$, which can be determined from known expressions by the formula (see \eqref{eq:model perturbations limit k-1})

\begin{equation}
\begin{split}
 \Lo{1}{\varphi^4}{x}=\lim_{y\to x}\Big[&\Lo{0}{\varphi}{x}\Lo{1}{\varphi^3}{y}+\Lo{1}{\varphi}{x}\Lo{0}{\varphi^3}{y}-\Co{1}{\varphi^4}{\varphi^3\, \varphi}(x-y)\Lo{0}{\varphi^4}{x}\\
 -&\Co{1}{\varphi^2}{\varphi^3\, \varphi}(x-y)\Lo{0}{\varphi^2}{x}-\Co{1}{\varphi\varphi_{1m}}{\varphi^3\, \varphi}(x-y)\Lo{0}{\varphi\varphi_{1m}}{x}\\ -&\Co{0}{\varphi^2}{\varphi^3\, \varphi}(x-y)\Lo{1}{\varphi^2}{x}- \Co{0}{\varphi\varphi_{1m}}{\varphi^3\, \varphi}(x-y)\Lo{1}{\varphi\varphi_{1m}}{x} \Big]\, .
\end{split}
\label{eq:model NLO Lphi4 limit}
\end{equation}
The counterterms here take the values (see appendix \ref{app:OPE table})

\begin{align}
 \Co{1}{\varphi^4}{\varphi^3\, \varphi}(x-y)&=60\log |x-y| \label{eq:model NLO Lphi4 CT logarithmic} \\
 \Co{1}{\varphi^2}{\varphi^3\, \varphi}(x-y)&= [(R_1)_{\varphi^3}(x-y)]^{\varphi^2}_{\varphi}\label{eq:model NLO Lphi4 CT unknown}\\
\Co{1}{\varphi\varphi_{1m}}{\varphi^3\, \varphi}(x-y)&=[(R_1)_{\varphi^3}(x-y)]^{\varphi\varphi_{1m}}_{\varphi}\label{eq:model NLO Lphi4 CT unknown2} \\
 \Co{0}{\varphi^2}{\varphi^3\, \varphi}(x-y)&=\frac{3}{|x-y|}
\label{eq:model NLO Lphi4 CT polynomial}
\end{align}
Recall that we have not found results for the OPE coefficients $\Co{1}{b}{\varphi^3\, a}$ with $g(a,b)=1$, so the two coefficients in eqs.\eqref{eq:model NLO Lphi4 CT unknown} and \eqref{eq:model NLO Lphi4 CT unknown2} are unknown. Clearly this is a problem if one wants to determine $\Lo{1}{\varphi^4}{x}$ completely. However, in this section our modest aim is to determine only the contribution acting by more than two ladder operators. As we shall see below, this can be achieved without knowledge of the mentioned counterterms.

By now the general procedure should be familiar: We exploit the commutation relations of the ladder operators (and possibly results of the previous sections as shortcuts) to bring the desired left representative into the form

\begin{results}[Left representative $\Lo{1}{\varphi^4}{x}$ in normal ordered form]\label{res:model NLO Lphi4}
 
\begin{equation}
\begin{split}
 \Lo{1}{\varphi^4}{x}=&4:\Lo{1}{\varphi}{x}\Lo{0}{\varphi^3}{x}:+6:(R_1)_{\varphi^2}(x)\Lo{0}{\varphi^2}{x}:\\
+&4:(R_1)_{\varphi^3}(x)\Lo{0}{\varphi}{x}:+(R_1)_{\varphi^4}(x)
\end{split}
\label{eq:model NLO Lphi4 NO form}
\end{equation}
where

\begin{equation}
 \begin{split}
 (R_1)_{\varphi^4}(x)=\lim_{y\to x}\Big[&\contraction[.5ex]{}{\Lo{1}{\varphi}{x}}{}{\Lo{0}{\varphi}{y}} \contraction[1ex]{}{\Lo{1}{\varphi}{x}}{\Lo{0}{\varphi}{y}}{\Lo{0}{\varphi}{y}} \contraction[1.5ex]{}{\Lo{1}{\varphi}{x}}{\Lo{0}{\varphi}{y}\Lo{0}{\varphi}{y}}{\Lo{0}{\varphi}{y}} \Lo{1}{\varphi}{x}\Lo{0}{\varphi}{y}\Lo{0}{\varphi}{y}\Lo{0}{\varphi}{y}+ 
\contraction{}{\Lo{0}{\varphi}{x}}{}{(R_1)_{\varphi^3}}\Lo{0}{\varphi}{x}(R_1)_{\varphi^3}(y)\\
-& \Co{1}{\varphi^2}{\varphi^3\, \varphi}(x-y)\Lo{0}{\varphi^2}{x}-\Co{1}{\varphi\varphi_{1m}}{\varphi^3\, \varphi}(x-y)\Lo{0}{\varphi\varphi_{1m}}{x} \Big]
 \end{split}
\label{eq:model NLO R1phi4}
\end{equation}\index{symbols}{R14@$(R_1)_{\varphi^4}(x)$}

\end{results}
Remark: According to the scheme outlined in section \ref{subsec:perturbation via field equation} all matrix elements of $(R_1)_{\varphi^4}(x)$ should be finite. It should be noted that we have not been able to check this convergence in this thesis.

\begin{proof}

The shortest way to this result is again to first substitute $\Lo{1}{\varphi^3}{x}$ in the form of eq.\eqref{eq:model NLO Lphi3 remainder form} into eq.\eqref{eq:model NLO Lphi4 limit}.

\begin{equation}
\begin{split}
 &\Lo{1}{\varphi^4}{x}=\lim_{y\to x}\Big[3\Lo{0}{\varphi}{x}:\Lo{0}{\varphi^2}{y}\Lo{1}{\varphi}{y}:+3\Lo{0}{\varphi}{x}:\Lo{0}{\varphi}{y}(R_1)_{\varphi^2}(y):\\+&\Lo{0}{\varphi}{x}(R_1)_{\varphi^3}(y)+\Lo{1}{\varphi}{x}\Lo{0}{\varphi^3}{y}-\Co{1}{\varphi^4}{\varphi^3\, \varphi}(x-y) \Lo{0}{\varphi^4}{x}
 -\Co{1}{\varphi^2}{\varphi^3\, \varphi}(x-y) \Lo{0}{\varphi^2}{x}\\-&\Co{1}{\varphi\varphi_{1m}}{\varphi^3\, \varphi}(x-y)\Lo{0}{\varphi\varphi_{1m}}{x} -\Co{0}{\varphi^2}{\varphi^3\, \varphi}(x-y)\Lo{1}{\varphi^2}{x}- \Co{0}{\varphi\varphi_{1m}}{\varphi^3\, \varphi}(x-y)\Lo{1}{\varphi\varphi_{1m}}{x} \Big]\, .
\end{split}
\label{eq:model NLO Lphi4 subst}
\end{equation}
Now let us bring all expressions with a positive sign in this equation into normal order:

\begin{equation}
\begin{split}
3\lim_{y\to x}\, \Big[\Lo{0}{\varphi}{x}\,&:\Lo{1}{\varphi}{y}\Lo{0}{\varphi^2}{y}:\Big]\, =3:\Lo{1}{\varphi}{x}\Lo{0}{\varphi^3}{x}:\\
+&3\lim_{y\to x}\Big[ :\contraction{}{\Lo{0}{\varphi}{x}}{}{\Lo{1}{\varphi}{y}} \Lo{0}{\varphi}{x}\Lo{1}{\varphi}{y}\Lo{0}{\varphi^2}{y}:+ 2:\contraction{}{\Lo{0}{\varphi}{x}}{\Lo{1}{\varphi}{y}}{\Lo{0}{\varphi}{y}}\Lo{0}{\varphi}{x} \Lo{1}{\varphi}{y}\Lo{0}{\varphi}{y}\Lo{0}{\varphi}{y}:\Big]
\end{split}
\label{eq:model NLO Lphi4 limit NO 1}
\end{equation}

\begin{equation}
\begin{split}
\lim_{y\to x}\, \Big[\Lo{1}{\varphi}{x}&\Lo{0}{\varphi^3}{y}\Big]=\, :\Lo{1}{\varphi}{x}\Lo{0}{\varphi^3}{x}: + \lim_{y\to x}\Big[3:\contraction{}{\Lo{1}{\varphi}{x}}{}{\Lo{0}{\varphi}{y}}   \Lo{1}{\varphi}{x}\Lo{0}{\varphi}{y}\Lo{0}{\varphi^2}{y}:\\
+&3 :\contraction[.5ex]{}{\Lo{1}{\varphi}{x}}{}{\Lo{0}{\varphi}{y}} \contraction[1ex]{}{\Lo{1}{\varphi}{x}}{\Lo{0}{\varphi}{y}}{\Lo{0}{\varphi}{y}}   \Lo{1}{\varphi}{x}\Lo{0}{\varphi}{y}\Lo{0}{\varphi}{y}\Lo{0}{\varphi}{y}:+ \contraction[.5ex]{}{\Lo{1}{\varphi}{x}}{}{\Lo{0}{\varphi}{y}} \contraction[1ex]{}{\Lo{1}{\varphi}{x}}{\Lo{0}{\varphi}{y}}{\Lo{0}{\varphi}{y}} \contraction[1.5ex]{}{\Lo{1}{\varphi}{x}}{\Lo{0}{\varphi}{y}\Lo{0}{\varphi}{y}}{\Lo{0}{\varphi}{y}}  \Lo{1}{\varphi}{x}\Lo{0}{\varphi}{y}\Lo{0}{\varphi}{y}\Lo{0}{\varphi}{y}\Big]
\end{split}
\label{eq:model NLO Lphi4 limit NO 2}
\end{equation}

\begin{equation}
\begin{split}
3\lim\limits_{x\to y}\, \Big[\Lo{0}{\varphi}{x}:\Lo{0}{\varphi}{y}&(R_1)_{\varphi^2}(y):\Big]=3 :\Lo{0}{\varphi^2}{x}(R_1)_{\varphi^2}(x):\\ +&3\lim_{y\to x}\Big[\contraction{}{\Lo{0}{\varphi}{x}}{}{\Lo{0}{\varphi}{y}}   \Lo{0}{\varphi}{x}\Lo{0}{\varphi}{y}(R_1)_{\varphi^2}(y)
+: \contraction{}{\Lo{0}{\varphi}{x}}{\Lo{0}{\varphi}{y}}{(R_1)_{\varphi^2}} \Lo{0}{\varphi}{x}\Lo{0}{\varphi}{y}(R_1)_{\varphi^2}(y) :\Big]
\end{split}
\label{eq:model NLO Lphi4 limit NO 3}
\end{equation}

\begin{equation}
\begin{split}
&\lim_{y\to x}\, \Big[\Lo{0}{\varphi}{x}(R_1)_{\varphi^3}(y)\Big]=\, :\Lo{0}{\varphi}{x}(R_1)_{\varphi^3}(x): + \lim_{y\to x} \contraction{}{\Lo{0}{\varphi}{x}}{}{(R_1)_{\varphi^3}} \Lo{0}{\varphi}{x}(R_1)_{\varphi^3}(y)
\end{split}
\label{eq:model NLO Lphi4 limit NO 4}
\end{equation}
Many of the contracted operator products in these equations have already been analysed in the previous chapters. To begin with, recall from \eqref{eq:model NLO Lphi3 subtree polynomial divergence}

\begin{equation}
 \lim_{y\to x} \contraction{}{\Lo{0}{\varphi}{x}}{}{\Lo{0}{\varphi}{y}} \Lo{0}{\varphi}{x}\Lo{0}{\varphi}{y} =\lim_{y\to x}\,\frac{1}{x-y}\, .
\label{eq:model NLO Lphi4 subtree polynomial divergence}
\end{equation}
which appears in eq.\eqref{eq:model NLO Lphi4 limit NO 1} and eq.\eqref{eq:model NLO Lphi4 limit NO 3}. With the definitions of $(R_1)_{\varphi^2}$ and $\Lo{1}{\varphi^2}{x}$, see eqs.\eqref{eq:model NLO Lphi2 infinite sums formula} and \eqref{eq:model NLO Lphi2 normal ordered}, it follows that

\begin{equation}
\begin{split}
3&\lim\limits_{y\to x}\Big[2:\contraction{}{\Lo{0}{\varphi}{x}}{\Lo{1}{\varphi}{y}}{\Lo{0}{\varphi}{y}}\Lo{0}{\varphi}{x} \Lo{1}{\varphi}{y}\Lo{0}{\varphi}{y}\Lo{0}{\varphi}{y}:+ \contraction{}{\Lo{0}{\varphi}{x}}{}{\Lo{0}{\varphi}{y}}   \Lo{0}{\varphi}{x}\Lo{0}{\varphi}{y}(R_1)_{\varphi^2}(y) \Big]\\
 &=\lim_{y\to x}\Big(6:\Lo{1}{\varphi}{y}\Lo{0}{\varphi}{y}:+3(R_1)_{\varphi^2}(y)\Big)\frac{1}{x-y}=\lim_{y\to x}\frac{3}{x-y}\Lo{1}{\varphi^2}{y}\, ,
\end{split}
\end{equation}
which cancels with the polynomial counterterms in the last line of eq.\eqref{eq:model NLO Lphi4 limit} after a Taylor expansion in $y$ around $x$. Further we find for the sum of the products with one contraction in eqs.\eqref{eq:model NLO Lphi4 limit NO 1} and \eqref{eq:model NLO Lphi4 limit NO 2}

\begin{equation}
3\lim_{y\to x}\Big[:\contraction{}{\Lo{0}{\varphi}{x}}{}{\Lo{1}{\varphi}{y}} \Lo{0}{\varphi}{x}\Lo{1}{\varphi}{y}\Lo{0}{\varphi^2}{y}:+ :\contraction{}{\Lo{1}{\varphi}{x}}{}{\Lo{0}{\varphi}{y}}   \Lo{1}{\varphi}{x}\Lo{0}{\varphi}{y}\Lo{0}{\varphi^2}{y}:  \Big]=3:\Lo{0}{\varphi^2}{x}(R_1)_{\varphi^2}(x):
\end{equation}
where we again applied the results of section \ref{subsubsec:L1phi^2}. Remembering the definition of $(R_1)_{\varphi^3}$, \eqref{eq:model NLO Lphi3 remainder form b},  we find for the following expressions from eqs.\eqref{eq:model NLO Lphi4 limit NO 2} and \eqref{eq:model NLO Lphi4 limit NO 3}

\begin{equation}
\begin{split}
 3&\lim_{y\to x}\Big[:\contraction[.5ex]{}{\Lo{1}{\varphi}{x}}{}{\Lo{0}{\varphi}{y}} \contraction[1ex]{}{\Lo{1}{\varphi}{x}}{\Lo{0}{\varphi}{y}}{\Lo{0}{\varphi}{y}}   \Lo{1}{\varphi}{x}\Lo{0}{\varphi}{y}\Lo{0}{\varphi}{y}\Lo{0}{\varphi}{y}:+ : \contraction{}{\Lo{0}{\varphi}{x}}{\Lo{0}{\varphi}{y}}{(R_1)_{\varphi^2}} \Lo{0}{\varphi}{x}\Lo{0}{\varphi}{y}(R_1)_{\varphi^2}(y) : \Big]\\
&=3:\Lo{0}{\varphi}{x}(R_1)_{\varphi^3}(x):+60\Lo{0}{\varphi^4}{x}\lim_{y\to x}\log |x-y|\quad .
\end{split}
\end{equation}
Here the divergence cancels with the logarithmic counterterm \eqref{eq:model NLO Lphi4 CT logarithmic}. We have now dealt with all divergent expressions in eqs.\eqref{eq:model NLO Lphi4 limit NO 1}-\eqref{eq:model NLO Lphi4 limit NO 4}, except for the products including three contractions in eq.\eqref{eq:model NLO Lphi4 limit NO 2} and the contraction with the remainder term $(R_1)_{\varphi^3}$ in eq.\eqref{eq:model NLO Lphi4 limit NO 4} (which is essentially also a threefold contraction, since $(R_1)_{\varphi^3}$ itself includes two contractions). Further, the only remaining counterterms are the ones in the second line of eq.\eqref{eq:model NLO Lphi4 limit}. Thus, we have found

\begin{equation}
 \begin{split}
 (R_1)_{\varphi^4}(x)=\lim_{y\to x}\Big[&\contraction[.5ex]{}{\Lo{1}{\varphi}{x}}{}{\Lo{0}{\varphi}{y}} \contraction[1ex]{}{\Lo{1}{\varphi}{x}}{\Lo{0}{\varphi}{y}}{\Lo{0}{\varphi}{y}} \contraction[1.5ex]{}{\Lo{1}{\varphi}{x}}{\Lo{0}{\varphi}{y}\Lo{0}{\varphi}{y}}{\Lo{0}{\varphi}{y}} \Lo{1}{\varphi}{x}\Lo{0}{\varphi}{y}\Lo{0}{\varphi}{y}\Lo{0}{\varphi}{y}+ 
\contraction{}{\Lo{0}{\varphi}{x}}{}{(R_1)_{\varphi^3}}\Lo{0}{\varphi}{x}(R_1)_{\varphi^3}(y)\\
-& \Co{1}{\varphi^2}{\varphi^3\, \varphi}(x-y)\Lo{0}{\varphi^2}{x}-\Co{1}{\varphi\varphi_{1m}}{\varphi^3\, \varphi}(x-y)\Lo{0}{\varphi\varphi_{1m}}{x} \Big]
 \end{split}
\end{equation}
and we may verify eq.\eqref{eq:model NLO Lphi4 NO form} by insertion of the above results into eq.\eqref{eq:model NLO Lphi4 subst}.
\end{proof}
Although we do not determine the concrete form of the new remainder term $(R_1)_{\varphi^4}$ in this thesis, results for a wide class of OPE coefficients $\Co{1}{b}{\varphi^4\, a}$ can be obtained. Namely

\begin{results}[OPE coefficients $\Co{1}{b}{\varphi^4\, a}$]\label{res:model NLO Lphi C1phi4}%
The matrix elements of the left representative $\Lo{1}{\varphi^4}{x}$ presented in result \ref{res:model NLO Lphi4} are
\begin{equation}
\Co{1}{b}{\varphi^4\, a}(x)=
                           0 \qquad \text{for }g(a,b)>8 \text{ or }g(a,b)=\text{odd}
\end{equation}
and
\begin{align}
\Co{1}{b}{\varphi^4\, a}(x)=& 4\sum\limits_{\mathfrak{A}\in\mathcal{I}^b_a(8)}\sum\limits_{\mathcal{P}_{5,3}[\mathfrak{A}]}\sum\limits_{J,J_1,J_2}
f^b_a[\mathfrak{A}]\, \D{1}{P_1}{J_1}\,\Delta_{0}[\mathfrak{P}_2]_{J_2}\,T[J_1,J_2]_{J}\, S_J(\hat{x})\, r^{\mathbf{d}}
\notag\\ +&6\sum\limits_{\mathfrak{B}\in\mathcal{I}^b_a(6)}\sum\limits_{\mathcal{P}_{4,2}[\mathfrak{B}]}\sum\limits_{J,J_1,J_2}
f^b_a[\mathfrak{B}]\, \Ro{1}{P_1}{J_1}{p}{2}\,\Delta_{0}[\mathfrak{P}_2]_{J_2}\,T[J_1,J_2]_{J}\, S_J(\hat{x})\cdot r^{\mathbf{d}}
\notag\\
+&4\sum\limits_{\mathfrak{C}\in\mathcal{I}^b_a(4)}\sum\limits_{\mathcal{P}_{3,1}[\mathfrak{C}]}\sum\limits_{J,J_1,J_2}
f^b_a[\mathfrak{C}]\, \Ro{1}{P_1}{J_1}{p}{3}\,\Delta_{0}[\mathfrak{P}_2]_{J_2}\,T[J_1,J_2]_{J}\, S_J(\hat{x})\cdot r^{\mathbf{d}}
\label{eq:model NLO Lphi C1phi4 result}
\end{align}
for $g(a,b)>2$.
\end{results}

Remark: The remaining classes of coefficients, namely $\Co{1}{b}{\varphi^4\, b}$ with $g(a,b)=2$ and $g(a,b)=0$, include contributions from $(R_1)_{\varphi^4}$ and are thus not treated here. Otherwise the derivation of the above result is analogous to previous sections, so we do not bother to give a ``proof'' here.

\subsection{Construction of \texorpdfstring{$\Lo{1}{\varphi^5}{x}$}{L(1)phi5}}

The last left representative at first perturbation order to be discussed in this thesis is $\Lo{1}{\varphi^5}{x}$. This is the final step in our algorithm before we can proceed to second order by the field equation. The general outline of this section is similar to the previous ones. As always, our starting point is the expression of $\Lo{1}{\varphi^5}{x}$ in terms of known left representatives,

\begin{equation}
\begin{split}
 \Lo{1}{\varphi^5}{x}=\lim_{y\to x}\Big[&\Lo{0}{\varphi}{x}\Lo{1}{\varphi^4}{y}+\Lo{1}{\varphi}{x}\Lo{0}{\varphi^4}{y}-\Co{1}{\varphi^5}{\varphi^4\, \varphi}(x-y)\Lo{0}{\varphi^5}{x}\\
-&\Co{0}{\varphi^3}{\varphi^4\, \varphi}(x-y)\Lo{1}{\varphi^3}{x}-\Co{0}{\varphi^2\varphi_{1m}}{\varphi^4\, \varphi}(x-y)\Lo{1}{(\varphi^2\varphi_{1m})}{x}\\
-&\Co{1}{\varphi^3}{\varphi^4\, \varphi}(x-y)\Lo{0}{\varphi^3}{x}-\Co{1}{\varphi^2\varphi_{1m}}{\varphi^4\, \varphi}(x-y)\Lo{0}{(\varphi^2\varphi_{1m})}{x}\\
-&\Co{1}{\varphi}{\varphi^4\, \varphi}(x-y)\Lo{0}{\varphi}{x}\Big]
\end{split}
\label{eq:model NLO Lphi5 limit}
\end{equation}
which is a consequence of our consistency condition. It should be remarked that in the last line of this equation no counerterms of the form $\Co{1}{\varphi_{lm}}{\varphi^4\, \varphi}$ with $l\neq 0$ appear, because this coefficient is zero according to the results of the previous section (since here $g(a,b)=$odd). The counterterms in the above equation take the values

\begin{eqnarray}
 \Co{1}{\varphi^5}{\varphi^4\, \varphi}(x-y)&=&120\log |x-y| \label{eq:model NLO Lphi5 CT logarithmic} \\
 \Co{0}{\varphi^3}{\varphi^4\, \varphi}(x-y)&=&\frac{4}{|x-y|}\label{eq:model NLO Lphi5 CT polynomial} \\
 \Co{1}{\varphi^3}{\varphi^4\, \varphi}(x-y)&=&4\Co{1}{\varphi^2}{\varphi^3\, \varphi}(x-y)+[(R_1)_{\varphi^4}(x-y)]^{\varphi^3}_{\varphi} \label{eq:model NLO Lphi5 CT mixed1}\\
\Co{1}{\varphi^2\varphi_{1m}}{\varphi^4\, \varphi}(x-y)&=& 4\left(\Co{1}{\varphi^2}{\varphi^3\, \varphi}(x-y)\cdot |x-y|S^{1m}(\hat{x})+\Co{1}{\varphi\varphi_{1m}}{\varphi^3\, \varphi}(x-y)\right)\notag\\
&&+4\Co{1}{\varphi^2\varphi_{1m}}{\varphi^3\, \mathds{1}}(x-y)\cdot\frac{1}{|x-y|}+[(R_1)_{\varphi^4}(x-y)]^{\varphi^2\varphi_{1m}}_{\varphi}\label{eq:model NLO Lphi5 CT mixed1 der}\\
 \Co{1}{\varphi}{\varphi^4\, \varphi}(x-y)&=& [(R_1)_{\varphi^4}]_{\varphi}^{\varphi} \label{eq:model NLO Lphi5 CT mixed2}
\end{eqnarray}
The above relation between the coefficients $\Co{1}{\varphi^3}{\varphi^4\, \varphi}$ and $\Co{1}{\varphi^2}{\varphi^3\, \varphi}$ may be derived as follows: The sets $\mathcal{I}^{\varphi^2}_{\varphi}(7)$ and $\mathcal{I}^{\varphi^2}_{\varphi}(5)$ are empty, so the first two lines in eq.\eqref{eq:model NLO Lphi C1phi3 result} do not give any contribution to $\Co{1}{\varphi^2}{\varphi^3\, \varphi}$. Thus, the only contribution comes from the matrix element of the remainder term $(R_1)_{\varphi^3}$, i.e.

\begin{equation}
 \Co{1}{\varphi^2}{\varphi^3\, \varphi}(x)=[(R_1)_{\varphi^3}(x)]^{\varphi^2}_{\varphi}\quad.
\label{eq:model NLO Lphi5 relation counterterms}
\end{equation}
Now let us come to the coefficient $\Co{1}{\varphi^3}{\varphi^4\, \varphi}$. Here the first two lines of eq.\eqref{eq:model NLO Lphi C1phi4 result} vanish, because the sets $\mathcal{I}^{\varphi^3}_{\varphi}(8)$ and $\mathcal{I}^{\varphi^3}_{\varphi}(6)$ are empty as well. Therefore, the coefficient at hand contains the following contributions:

\begin{equation}
\begin{split}
 \Co{1}{\varphi^3}{\varphi^4\, \varphi}(x)=&\bra{\varphi^3}4:(R_1)_{\varphi^3}(x)\Lo{0}{\varphi}{x}:\ket{\varphi}+\bra{\varphi^3}(R_1)_{\varphi^4}(x)\ket{\varphi}\\
 =&4[(R_1)_{\varphi^3}(x)]_{\mathds{1}}^{\varphi^3}\cdot\Co{0}{\mathds{1}}{\varphi\, \varphi}(x)+4[(R_1)_{\varphi^3}(x)]_{\varphi}^{\varphi^2}\cdot\Co{0}{\varphi^3}{\varphi\, \varphi^2}(x)+[(R_1)_{\varphi^4}(x)]_{\varphi}^{\varphi^3}\\
 =&4[(R_1)_{\varphi^3}(x)]_{\varphi}^{\varphi^2}+[(R_1)_{\varphi^4}(x)]_{\varphi}^{\varphi^3}=4\Co{1}{\varphi^2}{\varphi^3\, \varphi}(x)+[(R_1)_{\varphi^4}(x)]_{\varphi}^{\varphi^3}
\end{split}
\end{equation}
In the second line we decomposed the left representative $\Lo{0}{\varphi}{x}$ into a creation and an annihilation part. In addition, the results $\Co{0}{\varphi^3}{\varphi\, \varphi^2}(x)=1$ and $[(R_1)_{\varphi^3}(x)]_{\mathds{1}}^{\varphi^3}=0$ from the previous sections were applied. The final equality then follows from eq.\eqref{eq:model NLO Lphi5 relation counterterms} and confirms eq.\eqref{eq:model NLO Lphi5 CT mixed1}. Eq.\eqref{eq:model NLO Lphi5 CT mixed1 der} was derived in a similar manner.

The strategy is	 again to write $\Lo{1}{\varphi^5}{x}$ as a sum of some known normal ordered (and thus finite) expressions and an additional remainder term $(R_1)_{\varphi^5}$. The computation of this remainder term is the main effort that goes into the construction of the desired left representative.

\begin{results}[Left representative $\Lo{1}{\varphi^5}{x}$ in normal ordered form]\label{res:model NLO Lphi5}
\begin{equation}
\begin{split}
 \Lo{1}{\varphi^5}{x}=&5:\Lo{1}{\varphi}{x}\Lo{0}{\varphi^4}{x}:+10:(R_1)_{\varphi^2}(x)\Lo{0}{\varphi^3}{x}:+10:(R_1)_{\varphi^3}(x)\Lo{0}{\varphi^2}{x}:\\+ &5:(R_1)_{\varphi^4}(x)\Lo{0}{\varphi}{x}:+(R_1)_{\varphi^5}(x)
\end{split}
\label{eq:model NLO Lphi5 NO form}
\end{equation}
where

\begin{equation}
 \begin{split}
 &(R_1)_{\varphi^5}(x)=\\
&\lim_{y\to x}\Bigg[\contraction[.5ex]{}{\Lo{1}{\varphi}{x}}{}{\Lo{0}{\varphi}{y}} \contraction[1ex]{}{\Lo{1}{\varphi}{x}}{\Lo{0}{\varphi}{y}}{\Lo{0}{\varphi}{y}} \contraction[1.5ex]{}{\Lo{1}{\varphi}{x}}{\Lo{0}{\varphi}{y}\Lo{0}{\varphi}{y}}{\Lo{0}{\varphi}{y}} 
\contraction[2ex]{}{\Lo{1}{\varphi}{x}}{\Lo{0}{\varphi}{y}\Lo{0}{\varphi}{y}\Lo{0}{\varphi}{y}}{\Lo{0}{\varphi}{y}}
  \Lo{1}{\varphi}{x}\Lo{0}{\varphi}{y}\Lo{0}{\varphi}{y}\Lo{0}{\varphi}{y}\Lo{0}{\varphi}{y}+
\contraction{}{\Lo{0}{\varphi}{x}}{}{(R_1)_{\varphi^4}}\Lo{0}{\varphi}{x}(R_1)_{\varphi^4}(y)
-\Co{1}{\varphi}{\varphi^4\, \varphi}(x-y)\Lo{0}{\varphi}{x}\Bigg]
 \end{split}
\label{eq:model NLO R1phi5}
\end{equation}\index{symbols}{R15@$(R_1)_{\varphi^5}(x)$}

\end{results}
Remark: Again, we have not been able to verify convergence of this limit explicitly. However, our renormalization procedure implies

\begin{results}[Constraints on remainder terms]\label{res:model NLO Lphi5 remainder terms}
For the consistency condition \eqref{eq:model perturbations limit k-1} to hold, it is necessary that
 \begin{equation}
 [(R_1)_{\varphi^4}(x)]_{\varphi}^{\varphi^3}=0
\end{equation}
and
\begin{equation}
 [(R_1)_{\varphi^4}(x)]_{\varphi}^{\varphi^2\varphi_{1m}}=-160\cdot S^{1m}(\hat{x})\cdot(\log r+c)
\end{equation}
with $c\in\mathbb{C}$.
\end{results}

\begin{proof}[of results \ref{res:model NLO Lphi5} and \ref{res:model NLO Lphi5 remainder terms}]

Insertion of $\Lo{1}{\varphi^4}{x}$ in the form \eqref{eq:model NLO Lphi4 NO form} into our equation for $\Lo{1}{\varphi^5}{x}$ yields

\begin{equation}
\begin{split}
 &\Lo{1}{\varphi^5}{x}=\lim_{y\to x}\Big[4\Lo{0}{\varphi}{x}:\Lo{1}{\varphi}{y}\Lo{0}{\varphi^3}{y}:+6\Lo{0}{\varphi}{x}:(R_1)_{\varphi^2}(y)\Lo{0}{\varphi^2}{y}:\\+&4\Lo{0}{\varphi}{x}:(R_1)_{\varphi^3}(y)\Lo{0}{\varphi}{y}:+\Lo{0}{\varphi}{x}(R_1)_{\varphi^4}(y)+\Lo{1}{\varphi}{x}\Lo{0}{\varphi^4}{y}-\text{ ``counterterms'' }\Big]
\end{split}
\label{eq:model NLO Lphi5 subst}
\end{equation}
The next step is to bring these expressions into normal order and to keep track of the additional terms generated in the process. To be begin with, we pick up the additional expressions

\begin{equation}
\lim_{y\to x}\Big[4: \contraction{}{\Lo{0}{\varphi}{x}}{}{\Lo{1}{\varphi}{y}} \Lo{0}{\varphi}{x}\Lo{1}{\varphi}{y}\Lo{0}{\varphi^3}{y}:\, + 4: \contraction{}{\Lo{1}{\varphi}{x}}{}{\Lo{0}{\varphi}{y}} \Lo{1}{\varphi}{x}\Lo{0}{\varphi}{y}\Lo{0}{\varphi^3}{y}:  \Big] =4:\Lo{0}{\varphi^3}{x}(R_1)_{\varphi^2}(x):
\label{eq:model NLO Lphi5 normal ordering 1}
\end{equation}
from normal ordering of the first and the last product in eq.\eqref{eq:model NLO Lphi5 subst}. Further, the contractions

\begin{equation}
\begin{split}
&4\lim_{y\to x}\Big[\contraction{}{\Lo{0}{\varphi}{x}}{}{\Lo{0}{\varphi}{y}} \Lo{0}{\varphi}{x}\Lo{0}{\varphi}{y}\, \Big(3:\Lo{1}{\varphi}{y}\Lo{0}{\varphi^2}{y}:+3:(R_1)_{\varphi^2}(y)\Lo{0}{\varphi}{y}:+(R_1)_{\varphi^3}(y)\Big)\Big] \\
& =4\lim_{y\to x}\frac{1}{|x-y|}\Big(3:\Lo{1}{\varphi}{y}\Lo{0}{\varphi^2}{y}:+3:(R_1)_{\varphi^2}(y)\Lo{0}{\varphi}{y}:+(R_1)_{\varphi^3}(y)\Big)=\lim_{y\to x}\frac{4}{|x-y|}\Lo{1}{\varphi^3}{y}\, ,
\end{split}
\end{equation}
result from normal ordering of the first three summands in eq.\eqref{eq:model NLO Lphi5 subst}. The result cancels with the polynomial counterterms in the second line of \eqref{eq:model NLO Lphi5 limit}. The remaining expressions with two contractions are

\begin{equation}
\begin{split}
6&\lim_{y\to x}\Big[ :\contraction{}{\Lo{0}{\varphi}{x}}{}{(R_1)_{\varphi^2}}\Lo{0}{\varphi}{x}(R_1)_{\varphi^2}(y)\Lo{0}{\varphi^2}{y}:\, +\, :\contraction[.5ex]{}{\Lo{1}{\varphi}{x}}{}{\Lo{0}{\varphi}{y}}\contraction{}{\Lo{1}{\varphi}{x}}{\Lo{0}{\varphi}{y}}{\Lo{0}{\varphi}{y}} \Lo{1}{\varphi}{x}\Lo{0}{\varphi}{y}\Lo{0}{\varphi}{y}\Lo{0}{\varphi^2}{y}  :  \Big]\\
&\quad=6:(R_1)_{\varphi^3}(x)\Lo{0}{\varphi^2}{x}:+120\Lo{0}{\varphi^5}{x}\lim_{y\to x}\log|x-y|
\end{split}
\end{equation}
This divergence cancels with the logarithmic counterterm in eq.\eqref{eq:model NLO Lphi5 CT logarithmic}. Next consider the products with three contractions

\begin{equation}
\begin{split}
 &4\lim_{y\to x}\Big[ : \contraction{}{\Lo{0}{\varphi}{x}}{}{(R_1)_{\varphi^3}} \Lo{0}{\varphi}{x}(R_1)_{\varphi^3}(y)\Lo{0}{\varphi}{y} :\,+\, :\contraction[.5ex]{}{\Lo{1}{\varphi}{x}}{}{\Lo{0}{\varphi}{y}}\contraction{}{\Lo{1}{\varphi}{x}}{\Lo{0}{\varphi}{y}}{\Lo{0}{\varphi}{y}} \contraction[1.5ex]{}{\Lo{1}{\varphi}{x}}{\Lo{0}{\varphi}{y}\Lo{0}{\varphi}{y}}{\Lo{0}{\varphi}{y}} \Lo{1}{\varphi}{x}\Lo{0}{\varphi}{y}\Lo{0}{\varphi}{y}\Lo{0}{\varphi}{y}\Lo{0}{\varphi}{y}  : \Big]\\
=&4:(R_1)_{\varphi^4}(x)\Lo{0}{\varphi}{x}:\\+&4\lim_{y\to x}\Big[ :\Lo{0}{\varphi^2}{x}\Lo{0}{\varphi}{y}:\Co{1}{\varphi^2}{\varphi^3\, \varphi}(x-y)+\Co{1}{\varphi\varphi_{1m}}{\varphi^3\, \varphi}(x-y):\Lo{0}{\varphi\varphi_{1m}}{x}\Lo{0}{\varphi}{y}:  \Big] \quad ,
\end{split}
\label{eq:model NLO Lphi5 normal ordering 4}
\end{equation}
which follows from the definition of $(R_1)_{\varphi^4}$, see eq.\eqref{eq:model NLO R1phi4}. Here we encounter the problem that neither the OPE coefficients in the expression above, nor the counterterms in eqs.\eqref{eq:model NLO Lphi5 CT mixed1} and \eqref{eq:model NLO Lphi5 CT mixed1 der} are explicitly known. Thus it seems difficult to verify the cancellation of infinite terms in the limit above. This is not really a problem, however, since $\Lo{1}{\varphi^5}{x}$ \emph{is} finite by its very construction (see section \ref{subsec:perturbation via field equation}). Thus, we may change our point of view and from now on \emph{assume} that the counterterms render the left representative finite, instead of trying to show this with the help of results from previous sections. This yields the following constraints:

\begin{align}
 &\lim_{y\to x}\left[4\Co{1}{\varphi^2}{\varphi^3\, \varphi}(x-y)-\Co{1}{\varphi^3}{\varphi^4\, \varphi}(x-y)\right]=\text{finite}\label{eq:model NLO Lphi5 counterterms condition1}
\\
&\lim_{y\to x}\Big[4\left(\Co{1}{\varphi^2}{\varphi^3\, \varphi}(x-y)\cdot |x-y|S^{1m}(\hat{x})+\Co{1}{\varphi\varphi_{1m}}{\varphi^3\, \varphi}(x-y)\right)-\Co{1}{\varphi^2\varphi_{1m}}{\varphi^4\, \varphi}(x-y)\Big]=\text{finite}\label{eq:model NLO Lphi5 counterterms condition2}
\end{align}
These conditions were derived as follows: We performed a Taylor expansion of the operators $\Lo{0}{\varphi}{y}$ in eq.\eqref{eq:model NLO Lphi5 normal ordering 4} around the point $x$ and neglected all terms with positive scaling dimension in $|x-y|$, since these terms vanish in the limit. Eq.\eqref{eq:model NLO Lphi5 counterterms condition1} then follows if we demand that the resulting expressions proportional to $\Lo{0}{\varphi^3}{x}$ are rendered finite by the corresponding counterterm, eq.\eqref{eq:model NLO Lphi5 CT mixed1}. Similarly, eq.\eqref{eq:model NLO Lphi5 counterterms condition2} collects all terms multiplying the left representative $\Lo{0}{\varphi^2\varphi_{1m}}{x}$ and requires that subtraction of the corresponding counterterm, eq.\eqref{eq:model NLO Lphi5 CT mixed1 der}, yields a finite result.

Substitution of eq.\eqref{eq:model NLO Lphi5 CT mixed1} into the first condition above yields

\begin{equation}
 \lim_{y\to x}\left[-[(R_1)_{\varphi^4}(x-y)]_{\varphi}^{\varphi^3}\right]=\text{finite}\quad.
\end{equation}
Since the OPE coefficient $\Co{1}{\varphi^3}{\varphi^4\, \varphi}$ has scaling degree $1$, this is also true for the contribution from the remainder term, $[(R_1)_{\varphi^4}(x-y)]_{\varphi}^{\varphi^3}$. Thus, we conclude that $[(R_1)_{\varphi^4}(x-y)]_{\varphi}^{\varphi^3}$ is of the form $(c_1\log|x-y|+c_2)/|x-y|$, where $c_1,c_2\in\mathbb{C}$ are constants. This fact together with the condition above uniquely determines $[(R_1)_{\varphi^4}]_{\varphi}^{\varphi^3}$.

\begin{equation}
 [(R_1)_{\varphi^4}(x)]_{\varphi}^{\varphi^3}=0
\end{equation}
This confirms the first half of result \ref{res:model NLO Lphi5 remainder terms}. Now let us come to eq.\eqref{eq:model NLO Lphi5 counterterms condition2}. Here substitution of eq.\eqref{eq:model NLO Lphi5 CT mixed1 der} leads to

\begin{equation}
 \lim_{y\to x}\left[-[(R_1)_{\varphi^4}(x-y)]_{\varphi}^{\varphi^2\varphi_{1m}}-4[(R_1)_{\varphi^3}(x-y)]_{\mathds{1}}^{\varphi^2\varphi_{1m}}\frac{1}{|x-y|}\right]=\text{finite}\quad.
\end{equation}
Dimensional analysis of the expressions in brackets suggests that all summands are at most logarithmically divergent. Therefore, the above constraint is not as strong as the first constraint, in the sense that it will not allow for a unique determination of $[(R_1)_{\varphi^4}]_{\varphi}^{\varphi^2\varphi_{1m}}$. The term $[(R_1)_{\varphi^3}(x-y)]_{\mathds{1}}^{\varphi^2\varphi_{1m}}$ may be determined with the help of result \ref{res:model NLO Rphi3}. As mentioned above, we are only interested in the logarithmic contribution to this matrix element. We find

\begin{equation}
 [(R_1)_{\varphi^3}(x)]_{\mathds{1}}^{\varphi^2\varphi_{1m}}=40\cdot |x|\, S^{1m}(\hat{x})\cdot\log r+\text{polynomial contribution}
\end{equation}
which yields upon insertion into the constraint above

\begin{equation}
 [(R_1)_{\varphi^4}(x)]_{\varphi}^{\varphi^2\varphi_{1m}}=-160\cdot S^{1m}(\hat{x})\cdot(\log r+c)
\end{equation}
where $c\in\mathbb{C}$ is some constant. Thus, we have confirmed result \ref{res:model NLO Lphi5 remainder terms}.

It remains to analyze the genuinely new contributions containing four contractions and the remaining counterterm

\begin{equation}
\begin{split}
(R_1)_{\varphi^5}(x)=\lim_{y\to x}\Big[\contraction{}{\Lo{0}{\varphi}{x}}{}{(R_1)_{\varphi^4}} \Lo{0}{\varphi}{x}(R_1)_{\varphi^4}(y)+ \contraction[.5ex]{}{\Lo{1}{\varphi}{x}}{}{\Lo{0}{\varphi}{y}} \contraction[1ex]{}{\Lo{1}{\varphi}{x}}{\Lo{0}{\varphi}{y}}{\Lo{0}{\varphi}{y}} \contraction[1.5ex]{}{\Lo{1}{\varphi}{x}}{\Lo{0}{\varphi}{y}\Lo{0}{\varphi}{y}}{\Lo{0}{\varphi}{y}} 
\contraction[2ex]{}{\Lo{1}{\varphi}{x}}{\Lo{0}{\varphi}{y}\Lo{0}{\varphi}{y}\Lo{0}{\varphi}{y}}{\Lo{0}{\varphi}{y}}
  \Lo{1}{\varphi}{x}\Lo{0}{\varphi}{y}\Lo{0}{\varphi}{y}\Lo{0}{\varphi}{y}\Lo{0}{\varphi}{y}-\Co{1}{\varphi}{\varphi^4\, \varphi}(x-y)\Lo{0}{\varphi}{x} \Big]
\end{split}
\end{equation}
Restoring all the normal ordered products obtained in eqs.\eqref{eq:model NLO Lphi5 normal ordering 1}-\eqref{eq:model NLO Lphi5 normal ordering 4}, one verifies eq.\eqref{eq:model NLO Lphi5 NO form}.
\end{proof}

The OPE coefficients $\Co{1}{b}{\varphi^5\, a}$ with $g(a,b)>3$ can be determined without any knowledge of $(R_1)_{\varphi^4}$ and $(R_1)_{\varphi^5}$, so we will focus on these cases. 

\begin{results}[OPE coefficients $\Co{1}{b}{\varphi^5\, a}$]\label{res:model NLO Lphi C1phi5}%
\begin{equation}
\Co{1}{b}{\varphi^5\, a}(x)=
                           0 \qquad \text{for }g(a,b)>9 \text{ or }g(a,b)=\text{even}
\end{equation}
and
\begin{align}
\Co{1}{b}{\varphi^5\, a}(x)=& 5\sum\limits_{\mathfrak{A}\in\mathcal{I}^b_a(9)}\sum\limits_{\mathcal{P}_{5,4}[\mathfrak{A}]}\sum\limits_{J,J_1,J_2}
f_a^b[\mathfrak{A}]\, \D{1}{P_1}{J_1}\,\Delta_{0}[\mathfrak{P}_2]_{J_2}\,T[J_1,J_2]_{J}\, S_J(\hat{x})\, r^{\mathbf{d}}
\notag\\ +&10\sum\limits_{\mathfrak{B}\in\mathcal{I}^b_a(7)}\sum\limits_{\mathcal{P}_{4,3}[\mathfrak{B}]}\sum\limits_{J,J_1,J_2}
f_a^b[\mathfrak{B}]\, \Ro{1}{P_1}{J_1}{p}{2}\,\Delta_{0}[\mathfrak{P}_2]_{J_2}\,T[J_1,J_2]_{J}\, S_J(\hat{x})\cdot r^{\mathbf{d}}
\notag\\
+&10\sum\limits_{\mathfrak{C}\in\mathcal{I}^b_a(5)}\sum\limits_{\mathcal{P}_{3,2}[\mathfrak{C}]}\sum\limits_{J,J_1,J_2}
f_a^b[\mathfrak{C}]\, \Ro{1}{P_1}{J_1}{p}{3}\,\Delta_{0}[\mathfrak{P}_2]_{J_2}\,T[J_1,J_2]_{J}\, S_J(\hat{x})\cdot r^{\mathbf{d}}
\label{eq:model NLO Lphi C1phi5 result}
\end{align}
for $g(a,b)>3$.
\end{results}
The result may be derived from the form of the left representative $\Lo{1}{\varphi^5}{x}$ in analog to the previous sections.

\subsection{Construction of \texorpdfstring{$\Lo{2}{\varphi}{x}$}{L(2)phi}}

According to the algorithm outlined in section \ref{subsec:perturbation via field equation} it is possible to construct the second order left representatives $\Lo{2}{\varphi}{x}$, or equivalently the OPE coefficients $\Co{2}{b}{\varphi\, a}$, from the knowledge of the first order left representatives $\Lo{1}{\varphi^5}{x}$. In the previous chapters we have presented the iteration up to this point, so we are finally ready to exploit the field equation and proceed to second perturbation order. This process will be carried out in the present chapter.

The central equation of this chapter is

\begin{equation}
\begin{split}
 \Delta \Lo{2}{\varphi}{x}=\Lo{1}{\varphi^5}{x}=&5:\Lo{1}{\varphi}{x}\Lo{0}{\varphi^4}{x}:+10:(R_1)_{\varphi^2}(x)\Lo{0}{\varphi^3}{x}:\\ +&10:(R_1)_{\varphi^3}(x)\Lo{0}{\varphi^2}{x}:+5:(R_1)_{\varphi^4}(x)\Lo{0}{\varphi}{x}:+(R_1)_{\varphi^5}(x)
\end{split}
\label{eq:model NNLO field equation}
\end{equation}
which follows from eq.\eqref{eq:model perturbations relations}. Since we do not know the concrete form of the operators $(R_1)_{\varphi^4}$ and $(R_1)_{\varphi^5}$, we will only be able to analyze the contributions from the first three terms on the right side of the above equation. As we have seen in the previous chapter, this still allows for the computation of a large class of OPE coefficients $\Co{2}{b}{\varphi\, a}$, namely those with $g(a,b)>3$. With the help of the equation above and the definition (here we implicitly assume that at $n$-th perturbation order logarithms up to the power $n$ may appear, which will be proven in section \ref{subsec:higher order})

\begin{definition}[Gradings by powers of the logarithm]\index{symbols}{Yidjp@$\Lo{n}{\ket{v_a}}{x;d}_J^p$}%
The gradings of the vertex operators $\Lo{n}{\ket{v_a}}{x}$ and the remainder terms $(R_n)_{\varphi^k}(x)$ by scaling dimension $d$, ``spin'' $J$ and powers of logarithms $p$ are denoted by $\Lo{n}{\ket{v_a}}{x;d}_J^p$ and $(R_n)_{\varphi^k}(x,d)_J^p$, i.e.

\begin{equation}
 \Lo{n}{\ket{v_a}}{x}=\sum_{d=-\infty}^{\infty}\sum_{J=0}^{\infty}\sum_{p=0}^n \Lo{n}{\ket{v_a}}{x;d}_J^p\cdot r^d (\log r)^p S_J(\hat{x})
\end{equation}
and

\begin{equation}
 (R_n)_{\varphi^k}(x)=\sum_{d=-\infty}^{\infty}\sum_{J=0}^{\infty}\sum_{p=0}^n (R_n)_{\varphi^k}(x,d)_J^p\cdot r^d (\log r)^p S_J(\hat{x})\quad .
\end{equation}
Furthermore, let

\begin{equation}
\Do{1}{A}{J}{p}:=\frac{1}{p!}\frac{d^p}{(d\log r)^p}\D{1}{A}{J}\Big|_{\log r=0}
\label{eq:model NNLO Lphi Delta 1p}
\end{equation}
and

\begin{equation}
\Rop{1}{A}{J}{p}{k}:=\frac{1}{p!}\frac{d^p}{(d\log r)^p}\Ro{1}{A}{J}{p}{k}\Big|_{\log r=0}\quad.
\label{eq:model NNLO Lphi Lambda 1p}
\end{equation}
\end{definition}
we obtain the following partial result

\begin{results}[Left representative $\Lo{2}{\varphi}{x}$]\label{res:model NNLO Lphi}%
Using the operator $\Delta^{-1}$ to solve the differential equation \eqref{eq:model NNLO field equation}, we find
 \begin{equation}
\begin{split}
 \Lo{2}{\varphi}{x}=&5:\Delta^{-1}\left[\Lo{1}{\varphi}{x}\Lo{0}{\varphi^4}{x}\right]:\, +10:\Delta^{-1}\left[(R_1)_{\varphi^2}(x)\Lo{0}{\varphi^3}{x}\right]:\\
+&10:\Delta^{-1}\left[(R_1)_{\varphi^3}(x)\Lo{0}{\varphi^2}{x}\right]:\, +\text{ terms including }\leq\text{3 ladder operators}
\end{split}
\label{eq:model NNLO Lphi NO}
\end{equation}
\end{results}
where the terms on the right side are concretely

\begin{align}
 \Delta^{-1}\left[\Lo{1}{\varphi}{x}\Lo{0}{\varphi^4}{x}\right]=\sum_{d_1,d_2=-\infty}^{\infty}\sum_{J,J_1,J_2}\sum_{p=0}^1\mathcal{Y}_1(\varphi,x;d_1)_{J_1}^{p}\cdot\mathcal{Y}_0(\varphi^4,x;d_2)_{J_2}&\notag \\ T[J_1,J_2]_J\, S_J(\hat{x})\, r^{d_1+d_2+2}(\log r)^p\cdot D^{(p)}[d_1+d_2+2,J,r]&
\label{eq:model NNLO Lphi explicit NO}\\
\Delta^{-1}\left[(R_1)_{\varphi^2}(x)\Lo{0}{\varphi^3}{x}\right]=\sum_{d_1,d_2=-\infty}^{\infty}\sum_{J,J_1,J_2}\sum_{p=0}^1(R_1)_{\varphi^2}(x;d_1)_{J_1}^{p}\mathcal{Y}_0(\varphi^3,x;d_2)_{J_2}&\notag \\ T[J_1,J_2]_J\, S_J(\hat{x})\, r^{d_1+d_2+2}(\log r)^p\cdot D^{(p)}[d_1+d_2+2,J,r]&
\label{eq:model NNLO Lphi explicit remainder 1a}\\
\Delta^{-1}\left[(R_1)_{\varphi^3}(x)\Lo{0}{\varphi^2}{x}\right]=\sum_{d_1,d_2=-\infty}^{\infty}\sum_{J,J_1,J_2}\sum_{p=0}^1(R_1)_{\varphi^3}(x;d_1)_{J_1}^{p}\mathcal{Y}_0(\varphi^2,x;d_2)_{J_2}&\notag \\ T[J_1,J_2]_J\, S_J(\hat{x})\, r^{d_1+d_2+2}(\log r)^p\cdot D^{(p)}[d_1+d_2+2,J,r]&
\label{eq:model NNLO Lphi explicit remainder 2}
\end{align}
with $D^{(q)}(d,J,r)$ defined as in eq.\eqref{eq:app diff eq D} (the special cases $q=0$ and $q=1$ can be found in eqs.\eqref{eq:app diff eq D0} and \eqref{eq:app diff eq D1}).

\begin{proof}
 We simply have to use the gradings introduced above, the coupling rules of the spherical harmonics as discussed in appendix \ref{app:3-dim symmetries} and the solution to the resulting differential equation from appendix \ref{app:characteristic diff eq}.
\end{proof}

As in the previous chapter, this knowledge allows for the computation of the OPE coefficients $\Co{2}{b}{\varphi\, a}$ with $g(a,b)=9$ and $g(a,b)=7$ in full generality, and for $g(a,b)=5$ in the cases where $(R_1)_{\varphi^3}$ is known. 

\begin{results}[OPE coefficients $\Co{2}{b}{\varphi\, a}$]\label{res:model NLO Lphi C2phi}%
The matrix elements of the left representative $\Lo{2}{\varphi}{x}$ given in result \ref{res:model NNLO Lphi} are
\begin{equation}
\Co{2}{b}{\varphi\, a}(x)=
                           0 \qquad \text{for }g(a,b)>9 \text{ or }g(a,b)=\text{even}
\label{eq:model NNLO Lphi C2phi result vanishing}
\end{equation}
and
\begin{align}
&\Co{2}{b}{\varphi\, a}(x)=\Bigg[ 5\sum\limits_{\mathfrak{A}\in\mathcal{I}^b_a(9)}\sum\limits_{\mathcal{P}_{5,4}[\mathfrak{A}]}\sum\limits_{J,J_1,J_2}\sum\limits_{p=0}^1f_a^b[\mathfrak{A}]\, \Do{1}{P_1}{J_1}{p}\,\Delta_{0}[\mathfrak{P}_2]_{J_2}\,T[J_1,J_2]_{J}\, S_J(\hat{x})\, r^{\mathbf{d}} (\log r)^p\notag\\ +&10\sum\limits_{\mathfrak{B}\in\mathcal{I}^b_a(7)}\sum\limits_{\mathcal{P}_{4,3}[\mathfrak{B}]}\sum\limits_{J,J_1,J_2}\sum\limits_{p=0}^1f_a^b[\mathfrak{B}]\, \Rop{1}{P_1}{J_1}{p}{2}\,\Delta_{0}[\mathfrak{P}_2]_{J_2}\,T[J_1,J_2]_{J}\, S_J(\hat{x})\cdot r^{\mathbf{d}} (\log r)^p\notag\\
+&10\sum\limits_{\mathfrak{C}\in\mathcal{I}^b_a(5)}\sum\limits_{\mathcal{P}_{3,2}[\mathfrak{C}]}\sum\limits_{J,J_1,J_2}\sum\limits_{p=0}^1f_a^b[\mathfrak{C}]\, \Rop{1}{P_1}{J_1}{p}{3}\,\Delta_{0}[\mathfrak{P}_2]_{J_2}\,T[J_1,J_2]_{J}\, S_J(\hat{x})\cdot r^{\mathbf{d}} (\log r)^p\Bigg]\cdot D^{(p)}[{\mathbf{d}},J,r]
\label{eq:model NLO Lphi C2phi result}
\end{align}
for $g(a,b)>3$.
\end{results}

\begin{proof}
All we have to do is invert the Laplace operator on the OPE coefficients $\Co{1}{b}{\varphi^5\, a}$, see result \ref{res:model NLO Lphi C1phi5}. This effectively means that we have to multiply logarithmic expressions by $D^{(1)}$ and polynomial expressions by $D^{(0)}$. This procedure yields eq.\eqref{eq:model NLO Lphi C2phi result}.
\end{proof}

\section[Some higher order results] {Some higher order results\footnote{As above, we use the operator $\Delta^{-1}$ to solve the field equation in this section. For other solutions the results of this section might not hold.}}\label{subsec:higher order}

One aim of this thesis is to recognize patterns in our iterative scheme and in this way to extrapolate our knowledge of low perturbation orders to gain some insight into higher orders. The present section, which is dedicated to precisely this topic, is structured as follows: First we extend our results for the simplest class of OPE coefficients to arbitrary order in perturbation theory (still in the 3-dimensional model considered in the previous sections). Then the general structure of more complicated higher order coefficients is discussed.

We begin our discussion of results for arbitrary orders with the analysis of vanishing OPE coefficients. We would first like to show

\begin{proposition}[Maximum number of ladder operators in left representatives]\label{propos:max ladder operators}
The left representative $\Lo{n}{\varphi^k}{x}$ contains products of no more than $4n+k$ ladder operators.
\end{proposition}

\begin{proof}
 For the left representatives of the free theory this follows simply from eq.\eqref{eq:model free field La}. Now suppose we know the claim holds at order $n-1$. Then we know that the left representative $\Lo{n-1}{\varphi^5}{x}$ is related to $\Lo{n}{\varphi}{x}$ by the field equation \eqref{eq:model perturbations OPE field equation}. Hence $\Lo{n}{\varphi}{x}$ contains at most $4(n-1)+5=4n+1$ ladder operators just as we claimed. In order to construct the other $n$-th order left representatives we use the consistency condition. Let us now assume the claim holds for $\Lo{n}{\varphi^{k-1}}{x}$. Then the consistency condition yields

\begin{equation}
\begin{split}
 \Lo{n}{\varphi^k}{x}=\lim_{y\to x}&\left[\sum_{i=0}^{n}\Lo{i}{\varphi}{x}\Lo{n-i}{\varphi^{k-1}}{y}-\sum_{i=1}^{n} \Co{i}{c}{\varphi^{k-1}\varphi}(x-y)\Lo{n-i}{\ket{v_c}}{x}\right. \\-&\left.\Co{0}{\varphi^{k-2}}{\varphi^{k-1}\varphi}(x-y) \Lo{n}{\varphi^{k-2}}{x}-\Co{0}{(\varphi^{k-3}\varphi_{1m})}{\varphi^{k-1}\varphi}(x-y) \Lo{n}{\varphi^{k-3}\varphi_{1m}}{x}\right]
\end{split}
\label{model:higher orders consistency condition}
\end{equation}
Since the left representatives up to $\Lo{n}{\varphi^{k-1}}{x}$ fulfill the proposition by assumption, it is easy to check that the product of the left representatives on the right side of this equation contains at most $4i+1+4(n-i)+k-1=4n+k$ ladder operators. Further, the OPE coefficient $\Co{i}{c}{\varphi^{k-1}\varphi}$ vanishes in the limit $y\to x$ if $\ket{v_c}$ contains higher powers than $\varphi^{k}$. Thus, the left representative $\Lo{n-i}{\ket{v_c}}{x}$ multiplying this coefficient contains at most $4(n-i)+k$ ladder operators, where $i> 0$. It remains to discuss the left representatives in the second line, which also fulfill the desired property by assumption. Therefore, our claim also holds for $\Lo{n}{\varphi^k}{x}$ and by iteration of this procedure for arbitrary left representatives.

\end{proof}
In a similar manner, it can be shown that

\begin{proposition}\label{propos:evenodd ladder operators}%
The left representative $\Lo{n}{\varphi^k}{x}$ contains only products of an even number of ladder operators if $k$ is even, and an odd number of ladder operators if $k$ is odd.
\end{proposition}

\begin{proof}
 Again there is nothing to show at zeroth perturbation order due to eq.\eqref{eq:model free field La}. Also, assuming the claim holds at order $n-1$, it holds for $\Lo{n}{\varphi}{x}$ if we use $\Delta^{-1}$ to solve the field equation. Thus, it remains to check whether the consistency condition respects our proposition. Assuming the left representatives up to $\Lo{n}{\varphi^{k-1}}{x}$ satisfy the proposition, we can deduce that both factors in the product of the left representatives on the right side of eq.\eqref{model:higher orders consistency condition} satisfy our claim, so the product as a whole does so as well. It remains to investigate the counterterms. Let for the moment $k=$even. Then, since $\Lo{i}{\varphi^{k-1}}{x}$ acts by an odd number of ladder operators for $i\leq n$, one can show (see result \ref{result:vanishing general OPE coeff}) that the OPE coefficient $\Co{i}{c}{\varphi^{k-1}\varphi}$ vanishes if $g(\ket{v_c},\varphi)=even$, i.e. if $\ket{v_c}$ is constructed from $\varphi$ by an even number of ladder operators. In other words, for the coefficient not to vanish, $\ket{v_c}$ has to be of the form $\varphi_{l_1m_1}\cdots\varphi_{l_jm_j}$ where $j$ is an even number. Therefore, the left representative $\Lo{n-i}{\ket{v_c}}{x}$, which multiplies this OPE coefficient, is obtained from $\Lo{n-i}{\varphi^j}{x}$ by taking the appropriate derivatives, see eq.\eqref{eq:model free field La}. For $j=$even this left representative also acts by an even number of ladder operators due to our assumption, so the proposition does indeed hold. On the other hand, if $k=$odd, one can follow the same argumentation to show that $\Lo{n-i}{\ket{v_c}}{x}$ acts by an odd number of ladder operators. As the left representatives in the second line of eq.\eqref{model:higher orders consistency condition} fulfill the desired property by assumption, the iteration is complete.
\end{proof}
As a simple conclusion from these propositions, we find

\begin{results}[Vanishing OPE coefficients]\label{result:vanishing general OPE coeff}%
At arbitrary perturbation order $n\in\mathbb{N}$ and for any exponent $k\in\mathbb{N}$ the equation

\begin{equation}
 \Co{n}{b}{\varphi^k\, a}(x)=0\qquad \text{ for }g(a,b)>4n+k \text{ or }g(a,b)+k=\text{odd}
\end{equation}
holds.
\end{results}

\begin{proof}
 By proposition \ref{propos:max ladder operators} the left representative $\Lo{n}{\varphi^k}{x}$ acts by at most $4n+k$ ladder operators. Further, $g(a,b)>4n+k$ means that more than $4n+k$ ladder operators are needed to transform $\ket{v_a}$ into $\ket{v_b}$. Thus, the matrix element $\bra{v_b}\Lo{n}{\varphi^k}{x}\ket{v_a}=\Co{n}{b}{\varphi^k\, a}(x)$ vanishes, due to orthonormality of our basis. 

Now we come to the second part of the result. Assume $g(a,b)=$even for the moment, i.e. we need an even number of ladder operators to transform $\ket{v_a}$ into $\ket{v_b}$. Then only the part of $\Lo{n}{\varphi^k}{x}$ that acts by an even number of ladder operators contributes to the coefficient $\Co{n}{b}{\varphi^k\, a}$. Proposition \ref{propos:evenodd ladder operators} tells us that for $k=$odd, this left representative does not contain any contribution of this kind, so for $g(a,b)+k=$odd the OPE coefficient under consideration vanishes. If on the other hand $g(a,b)=$odd, we find by the same arguments that the coefficient vanishes for $k=$even, which finishes the proof.
\end{proof}
This result implies that only the coefficients $\Co{n}{b}{\varphi^k\, a}$ with $g(a,b)=4n+k-2i$ where $i\in\mathbb{N}$ are non-zero.
The difficulty in the computation of these remaining coefficients depends strongly on the value of $g(a,b)$, as we have also seen in the constructions of the previous sections. This is due to the fact that for $g(a,b)=4n+k-2i$ we have to contract $i$ pairs of ladder operators, which essentially means that we have to solve an $i$-fold infinite sum. Thus, it is natural to first consider the coefficients $\Co{n}{b}{\varphi^k\, a}$ with $g(a,b)=4n+k$, since here no infinite sums appear. In this simple case it is possible to give a closed form expression with the help of the following generalizations of our notation:

\begin{definition}[Generalized partitions]\index{symbols}{Pijgen@$\mathcal{P}_{(a_1,\ldots,a_n)}[\mathfrak{A}]$}%
Let $\mathcal{P}_{(a_1,\ldots,a_n)}\left[\mathfrak{A}\right]$ be the set of partitions of any multiset $\mathfrak{A}$ of cardinality $a_1+\ldots+a_n$ into $n$ submultisets of cardinality $a_1, a_2,\ldots\text{ and }a_n$ respectively, whose sum is $\mathfrak{A}$, i.e.
 \begin{equation}
\begin{split}
&\mathcal{P}_{(a_1,\ldots,a_n)}\left[\mathfrak{A}\right]=\left\{\mathfrak{P}_1,\ldots,\mathfrak{P}_n\Big| \operatorname{card}\mathfrak{P}_i=a_i \, \forall i\in \{1,\ldots,n\}\text{ and }  \mathfrak{P}_1\uplus \cdots\uplus \mathfrak{P}_n=\mathfrak{A} \right\}
\end{split}
\end{equation}
\end{definition}

\begin{definition}[Notation for higher orders]
For $n> 0$ we define recursively
 \begin{equation}
\begin{split}
&\Delta_n[\mathfrak{A}=\Lbag \mathfrak{l}^q_1,\ldots,\mathfrak{l}^q_{4n+1}\Rbag,r]_J:=\sum_{n_1,\ldots, n_5}^{n_1+\ldots+n_5=n-1}\sum_{\mathcal{P}_{(4n_1+1,\ldots,4n_5+1)}[\mathfrak{A}]}\sum_{p_1=0}^{n_1}\cdots\sum_{p_5=0}^{n_5}\sum_{J_1,\ldots,J_5}\\
&\times (\log r)^{p_1+\ldots+p_5}  D^{\left(p_1+\ldots+p_5\right)}\left[d_{\mathfrak{A}},J,r\right]\cdot \Do{n_1}{P_1}{J_1}{p_1}\cdots \Do{n_5}{P_5}{J_5}{p_5}\, T[J_1,\ldots,J_5]_{J}
\end{split}
\label{model:higher orders definition Dn}
\end{equation}\index{symbols}{Deltan@$\Delta_n[\mathfrak{A},r]_J$}
where $D^{(n)}[d_{\mathfrak{A}},J,r]$ is defined as in eq.\eqref{eq:app diff eq D} and with

\begin{equation}
 \Do{n}{A}{J}{p}:=\frac{1}{p!}\frac{\text{d}^p}{(\text{d}\log r)^p}\, \Delta_n[\mathfrak{A},r]_J\Big|_{\log r=0}\quad .
\label{model:higher orders definition Dn grading}
\end{equation}\index{symbols}{Deltanp@$\Do{n}{A}{J}{p}$}
As in previous sections, let

\begin{equation}
\Delta_0[\mathfrak{A}]_J=\Delta_0^0[\mathfrak{A}]_J= T[\mathfrak{A}]_J\, s[\mathfrak{A}]\quad .
\end{equation}
\end{definition}
Remark: This definition of $\Delta_1^p$ is consistent with the formula given in eq.\eqref{eq:model NNLO Lphi Delta 1p}. This can be seen as follows: Note that for $n=1$ the parameters $n_1,\ldots,n_5$ are all restricted to be equal to zero. Thus, eq.\eqref{model:higher orders definition Dn} takes the form

\begin{equation}
 \Delta_1[\mathfrak{A}=\Lbag\mathfrak{l}^q_1,\ldots,\mathfrak{l}^q_5\Rbag,r]_J=\sum_{\mathcal{P}_{(1,1,1,1,1)}[\mathfrak{A}]} D^{(0)}[d_{\mathfrak{A}},J,r]\, T[\mathfrak{A}]_{J}
\end{equation}
where we also used the fact that $\Delta_0[\mathfrak{A}=\Lbag\mathfrak{l}^q_i\Rbag]_{J}=\delta_{l_i,J}$ according to the definition above. The sum over partitions of $\mathfrak{A}$ into submultisets of cardinality 1 is equivalent to a sum over permutations of the elements of $\mathfrak{A}$. The right side of the above equation is invariant under such permutations (the submultisets $\mathfrak{P}_i$ do not appear) so we may replace this sum by a symmetry factor, which by definition is just $s[\mathfrak{A}]$. Therefore, eq.\eqref{eq:model NNLO Lphi Delta 1p} is equivalent to the definition above.

\begin{results}[The simplest class of non-vanishing OPE coefficients]\label{result:simple general OPE coeff}
 For $g(a,b)=4n+k$, $b=a+\sum_{i=1}^{4n+k}e_{\mathfrak{l}_i^q}$ and $\mathfrak{A}=\Lbag\mathfrak{l}_i^q,\ldots,\mathfrak{l}_{4n+k}^q\Rbag$ the equation

\begin{equation}
 \begin{split}
  \Co{n}{b}{\varphi^k\, a}(x)=&f_a^b[\mathfrak{A}] r^{\mathbf{d}}\sum_{n_1, \ldots, n_k=0}^{n_1+\ldots+n_k=n}\, \sum_{\mathcal{P}_{(4n_1+1,\ldots,4n_k+1)}\left[\mathfrak{A}\right]}\sum_{J}\sum_{J_1,\ldots,J_5}\\ &\times 	S_{JM}(\hat{x})T[J_1,\ldots,J_k]_{J}
  \cdot \Delta_{n_1}[\mathfrak{P}_1,r]_{J_1}\cdots \Delta_{n_k}[\mathfrak{P}_k,r]_{J_k}
 \end{split}
\label{model:higher orders CT coeff}
\end{equation}
holds.
\end{results}
Before we give the proof of this result, let us first present the following

\begin{lemma}[Maximum number of ladder operators in counterterms]\label{lemma:max operators in CT}%
The counterterms appearing in the construction of an arbitrary left representative $\Lo{n}{\varphi^k}{x}$ act by less than $4n+k$ ladder operators.
\end{lemma}

\begin{proof}
We have already shown this in the proof of proposition \ref{propos:max ladder operators}. There we have argued that the counterterms in the first line of eq.\eqref{model:higher orders consistency condition} contain at most products of $4(n-i)+k$ ladder operators with $i>0$. For the counterterms in the second line of eq.\eqref{model:higher orders consistency condition} the lemma holds trivially.
\end{proof}
Now we are ready for the

\begin{proof}[of result \ref{result:simple general OPE coeff}]
Let $g(a,b)=4n+k$. Then it follows from lemma \ref{lemma:max operators in CT} that we may write\footnote{Here we do not have to require normal ordering, since any contribution containing a product of contracted ladder operators vanishes due to proposition \ref{propos:max ladder operators}.}

\begin{equation}
 \bra{v_b}\Lo{n}{\varphi^k}{x}\ket{v_a}=\bra{v_b}\sum_{i=0}^n \Lo{n-i}{\varphi}{x}\Lo{i}{\varphi^{k-1}}{x}\ket{v_a}\quad,
\end{equation}
since the matrix elements of all the counterterms in eq.\eqref{model:higher orders consistency condition} vanish, which also allows us to perform the limit $y\to x$. Repetition of this procedure yields

\begin{equation}
 \bra{v_b}\Lo{n}{\varphi^k}{x}\ket{v_a}=\bra{v_b}\sum_{i=0}^n\sum_{j=0}^i \Lo{n-i}{\varphi}{x}\Lo{i-j}{\varphi}{x}\Lo{j}{\varphi^{k-2}}{x}\ket{v_a}\quad.
\end{equation}
This process can be further iterated to obtain the factorized form

\begin{equation}
 \bra{v_b}\Lo{n}{\varphi^k}{x}\ket{v_a}=\bra{v_b}\sum_{\atop{n_1,\ldots,n_k=0}{n_1+\ldots+n_k=n}}^n \Lo{n_1}{\varphi}{x}\cdots \Lo{n_k}{\varphi}{x}\ket{v_a}
\label{model:higher orders CT factorization}
\end{equation}
Hence, we can reduce the problem to finding an expression for

\begin{equation}
 \bra{v_b}\Lo{n}{\varphi}{x}\ket{v_a}=\Co{n}{b}{\varphi\, a}(x)
\end{equation}
with $g(a,b)=4n+1$. Recall that we may use the field equation in order to determine this coefficient from

\begin{equation}
 \Co{n}{b}{\varphi\, a}(x)=\Delta^{-1} \Co{n-1}{b}{\varphi^5\, a}(x)=\sum_{n_1,\ldots,n_5}^{n_1+\ldots+n_5=n-1}\Delta^{-1} \bra{v_b}\Lo{n_1}{\varphi}{x}\cdots\Lo{n_5}{\varphi}{x}\ket{v_a}\quad .
\label{model:higher orders CT DGL}
\end{equation}
In the second step we again used the factorization property \eqref{model:higher orders CT factorization}. With the help of this relation we can establish an iteration: We start at $n=1$ with the formula

\begin{equation}
 \Co{1}{b}{\varphi\, a}(x)=\Delta^{-1} \Co{0}{b}{\varphi^5\, a}(x)=\Delta^{-1} \bra{v_b}\Lo{0}{\varphi}{x}\cdots\Lo{0}{\varphi}{x}\ket{v_a}
\end{equation}
with $g(a,b)=5$, which is familiar from section \ref{subsubsec:L(1)phi}. There we have found the result

\begin{equation}
 \Co{1}{b}{\varphi\, a}(x)=f^b_a[\mathfrak{A}]\sum_J \Delta_1[\mathfrak{A},r]_J\, S_J(\hat{x})\, r^{\mathbf{d}}
\label{model:higher orders CT iteration start old}
\end{equation}
with $\mathfrak{A}=\Lbag \mathfrak{l}^q_1,\ldots,\mathfrak{l}^q_5 \Rbag$ and $b=a+\sum_{i=1}^5e_{\mathfrak{l}^q_i}$ (recall that for $g(a,b)=n$ the set $\mathcal{I}^b_a(n)$ consists of only one element). This is in accordance with eq.\eqref{model:higher orders CT coeff}.

Now suppose eq.\eqref{model:higher orders CT coeff} holds for all OPE coefficients up to $\Co{n-1}{b}{\varphi\, a}$. Then the right side of eq.\eqref{model:higher orders CT DGL} can be written as

\begin{equation}
\begin{split}
 \Delta^{-1} \Co{n-1}{b}{\varphi^5\, a}(x)= &\Delta^{-1}\Big( f^b_a[\mathfrak{A}] r^{d_{\mathfrak{A}}-2}\sum_{n_1,\ldots,n_5}^{n_1+\ldots+n_5=n-1}\, \sum_{\mathcal{P}_{(4n_1+1,\ldots,4n_5+1)}\left[\mathfrak{A}\right]}\sum_{J}\sum_{J_1,\ldots,J_5} \sum_{p_1=0}^{n_1}\cdots\sum_{p_5=0}^{n_5}
  \\ &\quad\times (\log r)^{p_1+\ldots+p_5}\Do{n_1}{P_1}{J_1}{p_1}\cdots \Do{n_5}{P_5}{J_5}{p_5}\, S_{J}(\hat{x})T[J_{1},\ldots,J_{5}]_{J}  \Big)\\
=& f^b_a[\mathfrak{A}] r^{d_{\mathfrak{A}}}\sum_{n_1,\ldots,n_5}^{n_1+\ldots+n_5=n-1}\, \sum_{\mathcal{P}_{(4n_1+1,\ldots,4n_5+1)}\left[\mathfrak{A}\right]}\sum_{J} \sum_{J_1,\ldots,J_5} \sum_{p_1=0}^{n_1}\cdots\sum_{p_5=0}^{n_5} (\log r)^{p_1+\ldots+p_5}
  \\ &\quad\times D^{(p_1+\ldots+p_5)}[d_{\mathfrak{A}},J,r]\, \Do{n_1}{P_1}{J_1}{p_1}\cdots \Do{n_5}{P_5}{J_5}{p_5}\, S_{J}(\hat{x})T[J_{1},\ldots,J_{5}]_{J}\\
=&f^b_a[\mathfrak{A}] r^{d_{\mathfrak{A}}}\sum_{J} S_{J}(\hat{x}) \Delta_n[\mathfrak{A},r]_J= \Co{n}{b}{\varphi\, a}(x)
\end{split}
\end{equation}
where $g(a,b)=4n+1$ and $\mathfrak{A}=\Lbag\mathfrak{l}^q_1,\ldots,\mathfrak{l}^q_{4n+1}\Rbag$ with $b=\sum_{i=1}^{4n+1}e_{\mathfrak{l}^q_i}$. In the second step we used the definition of $D^{(n)}$, eq.\eqref{eq:app diff eq D}, in order to solve the differential equation, and in the last line we used the definition of $\Delta_n$, see eq.\eqref{model:higher orders definition Dn}. Therefore, eq.\eqref{model:higher orders CT coeff} holds for all coefficients of the form $\Co{n}{b}{\varphi\, a}$ with $g(a,b)=4n+1$, and hence for all coefficients $\Co{n}{b}{\varphi^k\, a}$ with $g(a,b)=4n+k$ due to the factorization property, eq.\eqref{model:higher orders CT factorization}.

\end{proof}

As mentioned above, the construction of OPE coefficients $\Co{n}{b}{\varphi^k\, a}$ becomes increasingly difficult for decreasing values of $g(a,b)$, so it will be considerably more complicated to extend the above result so smaller values of $g(a,b)$. Thus, instead of trying to determine the concrete form of these coefficients, we will spend the rest of this section discussing some general properties of arbitrary OPE coefficients, which follow from the patterns observed in our low order computations.

\subsection*{Powers of logarithms}

Here we want to prove the familiar claim

\begin{proposition} {\ \\}
 At $n$-th order in perturbation theory, OPE coefficients $\Co{n}{c}{ab}(x)$ contain at most the $n$-th power of $\log r$.
\end{proposition}

\begin{proof}
 We prove this statement iteratively. At zeroth-order it is obviously true, as can be seen from our explicit construction of the general left representative $\Lo{0}{\ket{v_a}}{x}$ of the free theory. Matrix elements of this normal ordered operator contain only finite sums of polynomial terms, and hence no logarithms. Now suppose the proposition is true at order $n$. Then according to our algorithm we proceed to order $n+1$ by inverting the Laplace operator on $\Co{n}{b}{\varphi^5\, a}$. By assumption, this coefficient contains no higher powers than $(\log r)^n$. Now according to eq.\eqref{eq:app diff eq D}, inversion of the Laplace operator on such an expression can increase the power of $\log r$ at most by one, which implies that our claim also holds for $\Co{n+1}{b}{\varphi\, a}$. The next step in our scheme is to determine $\Co{n+1}{b}{\varphi^2\, a}$ using the consistency condition

\begin{equation}
 \Co{n+1}{b}{\varphi^2\, a}(x)=\lim_{y\to x}\left[\sum_{i=0}^{n+1}\Co{i}{c}{\varphi\, a}(y)\Co{n+1-i}{b}{\varphi\, c}(x)-\text{ counterterms }\right]
\end{equation}
All expressions in this formula are known, i.e. only coefficients up to $\Co{n+1}{b}{\varphi\, a}$ appear. Thus, we know that each summand on its own fulfills our claim, so the only possible source for an additional power of the logarithm is the infinite sum over $c$. Despite our lack of knowledge of the explicit form of the coefficients in this sum, dimensional analysis\footnote{Recall that the OPE coefficients obtained with $\Delta^{-1}$ are graded by dimension, i.e. $sd\, \Co{i}{c}{ab}=|a|+|b|-|c|$.} allows us to put it into the form

\begin{equation}
 \Co{i}{c}{\varphi\, a}(y)\Co{n+1-i}{b}{\varphi\, c}(x)\propto\sum_c \left(\frac{|y|}{|x|}\right)^{|c|}\cdot \frac{|x|^{|b|-1/2}}{|y|^{|a|+1/2}} (\log |x|)^q(\log |y|)^p \qquad, p+q\leq n+1
\end{equation}
From this estimate we can see that if the infinite sum over $c$ is to produce additional powers of logarithms, the argument of this logarithm will clearly be $1-|y|/|x|$. However, in the limit $y\to x$ this expression is divergent and thus has to be cured by subtraction of an appropriate counterterm. Let us suppose the sum over $c$ diverges as $(\log 1-|y|/|x|)^r$, then the counterterm has to be proportional to $(\log |x|)^{p+q}(\log|x-y|)^r |x|^{|b|-|a|-1}$. After cancellation of the infinite parts, we are left with a finite contribution of the form $(\log |x|)^{p+q+r} |x|^{|b|-|a|-1}$. Now recall that every counterterm is a product of two OPE coefficients of order $i$ and $n+1-i$ respectively, which both satisfy our proposition. In other words, the combined power of logarithms in this product may not exceed $n+1$. Therefore, for our counterterm of the form $(\log |x|)^{p+q}(\log|x-y|)^r |x|^{|b|-|a|-1}$ we find the condition $p+q+r\leq n+1$, and it follows that also the finite result will fulfill our proposition.

This argumentation can be straightforwardly generalized to show that all coefficients $\Co{n+1}{b}{\varphi^k\, a}$, and thus also the general coefficient $\Co{n+1}{c}{a\, b}$, fulfill our proposition, which completes the iteration.
\end{proof}

\section{Comparison to customary method}\label{subsec:comparison to alternative}

In the standard approach to quantum field theory OPE coefficients are determined via certain renormalized Feynman integrals \cite{Collins1984}. In this section an exemplary computation of this type for a first order coefficient is presented. It will be shown that our method, i.e. the scheme outlined above, does indeed yield an equivalent result.

We want to determine the three-point coefficient $\Co{1}{\varphi^3}{\varphi,\varphi,\varphi}(x_1,x_2,x_3)$ again in our three dimensional toy model with $\varphi^6$ interaction. In the usual approach this means that we have to perform the integrals

\begin{equation}
\begin{split}
 	\Co{1}{\varphi^3}{\varphi,\varphi,\varphi}(x_1,x_2,x_3)&=\left[\quad	 \parbox{10mm}{\vspace{.7cm}
 \begin{fmffile}{gluon3}
	        \begin{fmfgraph*}(8,15)
\fmfpen{thick}
	            \fmftop{i1,i2,i3}
	            \fmfbottom{o1,o2,o3}
	            \fmf{plain}{i1,v1,o1}
	        \fmf{plain}{i2,v1,o2}
\fmf{plain}{i3,v1,o3}
\fmfdot{v1}
\fmfdotn{i}{3}
\fmflabel{$x_1$}{i1}
\fmflabel{$x_2$}{i2}
\fmflabel{$x_3$}{i3}
\fmflabel{$y$}{v1}
\end{fmfgraph*}
\end{fmffile}} - \parbox{10mm}{\vspace{.7cm}
 \begin{fmffile}{gluon4}
	        \begin{fmfgraph*}(8,15)
\fmfpen{thick}
	            \fmftop{i1}
	            \fmfbottom{o1,o2,o3}
	            \fmf{plain,left=.7}{i1,v1}
	        \fmf{plain}{i1,v1}
\fmf{plain,right=.7}{i1,v1}
\fmf{plain,tension=1.3}{v1,o1}
\fmf{plain,tension=1.3}{v1,o2}
\fmf{plain,tension=1.3}{v1,o3}
\fmfdot{v1}
\fmfdot{i1}
\fmflabel{$x_3$}{i1}
\fmfv{label=$y$,label.angle=0}{v1}
\end{fmfgraph*}
\end{fmffile}}\quad \right]_{\text{UV-renormalized}} \\
&=-\frac{120}{24\pi} \int_{\mathbb{R}^3} \Big(G_F(x_1,y)G_F(x_2,y)G_F(x_3,y)-G_F^3(x_3,y)\Big)_{\text{UV-ren.}}\, d^3y\\
	&=\frac{5}{\pi}  \int_{\mathbb{R}^3} \left(\frac{1}{\sqrt{(y-x_3)^2}^3}-\frac{1}{\sqrt{(y-x_1)^2(y-x_{2})^2(y-x_{3})^2}}\right)_{\text{UV-ren.}} \, d^3y\\
	&=\frac{5}{\pi}  \int_{\mathbb{R}^3} \left(\frac{1}{\sqrt{y^2}^3}-\frac{1}{\sqrt{y^2(y-x_{13})^2(y-x_{23})^2}}\right)_{\text{UV-ren.}} \, d^3y\quad ,
\end{split}
\label{model:comparison old method1}
\end{equation}
with $x_{ij}:=x_i-x_j$ and where

\begin{equation}
 G_F(x,y)=\frac{1}{|x-y|}
\label{model:comparison old method propagator}
\end{equation}
is the propagator in our theory. Here we used the Feynman rules corresponding to the Lagrangian

\begin{equation}
\mathcal{L}(\varphi,\del_{\mu}\varphi)=-\frac{1}{4\pi}\int \left(\del_\mu\varphi(y)\del^\mu\varphi(y)+\frac{\lambda}{6}\varphi^{6}(y)\right)\, d^3y\quad,
\end{equation}\index{symbols}{Lagrangian@$\mathcal{L}$}%
see eqs.\eqref{eq:model free field lagrangian} and \eqref{eq:model perturbations interaction lagrangian}. In the last step of eq.\eqref{model:comparison old method1} we simply shifted the integration variable $y\to y+x_3$. Power counting suggests that the integrals are logarithmically infrared-divergent. Therefore, we introduce a \emph{cutoff} as regularization and treat the integrals separately. In the end, as the cutoff is removed, we will obtain a finite result for eq.\eqref{model:comparison old method1}.

Let us start with the first integral in eq.\eqref{model:comparison old method1} and assume without loss of generality $r_{13}\leq r_{23}$, where $r_{ij}=|x_{ij}|$. Then we can solve the integral using the \emph{Gegenbauer polynomial technique} \cite{chetyrkin1980nae}\cite{Smirnov2004}. Let $r_y=|y|$, $d\Omega=\sin\Theta\, d\Theta\, d\phi$ and $\Lambda\in\mathbb{R}$. Then we find

\begin{equation}
\begin{split}
 &\int \frac{1}{\sqrt{y^2(y-x_{13})^2(y-x_{23})^2}} \, d^3y\, \Big |_{r_y<\Lambda} =\\ &\int\limits_0^{r_{13}}dr_y\int d\Omega \, \left(
\frac{r_y^2}{r_y\cdot r_{13}\cdot r_{23}}\cdot \frac{1}{\sqrt{1+\frac{r_y^2}{r_{13}^2}-2\frac{r_y}{r_{13}}\, \hat{y}\cdot\hat{x}_{13}}}\cdot \frac{1}{\sqrt{1+\frac{r_y^2}{r_{23}^2}-2\frac{r_y}{r_{23}}\, \hat{y}\cdot\hat{x}_{23}}}\right)\\
+&\int\limits_{r_{13}}^{r_{23}}dr_y\int d\Omega \, \left(
\frac{r_y^2}{r_y^2\cdot r_{23}}\cdot \frac{1}{\sqrt{1+\frac{r_{13}^2}{r_y^2}-2\frac{r_{13}}{r_y}\, \hat{y}\cdot\hat{x}_{13}}}\cdot \frac{1}{\sqrt{1+\frac{r_y^2}{r_{23}^2}-2\frac{r_y}{r_{23}}\, \hat{y}\cdot\hat{x}_{23}}}\right)\\
+&\int\limits_{r_{23}}^{\Lambda}dr_y\int d\Omega \, \left(
\frac{r_y^2}{r_y^3}\cdot \frac{1}{\sqrt{1+\frac{r_{13}^2}{r_y^2}-2\frac{r_{13}}{r_y}\, \hat{y}\cdot\hat{x}_{13}}}\cdot \frac{1}{\sqrt{1+\frac{r_{23}^2}{r_y^2}-2\frac{r_{23}}{r_y}\, \hat{y}\cdot\hat{x}_{23}}}\right)
\end{split}
\label{model:comparison old method cutoff}
\end{equation}
Here we split the radial integration into three parts and introduced the cutoff parameter $\Lambda$. The original integral is restored in the limit $\Lambda\to\infty$. The square root expressions under the integrals can now be recognized as generating functions of the Legendre polynomial (see eq.\eqref{eq:app spherical symm D=3 legendreP generating function}). Hence,

\begin{equation}
 \begin{split}
  &\int \frac{1}{\sqrt{y^2(y-x_{13})^2(y-x_{23})^2}} \, d^3y\, \Big |_{r_y<\Lambda} =\\ =&\int\limits_0^{r_{13}}dr_y\int d\Omega \, \left[
\frac{r_y}{r_{13}\cdot r_{23}}\cdot \sum_{n=0}^{\infty}P_n(\hat{y}\cdot\hat{x}_{13})\left(\frac{r}{r_{13}}\right)^{n}\cdot \sum_{m=0}^{\infty}P_m(\hat{y}\cdot\hat{x}_{23})\left(\frac{r}{r_{23}}\right)^{m}\right]\\
+&\int\limits_{r_{13}}^{r_{23}}dr_y\int d\Omega \, \left[
\frac{1}{\sqrt{r_{23}^2}}\cdot \sum_{n=0}^{\infty}P_n(\hat{y}\cdot\hat{x}_{13})\left(\frac{r_{13}}{r_y}\right)^{n}\cdot \sum_{m=0}^{\infty}P_m(\hat{y}\cdot\hat{x}_{23})\left(\frac{r_y}{r_{23}}\right)^{m}\right]\\
+&\int\limits_{r_{23}}^{\Lambda} dr_y\int d\Omega \, \left[
\frac{1}{r_y}\cdot \sum_{n=0}^{\infty}P_n(\hat{y}\cdot\hat{x}_{13})\left(\frac{r_{13}}{r_y}\right)^{n}\cdot \sum_{m=0}^{\infty}P_m(\hat{y}\cdot\hat{x}_{23})\left(\frac{r_{23}}{r_y}\right)^{m}\right]
 \end{split}
\label{model:comparison old method generating fct}
\end{equation}
Now the angular integration can be performed conveniently with the help of the orthogonality relation of the Legendre polynomials

\begin{equation}
 \int d\hat{y}\, P_n(\hat{y} \cdot \hat{x}_1)P_m(\hat{y}\cdot \hat{x}_2)=\delta_{n,m}\frac{1}{2n+1} P_n(\hat{x}_1\cdot \hat{x}_2)\, ,
\label{model:comparison old method orthogonality}
\end{equation}
which yields

\begin{equation}
 \begin{split}
  &\int \frac{1}{\sqrt{y^2(y-x_{13})^2(y-x_{23})^2}} \, d^3y\, \Big |_{r_y<\Lambda}=4\pi \int\limits_0^{r_{13}}dr_y \, 
\frac{r_y}{r_{13}\cdot r_{23}}\cdot \sum_{n=0}^{\infty}\frac{P_n(\hat{x}_{13}\cdot \hat{x}_{23})}{2n+1}\left(\frac{r_y^2}{r_{13}r_{23}}\right)^{n}\\
+&4\pi \int\limits_{r_{13}}^{r_{23}}dr_y \, 
\frac{1}{r_{23}}\cdot \sum_{n=0}^{\infty}\frac{P_n(\hat{x}_{13}\cdot \hat{x}_{23})}{2n+1}\left(\frac{r_{13}}{r_{23}}\right)^{n}+4\pi \int\limits_{r_{23}}^{\Lambda} dr_y \, 
\frac{1}{r_y}\cdot \sum_{n=0}^{\infty}\frac{P_n(\hat{x}_{13}\cdot \hat{x}_{23})}{2n+1}\left(\frac{r_{13}r_{23}}{r_y^2}\right)^{n}
 \end{split}
\label{model:comparison old method angular int}
\end{equation}
Finally, we are ready to perform the radial integration, which is trivial in our present form of the integral.

\begin{equation}
 \begin{split}
  \int &\frac{1}{\sqrt{y^2(y-x_{13})^2(y-x_{23})^2}} \, d^3y\, \Big |_{r_y<\Lambda} =\\
&4\pi \left[\sum_{n=0}^{\infty}P_n(\hat{x}_{13}\cdot \hat{x}_{23})\frac{1}{2n+1}\cdot\frac{1}{2n+2}\left(\frac{r_{13}}{r_{23}}\right)^{n+1}-0\right]\\
+&4\pi \left[\sum_{n=0}^{\infty}P_n(\hat{x}_{13}\cdot \hat{x}_{23})\frac{1}{2n+1}\left(\frac{r_{13}}{r_{23}}\right)^{n}-\sum_{n=0}^{\infty}P_n(\hat{x}_{13}\cdot \hat{x}_{23})\frac{1}{2n+1}\left(\frac{r_{13}}{r_{23}}\right)^{n+1}\right]\\
+&4\pi\left[\log\Lambda-\log r_{23}\right]\\
+&4\pi \left[-\sum_{n=1}^{\infty}P_n(\hat{x}_{13}\cdot \hat{x}_{23})\frac{1}{2n+1}\cdot\frac{1}{2n}\left(\frac{r_{13}r_{23}}{\Lambda^2}\right)^{n}+\sum_{n=1}^{\infty}P_n(\hat{x}_{13}\cdot \hat{x}_{23})\frac{1}{2n+1}\cdot\frac{1}{2n}\left(\frac{r_{13}}{r_{23}}\right)^{n}\right]
\end{split}
\label{model:comparison old method radial int}
\end{equation}
Now consider the other integral in eq.\eqref{model:comparison old method1}. In addition to the infrared-divergence, which we will again control using a cutoff, this integral is also ultraviolet-divergent. This divergence may be cured using \emph{differential renormalization} \cite{SmirnovZav'yalov1993, FreedmanJohnsonLatorre1992, LatorreManuelVilasis-Cardona1994}, which works as follows: We may replace the integrand using the identity

\begin{equation}
 \frac{1}{r_y^3}=-\Delta \frac{\log (\mu r_y)}{r_y}\quad ,
\label{model:comparison old method diff renorm}
\end{equation}
with some renormalization parameter $\mu\in\mathbb{C}$, which holds for $r\neq 0$. Thus, we obtain for the integral under consideration

\begin{equation}
\begin{split}
 \int\frac{1}{\sqrt{y^2}^3}\, dy^3\, \Big |_{r_y<\Lambda} &= -\int\Delta\left(\frac{\log (\mu r_y)}{r_y}\right)\, dy^3\, \Big |_{r_y<\Lambda}\\
 &= -\int d\Omega\, r_y^2\del_{r_y} \left(\frac{\log (\mu r_y)}{r_y}\right)\, \Big |_{r_y=\Lambda}=4\pi \left(\log(\mu\Lambda)-1\right)\, ,
\end{split}
\label{model:comparison old method diff renorm result}
\end{equation}
where Gauss'-theorem was applied in the second step. Subtraction of this result from eq.\eqref{model:comparison old method radial int} shows that the logarithmic divergences cancel out. Hence we may safely remove the cutoff, i.e. take the limit $\Lambda\to\infty$, and arrive at the result

\begin{equation}
\begin{split}
 \Co{1}{\varphi^3}{\varphi,\varphi,\varphi}(x_1,x_2,x_3)&=20 \left(\sum_{n=0}^{\infty}P_n(c)s^{n+1}\left(\frac{1}{2(n+1)}\right)-\sum_{n=1}^{\infty}P_n(c)s^n\cdot \frac{1}{2n}+\log r_{23}+\log\mu-1\right)\\
 &=10 \left(\sum_{n=0}^{\infty}\frac{1}{n+1}s^{n+1}\left(P_{n}(c)-P_{n+1}(c)\right)+\log (\mu^2r_{23}^2)-2\right)
\end{split}
\label{model:comparison old method OPE result}
\end{equation}
where the abbreviations

\begin{equation}
 s:=\frac{r_{13}}{r_{23}}\qquad,\qquad c:=\hat{x}_{13}\cdot \hat{x}_{23}
\label{model:comparison old method abbreviations}
\end{equation}
were used. Now let us compute the same coefficient in our framework. First, the coherence theorem, thm. \ref{thm:coherence theorem}, states that the desired three-point coefficient can be uniquely determined just from the knowledge of the two-point coefficients. This can be seen by application of the factorization axiom (assuming $r_{13}\leq r_{23}$, i.e. $s\leq 1$)

\begin{equation}
 \Co{1}{\varphi^3}{\varphi,\varphi,\varphi}(x_1,x_2,x_3)=\sum_c \Co{1}{v_c}{\varphi\varphi}(x_1,x_3)\Co{0}{\varphi^3}{\varphi\, v_c}(x_2,x_3)+\sum_c \Co{0}{v_c}{\varphi\varphi}(x_1,x_3)\Co{1}{\varphi^3}{\varphi\, v_c}(x_2,x_3)
\label{model:comparison new factorization}
\end{equation}
where the sums go over all basis elements $\ket{v_c}\in V$. The coefficients on the right side have been determined in section \ref{subsubsec:L(1)phi}, see results \ref{res:model NLO Lphi C0phik} and \ref{res:model NLO Lphi C1phi}. With the help of these results we can reduce the sums to the form

\begin{equation}
\begin{split}
 \Co{1}{\varphi^3}{\varphi,\varphi,\varphi}(x_1,x_2,x_3)=&\sum_{n=0}^{\infty} \Co{1}{(\del^n\varphi)\varphi^3}{\varphi\varphi}(x_1,x_3)\Co{0}{\varphi^3}{\varphi,(\del^n\varphi)\varphi^3}(x_2,x_3)\\+&\sum_{n=0}^{\infty} \Co{0}{(\del^n\varphi)\varphi}{\varphi\varphi}(x_1,x_3)\Co{1}{\varphi^3}{\varphi,(\del^n\varphi)\varphi}(x_2,x_3)
\end{split}
\label{model:comparison new factorization reduced}
\end{equation}
since the coefficients vanish for all other (linearly independent) choices of $v_c$. Explicitly, we find for the coefficients on the right side (using the mentioned results from section \ref{subsubsec:L(1)phi} or appendix \ref{app:OPE table})

\begin{eqnarray}
 \Co{0}{\varphi^3}{\varphi,(\del^n\varphi)\varphi^3}(x_1,x_2)=&r_{12}^{-n-1}\sum_m S_{nm}(\hat{x}_{12})\cdot (1+3\delta_{n,0})\\
 \Co{0}{(\del^n\varphi)\varphi}{\varphi\varphi}(x_1,x_2)=&r_{12}^{n} \sum_m S^{nm}(\hat{x}_{12})\\
 \Co{1}{(\del^n\varphi)\varphi^3}{\varphi\varphi}(x_1,x_2)=&\frac{r_{12}^{n+1}}{n+1}\sum_m S^{nm}(\hat{x}_{12})\cdot (10-\frac{15}{2}\delta_{n,0})\\
 \Co{1}{\varphi^3}{\varphi,(\del^n\varphi)\varphi}(x_1,x_1)=& \begin{cases}
\frac{-10}{n}r_{12}^{-n}\sum_{m}S_{nm}(\hat{x}_{12}) \quad &\text{for }n>0\\20\log r_{12} &\text{for }n=0
                                                              \end{cases}
\label{model:comparison new 2pt coeffs}
\end{eqnarray}
Thus we have by substitution into eq.\eqref{model:comparison new factorization reduced}

\begin{equation}
 \begin{split}
  \Co{1}{\varphi^3}{\varphi,\varphi,\varphi}(x_1,x_2,x_3)&=\sum_{n=0}^{\infty}\sum_{m=-n}^n \frac{10}{n+1} \left(\frac{r_{13}}{r_{23}}\right)^{n+1} S^{nm}(\hat{x}_{13})S_{nm}(\hat{x}_{23})\\ &-\sum_{n=1}^{\infty}\sum_{m=-n}^n \frac{10}{n} \left(\frac{r_{13}}{r_{23}}\right)^n S^{nm}(\hat{x}_{13})S_{nm}(\hat{x}_{23})+20\log r_{23}
\end{split}
\end{equation}
Finally, using the addition theorem of the spherical harmonics and the abbreviations $s$ and $c$ as introduced above, we arrive at the formula

\begin{equation}
 \Co{1}{\varphi^3}{\varphi,\varphi,\varphi}(x_1,x_2,x_3)=10\sum_{n=0}^{\infty}\frac{s^{n+1}}{n+1}\Big(P_n(c)-P_{n+1}(c)\Big)+10\log r_{23}^2
\end{equation}
Comparing this result to eq.\eqref{model:comparison old method OPE result} we find that the difference can be absorbed in the arbitrary choice of renormalization parameter\footnote{Recall from sec.\ref{subsec:ambiguities} that we may also introduce an analog to this parameter in our approach by a redefinition of $\Delta^{-1}$.} $\mu$. Therefore, the results obtained from the two methods are indeed equivalent.

\chapter{Conclusions and outlook}\label{sec:conclusions}

This thesis contains the first concrete results on the construction of an interacting, perturbative quantum field theory in the framework of \cite{Hollands2008} (see also chaper \ref{sec:OFT in terms of consistency}). Adopting a Fock-space representation and diagrammatic notation from Hollands and Olbermann \cite{Hollandsa} for the free theory, we have explicitly constructed all OPE coefficients of the form $\Co{1}{b}{\varphi\, a}$ and $\Co{1}{b}{\varphi^2\, a}$, as well as a large class of coefficents $\Co{1}{b}{\varphi^k\, a}$, $k>2$ and $\Co{2}{b}{\varphi\, a}$ in a model theory with $\varphi^6$-interaction on 3-dimensional Euclidean space. 

These results were obtained with the help of an iterative algorithm first proposed in \cite{Hollands2008} (see also \ref{subsec:perturbation via field equation}
), which contains an inherent analog of renormalization via subtraction of counterterms. It was found by Hollands and Olbermann that this procedure could be neatly replaced in the free theory by a normal ordering prescription on the ladder operators in the mentioned Fock-space. As one might expect, however, the process of renormalization in interacting theories is considerably more complicated, in analog to usual formulations of perturbative quantum field theory. The constructions of section \ref{sec:the model} constitute the first explicit example of this non-trivial process. 

We have found that one can again incorporate renormalization by bringing all ladder operators into normal order. However, in the interacting case there is a finite difference between this procedure and the subtraction of counterterms arising from the general algorithm (see e.g. eq.\eqref{eq:model perturbations limit k-1}). In order to compensate for this difference, we have to add additional \emph{remainder operators}, which in a sense contain all the non-trivial information on the remormalization procedure, and are the main computational obstacle to proceed the iterative scheme. In order to find explicit expressions for these remainder terms, one has to perform divergent multiple series, subtract divergent counterterms and extract the finite difference. In result \ref{res:model NLO Lphi2} we have defined the first operator of this kind, and in result \ref{res:model NLO Rphi3} the second one is partially given. The computational machinery applied to obtain these results consists of identities of special functions and of hypergeometric series. Most notably, the results of \cite{Din:1981ey} and \cite{szmytkowski2006} have been of great help in the analysis of the mentioned infinite sums. At the heart of these identities lies \emph{Dougall's formula}, see eq. \eqref{eq:app character sum dougall}, which may be generalized to arbitrary dimension and should thus be of importance in more general models \cite{Hollandsa}. These findings constitute the first insights into the structure of the infinite sums appearing in the framework.

Apart from these specific results on the OPE coefficients at first and second perturbation order, we have observed some general structures in the construction, which also apply to higher orders. These results are presented in section \ref{subsec:higher order}. We were able to give a result for a particularly simple class of OPE coefficients to arbitrary orders by extrapolating the knowledge gathered in first order computations. In addition, a general statement about the powers of logarithms appearing in arbitrary OPE coefficients was made. Finally, in section \ref{subsec:comparison to alternative} we have shown in an example that our approach does indeed yield the same results as standard ones.\\

Future research will be aimed at a better conceptual understanding of the underlying mathematics of the framework, e.g. the interpretation of the \emph{left representatives} as vertex operators \cite{Hollandsa}, but also at a better understanding of the explicit computational obstacles. In particular, one would like to find ways to treat the infinite sums inherent in the framework in a general way. The results of this thesis may provide a first step into this direction, and future work could extend the results to higher orders and also to more general models in terms of spacetime dimension and type of coupling. Furthermore, a generalization of the framework to arbitrary (globally hyperbolic) background manifolds would be of interest, since the OPE is expected to play a fundamental role in the formulation of quantum field theory on curved spacetimes \cite{HollandsWald}. It may also be fruitful to study the relation to standard renormalization theory more deeply, and possibly to incorporate renormalization group techniques.

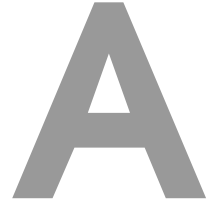
\noappendicestocpagenum
\begin{appendices}
\updatechaptername
\numberwithin{equation}{chapter}

\chapter{Table of OPE coefficients}\label{app:OPE table}

In the following table explicit results for OPE coefficients of the form $\Co{n}{c}{ab}$ are summarized.

\begin{longtable}{|c|c|c|c||c|}
\hline
 $n$ & $a$ & $b$ & $c$ & $\Co{n}{c}{ab}(x)$\\ \hline \hline \endhead
$0$ & $\varphi^p$ & $\varphi^q$ & $\varphi^{p+q}$ & $1$\\ \hline
$0$ & $\varphi^{p+q}$ & $\varphi^q$ & $\varphi^p$ & $\frac{(p+q)!}{p!}r^{-q}$\\ \hline
$0$ & $\varphi^q$  & $\varphi^{p+q}$ & $\varphi^p$ & $\frac{(p+q)!}{p!}r^{-q}$\\ \hline
$0$ & $\varphi$  & $\mathds{1}$ & $\partial^l\varphi$ & $r^l\sum_m S_{lm}(\hat{x})$\\ \hline
$0$ & $\varphi$  & $\partial^l\varphi$ & $\mathds{1}$ & $r^{-l-1}\sum_m S^{lm}(\hat{x})$\\ \hline
$0$ & $\varphi$  & $\ket{v_a}$ & $\ket{v_a+e_{lm}}$ & $(a_{lm}+1)r^l\sum_m S_{lm}(\hat{x})$\\ \hline
$0$ & $\varphi$  & $\ket{v_a}$ & $\ket{v_a-e_{lm}}$ & $(a_{lm})r^{-l-1}\sum_m S^{lm}(\hat{x})$\\
\hline
$0$ & $\varphi^k$  & $\ket{v_a}$ & $\ket{v_a+\sum\limits_{i=1}^ke_{\mathfrak{l}_i^q}}$ & $s[\mathfrak{A}=\Lbag\mathfrak{l}^q_1,\ldots,\mathfrak{l}^q_k\Rbag]f_a[\mathfrak{A}] r^{\left(\sum\limits_{i=q+1}^kl_i-\sum\limits_{j=1}^ql_j-q\right)} S_{JM}(\hat{x})T[\mathfrak{A}]_{JM}$\\ \hline
$1$ & $\varphi$ & $\varphi^p$ & $\varphi^{p+5}$ & $\frac{r^2}{6}$\\ \hline
$1$ & $\varphi$ & $\varphi^p$ & $\varphi^{p+3}$ & $5p\frac{r}{2}$\\ \hline
$1$ & $\varphi$ & $\varphi^p$ & $\varphi^{p+1}$ & $10p(p-1)\log r$\\ \hline
$1$ & $\varphi$ & $\varphi^{p+1}$ & $\varphi^p$ & $-10\frac{(p+1)!}{(p-2)!}\frac{\log r}{r}$\\ \hline
$1$ & $\varphi$ & $\varphi^{p+3}$ & $\varphi^p$ & $5\frac{(p+3)!}{(p-1)!}\frac{1}{2r^2}$\\ \hline
$1$ & $\varphi$ & $\varphi^{p+5}$ & $\varphi^p$ & $\frac{(p+5)!}{p!} \frac{1}{6r^3}$\\ \hline
$1$ & $\varphi$ & $\varphi^{p}$ & $\varphi^{p+4}\del^l\varphi$ & $\frac{5 r^{l+2}}{4l+6}\sum_m S^{lm}(\hat{x})$\\ \hline
$1$ & $\varphi$ & $\varphi^{p}$ & $\varphi^{p+2}\del^l\varphi$ & $p\frac{10 r^{l+1}}{l+1}\sum_m S^{lm}(\hat{x})$\\ \hline
$1$ & $\varphi$ & $\varphi^{p}$ & $\varphi^{p}\del^l\varphi$ & $p(p-1)\frac{30 r^{l}\log r}{2l+1}\sum_m S^{lm}(\hat{x})$\\ \hline
$1$ & $\varphi$ & $\varphi^{p+2}$ & $\varphi^{p}\del^l\varphi$ & $-\frac{(p+2)!}{(p-1)!}\frac{10 r^{l-1}}{l}\sum_m S_{lm}(\hat{x})$\\ \hline
$1$ & $\varphi$ & $\varphi^{p+4}$ & $\varphi^{p}\del^l\varphi$ & $-\frac{(p+4)!}{p!}\frac{5 r^{l-2}}{4l-2}\sum_m S_{lm}(\hat{x})$\\ \hline
$1$ & $\varphi^2$ & $\varphi^{p}$ & $\varphi^{p+6}$ & $\frac{1}{3}\cdot r^2$\\ \hline
$1$ & $\varphi^2$ & $\varphi^{p}$ & $\varphi^{p+4}$ & $\frac{16}{3}p\cdot r$ \\ \hline
$1$ & $\varphi^2$ & $\varphi^{p}$ & $\varphi^{p+2}$ & $20p^2\log r +5 p(p-1)$ \\ \hline
$1$ & $\varphi^2$ & $\varphi^{p+2}$ & $\varphi^{p}$ & $(p+2)(p+1)\left[5p-20(p+1)\log r\right]r^{-2}$ \\ \hline
$1$ & $\varphi^2$ & $\varphi^{p+4}$ & $\varphi^{p}$ & $\frac{16}{3}p\frac{(p+4)!}{p!}\cdot r^{-3}$ \\ \hline
$1$ & $\varphi^2$ & $\varphi^{p+6}$ & $\varphi^{p}$ & $\frac{1}{3}\frac{(p+6)!}{p!}\cdot r^{-4}$\\ \hline
$1$ & $\varphi^2$ & $\varphi^{p}$ & $\varphi^{p+5}\del^l\varphi$ & $r^{l+2}S^{lm}(\hat{x})\left(\frac{1}{3}+\frac{5}{2l+3}\right)$\\ \hline
$1$ & $\varphi^2$ & $\varphi^{p}\del^l\varphi$ & $\varphi^{p+5}$ & $r^{1-l}S_{lm}(\hat{x})\left(\frac{1}{3}-\frac{5}{2l-1}\right)$\\ \hline
$1$ & $\varphi^2$ & $\varphi^{p}$ & $\varphi^{p+3}\del^l\varphi$ & $p \cdot r^{l+1}S^{lm}(\hat{x})\left(5+\frac{20}{l+1}+\frac{5}{2l+3}\right)$\\ \hline
\multirow{2}{*}{$1$} & \multirow{2}{*}{$\varphi^2$} & \multirow{2}{*}{$\varphi^{p}$} & \multirow{2}{*}{$\varphi^{p+1}\del^l\varphi$} & $20p(p-1)r^{l} S^{lm}(\hat{x})\left(\frac{3\log r}{2l+1}+\log r+\frac{1}{l+1}\right)$\\
& & & & $+60pr^lS^{lm}(\hat{x})\sum\limits_{l'=0}^l\sum\limits_{J=l-l'}^{l+l'}\braket{ll'\, 00}{J\, 0}^2D[l-l'-2,J,r]$ \\ \hline
\multirow{2}{*}{$1$} & \multirow{2}{*}{$\varphi^2$} & \multirow{2}{*}{$\varphi^{p+1}\del^l\varphi$} & \multirow{2}{*}{$\varphi^{p}$} & $-20(p+1)p(p-1)r^{-l-2} S_{lm}(\hat{x})\left(\frac{3\log r}{2l+1}+\log r-\frac{1}{l+1}\right)$\\
& & & & $+60(p+1)pr^{-l-2}S_{lm}(\hat{x})\sum\limits_{l'=0}^{l-2}\sum\limits_{J=l-l'}^{l+l'}\braket{ll'\, 00}{J\, 0}^2D[l-l'-4,J,r]$ \\ \hline
$1$ & $\varphi^2$ & $\varphi^{p+3}\del^l\varphi$ & $\varphi^{p}$ & $p\frac{(p+3)!}{p!} \cdot r^{-l-3}S_{lm}(\hat{x})\left(5+\frac{20}{l+1}+\frac{5}{2l+3}\right)$\\ \hline
$1$ & $\varphi^2$ & $\varphi^{p+5}\del^l\varphi$ & $\varphi^{p}$ & $\frac{(p+5)!}{p!} r^{-l-4}S_{lm}(\hat{x})\left(\frac{1}{3}+\frac{5}{2l+3}\right)$\\ \hline
 $1$ & $\varphi^3$ & $\varphi^p$ & $\varphi^{p+7}$ & $\frac{r^2}{2}$\\ \hline
$1$ & $\varphi^3$ & $\varphi^p$ & $\varphi^{p+5}$ & $\frac{17}{2}p\cdot r$\\ \hline
$1$ & $\varphi^3$ & $\varphi^p$ & $\varphi^{p+3}$ & $p(p-1)\frac{31}{2}+p(p+1) 30\log r$ \\ \hline
$1$ & $\varphi^3$ & $\varphi^{p+3}$ & $\varphi^p$ &  $\left(p(p-1)\frac{31}{2}-p(p+1) 30\log r-40\log r\right) \frac{(p+3)!}{p!}r^{-3}$\\ \hline
$1$ & $\varphi^3$ & $\varphi^{p+5}$ & $\varphi^p$ & $\frac{17}{2}p\frac{(p+5)!}{p!}\cdot r^{-4}$\\ \hline
$1$ & $\varphi^3$ & $\varphi^{p+7}$ & $\varphi^p$ & $\frac{(p+7)!}{p!} \frac{r^{-5}}{2}$\\ \hline
$1$ & $\varphi^4$ & $\varphi^p$ & $\varphi^{p+8}$ & $\frac{2r^2}{3}$\\ \hline
$1$ & $\varphi^4$ & $\varphi^p$ & $\varphi^{p+6}$ & $12p\cdot r$\\ \hline
$1$ & $\varphi^4$ & $\varphi^p$ & $\varphi^{p+4}$ & $32p(p-1)+40p(p+2)\log r$\\ \hline
$1$ & $\varphi^4$ & $\varphi^{p+4}$ & $\varphi^p$ & $\Big(32p(p-1)-40p(p+2)\log r-160\log r\Big)\frac{(p+4)!}{p!}r^{-4}$\\ \hline
$1$ & $\varphi^4$ & $\varphi^{p+6}$ & $\varphi^p$ & $12p\frac{(p+6)!}{p!}\cdot r^{-5}$\\ \hline
$1$ & $\varphi^4$ & $\varphi^{p+8}$ & $\varphi^p$ & $\frac{(p+8)!}{p!}\frac{2r^{-6}}{3}$\\ \hline
$1$ & $\varphi^5$ & $\varphi^p$ & $\varphi^{p+9}$ & $\frac{5r^2}{6}$\\ \hline
$1$ & $\varphi^5$ & $\varphi^p$ & $\varphi^{p+7}$ & $\frac{95}{6}p\cdot r$\\ \hline
$1$ & $\varphi^5$ & $\varphi^p$ & $\varphi^{p+5}$ & $55p(p-1)+50p(p+3)\log r$\\ \hline
$1$ & $\varphi^5$ & $\varphi^{p+5}$& $\varphi^p$  & $\Big(55p(p-1)-50p(p+3)\log r-640\log r\Big)\frac{(p+5)!}{p!}r^{-5}$\\ \hline
$1$ & $\varphi^5$ & $\varphi^{p+7}$ & $\varphi^p$ & $\frac{95}{6}p\frac{(p+7)!}{p!}\cdot r^{-6}$\\ \hline
$1$ & $\varphi^5$ & $\varphi^{p+9}$ & $\varphi^p$ & $\frac{(p+9)!}{p!}\frac{5r^{-7}}{6}$\\ \hline
$1$ & $\varphi^k$ & $\varphi^p$ & $\varphi^{k+p+4}$ & $\frac{k}{6}\, r^{2}$\\ \hline
$1$ & $\varphi^k$ & $\varphi^p$ & $\varphi^{k+p+2}$ & $\left(\frac{5}{2}+\frac{k-1}{6}\right)\, kr$\\ \hline
$2$ & $\varphi$ & $\varphi^p$ & $\varphi^{p+9}$ & $\frac{r^4}{24}$\\ \hline
$2$ & $\varphi$ & $\varphi^p$ & $\varphi^{p+7}$ & $\frac{95}{72}p\cdot r^3$\\ \hline
$2$ & $\varphi$ & $\varphi^p$ & $\varphi^{p+5}$ & $\left(\frac{55}{6}p(p-1)+10p(p+3)\log r\right)r^2$\\ \hline
$2$ & $\varphi$ & $\varphi^{p+5}$ & $\varphi^p$  & $\Big(\frac{55}{6}p(p-1)+10p(p+3)\log r+128\log r\Big)\frac{(p+5)!}{p!}r^{-3}$\\ \hline
$2$ & $\varphi$ & $\varphi^{p+7}$ & $\varphi^p$ & $\frac{95}{72}p\frac{(p+7)!}{p!}\cdot r^{-4}$\\ \hline
$2$ & $\varphi$ & $\varphi^{p+9}$ & $\varphi^p$ & $\frac{(p+9)!}{p!}\frac{r^{-5}}{24}$\\ \hline
\end{longtable}

\chapter{Multisets}\label{app:multisets}

A \emph{multiset} is a generalization of the notion \emph{set}, where finite occurrences of any element are allowed. Equivalently, one may view a multiset as an \emph{unordered tuple}. The concept was first used by Dedekind in 1888 and has since found applications in various fields of applied mathematics \cite{Calude2001, Blizard1989}. A formal definition is given by

\begin{definition}[Multiset]\label{app:multisets definition}{\ \\}\index{symbols}{1multisets@$\mathfrak{A},\mathfrak{B},\mathfrak{C},\ldots$}
 Let $D$ be a set. A multiset over $D$ is a pair $(D,f)$, where $f:\, D\to \mathbb{N}$ is a function.
\end{definition}
Remark: Any set $A$ is a multiset $(A,\chi_A)$, where $\chi_A$ is the characteristic function. Throughout this thesis, we denote multisets by capital fraktur letters $\mathfrak{A,B,C,\ldots}$.\\

One may also define a multiset by giving the list of its elements. In order to avoid confusion, we will use $\Lbag \cdot\Rbag$ as brackets. Then for example the multiset $\mathfrak{A}=(D=\{a,b,c\},f)$ with

\begin{equation}
 f(x\in D)=\begin{cases}
            1 \quad & \text{for }x=b \text{ or }x=c\\
2 & \text{for }x=a
           \end{cases}
\end{equation}
may equivalently be written as

\begin{equation}
 \mathfrak{A}=\Lbag a,b,c,a\Rbag=\Lbag a,a,b,c\Rbag=\Lbag a,b,a,c\Rbag=\ldots
\label{eq:app multiset example}
\end{equation}\index{symbols}{0multisets@$\Lbag\cdot\Rbag$}
Many properties of sets can be naturally generalized to multisets. Here we only need the notion of cardinality.

\begin{definition}[Cardinality of a multiset]\label{app:cardinality definition}{\ \\}
 Let $\mathfrak{A}=(A,f)$ be a multiset; its cardinality, denoted by $\operatorname{card}(\mathfrak{A})$, is defined as

\begin{equation}
 \operatorname{card}(\mathfrak{A})=\sum_{a\in A} f(a)
\label{eq:app multiset cardinality}
\end{equation}

\end{definition}
As an example, the cardinality of the multiset given in eq.\eqref{eq:app multiset example} is 4. Further, one can define the following operations between multisets:

\begin{definition}[Sum of multisets]\label{app:sum definition}{\ \\}
 Suppose that $\mathfrak{A}=(A,f)$ and $\mathfrak{B}=(A,g)$ are multisets. Their sum, written as $\mathfrak{A}\uplus\mathfrak{B}$, is the multiset $\mathfrak{C}=(A,h)$, where for all $a\in A$:
\begin{equation}
 h(a)=f(a)+g(a)
\label{eq:app multiset sum}
\end{equation}

\end{definition}
One can show that this sum operation has the following properties:

\begin{enumerate}
 \item Commutativity: $\mathfrak{A}\uplus\mathfrak{B}=\mathfrak{B}\uplus\mathfrak{A}$
 \item Associativity: $\left(\mathfrak{A}\uplus\mathfrak{B}\right)\uplus\mathfrak{C}=\mathfrak{A}\uplus\left(\mathfrak{B}\uplus\mathfrak{C}\right)$
 \item There exists a multiset, called \emph{null multiset} $\emptyset$, such that $\mathfrak{A}\uplus\emptyset=\mathfrak{A}$
\end{enumerate}
However, there is no inverse multiset, so this structure is not an Abelian group.

\begin{definition}[Union of multisets]\label{app:union definition}{\ \\}
 Let $\mathfrak{A}=(A,f)$ and $\mathfrak{B}=(A,g)$ be multisets. Their union, denoted $\mathfrak{A}\cup\mathfrak{B}$, is the multiset $\mathfrak{C}=(A,h)$, where for all $a\in A$:

\begin{equation}
 h(a)=\max\Big(f(a),g(a)\Big)
\label{eq:app multiset union}
\end{equation}

\end{definition}

\begin{definition}[Intersection of multisets]\label{app:intersection definition}{\ \\}
 Let $\mathfrak{A}=(A,f)$ and $\mathfrak{B}=(A,g)$ be multisets. Their intersection, denoted $\mathfrak{A}\cap\mathfrak{B}$, is the multiset $\mathfrak{C}=(A,h)$, where for all $a\in A$:

\begin{equation}
 h(a)=\min\Big(f(a),g(a)\Big)
\label{eq:app multiset intersection}
\end{equation}

\end{definition}
As in the case of ordinary sets, the notions of union and intersection are commutative, associative, idempotent and distributive. One can illustrate these operations on multisets with the following example. Let $\mathfrak{A}$ be the multiset given in \eqref{eq:app multiset example} and further let

\begin{equation}
 \mathfrak{B}=\Lbag a,c,d \Rbag\quad .
\end{equation}
Then the definitions above imply

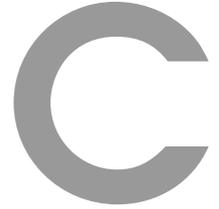
\begin{align}
 \mathfrak{A}\uplus\mathfrak{B}&=\Lbag a,a,a,b,c,c,d \Rbag\\
 \mathfrak{A}\cup\mathfrak{B}&=\Lbag a,a,b,c,d \Rbag\\
 \mathfrak{A}\cap\mathfrak{B}&=\Lbag a,c \Rbag
\end{align}

\chapter{Hypergeometric series}\label{app:hypergeoemtric functions}

Due to their appearance in calculations in all fields of physics, hypergeometric series have received increasing attention over the last decades. In this chapter, after a brief introduction into the notation, definitions and basic results on hypergeometric series, we state the identities used in the computations of this thesis. For proofs and additional results on the topic we refer the reader to the literature \cite{Erdelyi1953, Buehring2001}.

The series

\begin{equation}
 _2F_1(a,b;c;z):=\sum_{n=0}^{\infty}\frac{(a)_n(b)_n}{(c)_n}\, \frac{z^n}{n!}\qquad ; a,b,c,z\in\mathbb{C}
\label{eq:app hypergeos gaussian definition}
\end{equation}
is called \emph{hypergeometric series} or also \emph{Gaussian hypergeometric series}, as it was introduced into analysis by Gauss in 1812. Here the \emph{Pochhammer symbol}

\begin{equation}
 (\lambda)_n:=\frac{\Gamma(\lambda+n)}{\Gamma(\lambda)}\quad ,
\label{eq:app hypergeos pochhammer symbol}
\end{equation}
where $\Gamma(z)$ is the Gamma function

\begin{equation}
 \Gamma(z)=\left\{\begin{array}{ll}
 \int_0^{\infty} t^{z-1}e^{-t}\, \text{d}t,\quad & \operatorname{Re}(z)>0\\
 \Gamma(z+1)/z, &\operatorname{Re}(z)<0;\, z\neq -1,-2,-3,\ldots
\end{array}
\right.\quad ,
\label{eq:app hypergeos gamma fct}
\end{equation}\index{symbols}{Gamma@$\Gamma(x)$}
was employed for convenience. We call $a,b$ and $c$ the parameters of the hypergeometric series and $z$ its argument. A natural generalization of eq.\eqref{eq:app hypergeos gaussian definition} is

\begin{equation}
 _pF_q(\alpha_1,\ldots,\alpha_p;\beta_1,\ldots,\beta_q;z)=\pFq{p}{q}{\alpha_1,\ldots,\alpha_p}{\beta_1,\ldots,\beta_q}{z}:= \sum_{n=0}^{\infty}\frac{(\alpha_1)_n\ldots(\alpha_p)_n}{(\beta_1)_n\ldots(\beta_q)_n}\, \frac{z^n}{n!}\qquad  ,
\label{eq:app hypergeos gernealized definition}
\end{equation}\index{symbols}{Fpq@$\pFq{p}{q}{\alpha_1,\ldots,\alpha_p}{\beta_1,\ldots,\beta_q}{z}$}
with $\alpha_i,\beta_j\in\mathbb{C}\forall i\in\{1,\ldots p\},j\in\{1,\ldots q\}$ and $p,q\in\mathbb{N}$, which is known as the \emph{generalized hypergeometric series}.

\section*{Convergence}

The generalized hypergeometric function, eq.\eqref{eq:app hypergeos gernealized definition}, converges for $|z|<\infty$ if $p\leq q$ and for $|z|<1$ if $p=q+1$. It diverges for all $z\neq 0$ if $p>q+1$. Furthermore, if we set

\begin{equation}
 \omega=\sum_{j=1}^q\beta_j-\sum_{j=1}^p\alpha_j\quad ,
\label{eq:app hypergeos omega balanced}
\end{equation}
then the series $_pF_{p+1}$ with $|z|=1$ is absolutely convergent if $\operatorname{Re}(\omega)>0$ and divergent if $\operatorname{Re}(\omega)\leq -1$. If $\omega\in\mathbb{Z}$, then it is customary to refer to the corresponding class of hypergeometric series as \emph{$\omega$-balanced hypergeometric series}. A one-balanced series is also called \emph{Saalschützian} series.

It is often necessary in the calculations of chapter \ref{subsec:low orders} to analyze the divergent behavior of zero-balanced hypergeometric series in the limit $z\to 1$. Therefore, the formula

\begin{equation}
 \frac{\Gamma(\alpha_1)\cdots\Gamma(\alpha_{p+1})}{\Gamma(\beta_1)\cdots\Gamma(\beta_p)}\, \pFq{p+1}{p}{\alpha_1,\ldots,\alpha_{p+1}}{\beta_1,\ldots,\beta_p}{z}=L(p)[1+O(1-z)]-\log(1-z)[1+O(1-z)]
\label{eq:app hypergeos zero balanced limit}
\end{equation}
will be of great value. Here we defined (see \cite{Buehring2001})

\begin{equation}
 L(p):=2\psi(1)-\psi(\alpha_1)-\psi(\alpha_2)+B(p)
\label{eq:app hypergeos L(p)}
\end{equation}\index{symbols}{Lp@$L(p)$}
and

\begin{equation}
\begin{split}
 &B(p):=\frac{\beta_1-\alpha_3}{\alpha_1\alpha_2}\left(\sum_{j=2}^p\beta_j-\sum_{j=3}^{p+1}a_j\right)\, \pFq{4}{3}{\beta_1-\alpha_3+1,\sum\limits_{j=2}^p\beta_j-\sum\limits_{j=3}^{p+1}a_j+1,1,1}{\alpha_1+1,\alpha_2+1,2}{1}\\
+&\sum_{k=3}^p\frac{(\beta_{k-1}-\alpha_{k+1})\left(\sum\limits_{j=k}^p\beta_j-\sum\limits_{j=k+1}^{p+1}a_j\right)\Gamma(\alpha_1)\cdots\Gamma(\alpha_k)}{\Gamma(\beta_1)\cdots\Gamma(\beta_{k-2})\Gamma\left(\sum\limits_{j=k-1}^p\beta_j-\sum\limits_{j=k+1}^{p+1}a_j+1\right)}\times\\
 &\sum_{l=0}^{\infty}\frac{(\beta_{k-1}-\alpha_{k+1}+1)_l\left(\sum\limits_{j=k}^p\beta_j-\sum\limits_{j=k+1}^{p+1}a_j+1\right)_l}{(l+1)!\left(\sum\limits_{j=k-1}^p\beta_j-\sum\limits_{j=k+1}^{p+1}a_j+1\right)_l}\, \pFq{k}{k-1}{\alpha_1,\ldots,\alpha_k}{\beta_1,\ldots,\beta_{k-2},\sum\limits_{j=k-1}^p\beta_j-\sum\limits_{j=k+1}^{p+1}\alpha_j+1+l}{1}
\end{split}
\label{eq:app hypergeos B(p)}
\end{equation}

\numberwithin{equation}{section}

\chapter{Spherical symmetries in Euclidean spaces}\label{app:spherical symmetries}

As probably the most intuitive kind of symmetry that is frequently present in various problems in physics, rotational symmetries have naturally been studied extensively, as documented in the vast amount of literature on the topic. By axiom \ref{ax:Euclidean invariance}, symmetries of this kind will also appear in our framework, and should thus be expected to play a prominent role in simplifications of explicit calculations. In the following two subsections, we want to recall some basic results from the analysis of such symmetries, first for the general case of arbitrary dimensional space and then for the special 3-dimensional case of relevance for the calculations of this thesis.

The first subsection will mainly be concerned with some special functions related to spherical symmetries, most notably the \emph{spherical harmonics} in $D$ dimensions \cite{Muller1998,Muller1966}. After stating the definitions and basic properties of these functions, we will derive some results needed in section \ref{subsec:free field}, in particular concerning the relationship between a basis of spherical harmonics and traceless, totally symmetric tensors. Then additional properties of the special $D=3$ dimensional case are recalled \cite{Brink1962, Varshalovich1988, Devanathan1999}. Emphasis will be shifted to group theoretical results, which were first introduced into physics in the context of quantum mechanics.

\section{Euclidean spaces of arbitrary dimension D}

In a setting where spherical symmetries are present, it is often desirable to express formulas in terms of functions, which are \emph{invariant} under these symmetries. Let $C(S^{D-1})$ be the (pre-) Hilbert space of continuous functions $S^{D-1}\to\mathbb{C}$, where $S^{D-1}$ is the $(D-1)$-dimensional unit sphere, with scalar product

\begin{equation}
 \langle f, g \rangle_{(D)}:=\int_{S^{D-1}}\, f\bar{g}\, \text{d}S^{D-1}\quad .
\end{equation}
Then the following three definitions describe convenient properties we would like these functions to employ.

\begin{definition}[Invariant spaces]
 {\ \\} Let $O(D)$ be the orthogonal group of degree $D$, i.e. the group of all real $D\times D$ matrices $A$ with $A^T=A^{-1}$. A linear space $\mathcal{I}\in C(S^{D-1})$ is called \textbf{invariant} or \textbf{stable} if for all $f\in\mathcal{I}$ and all $A\in O(D)$ we have $f(Ax)\in\mathcal{I}$.
\end{definition}

\begin{definition}[Irreducibility]
{\ \\} An invariant space $\mathcal{I}$ is called \textbf{reducible} if it can be split into two nontrivial invariant subspaces $\mathcal{I}_1$, $\mathcal{I}_2$ with $\mathcal{I}_1 \perp\mathcal{I}_2$. Otherwise, the space is called \textbf{irreducible}.
\end{definition}

\begin{definition}[Primitive spaces]
 {\ \\} A space $\mathcal{I}$ is called \textbf{primitive} if it is invariant and irreducible.
\end{definition}

In the following we turn to the explicit construction of a primitive space in arbitrary dimension $D$, namely the space of spherical harmonics. The latter are most straightforwardly introduced with the help of the linear space $\mathcal{H}_n(D)$ of \emph{homogeneous polynomials of degree $n$} in $D$-dimensions, which consists of elements of the form

\begin{equation}
 H_n(x_1,\ldots,x_D)=\sum_{|i_{(D)}|=n} a_{i_1,\ldots,i_D} x_1^{i_1}\cdots x_D^{i_D}
\label{eq:app spherical symm arbitrary homogeoneous polynomial}
\end{equation}
with $a_{i_1,\ldots,i_D}\in\mathbb{C}$, $i_j\in\{1,\ldots,n\}$ and $|i_{(D)}|=i_1+i_2+\cdots+i_D$. Furthermore, a homogeneous polynomial of degree $n\geq 2$ that satisfies $\square H_n=0$ is called \emph{homogeneous- or solid harmonic}. Now we are ready for the central definition of this section:

\begin{definition}[Spherical harmonics in $D$ dimensions]
\index{symbols}{Ynmq@$Y_{nm}(D,\hat{x})$}
 {\ \\} The restriction $Y_n(D;\hat{x})$ of the homogeneous harmonic $H_n(x)$ to $S^{D-1}$, i.e.

\begin{equation}
 H_n(r\hat{x})=r^nH_n(\hat{x})=:r^nY_n(D;\hat{x})\, ,
\label{eq:app spherical symm arbitrary spherical harmonic def}
\end{equation}
is called \textbf{spherical harmonic of order $n$ in $D$ dimensions}. The space of such functions is denoted by $\mathcal{Y}_n(D)$ with basis elements $Y_{nm}(D)$ labeled by an additional parameter $m\in \mathbb{Z}$.
\end{definition}
Alternatively one might give a more abstract, but also more easily extendible definition of spherical harmonics. Namely, we can define the space $\mathcal{Y}_n(D)$ as the space of eigenfunctions of the Beltrami operator $\Delta^*_{(D-1)}$ on $S^{D-1}$, defined via $\square_{(D)}=\del_r^2+\frac{D-1}{r}\del_r+\frac{1}{r^2}\Delta^*_{(D-1)}$, to the eigenvalue $-n(n+D-2)$.

In the following, we sum up some basic properties of spherical harmonics without giving any proofs, which may be looked up in the cited literature:

\begin{itemize}
 \item the space $\mathcal{Y}_n(D)$ has dimension $N(n,D)$, where
\begin{equation}
 N(n,D)=\left\{\begin{array}{clcl} &1  &\text{for }l=0 \\ 
 &\frac{(2l+D-2)(l+D-3)!}{(D-2)!l!}\: &\text{for }l>0 \end{array}\right.\, ,
\label{eq:app spherical symm arbitrary dimension N}
\end{equation}\index{symbols}{NlD@$N(l,D)$}
hence the parameter $m$ labeling the basis elements $Y_{nm}(D)$ goes from $1$ to $N(n,D)$

\item the spaces $\mathcal{Y}_n(D)$ for $n=0,1,\ldots$ are primitive with respect to $O(D)$

\item the spherical harmonics are \emph{complete} in $C(S^{D-1})$, i.e. the set of linear combinations of spherical harmonics are dense in $C(S^{D-1})$

\item from here on, we will denote by $Y_{nm}(D;\hat{x})$ an orthonormal basis of $\mathcal{Y}_n(D)$, i.e.

\begin{equation}
 \int_{S^{D-1}}d\Omega\,  Y_{nm}(D;x)\overline{Y_{n'm'}}(D;x)=\delta_{nn'}\delta_{mm'}\, ,
\label{eq:app spherical symm arbitrary spher harm orthonormal}
\end{equation}
where the bar denotes complex conjugation, and by 
\begin{equation}
S_{nm}(D;\hat{x})=\left(\frac{\sigma_D}{N(n,D)}\right)^{1/2}\, Y_{nm}(D;\hat{x})
\label{eq:app spherical symm arbitrary spher harm modified}
\end{equation}\index{symbols}{SnmD@$S_{nm}(D;\hat{x})$}
a reparametrization of this basis (which is not normalized anymore) with $\sigma_D=|S^{D-1}|$\index{symbols}{sigmaD@$\sigma_D$} being the surface area of the $D-1$-dimensional sphere

\item the following addition theorem holds for the basis elements defined above:

\begin{eqnarray}
 \sum_{m=1}^{N(n,D)} \overline{Y_{nm}}(D;\hat{x}_1)Y_{nm}(D;\hat{x}_2)&=&\frac{N(n,D)}{\sigma_D}\, P_n(D;\hat{x}_1\cdot\hat{x}_2)\label{eq:app spherical symm arbitrary addition theorem} \\
 \sum_{m=1}^{N(n,D)} \overline{S_{nm}}(D;\hat{x}_1)S_{nm}(D;\hat{x}_2)&=& P_n(D;\hat{x}_1\cdot\hat{x}_2)\, ,
\end{eqnarray}
where $P_n(D;x)$ is the Legendre polynomial in $D$ dimensions defined as
\end{itemize}

\begin{definition}[Legendre polynomial in $D$ dimensions]
\index{symbols}{PnD@$P_{n}(D;t)$}
 {\ \\} Let $L_n(D;x)$ be a homogeneous harmonic with the following two properties: 

\begin{itemize}
\item $L_n(D;x)$ is isotropically invariant with respect to the axis $(-\epsilon_D,\epsilon_D)$, i.e. $L_n(D;Ax)=L_n(D;x)$ for all $A\in$O$(D)$ satisfying $A\epsilon_D=\epsilon_D$ for a vector $\epsilon_D\in S^{D-1}$

\item $L_n(D;\epsilon_D)=1$. 
\end{itemize}
Then $L_n(D;x)$ is called the \textbf{Legendre harmonic} of degree $n$ in $D$ dimensions and the \textbf{Legendre polynomial} $P_n(D;t)$ with $t\in\mathbb{R}$ is defined via $L_n(D;\hat{x})=P_n(D;t)$, where polar coordinates $\hat{x}=t\epsilon_D+\sqrt{1-t^2}\, \hat{x}_{D-1}$ were used.
\end{definition}
Remarks:
\begin{itemize}

\item $P_n(D;t)$ is a polynomial of degree $n$ in $t$ with the properties $P_n(D;1)=1$ and $P_n(D;-t)=(-1)^n P_n(D;t)$

\item the generating function of the Legendre polynomials is

\begin{equation}
 \sum_{n=0}^{\infty} \left(\atop{n+D-3}{D-3}\right) x^n P_n(D;t)=\left(\frac{1}{1+x^2-2xt}\right)^{(D-2)/2}
\label{eq:app spherical symm arbitrary legendre pol generating function}
\end{equation}
where $D\geq 3$, $0\leq x < 1$ and $-1\leq t\leq 1$.

\item Legendre polynomials obey the orthogonality relation

\begin{equation}
 \int_{-1}^1 dx\, (1-x^2)^{(D-3)/2} P_l(D;x)P_{l'}(D;x)=\delta_{ll'}\frac{\sigma_D}{\sigma_{(D-1)}}\frac{1}{N(l,D)}
\label{eq:app spherical symm arbitrary legendre pol orthogonality}
\end{equation}

\end{itemize}
Of course, many additional properties and alternative definitions of spherical harmonics and Legendre polynomials can be found in the literature on the subject  \cite{Muller1998,Muller1966}. In the context of this thesis, however, the above general relations should be sufficient, and we now focus on some specific results that are needed in the calculations of our framework. First, in section \ref{sec:the model} we made use of the following isomorphism between spherical harmonics $Y_{nm}$ and totally symmetric, traceless, orthonormal tensors $t_{lm}$\index{symbols}{tlm@$t_{lm}$}

\begin{equation}
 (t_{lm})^{\mu_1\ldots\mu_l}=c_l\int_{S^{D-1}} d\Omega \: x^{\{\mu_1}\cdots x^{\mu_l\}} Y_{lm}(x)\: .
\label{eq:app spherical symm arbitrary isomorphism tensors-harmonics}
\end{equation}
Here we want to derive the normalization factor $c_l$. Orthonormality of the $t_{lm}$ means

\begin{equation}
 \overline{t_{lm}}t_{lm}=1\, .
\label{eq:app spherical symm arbitrary tensor orthonormal}
\end{equation}
As mentioned above, the space of spherical harmonics for given $l$ has dimension $N(l,D)$, which implies by the isomorphism, eq. \eqref{eq:app spherical symm arbitrary isomorphism tensors-harmonics}, that this is also true for the $t_{lm}$. With this information we can determine $c_l$ as follows:

\begin{equation}
 \begin{split}
  N(l,D)=&\sum_m \overline{t_{lm}}t_{lm} \\
 =&|c_l|^2 \sum_{m;\mu_1,\ldots,\mu_l}\int_{S^{D-1}}d\Omega_x d\Omega_y\, Y^{lm}(D;x)Y_{lm}(D;y)x_{\{\mu_1}\ldots x_{\mu_l\}}y^{\{\mu_1}\ldots y^{\mu_l\}}\\
 =&|c_l|^2|\sigma_D\sum_m \int_{S^{D-1}}d\Omega_x\, Y^{lm}(D;x)Y_{lm}(D;\hat{e})x_e^l\\
 =&|c_l|^2\sigma_{(D-1)}N(l,D) \int_{-1}^1 dx_e\, P_l(D;x_e)x_e^l (1-x_e^2)^{(D-3)/2}\\
 =&|c_l|^2\sigma_DN(l,D)\frac{\Gamma(D/2)\, l!}{\Gamma(l+D/2)\, 2^l}
 \end{split}
\label{eq:app spherical symm arbitrary derivation cl}
\end{equation}
The third equality holds because of rotational invariance, i.e. we simply performed the angular integral over $y$ by replacing $y\to\hat{e}$, where $\hat{e}$ is any unit vector, and multiplying with the surface area of the $D-1$-dimensional sphere. We also abbreviated $x\cdot \hat{e}$ by $x_e$. Then the addition theorem, eq.\eqref{eq:app spherical symm arbitrary addition theorem}, was used and the integral over the $D-2$-dimensional sphere was performed. Finally, by the orthogonality relation for the Legendre polynomials \eqref{eq:app spherical symm arbitrary legendre pol orthogonality} and the expansion (see \cite{Muller1998})

\begin{equation}
 P_l(D;x)=\frac{1}{N(l,D)}\frac{\Gamma(l+D/2)\, 2^l}{\Gamma(D/2)\, l!} x^l + O(x^{l-2})
\label{eq:app spherical symm arbitrary expansion legendreP}
\end{equation}
we obtain the last equality. This confirms our definition\index{symbols}{cl@$c_l$}

\begin{equation}
c_l=\left(\frac{2^l\, \Gamma(l+D/2)}{l!\,\Gamma(D/2)\sigma_D}\right)^{1/2}
\label{eq:app spherical symm arbitrary definition cl}
\end{equation}
of section \ref{subsec:free field}. In another calculation of that section we used the identity

\begin{equation}
 (-1)^{l}\del_{\mu_1}\cdots\del_{\mu_{l}} r^{2-D}=c_l^{-1}2^l\,\frac{\Gamma(l+(D-2)/2)}{\Gamma(D/2-1} (t^{lm})_{\mu_1\ldots\mu_{l}}Y_{lm}(\hat{x})r^{2-D-l}\, ,
\label{eq:app spherical symm arbitrary second identity}
\end{equation}
which we show to be valid now. First observe that

\begin{equation}
 \del_{\mu_i}=\frac{\del r}{\del x_{\mu_i}}\frac{\del}{\del r}=\frac{x_{\mu_i}}{r}\frac{\del}{\del r}\, .
\label{eq:app spherical symm arbitrary chain rule}
\end{equation}
Successive application of this relation yields

\begin{equation}
 (-1)^{l}\del_{\mu_1}\cdots\del_{\mu_{l}} r^{2-D}=(-1)^l \frac{x_{\mu_1}}{r}\cdots\frac{x_{\mu_l}}{r} (2-D)\cdot (-D)\cdots (4-2l-D) r^{2-D-l}\, .
\label{eq:app spherical symm arbitrary derivation second identity}
\end{equation}
From eq.\eqref{eq:app spherical symm arbitrary isomorphism tensors-harmonics} one can easily derive (as also noted in eq.\eqref{eq:model free field simplification 1})
\begin{equation}
 x_{\{\mu_1}\cdots x_{\mu_{l}\}}=c_l^{-1}(t^{lm})_{\mu_1\ldots\mu_l}r^l Y_{lm}(\hat{x})\, .
\label{eq:app spherical symm arbitrary simplification 1}
\end{equation}
Substitution of this formula into eq.\eqref{eq:app spherical symm arbitrary derivation second identity} and simple algebraic manipulation gives the proclaimed result.

\section{3-dimensional Euclidean space}\label{app:3-dim symmetries}

Having discussed the special functions appearing in the analysis of spherical symmetries in an Euclidean space of arbitrary dimension $D$ in the previous section, we from here on focus on the case $D=3$, which describes the space our toy model of section \ref{sec:the model} lives on. First, some of the above general results are repeated in the special $D=3$ version, before we come to additional group theoretic results familiar from the analysis of angular momenta in quantum mechanics.

Let us rewrite the central formulas of the previous section for $D=3$. First, the space $\mathcal{Y}_l(D)$ of spherical harmonics of degree $l$ has dimension $N(l,3)=2l+1$. The basis elements $Y_{lm}$ of this space are parametrized by $m\in\mathbb{Z}$, with $-l\leq m\leq l$. Further, $\mathcal{Y}_l(D)$ is an irreducible, $O(3)$-invariant, space of functions for any $l$, and the spherical harmonics are complete in $C(S^2)$. In the following, we omit the label $D$ and implicitly assume $D=3$, e.g. we write $Y_{lm}(\hat{x})$ instead of $Y_{lm}(3;\hat{x})$. The addition theorem then takes the form

\begin{eqnarray}
 \sum_{m=-l}^{l} \overline{Y_{lm}}(\hat{x}_1)Y_{lm}(\hat{x}_2)&=&\frac{2l+1}{4\pi}\, P_l(\hat{x}_1\cdot\hat{x}_2) \\
 \sum_{m=-l}^{l} \overline{S_{lm}}(\hat{x}_1)S_{lm}(\hat{x}_2)&=& P_l(\hat{x}_1\cdot\hat{x}_2)\, ,
\label{eq:app spherical symm D=3 addition theorem}
\end{eqnarray}
where

\begin{equation}
S_{lm}(\hat{x})=\left(\frac{4\pi}{2l+1}\right)^{1/2}\, Y_{lm}(\hat{x})\, .
\label{eq:app spherical symm D=3 spher harm modified}
\end{equation}
Here $P_l$ is the usual Legendre polynomial with generating function (see \cite{Erdelyi1953})

\begin{equation}
 \left(\frac{1}{1+x^2-2xt}\right)^{1/2}=\left\{\begin{array}{clcl} &\sum\limits_{l=0}^{\infty} x^l P_l(t)  &\text{for }|x|<\min|t\pm(t^2-1)^{1/2}| \\ 
 &\sum\limits_{l=0}^{\infty} x^{-l-1} P_l(t) &\text{for }|x|>\max|t\pm(t^2-1)^{1/2}| \end{array}\right.\, .
\label{eq:app spherical symm D=3 legendreP generating function}
\end{equation}

As mentioned in section \ref{subsec:perturbation via field equation}, we chose a toy model on 3-dimensional Euclidean space, because the representation theory of the corresponding symmetry group is comparably simple and familiar from the quantum mechanics of angular momentum \cite{Brink1962, Varshalovich1988, Devanathan1999}, where spherical harmonics are the eigenfunctions of the \emph{operator of orbital angular momentum}. In the remainder of this section, we will be concerned with the decomposition of products of spherical harmonics into irreducible parts. By this procedure, we can put our OPE coefficients into a convenient form that simplifies the differential equations \eqref{eq:model perturbations relations}. 

For this purpose, let us briefly recall the \emph{addition (or coupling) of angular momenta} from quantum mechanics. Given two systems with angular momentum quantum numbers $(j_1,m_1)$ and $(j_2,m_2)$ and corresponding eigenstates $\ket{j_1m_1}$ and $\ket{j_2m_2}$ of the angular momentum operators $\mathbf{J}_1^2$, $J_{1z}$ and $\mathbf{J}_2^2$, $J_{2z}$, there are different ways to express the combined system. On the one hand, one may use the direct product $\ket{j_1m_1}\otimes \ket{j_2m_2}=\ket{j_1j_2m_1m_2}$ of the constituent states, which is an eigenstate of all four operators $\mathbf{J}_1^2$, $J_{1z}$, $\mathbf{J}_2^2$ and $J_{1z}$. It is easy to see that this product state is reducible, despite the irreducibility of both $\ket{j_1m_1}$ and $\ket{j_2m_2}$. Alternatively, one may choose the system to be described by eigenstates $\ket{j_1j_2 JM}$ of the operators $\mathbf{J}_1^2$, $\mathbf{J}_2^2$, $\mathbf{J}^2=(\mathbf{J}_1+\mathbf{J}_2)^2$ and $J_{z}=J_{1z}+J_{2z}$. Contrary to the direct product case above, these states are also eigenstates of the \emph{total angular momentum operator} $\mathbf{J}^2$ and hence irreducible. It is an important fact for quantum mechanics that there exists a unitary transformation between the two mentioned sets of states describing the \emph{coupled} system, which has the form

\begin{equation}
\begin{array}{lcclc}
 \ket{j_1j_2JM}&=&\sum\limits_{m_1,m_2}&\ket{j_1j_2m_1m_2} &\braket{j_1j_2m_1m_2}{j_1j_2JM}\\
 \ket{j_1j_2m_1m_2}&=&\sum\limits_{J,M}&\ket{j_1j_2JM} &\braket{j_1j_2JM}{j_1j_2m_1m_2}
\end{array}
\label{eq:app spherical symm D=3 CG coefficients 1}
\end{equation}\index{symbols}{0clebschg@$\braket{j_1j_2m_1m_2}{j_1j_2JM}$}
where the expansion coefficients $\braket{j_1j_2m_1m_2}{j_1j_2JM}=\braket{j_1j_2JM}{j_1j_2m_1m_2}$ are called \emph{Clebsch-Gordan coefficients} (CG coefficients). For the sake of brevity, we will often use the notation $\braket{j_1j_2m_1m_2}{JM}$ instead. The \emph{Wigner 3j-symbol} defined by

\begin{equation}
 \braket{j_1j_2m_1m_2}{J-M}=(-1)^{j_1-j_2-M}(2J+1)^{1/2}\begin{pmatrix}
                                                                 j_1 & j_2 & J \\
								 m_1 & m_2 & M
                                                               \end{pmatrix}
\label{eq:app spherical symm D=3 3j symbol}
\end{equation}\index{symbols}{0clebschg2@$\begin{pmatrix}
                                                                 j_1 & j_2 & J \\
								 m_1 & m_2 & M
                                                               \end{pmatrix}$}
is also widely used because of its additional symmetry properties. Let us briefly recall some basic features of CG coefficients:

\begin{itemize}
 \item CG coefficients satisfy the orthogonality relations
\begin{eqnarray}
 \sum_{m_1m_2}\braket{j_1j_2JM}{j_1j_2m_1m_2}  \braket{j_1j_2m_1m_2}{j_1j_2KQ}&=&\delta_{JK}\delta_{MQ}\\
 \sum_{JM} \braket{j_1j_2m_1m_2}{j_1j_2JM}\braket{j_1j_2JM}{j_1j_2m_1'm_2'} &=&\delta_{m_1m_1'}\delta_{m_2m_2'}
\label{eq:app spherical symm D=3 CG orthogonality}
\end{eqnarray}

\item they vanish unless the triangle inequality

\begin{equation}
 |j_1-j_2|\leq J\leq j_1+j_2
\label{eq:app spherical symm D=3 CG triangle inequality}
\end{equation}
and the condition $m_1+m_2=M$ are fulfilled

\end{itemize}

Before we come to the desired relation transforming a product of spherical harmonics into irreducible parts, we study the \emph{rotation matrices} $D^J_{MM'}(\alpha,\beta,\gamma)$ defined by

\begin{equation}
 \psi_{JM}(x')=\sum_{M'} D^J_{MM'}(\alpha,\beta,\gamma) \psi_{JM'}(x)\,
\label{eq:app spherical symm D=3 rotation matrix}
\end{equation}
where $\psi_{JM}$ is the wavefunction of a quantum mechanical system with angular momentum quantum numbers $J$ and $M$, and $x'$ is obtained from $x$ by a rotation of the coordinate system by the Euler angles $(\alpha,\beta,\gamma)$. Thus, rotating the coordinate system on both sides of eq.\eqref{eq:app spherical symm D=3 CG coefficients 1} by $\omega=(\alpha,\beta,\gamma)$ and using orthogonality of the states, we arrive at the \emph{Clebsch-Gordan series}

\begin{equation}
 D^{j_1}_{m_1n_1}(\omega)D^{j_2}_{m_2n_2}(\omega)=\sum_{JMN}\braket{j_1j_2m_1m_2}{JM}\, D^J_{MN}(\omega)\, \braket{JN}{j_1j_2n_1n_2}
\label{eq:app spherical symm D=3 CG series 1}
\end{equation}
and equivalently
\begin{equation}
 D^J_{MN}(\omega)=\sum_{m_1,m_2,n_1,n_2}\braket{JM}{j_1j_2m_1m_2}\, D^{j_1}_{m_1n_1}(\omega)D^{j_2}_{m_2n_2}(\omega)\, \braket{j_1j_2n_1n_2}{JN}\, .
\label{eq:app spherical symm D=3 CG series 2}
\end{equation}
Now let us draw the connection to the coupling of spherical harmonics and consider the rotation

\begin{equation}
 \sum_{m}D^l_{m0}(\phi_1,\Theta_1,0)Y_{lm}(\phi_2,\Theta_2)=Y_{l0}(\Theta,0)=\left(\frac{2l+1}{4\pi}\right)^{1/2}P_l(\cos\Theta)\, ,
\label{eq:app spherical symm D=3 spherical harmonics rotation}
\end{equation}
where the second equality is a standard identity from the theory of special functions. Comparing this equation to the addition formula \eqref{eq:app spherical symm D=3 addition theorem}, we immediately obtain

\begin{equation}
 D^l_{m0}(\phi_1,\Theta_1,0)=\left(\frac{4\pi}{2l+1}\right)^{1/2}\overline{Y_{lm}}(\phi_1,\Theta_1)\, .
\label{eq:app spherical symm D=3 spherical harmonics rotation matrix}
\end{equation}
Substituting this connection between rotation matrices and spherical harmonics into eqs.\eqref{eq:app spherical symm D=3 CG series 1} and \eqref{eq:app spherical symm D=3 CG series 2} we finally arrive at the desired \emph{coupling rule} for spherical harmonics with the same argument.

\begin{equation}
\begin{array}{lcl}
 S_{l_1m_1}(\hat{x})S_{l_2m_2}(\hat{x})&=&\sum\limits_l \braket{l_1l_2m_1m_2}{lm}\braket{l0}{l_1l_200}\, S_{lm}(\hat{x})\\
 S_{lm}(\hat{x})\braket{l_1l_200}{l0} &=&\sum\limits_{m_1,m_2} \braket{lm}{l_1l_2m_1m_2}\, S_{l_1m_1}(\hat{x})S_{l_2m_2}(\hat{x})
\end{array}
\label{eq:app spherical symm D=3 spherical harmonics coupling}
\end{equation}
Here we used the unnormalized spherical harmonics for convenience. The CG coefficient with \emph{magnetic quantum numbers} equal to zero, i.e. the coefficient of the form $\braket{l_1l_200}{l0}$, is sometimes called \emph{parity coefficient} in the literature (see e.g.\cite{Devanathan1999}), because it vanishes unless the sum of its entries is an even number, i.e.

\begin{equation}
 \braket{l_1l_200}{l0}=0 \qquad \text{if }l_1+l_2+l=2n+1\text{ with }n\in\mathbb{N}\, .
\label{eq:app spherical symm D=3 CG parity}
\end{equation}
Further, this coefficient is related to the Legendre polynomials by

\begin{equation}
 \braket{l_1l_2\, 00}{J0}^2=\frac{2J+1}{2}\int^1_{-1}\text{d}y P_J(y)P_{l_1}(y)P_{l_2}(y)\, ,
\label{eq:app spherical symm D=3 CG <-> LegendreP}
\end{equation}
and explicitly takes the values

\begin{equation}
 \braket{l_1l_2\, 00}{J0}=(-1)^{\frac{l_1+l_2+3J}{2}}\left(\frac{2J+1}{2\pi}\frac{\Gamma(\frac{l_1+l_2-J+1}{2})\Gamma(\frac{l_1-l_2+J+1}{2})\Gamma(\frac{l_2-l_1+J+1}{2})\Gamma(\frac{l_1+l_2+J}{2}+1)}{\Gamma(\frac{l_1+l_2-J}{2}+1)\Gamma(\frac{l_1-l_2+J}{2}+1)\Gamma(\frac{l_2-l_1+J}{2}+1)\Gamma(\frac{l_1+l_2+J+3}{2})}\right)^{1/2}
\label{eq:app spherical symm D=3 CG <-> Gamma fct}
\end{equation}
Eqs.\eqref{eq:app spherical symm D=3 spherical harmonics coupling} are the central formulae of this section. It is obvious that by successive application of these equations, one may decompose products of any number of spherical harmonics.

\begin{equation}
\begin{split}
 S_{l_1m_1}(\hat{x})S_{l_2m_2}(\hat{x})\cdots S_{l_nm_n}(\hat{x})=\prod_{i=2}^n\sum_{l_{12\ldots i}} &\braket{l_{12\ldots (i-1)}l_im_{12\ldots (i-1)}m_i}{l_{12\ldots i}m_{12\ldots i}}\braket{l_{12\ldots i}0}{l_{12\ldots (i-1)}l_i00}\\
 &\times  S_{l_{12\ldots n}m_{12\ldots n}}(\hat{x})\\
=:\sum_{J} &T[-(l_1m_1),-(l_2m_2),\ldots, -(l_nm_n)]_{JM}\, S_{JM}(\hat{x})\, 
\end{split}
\label{eq:app spherical symm D=3 spherical harmonics coupling genral}
\end{equation}
where we defined the coupling tensor

\begin{equation}
\begin{split}
 T[-(l_1m_1),-(l_2m_2),\ldots, -(l_nm_n)]_{JM}= &\braket{l_{12\ldots (n-1)}l_nm_{12\ldots (n-1)}m_n}{JM}\\
\prod_{i=2}^{n-1}\sum_{l_{12\ldots i}}&\braket{l_{12\ldots (i-1)}l_im_{12\ldots (i-1)}m_i}{l_{12\ldots i}m_{12\ldots i}}\braket{l_{12\ldots i}0}{l_{12\ldots (i-1)}l_i00}\, .
\end{split}
\label{eq:app spherical symm D=3 spherical harmonics coupling tensor T}
\end{equation}
Products containing complex conjugate (or \emph{contragredient}) spherical harmonics may be treated analogously. We choose the so called \emph{Condon-Shortley phase convention}

\begin{equation}
 \overline{S_{lm}}(\hat{x})=(-1)^mS_{l(-m)}(\hat{x})=:S^{lm}(\hat{x})
\label{eq:app spherical symm D=3 spherical harmonics phase convention}
\end{equation}
and write

\begin{equation}
\begin{split}
 S^{l_1m_1}(\hat{x})\cdots S^{l_am_a}(\hat{x})&S_{l_{a+1}m_{a+1}}(\hat{x})\cdots S_{l_bm_b}(\hat{x})=\\
 &\sum_J T[+(l_1m_1),\ldots +(l_am_a),-(l_{a+1}m_{a+1}),\ldots. -(l_bm_b)]_{JM}\, S_{JM}(\hat{x})
\end{split}
\label{eq:app spherical symm D=3 spherical harmonics coupling completely general}
\end{equation}
with

\begin{equation}
\begin{split}
 &T[+(l_1m_1),\ldots +(l_am_a),-(l_{a+1}m_{a+1}),\ldots. -(l_bm_b)]_{JM}=\\
 &\qquad(-1)^{m_1+\ldots+m_a}T[+(l_1 (-m_1)),\ldots +(l_a (-m_a)),+(l_{a+1}m_{a+1}),\ldots. +(l_bm_b)]_{JM}
\end{split}
\label{eq:app spherical symm D=3 spherical harmonics coupling tensor T general}
\end{equation}
\index{symbols}{Tl@$T[\pm(l_1m_1),\ldots \pm(l_nm_n)]_{JM}$}

\numberwithin{equation}{chapter}

\chapter{The characteristic differential equation}\label{app:characteristic diff eq}

As we discussed in section \ref{subsec:perturbation via field equation}, our iterative scheme for the construction of OPE coefficients consists basically of two steps: Perform infinite sums of the form of eq.\eqref{eq:model perturbations limit k-1} to determine all coefficients at a given order and use the field equation, or more precisely eq.\eqref{eq:model perturbations relations}, to proceed to the next perturbation order. As mentioned in that section, most calculational effort goes into the former step, i.e. the infinite sums. In the present section we show how to perform the latter step, i.e. solving the differential equation

\begin{equation}
\square\Lo{i}{\varphi}{x}= \Lo{i-1}{\varphi^k}{x}
\label{eq:app diff eq main}
\end{equation}
for any $i,k\in\mathbb{N}$ and in arbitrary spacetime dimension $D$. We assume that any left representative takes values in the ring of functions

\begin{equation}
\mathbb{Y}(x)=\mathbb{C}\llbracket r,r^{-1},\log r \rrbracket\otimes \{Y_n(\hat{x};D)\}\otimes\End(V)\quad.
\label{eq:app diff eq ring Y}
\end{equation}\index{symbols}{y@$\mathbb{Y}(x)$}
Thus, an arbitrary element $\Lo{i}{\ket{v}}{x}\in\mathbb{Y}$ is of the form

\begin{equation}
\Lo{i}{\ket{v}}{x}\in\mathbb{Y}=\sum A_{i,d,q,J,M}(\ket{v})r^d(\log r)^q Y_{JM}(\hat{x},D)
\end{equation}
with $A_{i,d,q,J,M}\in\End(V)$. Note that this assumption is consistent with our free theory results. Now in order to find a solution to the differential equation \eqref{eq:app diff eq main} we define a right inverse to the Laplacian, i.e. an operator $\square^{-1}\in\End(\mathbb{Y})$ satisfying

\begin{equation}
\square\circ\square^{-1}=id\quad.
\end{equation}
A solution to the differential equation is then simply found by the application of this operator on $\Lo{i-1}{\varphi^k}{x}$. This can be seen from

\begin{equation}
\square\Lo{i}{\varphi}{x}= \Lo{i-1}{\varphi^k}{x}=\square\Big[\square^{-1}\Lo{i-1}{\varphi^k}{x}\Big]
\end{equation}
and therefore

\begin{equation}
\Lo{i}{\varphi}{x}=\square^{-1}\Lo{i-1}{\varphi^k}{x}
\end{equation}
is the desired solution. Our explicit choice for $\square^{-1}$ is

\begin{equation}
\square^{-1}\Big[r^d Y_{JM}(\hat{x};D)(\log r)^p \Big]=r^{d+2}Y_{JM}(\hat{x};D)(\log r)^p\cdot D^{p}(d+2,J,r)
\label{eq:app diff eq inverse}
\end{equation}\index{symbols}{Laplaceinverse@$\square^{-1}$}
with

\begin{equation}
D^{(p)}(d,J,r)=\left\{\begin{array}{clll} &\sum\limits_{i=0}^p(-1)^i\frac{p!\left(\log r\right)^{1-i}}{(p+1-i)!\left(2d+1\right)^{i+1}}  &\text{for }J=\min\{|d|,|d+D-2|\} \\
 &\sum\limits_{i=0}^{p}\sum\limits_{n=0}^i(-1)^i\frac{p!\left(\log r\right)^{-i}}{(p-i)!(d-J)^{n+1}(d+D-2+J)^{i-n+1}}\: &\text{otherwise } \end{array}\right.
\label{eq:app diff eq D}
\end{equation}\index{symbols}{D9@$D^{(p)}(d,J,r)$}
and

\begin{equation}
 \square^{-1}\, (0)=0\quad.
\end{equation}
Expressing the Laplace operator in spherically symmetric form, i.e. using the form of $\square$ given below def.\ref{def:Spherical harmonics in $D$ dimensions}, one can check by straightforward calculation that this definition does indeed yield a right inverse to the Laplacian. Note however that this choice is not unique, since we may add any $A\in\mathbb{Y}$ with $A\in\ker\square$ and would again obtain a right inverse. These functions $A$ are the harmonic polynomials in $x$ with values in $\End(V)$.\\

In the computations of section \ref{subsec:low orders} we will need the explicit form of $D^{(p)}(d,J,r)$ in $D=3$ dimensions and for $p\in\{0,1\}$, which are

\begin{equation}
 D^{(0)}(d,J,r)=:D(d,J,r)=\left\{\begin{array}{cl}
 \frac{1}{d(d+1)-J(J+1)}\quad & \text{for }\min\{|d|,|d+1|\}\neq J\\ & \\
 \frac{\log r}{2d+1} & \text{for }\min\{|d|,|d+1|\}= J
\end{array}
\right.
\label{eq:app diff eq D0}
\end{equation}\index{symbols}{D0@$D(d,J,r), D^{(0)}(d,J,r)$}
and

\begin{equation}
 D^{(1)}(d,J,r)=\left\{\begin{array}{clll} &\left(\frac{\log r}{2(2d+1)}-\frac{1}{(2d+1)^2}\right)  &\text{for }J=\min\{|d|,|d+1|\} \\
 &\left(\frac{1}{d(d+1)-J(J+1)}-\frac{(2d+1)(\log r)^{-1}}{[d(d+1)-J(J+1)]^2}\right)\: &\text{otherwise } \end{array}\right.\quad.
\label{eq:app diff eq D1}
\end{equation}\index{symbols}{D1@$D^{(1)}(d,J,r)$}

\chapter{The characteristic sum}\label{app:characteristic sum}

Due to the iterative nature of the construction described in section \ref{sec:the model}, one expects certain patterns to appear in the calculation of OPE coefficients. In this section we analyze one such expression, which characteristically shows up in first order calculations. Namely, as we saw in section \ref{subsec:low orders}, sums of the general form

\begin{equation}
 S(l_1,l_2;a):=\sum_{\atop{J=|l_1-l_2|}{J\neq a}}^{l_1+l_2} \frac{\braket{l_1 l_2\, 00}{J 0}^2}{a(a+1)-J(J+1)}
 \label{eq:app character sum definition}
\end{equation}
with $l_1,l_2,a\in \mathbb{N}$ are typically present at first perturbation order. The denominator is familiar from the solution of the differential equation relating the coefficients of the free theory to the first order coefficients, see eq.\eqref{eq:app diff eq D0} while the CG coefficient results from the coupling of angular momenta as discussed in appendix \ref{app:3-dim symmetries}. In the following we will first give some general simplifications of this sum, and afterwards distinguish different special cases of the parameters. This analysis is based on the results of \cite{Din:1981ey,Askey1982,Askey1986} and \cite{szmytkowski2006}.

First note that we may extend the summation limits arbitrarily, as the CG coefficients automatically vanish if the triangular inequality \eqref{eq:app spherical symm D=3 CG triangle inequality} is not satisfied. Thus we may write

\begin{equation}
  S(l_1,l_2;a):=\sum_{\atop{J=0}{J\neq a}}^{\infty} \frac{\braket{l_1 l_2\, 00}{J 0}^2}{a(a+1)-J(J+1)}\, .
 \label{eq:app character sum definition extended}
\end{equation}
Further, we may express the CG coefficients through an integral over Legendre polynomials by eq.\eqref{eq:app spherical symm D=3 CG <-> LegendreP}, which yields

\begin{equation}
 S(l_1,l_2;a)=\frac{1}{2}\sum_{\atop{J=0}{J\neq a}}^{\infty}\frac{2J+1}{a(a+1)-J(J+1)}\int^1_{-1}\text{d}y P_J(y)P_{l_1}(y)P_{l_2}(y)\, .
\label{eq:app character sum integral}
\end{equation}
In order to get rid of the sum over $J$ we would now like to apply \emph{Dougall's formula} (see \cite{Erdelyi1953})

\begin{equation}
 \sum_{k=0}^{\infty} \frac{2k+1}{\nu(\nu+1)-k(k+1)}P_k(y)=\frac{\pi}{\sin(\pi \nu)}P_{\nu}(-y)\hspace{2cm} (\nu\notin \mathbb{Z})\, ,
\label{eq:app character sum dougall}
\end{equation}
but as $a\in\mathbb{Z}$, this is clearly not possible at this stage. Hence, we first have to use the little trick of writing our sum as the limit

\begin{equation}
 S(l_1,l_2;a)=\frac{1}{2}\lim_{\nu\to a}\left[\int^1_{-1}\text{d}y \left(\frac{\partial}{\partial \nu}(\nu-a)\sum_{J=0}^{\infty}\frac{2J+1}{\nu(\nu+1)-J(J+1)}P_{\nu}(y)+\frac{1}{2a+1}P_a(y)\right)P_{l_1}(y)P_{l_2}(y)\right]\, .
\label{eq:app character sum limit}
\end{equation}
Here we may use eq.\eqref{eq:app character sum dougall}, and after carrying out some derivations and the limit we arrive at the convenient form

\begin{equation}
 S(l_1,l_2;a)=\frac{1}{2}\int^1_{-1}\text{d}y P_{l_1}(y)P_{l_2}(y) \left[(-1)^a\frac{\partial P_{\nu}(-y)}{\partial \nu}\Big |_{\nu=a}+\frac{1}{2a+1}P_a(y)\right]\, .
\label{eq:app character sum deriv degree}
\end{equation}
The derivative of the Legendre function with respect to its degree has been discussed in \cite{szmytkowski2006}, where the explicit form

\begin{equation}
 \pdd{P_{\nu}(x)}{\nu}\Big |_{\nu=n}=P_n(x)\log\frac{1+x}{2}+R_n(x)
\label{eq:app character sum deriv degree formula}
\end{equation}
was derived. Here $R_n$ is a polynomial defined as

\begin{equation}
 R_n(x):=2[\psi(2n+1)-\psi(n+1)]P_n(x)+2\sum_{k=0}^{n-1}(-1)^{n+k}\frac{2k+1}{(n-k)(n+k+1)}P_k(x)\, ,
\label{eq:app character sum deriv degree Rn}
\end{equation}
where $\psi$ is the digamma function \cite{Erdelyi1953, Abramowitz1965}

\begin{equation}
 \psi(n+1)=-\gamma+\sum_{k=1}^n\frac{1}{k}
\label{eq:app character sum digamma psi}
\end{equation}\index{symbols}{psi@$\psi(n)$}
with the \emph{Euler-Mascheroni constant} $\gamma$. Using this explicit form of the derivative in eq.\eqref{eq:app character sum deriv degree}, we obtain yet another expression for our sum.

\begin{equation}
\begin{split}
 S(l_1,l_2;a)=\frac{1}{2}\int_{-1}^1 P_{l_1}(y)P_{l_2}(y)&\Big(  2\sum_{k=|l_1-l_2|}^{a-1}\frac{2k+1}{(a(a+1)-k(k+1))}P_{k}(y) + \\ +& P_a(y)\left[\log \frac{1-y}{2}+\psi(2a+2)+\psi(2a+1)-2\psi(a+1)\right]\Big)\, \text{d}y
\end{split}
\label{eq:app character sum explicit general}
\end{equation}
This concludes our discussion of the general form of the sum $S$, and we will now use these results in order to further simplify the sum for special choices of the parameters $l_1,l_2$ and $a$.

\numberwithin{equation}{section}

\section{The cases \texorpdfstring{$a<|l_1-l_2|$}{a<l1-l2} and \texorpdfstring{$a>l_1+l_2$}{a>l1+l2}}\label{app:subsec:triangle}

Let us first consider the case $a<|l_1-l_2|$ in eq.\eqref{eq:app character sum explicit general}. First we note that there is no sum over $k$ in this case. Further, all expressions containing the integral $\int_{-1}^1P_{l_1}(y)P_{l_2}(y)P_{a}(y)dy$ vanish, as the triangle inequality \eqref{eq:app spherical symm D=3 CG triangle inequality} is not satisfied. Therefore, only the term containing the logarithm remains.

\begin{equation}
 S(l_1,l_2;a<|l_1-l_2|)=\frac{1}{2}\int_{-1}^1 P_{l_1}(y)P_{l_2}(y)P_a(y)\log(1-y)\, \text{d}y
\label{eq:app character sum smaller}
\end{equation}
Similarly, the integral $\int_{-1}^1P_{l_1}(y)P_{l_2}(y)P_{a}(y)dy$ also vanishes if $a>l_1+l_2$. In this case, the sum over $k$ in eq.\eqref{eq:app character sum explicit general} goes from $|l_1-l_2|$ to $l_1+l_2$. Recalling the original form of our sum $S$ from eq.\eqref{eq:app character sum definition}, we notice that the sum over $k$ is just $2\cdot S$ in the case at hand. Subtracting $2\cdot S$ on both sides of equation \eqref{eq:app character sum explicit general} and multiplying with $(-1)$, we obtain the simple form

\begin{equation}
 S(l_1,l_2;a>l_1+l_2)=-\frac{1}{2}\int_{-1}^1 P_{l_1}(y)P_{l_2}(y)P_a(y)\log(1-y)\, \text{d}y\, .
\label{eq:app character sum greater}
\end{equation}

\section{The case \texorpdfstring{$a+l_1+l_2=\text{odd}$}{a+l1+l2=odd}}\label{app:subsec: odd}

If the sum of the parameters $l_1,l_2$ and $a$ is an odd number, the sum $S$ simplifies drastically. This is also the case that has been studied most extensively in the literature (see \cite{Din:1981ey,Askey1982,Askey1986}). The simplification is most easily derived from the form \eqref{eq:app character sum deriv degree} of our sum. To begin with, note that the second summand, i.e. the term without the derivative, vanishes, due to the parity requirement of the CG coefficients \eqref{eq:app spherical symm D=3 CG parity} (alternatively, one may deduce this result from the fact that we integrate over a function of odd degree). Additionally, as we will show in the following, the derivative of the Legendre function in the remaining term may be written in a convenient form in the underlying case.

With the help of the identity (see e.g.\cite{Erdelyi1953} or \cite{Abramowitz1965})

\begin{equation}
 P_{\rho}(-x)=\cos \rho \pi P_{\rho}(x)-\frac{2}{\pi}\sin\rho\pi Q_{\rho}(x)\, ,
\label{eq:app character sum legendreF negative argument}
\end{equation}
where $Q_{\rho}$ is the Legendre function of the second kind, and the special case $P_n(-x)=(-1)^nP_n(x)$ for the Legendre polynomial, we may perform the following simple transformations

\begin{equation}
\begin{split}
 \int_{-1}^1 &\left(\frac{\partial}{\partial\nu}P_{\nu}(-x)\right)_{\nu=a} P_{l_1}(x)P_{l_2}(x) \text{d}x=\int_{-1}^1 \frac{\partial}{\partial\nu}\left(\cos \nu \pi P_{\nu}(x)-\frac{2}{\pi}\sin\nu\pi Q_{\nu}(x)\right)_{\nu=a} P_{l_1}(x)P_{l_2}(x) \text{d}x\\
=&\int_{-1}^1 \left[(-1)^a\left(\frac{\partial}{\partial\nu}P_{\nu}(x)\right)_{\nu=a}+2\cdot (-1)^{a+1}Q_a(x)\right] P_{l_1}(x)P_{l_2}(x) \text{d}x\\
=&\int_{-1}^1 \left[(-1)^{a+l_1+l_2}\left(\frac{\partial}{\partial\nu}P_{\nu}(x)\right)_{\nu=a}P_{l_1}(-x)P_{l_2}(-x)+2\cdot (-1)^{a+1}Q_a(x)P_{l_1}(x)P_{l_2}(x)\right] \text{d}x\\
=&-\int_{-1}^1 \left(\frac{\partial}{\partial\nu}P_{\nu}(-x)\right)_{\nu=a}P_{l_1}(x)P_{l_2}(x)\text{d}x+2\cdot \int_{-1}^1 Q_a(x)P_{l_1}(x)P_{l_2}(x) \text{d}x\, .
\end{split}
\label{eq:app character sum derivative -> legendreQ}
\end{equation}
In the last step we used the fact that $a+l_1+l_2$ is odd by assumption. Comparing both sides of this series of equations, we observe that in the underlying case we may simply replace

\begin{equation}
 (-1)^{a}\frac{\partial}{\partial \nu}P_{\nu}(-y)\Big |_{\nu=a}\to -Q_{a}(y)
\label{eq:app character sum derivative -> legendreQ replace}
\end{equation}
under the integral in eq.\eqref{eq:app character sum deriv degree}. In summary, we have just found the simple result

\begin{equation}
 S(l_1,l_2;a| l_1+l_2+a=\text{odd})=-\frac{1}{2}\int_{-1}^1 \text{d}y\, P_{l_1}(y)P_{l_2}(y)Q_a(y)\: .
\label{eq:app character sum odd}
\end{equation}
It was shown in \cite{Din:1981ey} that this integral vanishes for $|l_1-l_2|\leq a\leq l_1+l_2$, which means

\begin{equation}
 S(l_1,l_2;a| l_1+l_2+a=\text{odd}, |l_1-l_2|\leq a\leq l_1+l_2)=0\: .
\label{eq:app character sum odd vanish} 
\end{equation}
If the parameter $a$ does not lie within this range, the solutions

\begin{equation}
 \begin{split}
  -\int_{-1}^1 P_{l_1}(y)P_{l_2}(y)Q_{l_1+l_2+2j+1}(y)\, \text{d}y&=\int_{-1}^1 P_{l_1}(y)Q_{l_2}(y)P_{l_1+l_2+2j+1}(y)\, \text{d}y=\\
 =&\frac{\Gamma(j+\frac{1}{2})\Gamma(j+l_1+1)\Gamma(j+l_2+1)\Gamma(j+l_1+l_2+\frac{3}{2})}{2\Gamma(j+1)\Gamma(j+l_1+\frac{3}{2})\Gamma(j+l_2+\frac{3}{2})\Gamma(j+l_1+l_2+2)}
 \end{split}
\label{eq:app character sum odd outside}
\end{equation}
with $j\in\mathbb{N}$ are obtained (see \cite{Askey1982} and \cite{Askey1986}).

\chapter{Proofs}\label{app:proofs}
Here we gather some lengthy, more involved computations for the sake of readability of the main text.
\section{Proof of result 3.4.4}

Translating the corresponding diagrams into an explicit equation using the rules of section \ref{subsec:diagrammatic notation}, we arrive at the formula

\begin{equation}
\begin{split}
&\begin{tikzpicture}
\node at (0,1) {$x$} child[level distance = 1cm,sibling distance=1.5cm]{[fill] circle (1.5pt) node(1){} } child[level distance=.5cm,sibling distance=.75cm] {[fill] circle (1.5pt) child[sibling distance=.25cm] {[fill] circle (1.5pt)node(2){}  } child[sibling distance=.25cm] {[fill] circle (1.5pt) } child[sibling distance=.25cm] {[fill] circle (1.5pt) } child[sibling distance=.25cm] {[fill] circle (1.5pt) } child[sibling distance=.25cm] {[fill] circle (1.5pt) }};
\node at (3,1) {$x$}  child[level distance=.5cm,sibling distance=.75cm] {[fill] circle (1.5pt) child[sibling distance=.25cm] {[fill] circle (1.5pt) } child[sibling distance=.25cm] {[fill] circle (1.5pt) } child[sibling distance=.25cm] {[fill] circle (1.5pt) } child[sibling distance=.25cm] {[fill] circle (1.5pt) } child[sibling distance=.25cm] {[fill] circle (1.5pt)node(3){}  }}child[level distance = 1cm,sibling distance=1.5cm]{[fill] circle (1.5pt)node(4){}  };
\node at (1.5,.5){$+$};
\node at (-2.5,.5){$(R_1)_{\varphi^2}(x)=5\Big($};
\node at (7.5,.5){$\Big)=\contraction{}{\Lo{0}{\varphi}{x}}{}{\Lo{1}{\varphi}{x}}\Lo{0}{\varphi}{x}\Lo{1}{\varphi}{x}+\contraction{}{\Lo{1}{\varphi}{x}}{}{\Lo{0}{\varphi}{x}}\Lo{1}{\varphi}{x}\Lo{0}{\varphi}{x}$};
\path [<-, draw, thick, dashed] (1) -- +(0,-.5) -| (2);
\path [<-, draw, thick, dashed] (3) -- +(0,-.5) -| (4);
\end{tikzpicture}\\
 &=5\sum_{d=-\infty}^{\infty}\sum_{j=0}^{\infty} r^{d+1} S_{jm}(\hat{x})\Lo{0}{\varphi^4}{x;d}_{jm}\times\\
 &+\sum_{l=0}^{\infty}\left(\sum_{\atop{J=|l-j|}{J\neq\mathfrak{m}(l+d)}}^{l+j}  \frac{\braket{j l 0 0}{J 0}^2\, }{(l+d+2)(l+d+3)-J(J+1)}+\frac{\log r}{2(l+d)+5}\braket{jl\, 00}{\mathfrak{m}(l+d)\, 0}^2 \right.\\
&+\left.\sum_{\atop{J=|l-j|}{J\neq\mathfrak{m}(l-d-1)}}^{l+j}\frac{\braket{j l 0 0}{J 0}^2}{(d+1-l)(d+2-l)-J(J+1)} + \frac{\log r}{2(d-l)+3}\braket{jl\, 00}{\mathfrak{m}(d-l-1)\, 0}^2\right)
\label{eq:model NLO Lphi2 infinite sums formula}
\end{split}
\end{equation}
where the identity \eqref{eq:app spherical symm D=3 spherical harmonics coupling} was used to couple the spherical harmonics and with $\mathfrak{m}(d)$ defined as

\begin{equation}
 \mathfrak{m}(d):=\min(|d+2|,|d+3|)\quad .
\label{eq:model NLO Lphi2 m definition}
\end{equation}\index{symbols}{md@$\mathfrak{m}(d)$}%
We are especially interested in the sum over the contraction index $l$, since only this sum may contain infinite expressions after taking matrix elements. Therefore we omit the first line of the above equation in the following calculations and may easily restore it in the end. By a straightforward computation one can check that if we replace $d+2\to -d-2$, then the expression in brackets stays the same except for the sign in front of the logarithmic term in the last line, which then  becomes a minus.  Thus, it is sufficient to consider only the case $d+2\geq 0$, since the other values of $d$ can then simply be determined by the mentioned symmetry.

\subsubsection*{The logarithmic terms}{\ \\}

Let us first focus on the logarithmic terms in eq.\eqref{eq:model NLO Lphi2 infinite sums formula}. With the assumption $d+2\geq 0$ we can deduce

\begin{equation}
 \mathfrak{m}(l+d)=l+d+2
\label{eq:model NLO Lphi2 infinite logs m1}
\end{equation}

and

\begin{equation}
 \mathfrak{m}(d-l-1)=\left\{\begin{array}{ll} l-d-2\quad & \text{for }l\geq d+2\\
 d+1-l & \text{for }l< d+2
\end{array} \right.
\end{equation}
and perform the following manipulations

\begin{equation}
\begin{split}
 &\log r\left[\sum_{l=0}^{\infty}\frac{\braket{jl\, 00}{(l+d+2)\, 0}^2}{2(l+d)+5}+\sum_{l=0}^{d+1}\frac{\braket{jl\, 00}{(d+1-l)\, 0}^2}{2(d-l)+3}- \sum_{l=d+2}^{\infty}\frac{\braket{jl\, 00}{(l-d-2)\, 0}^2}{2(l-d)-3}\right]\\
=&\log r\, \left[\sum_{l=0}^{\infty}\begin{pmatrix}j & l & (l+d+2)\\ 0 & 0 & 0    \end{pmatrix}^2+\sum_{l=0}^{d+1} \begin{pmatrix}j & l & d+1-l\\ 0 & 0 & 0    \end{pmatrix}^2 - \sum_{l=d+2}^{\infty}\begin{pmatrix}j & l & (l-d-2)\\ 0 & 0 & 0    \end{pmatrix}^2\right]\\
=&\log r\, \sum_{l=0}^{d+1} \begin{pmatrix}j & l & d+1-l\\ 0 & 0 & 0    \end{pmatrix}^2\quad ,
\end{split}
\label{eq:model NLO Lphi2 infinite sums log result}
\end{equation}
which is clearly finite for given $d$ and $j$. Because of the mentioned antisymmetry around $d=-2$, the result for arbitrary $d$ is

\begin{equation}
\begin{split}
&\log r\sum_{l=0}^{\infty}\left(\frac{\braket{jl\, 00}{\mathfrak{m}(l+d)\, 0}^2}{2(l+d)+5}+\frac{\braket{jl\, 00}{\mathfrak{m}(d-l-1)\, 0}^2}{2(d-l)+3}\right)=\\
&\qquad \operatorname{sign}(d+2)\sum_{l=0}^{|d+2|-1} \begin{pmatrix}j & l & |d+2|-1-l\\ 0 & 0 & 0    \end{pmatrix}\log r\quad ,
\end{split}
\label{eq:model NLO Lphi2 infinite sums log result general}
\end{equation}
just as we claimed in eqs.\eqref{eq:model NLO Lphi2 remainder 2} and \eqref{eq:model NLO Lphi2 remainder 3}. Note that this expression is zero when $d+j=\text{even}$, since in this case the $3j$-symbol vanishes due to the parity condition.\\

The remaining expressions in eq.\eqref{eq:model NLO Lphi2 infinite sums formula} are of the typical form discussed in appendix \ref{app:characteristic sum}, where it was observed that sums of this type behave very differently for varying choices of parameters $l$, $j$ and $d$ (corresponding to $l_1,l_2$ and $a$ in eq.\eqref{eq:app character sum definition}). For that reason, we distinguish different cases in the analysis of the above formula: The case $d+2> j$, the case $d+2 \leq j$ and $d+j=\text{odd}$ and the case $d+2 \leq j$ and $d+j=\text{even}$ (still assuming $d+2\geq 0$). These cases are related to the number $q$ of annihilation operators in the left representative of eq.\eqref{eq:model NLO Lphi2 infinite sums formula}, as we will see in the following.

\subsubsection*{The case $q\neq 2$}{\ \\}

The fact that the result for $(R_1)_{\varphi^2}(q\neq 2)$ is much simpler than for $(R_1)_{\varphi^2}(q= 2)$ is related to the following lemma, which cannot be applied in the latter case.

\begin{lemma}\label{lem:model NLO Lphi2 infinite sums}{\ \\}
 \begin{equation}
\begin{split}
 \sum_{l=0}^{\infty} &\left(\sum_{\atop{J=|l-j|}{J\neq\mathfrak{m}(l+d)}}^{l+j} \frac{\braket{j l 0 0}{J 0}^2}{(l+d+2)(l+d+3)-J(J+1)}+\sum_{\atop{J=|l-j|}{J\neq\mathfrak{m}(l-d-1)}}^{l+j} \frac{\braket{j l 0 0}{J 0}^2}{(l-d-2)(l-d-1)-J(J+1)}\right)\\
 =&\begin{cases}\sum\limits_{l=0}^{|d+2|-1}\sum\limits_{\atop{J=|l-j|}{\text{denom.}\neq 0}}^{l+j}  \frac{\braket{j l 0 0}{J 0}^2}{(l-|d+2|)(l-|d+2|+1)-J(J+1)}\qquad & \text{if }|d+2|>j \text{ or } d+j=\text{odd}\\
 0 &  \text{if }|d+2|>j \text{ and } d+j=\text{even} \end{cases}
\end{split}
\label{eq:model NLO Lphi2 infinite sums formula result 1}
\end{equation}
\end{lemma}

\begin{proof}
 We assume $d+2>0$ for convenience (recall that the results for negative values of $d+2$ may be obtained from the symmetry around $d=-2$). Let us first consider the case $d+2> j$. Then also $l+d+2> l+j $, which according to eq.\eqref{eq:app character sum greater} yields the simplification

\begin{equation}
 \sum_{l=0}^{\infty}\sum_{J=|l-j|}^{l+j}  \frac{\braket{j l 0 0}{J 0}^2}{(l+d+2)(l+d+3)-J(J+1)}=-\frac{1}{2}\sum_{l=0}^{\infty}\int_{-1}^1P_{l}(y)P_j(y)P_{l+d+2}(y)\log(1-y)\, \text{d}y
\label{eq:model NLO Lphi2 infinite sums cancel 1}
\end{equation}
for the first sum we want to discuss. Note that here the denominator does not vanish for any value of $J$ due to our assumption for $d$. Now let us come to the second sum and first consider only the case $l\geq d+2$ neglecting the remaining finite sum. Then it is evident, that the inequality $l-d-2<|l-j|$ holds, so that we can apply eq.\eqref{eq:app character sum smaller} in order to obtain

\begin{equation}
\begin{split}
 \sum_{l=d+2}^{\infty}\sum_{J=|l-j|}^{l+j} \frac{\braket{j l 0 0}{J 0}^2}{(l-d-2)(l-d-1)-J(J+1)}=&\frac{1}{2}\sum_{l=d+2}^{\infty}\int_{-1}^1P_{l}(y)P_j(y)P_{l-d-2}(y)\log(1-y)\, \text{d}y\\
 =&\frac{1}{2}\sum_{l=0}^{\infty}\int_{-1}^1P_{l+d+2}(y)P_j(y)P_{l}(y)\log(1-y)\, \text{d}y\, .
\end{split}
\label{eq:model NLO Lphi2 infinite sums cancel 2}
\end{equation}
Comparing the two equations above, we see that these infinite sums cancel

\begin{equation}
\begin{split}
 \sum_{l=0}^{\infty} &\left(\sum_{J=|l-j|}^{l+j} \frac{\braket{j l 0 0}{J 0}^2}{(l+d+2)(l+d+3)-J(J+1)}+\sum_{\atop{J=|l-j|}{J\neq d+1-l}}^{l+j} \frac{\braket{j l 0 0}{J 0}^2}{(l-d-2)(l-d-1)-J(J+1)}\right)\\
 =&\sum_{l=0}^{d+1}\sum_{\atop{J=|l-j|}{J\neq d+1-l}}^{l+j}  \frac{\braket{j l 0 0}{J 0}^2}{(l-d-2)(l-d-1)-J(J+1)}\quad ,
\end{split}
\label{eq:model NLO Lphi2 infinite sums formula result 2}
\end{equation}
which together with the symmetry around $d=-2$ confirms the first statement of our lemma.

We proceed with the case $d+j=\text{odd}$, which is just the type of sum considered in section \ref{app:subsec: odd} of the appendix. Therefore we may now use eq.\eqref{eq:app character sum odd}, which leads to the following simplifications (again assuming $d+2>0$ for the moment):

\begin{equation}
 \sum_{l=0}^{\infty}\sum_{\atop{J=|l-j|}{J\neq l+d+2}}^{l+j}  \frac{\braket{j l 0 0}{J 0}^2}{(l+d+2)(l+d+3)-J(J+1)}=-\frac{1}{2}\sum_{l=0}^{\infty}\int_{-1}^1P_{l}(y)P_j(y)Q_{l+d+2}(y)\, \text{d}y
\label{eq:model NLO Lphi2 infinite sums cancel 1b}
\end{equation}

\begin{equation}
\begin{split}
 \sum_{l=d+2}^{\infty}\sum_{\atop{J=|l-j|}{J\neq l-d-2}}^{l+j} \frac{\braket{j l 0 0}{J 0}^2}{(l-d-2)(l-d-1)-J(J+1)}=&-\frac{1}{2}\sum_{l=d+2}^{\infty}\int_{-1}^1P_{l}(y)P_j(y)Q_{l-d-2}(y)\, \text{d}y\\
 =&-\frac{1}{2}\sum_{l=0}^{\infty}\int_{-1}^1P_{l+d+2}(y)P_j(y)Q_{l}(y)\, \text{d}y\, .
\end{split}
\label{eq:model NLO Lphi2 infinite sums cancel 2b}
\end{equation}
As mentioned in the appendix, these integrals vanish if the triangle inequality is satisfied by the parameters (see eq.\eqref{eq:app character sum odd vanish}), and they cancel each other if this inequality is not satisfied (see eq.\eqref{eq:app character sum odd outside}). Therefore, also in this case we observe that the infinite sums cancel and we again obtain the result \eqref{eq:model NLO Lphi2 infinite sums formula result 1} as claimed in the lemma.

It remains to show that the sum under investigation vanishes if $|d+2|>j$ and $d+j=$even. In view of the previous results, this suggests that we have to show

\begin{equation}
\sum\limits_{l=0}^{|d+2|-1}\sum\limits_{\atop{J=|l-j|}{\text{denom.}\neq 0}}^{l+j}  \frac{\braket{j l 0 0}{J 0}^2}{(l-|d+2|)(l-|d+2|+1)-J(J+1)}=0 
\end{equation}
if $|d+2|>j$ and $d+j=$even. Let us again assume $d+2>0$ for convenience. Since $d+j+1=$odd, we may apply eq.\eqref{eq:app character sum odd} obtaining

\begin{equation}
\sum\limits_{l=0}^{d+1}\sum\limits_{\atop{J=|l-j|}{\text{denom.}\neq 0}}^{l+j}  \frac{\braket{j l 0 0}{J 0}^2}{(d+2-l)(d+1-l)-J(J+1)}=\sum_{l=0}^{d+1}\int_{-1}^1\text{d}y\, P_j(y)P_l(y)Q_{d+1-l}(y)
\end{equation}
According to eq.\eqref{eq:app character sum odd vanish} this integral vanishes if the inequality

\begin{equation}
 |l-j|\leq d+1-l\leq l+j
\end{equation}
is satisfied. This implies that only two partial sums remain:

\begin{equation}
\begin{split}
 \sum_{l=0}^{d+1}\int_{-1}^1\text{d}y\, P_j(y)P_l(y)Q_{d+1-l}(y)=\left(\sum_{l=0}^{\frac{d-j}{2}}+\sum_{l=\frac{d+j}{2}+1}^{d+1}\right)\int_{-1}^1\text{d}y\, P_j(y)P_l(y)Q_{d+1-l}(y)&\\
=\sum_{l=0}^{\frac{d-j}{2}}\left(\int_{-1}^1\text{d}y\, P_j(y)P_l(y)Q_{d+1-l}(y)+\int_{-1}^1\text{d}y\, P_j(y)Q_l(y)P_{d+1-l}(y)\right)=0&
\end{split}
\end{equation}
The second line follows if we change the summation index $l\to d+1-l$, and in the last equality we used eq.\eqref{eq:app character sum odd outside}. Thus, the proof of the lemma is complete.

\end{proof}

We are now ready to prove our results for $(R_1)_{\varphi^2}(q\neq 2)$, eqs.\eqref{eq:model NLO Lphi2 remainder 2} and \eqref{eq:model NLO Lphi2 remainder vanish}. Let us start with the case $q\in\{0,4\}$, i.e. in \eqref{eq:model NLO Lphi2 infinite sums formula} we consider only the contribution including four creation or four annihilation operators, $\mathbf{b}_{\pm l_1}\cdots\mathbf{b}_{\pm l_4}$. The specific form of the left representative $\Lo{0}{\varphi^4}{x}=:\left(\Lo{0}{\varphi}{x}\right)^4:$ then suggests that the power of $r$ in these contributions is $d=l_1+\ldots+l_4$ for the contribution containing only creation operators, and $d=-l_1-\ldots-l_4-4$ for the part including four annihilation operators. Further, coupling of the spherical harmonics $S_{l_1m_1}\cdots S_{l_4m_4}$ restricts the possible values of the \emph{spin} $j$ in eq.\eqref{eq:model NLO Lphi2 infinite sums formula} to $j= l_1+\ldots+l_4-2k$ for $k\in\mathbb{N}$ (due to the parity condition and the triangle inequality satisfied by the intertwiners). Hence, for $q\in\{0,4\}$ we have $|d+2|>j$ and $d+j=\text{even}$, which implies that the logarithmic contribution vanishes (see discussion below eq.\eqref{eq:model NLO Lphi2 infinite sums log result general}) and that also the remainig sums vanish by lemma \ref{lem:model NLO Lphi2 infinite sums}. Thus, $(R_1)_{\varphi^2}(d,j,q\in\{0,4\})=0$.

 Analogously, we find for $q\in\{1,3\}$ that either $d=l_1+\ldots+l_3-l_4-1$ or $d=l_1-l_2-\ldots-l_4-3$, while $j$ takes the same values as before. Thus, we are adealing with the case $d+j=$odd, which means that we find eq.\eqref{eq:model NLO Lphi2 infinite sums log result general} for the logarithmic contribution and that we may apply the first case in lemma \ref{lem:model NLO Lphi2 infinite sums} for the remaining sums. The result is eq.\eqref{eq:model NLO Lphi2 remainder 2}.

\subsubsection*{The case $q=2$}{\ \\}

It remains to verify our result for $(R_1)_{\varphi^2}(q=2)$, which appears to be more complicated than the two previous ones. Here lemma \ref{lem:model NLO Lphi2 infinite sums} can not be applied, since $|d+2|>j$ does not hold in general. Therefore we have to analyze eq.\eqref{eq:model NLO Lphi2 infinite sums formula} in its original form, i.e. without any simplifications. Nevertheless, also in the case at hand we are able to show that infinities do indeed cancel. Consider the part of eq.\eqref{eq:model NLO Lphi2 infinite sums formula} with $l>j$. We can perform the following algebraic manipulations:

\begin{equation}
\begin{split}
 &\sum_{l=j}^{\infty} \left(\sum_{\atop{J=|l-j|}{J\neq l+d+2}}^{l+j} \frac{\braket{j l 0 0}{J 0}^2}{(l+d+2)(l+d+3)-J(J+1)}+\sum_{\atop{J=|l-j|}{J\neq l-d-2}}^{l+j}\frac{\braket{j l 0 0}{J 0}^2}{(l-d-2)(l-d-1)-J(J+1)}\right)\\
 =& \sum_{l=0}^{\infty}\sum_{\atop{J=l}{J\neq l+d+j+2}}^{l+2j}  \frac{\braket{j\,(l+j) 0 0}{J 0}^2}{(l+j+d+2)(l+j+d+3)-J(J+1)}\\
+&\sum_{l=0}^{\infty}\sum_{\atop{J=l}{J\neq l+j-d-2}}^{l+2j}\frac{\braket{j\, (l+j) 0 0}{J 0}^2}{(l+j-d-2)(l+j-d-1)-J(J+1)}\\
=& \sum_{l=0}^{\infty}\sum_{\atop{J=0}{J\neq (d+j+2)/2}}^{j} \frac{\braket{j\, (l+j) 0 0}{(2J+l) 0}^2}{2(2J+l)+1}\left(\frac{1}{j+d+2-2J}-\frac{1}{2l+j+d+2+2J+1}\right)\\
+&\sum_{l=0}^{\infty}\sum_{\atop{J=0}{J\neq (j-d-2)/2}}^{j} \frac{\braket{j\, (l+j) 0 0}{(2J+l) 0}^2}{2(2J+l)+1}\left(\frac{1}{j-d-2-2J}-\frac{1}{2l+j-d-2+2J+1}\right)
\end{split}
\label{eq:model NLO Lphi2 infinite sums d+j=even manipulation}
\end{equation}
Here we simply changed the summation limits of the sums over $l$ and $J$ and expanded partial fractions in the last step. In the last equation we also made use of the parity condition on the Clebsch-Gordan coefficient, which restricts the sum over $J$ to even values. We can now express the Clebsch-Gordan coefficient explicitly through gamma functions by eq.\eqref{eq:app spherical symm D=3 CG <-> Gamma fct} and write the resulting infinite sum as a hypergeometric series (see Appendix \ref{app:hypergeoemtric functions}) , which leads us to the rather messy formula

\begin{equation}
\begin{split}
&\sum_{l=j}^{\infty} \left(\sum_{\atop{J=|l-j|}{J\neq l+d+2}}^{l+j} \frac{\braket{j l 0 0}{J 0}^2}{(l+d+2)(l+d+3)-J(J+1)}+\sum_{\atop{J=|l-j|}{J\neq l-d-2}}^{l+j}\frac{\braket{j l 0 0}{J 0}^2}{(l-d-2)(l-d-1)-J(J+1)}\right)\\
 =&\sum_{\atop{J=0}{J\neq (d+j+2)/2}}^j \frac{\Gamma(J+\frac{1}{2})^2\Gamma(j-J+\frac{1}{2})\Gamma(J+j+1)}{\Gamma(J+1)^2\Gamma(j-J+1)\Gamma(J+j+\frac{3}{2})}
 \\ &\qquad\times\left(\frac{-\pFq{4}{3}{1,J+\frac{1}{2},J+j+1, J+\frac{j+d+3}{2}}{J+1, J+j+\frac{3}{2}, J+\frac{j+d+5}{2}}{1}}{2J+j+d+3} +
  \frac{\pFq{3}{2}{1,J+\frac{1}{2},J+j+1}{J+1, J+j+\frac{3}{2}}{1}}{j+d+2-2J}   \right)\\
+&\sum_{\atop{J=0}{J\neq (j-d-2)/2}}^j \frac{\Gamma(J+\frac{1}{2})^2\Gamma(j-J+\frac{1}{2})\Gamma(J+j+1)}{\Gamma(J+1)^2\Gamma(j-J+1)\Gamma(J+j+\frac{3}{2})}
 \\
 &\qquad \times\left( \frac{-\pFq{4}{3}{1,J+\frac{1}{2},J+j+1, J+\frac{j-d-1}{2}}{J+1, J+j+\frac{3}{2}, J+\frac{j-d+1}{2}}{1}}{2J+j-d-1} +\frac{\pFq{3}{2}{1,J+\frac{1}{2},J+j+1}{J+1, J+j+\frac{3}{2}}{1}}{j-d-2-2J}\right)
\end{split}
\label{eq:model NLO Lphi2 hypergeometrics}
\end{equation}
This form, despite its complicated appearance, allows us to analyze the divergent behavior of the infinite series by simply investigating the parameters of the hypergeometric functions involved. As mentioned in the appendix, hypergeometric series of the general form

\begin{equation}
 \pFq{p+1}{p}{\alpha_1,\ldots,\alpha_{p+1} }{\beta_1,\ldots,\beta_{p}}{1}
\end{equation}
converge for $\sum_j\beta_j-\sum_i\alpha_i=:k>0$. Thus, we immediately see that the hypergeometric functions $_4F_3$ in eq.\eqref{eq:model NLO Lphi2 hypergeometrics} are convergent, as they are \emph{1-balanced}, i.e. $k=1$. The series of the type $_3F_2$, on the other hand, are \emph{0-balanced}, and should thus not be expected to converge. The precise divergent behavior can be analyzed with the help of eq.\eqref{eq:app hypergeos zero balanced limit} from the appendix, which tells us that these hypergeometric series approach infinity as

\begin{equation}
\begin{split}
 \frac{\Gamma(\alpha_1)\Gamma(\alpha_2)\Gamma(\alpha_3)}{\Gamma(\beta_1)\Gamma(\beta_2)} &\pFq{3}{2}{\alpha_1, \alpha_2, \alpha_3}{\beta_1, \beta_2}{1-\varepsilon}= 2\psi(1)-\psi(\alpha_1)-\psi(\alpha_2)-\log\varepsilon\\
 +&\frac{(\beta_1-\alpha_3)(\beta_2-\alpha_3)}{\alpha_1\alpha_2} \pFq{4}{3}{1,1,\beta_1-\alpha_3+1,\beta_2-\alpha_3+1}{2,\alpha_1+1,\alpha_2+1}{1}+ O(\varepsilon)
\end{split}
\end{equation}
if the series on the left is zero balanced and if $\operatorname{Re}(\alpha_3)>0$. Applying this formula to eq.\eqref{eq:model NLO Lphi2 hypergeometrics} and extracting just the divergent part (i.e. the part proportional to the logarithm), we obtain

\begin{equation}
\begin{split}
 \lim_{\varepsilon\to 0}&\left(\sum_{\atop{J=0}{J\neq	(j+d+2)/2}}^j \frac{\Gamma(J+\frac{1}{2})\Gamma(j-J+\frac{1}{2})}{\Gamma(J+1)\Gamma(j-J+1)}\frac{\log\varepsilon}{j+d+2-2J}\right. \\ 
 &+\left.\sum_{\atop{J=0}{J\neq	(j-d-2)/2}}^j \frac{\Gamma(J+\frac{1}{2})\Gamma(j-J+\frac{1}{2})}{\Gamma(J+1)\Gamma(j-J+1)}\frac{\log\varepsilon}{j-d-2-2J}\right)=0
\end{split}
\end{equation}
Thus, we have indeed verified that all infinities cancel without the need of any exterior renormalization procedure. However, we do not obtain a result as simple as eq.\eqref{eq:model NLO Lphi2 infinite sums formula result 1} in this case, as the remaining hypergeometric functions of the type $_4F_3$ in eq.\eqref{eq:model NLO Lphi2 hypergeometrics} do not seem to cancel in general, leading to the complicated result

\begin{equation}
 \begin{split}
   &\sum_{l=0}^{\infty}  \left(\sum_{\atop{J=|l-j|}{J\neq j+d+2}}^{l+j}\frac{\braket{j l 0 0}{J 0}^2}{(l+d+2)(l+d+3)-J(J+1)}+\sum_{\atop{J=|l-j|}{J\neq\mathfrak{m}(d-l-1)}}^{l+j}\frac{\braket{j l 0 0}{J 0}^2}{(l-d-2)(l-d-1)-J(J+1)}\right)\\
 &=\sum_{l=0}^{j-1}  \left(\sum_{\atop{J=|l-j|}{J\neq j+d+2}}^{l+j}\frac{\braket{j l 0 0}{J 0}^2}{(l+d+2)(l+d+3)-J(J+1)}+\sum_{\atop{J=|l-j|}{J\neq\mathfrak{m}(d-l-1)}}^{l+j}\frac{\braket{j l 0 0}{J 0}^2}{(l-d-2)(l-d-1)-J(J+1)}\right)\\
&-\sum_{\atop{J=0}{J\neq(j+d+2)/2}}^j\frac{\Gamma(J+\frac{1}{2})^2\Gamma(j-J+\frac{1}{2})\Gamma(J+j+1)}{\Gamma(J+1)^2\Gamma(j-J+1)\Gamma(J+j+\frac{3}{2})}\frac{\pFq{4}{3}{1,J+\frac{1}{2},J+j+1, J+\frac{j+d+3}{2}}{J+1, J+j+\frac{3}{2}, J+\frac{j+d+5}{2}}{1}}{2J+j+d+3}\\
&-\sum_{\atop{J=0}{J\neq(j-d-2)/2}}^j\frac{\Gamma(J+\frac{1}{2})^2\Gamma(j-J+\frac{1}{2})\Gamma(J+j+1)}{\Gamma(J+1)^2\Gamma(j-J+1)\Gamma(J+j+\frac{3}{2})}\frac{\pFq{4}{3}{1,J+\frac{1}{2},J+j+1, J+\frac{j-d-1}{2}}{J+1, J+j+\frac{3}{2}, J+\frac{j-d+1}{2}}{1}}{2J+j-d-1}\\
&+\sum_{J=0}^j \left(\pFq{4}{3}{1,1,J+1,J+j+\frac{3}{2}}{2,J+\frac{3}{2}, J+j+2}{1}\, \frac{J(J+j+\frac{1}{2})}{(J+\frac{1}{2})(J+j+1)}+2\psi(1)-\psi(J+\frac{1}{2})-\psi(J+j+1)\right)\\
&\times \frac{\Gamma(J+\frac{1}{2})\Gamma(j-J+\frac{1}{2})}{\Gamma(J+1)\Gamma(j-J+1)}\, \left(\frac{1}{j+d+2-2J}\Big|_{J\neq\frac{j+d+2}{2}}+\frac{1}{j-d-2-2J}\Big|_{J\neq \frac{j-d-2}{2}}\right)\\
&=:R_{(q=2)}(d,j)
 \end{split}
\label{eq:model NLO Lphi2 infinite sums formula result general}
\end{equation}
This completes the proof of result \ref{res:model NLO Lphi2}.\par\hfill\qedsymbol\par
It should be noted that the above discussion, especially eq.\eqref{eq:model NLO Lphi2 infinite sums formula result general}, holds in general, i.e. in all three cases we distinguished, since we did not put any restrictions on $d$ or $j$. Thus, we could have shown finiteness for all choices of $j$ and $d$ in just this one step. However, the particularly simple result \eqref{eq:model NLO Lphi2 infinite sums formula result 1} does not seem to follow so easily from eq.\eqref{eq:model NLO Lphi2 infinite sums formula result general}. Furthermore, we wanted to emphasize the contrast between the calculational simplicity of the first two cases (due to the identities of sections \ref{app:subsec:triangle} and \ref{app:subsec: odd}) and the complicated expressions needed in the analysis of the last case.

\section{Proof of result 3.4.7}

Translating the three diagrams in eq.\eqref{eq:model NLO Lphi3 remainder form b} into an equation, we find

\begin{equation}
\begin{split}
&\begin{tikzpicture}
  \node at (4.6,0) {$x$} child[level distance = 1cm,sibling distance=1cm]{[fill] circle (1.5pt) node(5){} };
\node at (6,0) {$y$} child[level distance = 1cm,sibling distance=1.5cm]{[fill] circle (1.5pt) node(6){} } child[level distance=.5cm,sibling distance=1cm] {[fill] circle (1.5pt) child[sibling distance=.25cm] {[fill] circle (1.5pt) node(7){} } child[sibling distance=.25cm] {[fill] circle (1.5pt)  } child[sibling distance=.25cm] {[fill] circle (1.5pt)  } child[sibling distance=.25cm] {[fill] circle (1.5pt)  } child[sibling distance=.25cm] {[fill] circle (1.5pt) node(8){} }};
\node at (7.6,0) {$x$} child[level distance = 1cm,sibling distance=1cm]{[fill] circle (1.5pt) node(9){} }; \node at (9,0) {$y$} child[level distance=.5cm,sibling distance=.75cm] {[fill] circle (1.5pt) child[sibling distance=.25cm] {[fill] circle (1.5pt) node(10){} } child[sibling distance=.25cm] {[fill] circle (1.5pt)  } child[sibling distance=.25cm] {[fill] circle (1.5pt) } child[sibling distance=.25cm] {[fill] circle (1.5pt)  } child[sibling distance=.25cm] {[fill] circle (1.5pt) node(11){} }}child[level distance = 1cm,sibling distance=1cm]{[fill] circle (1.5pt) node(12){} };
\node at (7.35,-.5){$+$};
\node at (9.8,-.5){$+$};
\path [<-, draw,dashed] (5) -- +(0,-.5) -| (8);
\path [<-, draw,dashed] (6) -- +(0,-.3) -| (7);
\path [<-, draw,dashed] (9) -- +(0,-.5) -| (10);
\path [<-, draw,dashed] (11) -- +(0,-.5) -| (12);
\node at (10.6,0) {$x$}  child[level distance=.5cm,sibling distance=1cm] {[fill] circle (1.5pt) child[sibling distance=.25cm] {[fill] circle (1.5pt) node(13){} } child[sibling distance=.25cm] {[fill] circle (1.5pt)  } child[sibling distance=.25cm] {[fill] circle (1.5pt) } child[sibling distance=.25cm] {[fill] circle (1.5pt) } child[sibling distance=.25cm] {[fill] circle (1.5pt) node(14){} }};
\node at (12.2,0) {$y$} child[level distance = 1cm,sibling distance=1cm]{[fill] circle (1.5pt) node(15){} }  child[level distance = 1cm,sibling distance=1cm]{[fill] circle (1.5pt) node(16){} };
\path [<-, draw,dashed] (13) -- +(0,-.5) -| (16);
\path [<-, draw,dashed] (14) -- +(0,-.3) -| (15);
\node at (14.5,-.5){$-\log|x-y|\Lo{0}{\varphi^3}{x}\Big]=$};
\node at (2.6,-.5){$(R_1)_{\varphi^3}(x)=20\lim\limits_{y\to x}\Big[$};
\end{tikzpicture}\\
20\sum_{l,l'=0}^{\infty}&\sum_{\atop{J_1,J_2=0}{M_1,M_2}}^{\infty}\sum_{d=-\infty}^{\infty}\sum_j  (2J_1+1)(2J_2+1)\begin{pmatrix} j & l' & J_1 \\ 0& 0& 0 \end{pmatrix}^2 \begin{pmatrix} J_1 & l & J_2 \\ 0& 0& 0 \end{pmatrix}^2 \Lo{0}{\varphi^3}{x;d}_{jm}\, r^dS_{jm}(\hat{x})\\
 \times \Big(&\frac{\eta^{l'+1}}{(l+l'+d+2)(l+l'+d+3)-J_2(J_2+1)}\Big|_{J_2 \neq \mathfrak{m}(l+l'+d)}+ \frac{\eta^{l'+1}\log r_y}{2(l+l'+d)+5}\delta_{J_2,\mathfrak{m}(l+l'+d)}\\
 +&\frac{\eta^{l'+1}}{(l'-l+d+1)(l'-l+d+2)-J_2(J_2+1)}\Big|_{J_2 \neq \mathfrak{m}(d+l'-l-1)}+ \frac{\eta^{l'+1}\log r_y}{2(l'-l+d)+3}\delta_{J_2,\mathfrak{m}(d+l'-l-1)} \\
+&\frac{\eta^{l+l'+2}}{(l+l'-d-1)(l+l'-d)-J_2(J_2+1)}\Big|_{J_2 \neq \mathfrak{m}(d-l-l'-2)}- \frac{\eta^{l+l'+2}\log r_x}{2(l+l'-d)-1}\delta_{J_2,\mathfrak{m}(d-l-l'-2)}\Big)\\
&-20\log|x-y|\Lo{0}{\varphi^3}{x}
\end{split}
\label{eq:model NLO Lphi3 CT formula}
\end{equation}
with $\eta:=|y|/|x|$. Here we used the coupling rules for the spherical harmonics (eq.\eqref{eq:app spherical symm D=3 spherical harmonics coupling}) and inserted the specific form of the factor $D$ (eq.\eqref{eq:model NLO Lphi D}), which is obtained at the vertex in the diagrams. In the intermediate steps of the following calculation we omit the sums over $d$ and $j$ and the expression $\Lo{0}{\varphi^3}{x;d}_{jm}\, r^dS_{jm}(\hat{x})$, since this is independent of the contraction index $l$ and we are mainly interested in the sum over this index. We can easily restore the omitted expressions at the end of the calculation.

\subsubsection*{The logarithmic terms}

Let us start our inspection of eq.\eqref{eq:model NLO Lphi3 CT formula} with the terms containing logarithms. Replacing $d\to l'+d$ in eq.\eqref{eq:model NLO Lphi2 infinite sums log result}, we obtain just the first two logarithmic terms of the formula above (neglecting some prefactors independent of $l$). This suggests

\begin{equation}
\begin{split}
 &\sum_{l=0}^{\infty}\left(\braket{J_1l\, 00}{\mathfrak{m}(l+l'+d)\, 0}^2\frac{\eta^{l'+1}\log r_y}{2(l+l'+d)+5}+\braket{J_1l\, 00}{\mathfrak{m}(d+l'-l-1)\, 0}^2\frac{\eta^{l'+1}\log r_y}{2(l'-l+d)+3}\right)\\
 =&\operatorname{sign}(l'+d+2)\sum_{l=0}^{|l'+d+2|-1}\begin{pmatrix}
                        J_1 & l & |l'+d+2|-1-l\\ 0 & 0 & 0
                       \end{pmatrix}^2\eta^{l'+1}\log r_y
\end{split}
\label{eq:model NLO Lphi3 CT formula logs1}
\end{equation}
Let us assume $d+1\geq 0$ for the moment. Then the partial sum with $l+l'\geq d+1$ over the logarithmic term in the last line of eq.\eqref{eq:model NLO Lphi3 CT formula} becomes

\begin{equation}
\begin{split}
 -&\sum_{\atop{l,l'=0}{l+l'\geq d+1}}^{\infty}\sum_{J_1} \braket{jl\, 00}{J_1\, 0}^2 \braket{J_1l'\, 00}{(l+l'-d-1)\, 0}^2\frac{\eta^{l+l'+2}\log r_x}{2(l+l'-d)-1}\\
=-&\sum_{\atop{l,l'=0}{l'\geq d+1-l}}^{\infty}\sum_{J_1} \braket{jl\, 00}{J_1\, 0}^2 \begin{pmatrix}
                        J_1 & l' & l+l'-d-1\\ 0 & 0 & 0
                       \end{pmatrix}^2\, \eta^{l+l'+2}\log r_x\\
=-&\left(\sum_{l=0}^{d+1}\sum_{l'=d+1-l}^{\infty}+\sum_{l=d+2}^{\infty}\sum_{l'=0}^{\infty}\right)\sum_{J_1} \braket{jl\, 00}{J_1\, 0}^2 \begin{pmatrix}
                        J_1 & l' & l+l'-d-1\\ 0 & 0 & 0
                       \end{pmatrix}^2\, \eta^{l+l'+2}\log r_x
\end{split}
\end{equation}
Further, under this assumption equation \eqref{eq:model NLO Lphi3 CT formula logs1} takes the form

\begin{equation}
\begin{split}
 &\sum_{l'=0}^{\infty}\sum_{l=0}^{l'+d+1}\sum_{J_1} \braket{jl\, 00}{J_1\, 0}^2 \braket{J_1l'\, 00}{(l'+d+1-l)\, 0}^2\frac{\eta^{l'+1}\log r_y}{2(l'-l+d)+3}\\
=&\left(\sum_{l=0}^{d+1}\sum_{l'=0}^{\infty}+\sum_{l=d+2}^{\infty}\sum_{l'=l-(d+1)}^{\infty}\right)\sum_{J_1} \braket{jl\, 00}{J_1\, 0}^2 \begin{pmatrix}
                        J_1 & l' & l'+d+1-l\\ 0 & 0 & 0
                       \end{pmatrix}^2\, \eta^{l'+1}\log r_y
\end{split}
\end{equation}
Comparing the two equations above it is easy to verify that these sums cancel in the limit $\eta\to 1$ (i.e. $y\to x$), so that only the finite sum

\begin{equation}
 20 \sum_{l,l'=0}^{l+l'\leq d}\sum_{J_1} \braket{jl\, 00}{J_1\, 0}^2 \begin{pmatrix}
                        J_1 & l' & d-l-l'\\ 0 & 0 & 0
                       \end{pmatrix}^2\, \log r_x
\end{equation}
remains. Note that the sum of the three logarithmic series in eq.\eqref{eq:model NLO Lphi3 CT formula} is antisymmetric around $d=-3/2$. Hence, for $d<-1$ we may simply multiply our result by $\operatorname{sign}(d+3/2)$ and replace $d$ with $|d+3/2|-3/2$. Thus

\begin{equation}
 \begin{split}
  20&\sum_{d=-\infty}^{\infty}\sum_{j=0}\Lo{0}{\varphi^3}{x;d}_{jm}\, r^dS_{jm}(\hat{x})\sum_{l,l'=0}^{\infty}\sum_{J_1}\braket{j l'\, 00}{J_1\, 0}^2\, \log r_x\\
 \times&\left(\frac{\braket{J_1 l\, 00}{\mathfrak{m}(l+l'+d)\, 0}}{2(l+l'+d)+5}+\frac{\braket{J_1 l\, 00}{\mathfrak{m}(d+l'-l-1)\, 0}}{2(l'-l+d)+3}-\frac{\braket{J_1 l\, 00}{\mathfrak{m}(d-l-l'-2)\, 0}}{2(l+l'-d)-5}\right)\\
 =20&\sum_{d=-\infty}^{\infty}\sum_{j=0}\operatorname{sign}(d+3/2)\Lo{0}{\varphi^3}{x;d}_{jm} r^dS_{jm}(\hat{x})\\
 &\qquad\times\sum_{l,l'=0}^{l+l'\leq|d+\frac{3}{2}|-\frac{3}{2}}\sum_{J_1} \braket{jl\, 00}{J_1\, 0}^2 \begin{pmatrix}
                        J_1 & l' & |d+\frac{3}{2}|-\frac{3}{2}-l-l'\\ 0 & 0 & 0
                       \end{pmatrix}^2\, \log r_x
 \end{split}
\label{eq:model NLO Lphi3 CT formula logs result}
\end{equation}
which is in accordance with the logarithmic terms in eqs.\eqref{eq:model NLO Lphi3 remainder q=0} and \eqref{eq:model NLO Lphi3 remainder q=5}.

\subsubsection*{Vanishing partial sum}

Now consider the follwing partial sum of eq.\eqref{eq:model NLO Lphi3 CT formula}:

\begin{equation}
 \begin{split}
  &20\sum_{l,l'=0}^{\infty}\sum_{\atop{J_1,J_2}{M_1,M_2}}\frac{\eta^{l+1}\braket{jl\, 00}{J_1\, 0}^2\braket{J_1 l'\, 00}{J_2\, 0}^2}{(l+l'+d+2)(l+l'+d+3)-J_2(J_2+1)}\Big|_{\text{denom.}\neq 0} \\
+&20\sum_{l=0}^{\infty}\sum_{l'=l+d+2}^{\infty}\sum_{\atop{J_1,J_2}{M_1,M_2}}\frac{\eta^{l+1}\braket{jl\, 00}{J_1\, 0}^2\braket{J_1 l'\, 00}{J_2\, 0}^2}{(l'-l-d-2)(l'-l-d-1)-J_2(J_2+1)}\Big|_{\text{denom.}\neq 0}\, ,
 \end{split}
\end{equation}
 As in the calculation of the remainder $(R_1)_{\varphi^2}$ in the previous section, this expression again behaves very differently for varying values of $d$ and $j$. Hence, we again distinguish different values of the number of annihilation operators among the three operators constituting $\Lo{0}{\varphi^3}{x}$ in eq.\eqref{eq:model NLO Lphi3 CT formula}, where the notation $\Lo{0}{\varphi^3}{x;q}$ will be used in order to indicate this grading of the left representative.

Let us start with the case $q=0$, which implies $d\geq j$ and $d+j=\text{even}$. Thus, we may use eq.\eqref{eq:app character sum greater} from the appendix in order to simplify

\begin{equation}
\begin{split}
 &20\sum_{l,l'=0}^{\infty}\sum_{\atop{J_1,J_2}{M_1,M_2}}\frac{\braket{jl\, 00}{J_1\, 0}^2\braket{J_1 l'\, 00}{J_2\, 0}^2}{(l+l'+d+2)(l+l'+d+3)-J_2(J_2+1)}\Big|_{\text{denom.}\neq 0} \\
 =&-10\sum_{l,l'=0}^{\infty}\sum_{J=|j-l|}^{j+l}\braket{jl\, 00}{J\, 0}^2\int_{-1}^1 P_{J}(z)P_{l'}(z)P_{l+l'+d+2}(z)\log(1-z)\text{d}z
\end{split}
\end{equation}
and

\begin{equation}
\begin{split}
 &20\sum_{l=0}^{\infty}\sum_{l'=l+d+2}^{\infty}\sum_{\atop{J_1,J_2}{M_1,M_2}}\frac{\braket{jl\, 00}{J_1\, 0}^2\braket{J_1 l'\, 00}{J_2\, 0}^2}{(l'-l-d-2)(l'-l-d-1)-J_2(J_2+1)}\Big|_{\text{denom.}\neq 0} \\
 =&10\sum_{l=0}^{\infty}\sum_{l'=l+d+2}^{\infty}\sum_{J=|j-l|}^{j+l}\braket{jl\, 00}{J\, 0}^2\int_{-1}^1 P_{J}(z)P_{l'}(z)P_{l'-l-d-2}(z)\log(1-z)\text{d}z\\
=&10\sum_{l,l'=0}^{\infty}\sum_{J=|j-l|}^{j+l}\braket{jl\, 00}{J\, 0}^2\int_{-1}^1 P_{J}(z)P_{l'+l+d+2}(z)P_{l'}(z)\log(1-z)\text{d}z\, .
\end{split}
\end{equation}
Thus, we have found that the partial sum under consideration vanishes in the limit $y\to x$
\begin{equation}
 \begin{split}
  \lim_{y\to x}\Big[&20\sum_{l,l'=0}^{\infty}\sum_{\atop{J_1,J_2}{M_1,M_2}}\frac{\eta^{l+1}\braket{jl\, 00}{J_1\, 0}^2\braket{J_1 l'\, 00}{J_2\, 0}^2}{(l+l'+d+2)(l+l'+d+3)-J_2(J_2+1)}\Big|_{\text{denom.}\neq 0} \\
+&20\sum_{l=0}^{\infty}\sum_{l'=l+d+2}^{\infty}\sum_{\atop{J_1,J_2}{M_1,M_2}}\frac{\eta^{l+1}\braket{jl\, 00}{J_1\, 0}^2\braket{J_1 l'\, 00}{J_2\, 0}^2}{(l'-l-d-2)(l'-l-d-1)-J_2(J_2+1)}\Big|_{\text{denom.}\neq 0}\Big] =0	\, .
 \end{split}
\end{equation}
Note that this result is very similar in nature to the cancellation of eqs.\eqref{eq:model NLO Lphi2 infinite sums cancel 1} and \eqref{eq:model NLO Lphi2 infinite sums cancel 2} in the previous section. This similarity is no surprise, since the diagrams corresponding to this calculation of the previous section (see eq.\eqref{eq:model NLO Lphi2 infinite sums formula}) are subtrees of the first two graphs in eq.\eqref{eq:model NLO Lphi3 CT formula}.

\subsubsection*{Cancellation of infinities}

It remains to consider the partial sum

\begin{equation}
\begin{split}
 &20\sum_{l,l'}\sum_{\atop{J_1,J_2}{M_1,M_2}}\frac{\eta^{l+l'+2}\braket{jl\, 00}{J_1\, 0}^2\braket{J_1 l'\, 00}{J_2\, 0}^2}{(l+l'-d-1)(l+l'-d)-J_2(J_2+1)}\Big|_{\text{denom.}\neq 0}\\
 +&20\sum_{l=0}^{\infty}\sum_{l'=0}^{l+d+1}\sum_{\atop{J_1,J_2}{M_1,M_2}}\frac{\eta^{l+1}\braket{jl\, 00}{J_1\, 0}^2\braket{J_1 l'\, 00}{J_2\, 0}^2}{(l'-l-d-2)(l'-l-d-1)-J_2(J_2+1)}\Big|_{\text{denom.}\neq 0}\, .
\end{split}
\label{eq:model NLO Lphi3 cancellation sum}
\end{equation}
In the first sum the identity \eqref{eq:app character sum odd} from the appendix may be used for $l+l'\geq d+1$, since $J_1+l-d-1=2l+j-2k-d-1=\text{odd}$ (recall that $d-j=\text{even}$). Hence this sum simplifies to

\begin{equation}
\begin{split}
 &20\sum_{\atop{l,l'}{l+l'\geq d+1}}\sum_{\atop{J_1,J_2}{M_1,M_2}}\frac{\eta^{l+l'+2}\braket{jl\, 00}{J_1\, 0}^2\braket{J_1 l'\, 00}{J_2\, 0}^2}{(l+l'-d-1)(l+l'-d)-J_2(J_2+1)}\Big|_{\text{denom.}\neq 0} \\
 =-&10\sum_{\atop{l,l'}{l+l'\geq d+1}}\sum_{J_1,M_1}\eta^{l+l'+2} \braket{jl\, 00}{J_1\, 0}^2\int_{-1}^1 P_{J_1}(z)P_{l'}(z)Q_{l+l'-d-1}(z)\,\text{d}z
\end{split}
\label{eq:model NLO Lphi3 3rd sum}
\end{equation}
Now we consider the second sum. Here we may change the coupling order in the product of the Clebsch-Gordan coefficients, i.e. we exchange the role of $l$ and $l'$ in the intertwiners (we have this freedom since the coupling order was arbitrary). 
Application of \eqref{eq:app character sum odd} then yields

\begin{equation}
\begin{split}
 &20\sum_{\atop{l,l'}{l'-l\leq d+1}}\sum_{\atop{J_1,J_2}{M_1,M_2}}\frac{\eta^{l+1}\braket{jl'\, 00}{J_1\, 0}^2\braket{J_1 l\, 00}{J_2\, 0}^2}{(l-l'+d+1)(l-l'+d+2)-J_2(J_2+1)}\Big|_{\text{denom.}\neq 0} \\
 =-&10\sum_{\atop{l,l'}{l'-l\leq d+1}}\sum_{J_1,M_1}\eta^{l+1} \braket{jl'\, 00}{J_1\, 0}^2\int_{-1}^1 P_{J_1}(z)P_{l}(z)Q_{l-l'+d+1}(z)\,\text{d}z\\
=-&10\sum_{\atop{l,l'}{l+l'\geq d+1}}^{\infty}\sum_{J_1,M_1}\eta^{l+l'-d} \braket{jl'\, 00}{J_1\, 0}^2\int_{-1}^1 P_{J_1}(z)P_{l+l'-d-1}(z)Q_{l}(z)\,\text{d}z\\
=-&10\sum_{\atop{l,l'}{l+l'\geq d+1}}^{\infty}\sum_{J_1,M_1}\eta^{l+l'-d} \braket{jl\, 00}{J_1\, 0}^2\int_{-1}^1 P_{J_1}(z)P_{l+l'-d-1}(z)Q_{l'}(z)\,\text{d}z\\
\end{split}
\label{eq:model NLO Lphi3 2nd sum b}
\end{equation}
In the last step we simply exchanged the names of the summation indices. According to eq.\eqref{eq:app character sum odd vanish} from the appendix, the summands in eqs.\eqref{eq:model NLO Lphi3 3rd sum} and \eqref{eq:model NLO Lphi3 2nd sum b} vanish for the following constellation of parameters:

\begin{equation}
 |l+j-2k-l'|\leq l+l'-d-1\leq l+l'+j-2k
\end{equation}
where we wrote $J_1=l+j-2k$ with $k\in\{0,\ldots,\min(l,j)\}$. Therefore we are only interested in the cases where these inequalities are not satisfied. We first observe that the second inequality, $l+l'-d-1\leq l+l'+j-2k$ is equivalent to $j-2k+d+1\geq 0$, which always holds for $d\geq j$. The second condition, however, is not satisfied in the following two cases

\begin{equation}
 l+j-2k>l'\quad \text{and}\quad l'<\frac{j-2k+d+1}{2} \qquad \text{or}\qquad l<l'-(j-2k)\quad \text{and}\quad l<\frac{d+1-(j-2k)}{2}
\label{eq:model NLO Lphi3 3rd sum neq 0}
\end{equation}
Let us start with the latter condition, which implies

\begin{equation}
 J_1+l+l'-d-1<l'\, ,
\end{equation}
and allows us to use eq.\eqref{eq:app character sum odd outside} for the integrals in eqs.\eqref{eq:model NLO Lphi3 3rd sum} and \eqref{eq:model NLO Lphi3 2nd sum b}. Note that in the first equation, \eqref{eq:model NLO Lphi3 3rd sum}, the index $l'$ is attached to a Legendre polynomial of the first kind, $P_{l'}$, while in the other equation, \eqref{eq:model NLO Lphi3 2nd sum b}, it is attached a Legendre polynomial of the second kind, $Q_{l'}$. According to eq.\eqref{eq:app character sum odd outside} these integrals then differ by a minus sign, which tells us that eqs.\eqref{eq:model NLO Lphi3 3rd sum} and \eqref{eq:model NLO Lphi3 2nd sum b} cancel each other in the limit $\eta\to 1$ if the mentioned inequality is satisfied.

On the other hand, the second condition in eq.\eqref{eq:model NLO Lphi3 3rd sum neq 0} leads to the inequality
\begin{equation}
 l+2l'-d-1<J_1\, .
\end{equation}
Again we may apply eq.\eqref{eq:app character sum odd outside}, but this time the largest index, namely $J_1$, is attached to a Legendre polynomial of the first kind in both of the two equations. Since an exchange of $l_1$ and $l_2$ does not alter the result in eq.\eqref{eq:app character sum odd outside}, the two integrals are equal, so that instead of canceling each other, eqs.\eqref{eq:model NLO Lphi3 3rd sum} and \eqref{eq:model NLO Lphi3 2nd sum b} add up in the case at hand. Summing up the above discussion, we have found that (neglecting some prefactors $\eta^n$ with constant $n\in \mathbb{N}$, which are irrelevant in the limit $\eta\to 1$)

\begin{equation}
\begin{split}
 -10&\sum_{\atop{l,l'}{l+l'\geq d+1}}\sum_{J_1,M_1}\braket{jl\, 00}{J_1\, 0}^2\eta^{l+l'}\int_{-1}^1 P_{J_1}(z)\Big(P_{l'}(z)Q_{l+l'-d-1}(z)+P_{l+l'-d-1}(z)Q_{l'}(z)\Big)\,\text{d}z\\
 =-20&\left(\sum_{l=0}^{j}\sum_{k=0}^l\sum_{l'=d+1-l}^{(j-2k+d)/2}+\sum_{k=0}^j\sum_{l=j+1}^{d}\sum_{l'=d+1-l}^{(j-2k+d)/2}+\sum_{k=0}^j\sum_{l=d+1}^{\infty}\sum_{l'=0}^{(j-2k+d)/2}\right)\times\\
&\qquad\braket{jl\, 00}{l+j-2k\, 0}^2\eta^{l+l'}\int_{-1}^1 P_{l+j-2k}(z)P_{l'}(z)Q_{l+l'-d-1}(z)\,\text{d}z\\
\end{split}
\label{eq:model NLO Lphi3 3sums summed up}
\end{equation}
Clearly the first two triple sums in this expression are finite and it remains to search for infinities in the third one.

\begin{equation}
\begin{split}
 -20&\sum_{k=0}^j\sum_{l=d+1}^{\infty}\sum_{l'=0}^{(j-2k+d)/2}\braket{jl\, 00}{l+j-2k\, 0}^2\eta^{l+l'}\int_{-1}^1 P_{l+j-2k}(z)P_{l'}(z)Q_{l+l'-d-1}(z)\,\text{d}z\\
=-20&\sum_{k=0}^j\sum_{l=0}^{\infty}\sum_{l'=0}^{(j-2k+d)/2}\braket{j(l+d+1)\, 00}{(l+d+1+j-2k)\, 0}^2\\
 &\hspace{5cm}\times\int_{-1}^1 P_{l+d+1+j-2k}(z)P_{l'}(z)Q_{l+l'}(z)\,\text{d}z\eta^{l+l'}\\
=-\frac{10}{\pi}&\sum_{k=0}^j\sum_{l'=0}^{(j-2k+d)/2}\frac{\Gamma(k+\frac{1}{2})\Gamma(j-k+\frac{1}{2})\Gamma(\frac{d+j-2k+1}{2}-l')\Gamma(\frac{d+j-2k}{2}+1)}{\Gamma(k+1)\Gamma(j-k+1)\Gamma(\frac{d+j-2k}{2}-l'+1)\Gamma(\frac{d+j-2k+3}{2})}\, \eta^{l'}\\
\times &\pFq{6}{5}{1,d+j-2k+\frac{5}{2},d-k+\frac{3}{2},d+j-k+2,\frac{d+j-2k-l'}{2}+1,\frac{d+j-2k+l'+3}{2}}{d+j-2k+\frac{3}{2},d-k+2,d+j-k+\frac{5}{2},\frac{d+j-2k-l'+3}{2},\frac{d+j-2k+l'}{2}+2}{\eta}\\
\times &\frac{\Gamma(d+j-2k+\frac{5}{2})\Gamma(d-k+\frac{3}{2})\Gamma(d+j-k+2)\Gamma(\frac{d+j-2k-l'}{2}+1)\Gamma(\frac{d+j-2k+l'+3}{2})}{\Gamma(d+j-2k+\frac{3}{2})\Gamma(d-k+2)\Gamma(d+j-k+\frac{5}{2})\Gamma(\frac{d+j-2k-l'+3}{2})\Gamma(\frac{d+j-2k+l'}{2}+2)}
\end{split}
\label{eq:model NLO Lphi3 3sums result}
\end{equation}
Here we expressed the $3j$-coefficient and the integral through $\Gamma$-functions with the help of eqs.\eqref{eq:app spherical symm D=3 CG <-> Gamma fct} and \eqref{eq:app character sum odd outside}, and wrote the infinite sum over $l$ as a hypergeometric series. Taking a closer look at the parameters of this hypergeometric series, we find that it is zero-balanced, and therefore logarithmically divergent in the limit $\eta\to 1$, as expected. Eq.\eqref{eq:app hypergeos zero balanced limit} from the appendix allows for a more detailed characterization of this divergence. We find for the prefactor of the diverging logarithm

\begin{equation}
 \frac{10}{\pi}\sum_{k=0}^j\sum_{l'=0}^{(j-2k+d)/2}\frac{\Gamma(k+\frac{1}{2})\Gamma(j-k+\frac{1}{2})\Gamma(\frac{d+j-2k+1}{2}-l')\Gamma(\frac{d+j-2k}{2}+1)}{\Gamma(k+1)\Gamma(j-k+1)\Gamma(\frac{d+j-2k}{2}-l'+1)\Gamma(\frac{d+j-2k+3}{2})}\, \log(1-\eta)=20\log(1-\eta)\, .
\end{equation}
Finally, we have found the divergence that cancels with the logarithmic counterterm (see eq.\eqref{eq:model NLO Lphi3 CT log}). Since there is no further counterterm, and as we now have brought all expressions into normal order without any remaining divergences, we have verified that the renormalization procedure also works in the case at hand. 

\subsubsection*{Result}

Summing up all the results of the preceding discussion, we can now give the remainder for $q=0$ as

\begin{equation}
\begin{split}
 &(R_1)_{\varphi^3}(x;d,j,q=0)=20 \log r\left[\sum_{l,l'=0}^{l+l'\leq d}\sum_{J_1} \braket{jl\, 00}{J_1\, 0}^2 \begin{pmatrix}
                        J_1 & l' & d-l-l'\\ 0 & 0 & 0
                       \end{pmatrix}^2-1\right]\\
+&20\sum_{\atop{J_1,J_2}{M_1,M_2}}\left(\sum_{\atop{l,l'}{l+l'\leq d}}+2\sum_{l=0}^d\sum_{l'=d+1-l}^{\frac{J_1+d-l}{2}}\right)\frac{\braket{jl\, 00}{J_1\, 0}^2\braket{J_1 l'\, 00}{J_2\, 0}^2}{(l+l'-d-1)(l+l'-d)-J_2(J_2+1)}\Big|_{\text{denom.}\neq 0}\\
 &-\frac{10}{\pi}\sum_{k=0}^j\sum_{l'=0}^{(j-2k+d)/2}\frac{\Gamma(k+\frac{1}{2})\Gamma(j-k+\frac{1}{2})\Gamma(\frac{d+j-2k+1}{2}-l')\Gamma(\frac{d+j-2k}{2}+1)}{\Gamma(k+1)\Gamma(j-k+1)\Gamma(\frac{d+j-2k}{2}-l'+1)\Gamma(\frac{d+j-2k+3}{2})}\\
 &\times L_5\left[\atop{1,d+j-2k+\frac{5}{2},d-k+\frac{3}{2},d+j-k+2,\frac{d+j-2k-l'}{2}+1,\frac{d+j-2k+l'+3}{2}}{d+j-2k+\frac{3}{2},d-k+2,d+j-k+\frac{5}{2},\frac{d+j-2k-l'+3}{2},\frac{d+j-2k+l'}{2}+2}\right]\, ,
\end{split}
\label{eq:model NLO Lphi3 remainder q=0 b}
\end{equation}
where $L_5$ is the non-divergent part of the hypergeometric series $_6F_5$ in eq.\eqref{eq:model NLO Lphi3 3sums result}, see also eq.\eqref{eq:app hypergeos L(p)}. The first line of this result follows from eq.\eqref{eq:model NLO Lphi3 CT formula logs result}. Note however that there is an additional $-1$ in square brackets, which accounts for the finite contribution to the logarithmic counterterm \eqref{eq:model NLO Lphi3 CT log}. The second line contains two contributions: The first is due to the fact that in eq.\eqref{eq:model NLO Lphi3 3rd sum} we have only considered the partial sum with $l+l'\geq d+1$, so we have to recover the remaining partial sum. The other contribution comes from the finite sums in eq.\eqref{eq:model NLO Lphi3 3sums summed up}. Finally, the last two lines are the finite remainder of eq.\eqref{eq:model NLO Lphi3 3sums result} after the subtraction of the counterterm.

Fortunately, with this result we can find an expression for $(R_1)_{\varphi^3}(x;q=3)$ without any additional computational effort. For the parameters $d$ and $j$ the condition $q=3$ essentially means that now $d\leq -3$ and $|d+3|\geq j$. We have already mentioned that the partial sums including the logarithms in eq.\eqref{eq:model NLO Lphi3 CT formula} change the sign if we replace $d\to -d-3$. It is easy to see that the remaining terms in that equation are invariant under this transformation. Thus, we simply have to change the sign in the first line of eq.\eqref{eq:model NLO Lphi3 remainder q=0}, which leads to the result

\begin{equation}
\begin{split}
 &(R_1)_{\varphi^3}(x;-d-3,j,q=3)=-20 \log r\left[\sum_{l,l'=0}^{l+l'\leq d}\sum_{J_1} \braket{jl\, 00}{J_1\, 0}^2 \begin{pmatrix}
                        J_1 & l' & d-l-l'\\ 0 & 0 & 0
                       \end{pmatrix}^2+1\right]\\
+&20\sum_{\atop{J_1,J_2}{M_1,M_2}}\left(\sum_{\atop{l,l'}{l+l'\leq d}}+2\sum_{l=0}^d\sum_{l'=d+1-l}^{\frac{J_1+d-l}{2}}\right)\frac{\braket{jl\, 00}{J_1\, 0}^2\braket{J_1 l'\, 00}{J_2\, 0}^2}{(l+l'-d-1)(l+l'-d)-J_2(J_2+1)}\Big|_{\text{denom.}\neq 0}\\ &-\frac{10}{\pi}\sum_{k=0}^j\sum_{l'=0}^{(j-2k+d)/2}\frac{\Gamma(k+\frac{1}{2})\Gamma(j-k+\frac{1}{2})\Gamma(\frac{d+j-2k+1}{2}-l')\Gamma(\frac{d+j-2k}{2}+1)}{\Gamma(k+1)\Gamma(j-k+1)\Gamma(\frac{d+j-2k}{2}-l'+1)\Gamma(\frac{d+j-2k+3}{2})}\\
 &\times L_5\left[\atop{1,d+j-2k+\frac{5}{2},d-k+\frac{3}{2},d+j-k+2,\frac{d+j-2k-l'}{2}+1,\frac{d+j-2k+l'+3}{2}}{d+j-2k+\frac{3}{2},d-k+2,d+j-k+\frac{5}{2},\frac{d+j-2k-l'+3}{2},\frac{d+j-2k+l'}{2}+2}\right]\, ,
\end{split}
\label{eq:model NLO Lphi3 remainder q=5 b}
\end{equation}
and finishes the proof.\par\hfill\qedsymbol

\end{appendices}

\bibliographystyle{utphys}
\cleardoublepage
\addcontentsline{toc}{chapter}{Bibliography}
\bibliography{Diplom}

\printindex{symbols}{Symbols} \printindex{thms}{Axioms, Definitions, Theorems, etc}

\end{document}